\documentclass{article}
\usepackage{graphicx}
\usepackage{amsmath}
\usepackage{amssymb}
\usepackage{amsfonts}
\usepackage{amsthm}
\usepackage{bbm}
\usepackage{color}
\usepackage{cases}
\usepackage{setspace}
\usepackage{fullpage}
\usepackage[shortlabels]{enumitem}
\newtheorem{theorem}{Theorem}

\newtheorem{proposition}{Proposition}
\newtheorem{corollary}{Corollary}
\newtheorem*{corollary*}{Corollary}
\newtheorem*{theorem*}{Theorem}

\def\theequation{{\arabic{section}}.{\arabic{equation}}}

\def\cD{\mathcal{D}}

\usepackage{hyperref}

\usepackage{subcaption}
\usepackage{algorithm}
\usepackage{algorithmic}

\newcommand{\nc}{\newcommand}
\nc{\dsp}{\displaystyle}
\nc{\R}{\Bbb{R}}
\nc{\Z}{\Bbb{Z}}
\nc{\Pp}{\Bbb{P}}
\nc{\Ap}{\Bbb{A}}
\nc{\Wp}{\Bbb{W}}
 \nc{\brho}{\boldsymbol \rho}
\nc{\va}{\vec{\boldsymbol \mu}}
\nc{\ve}{\vec{\boldsymbol \epsilon}}
\nc{\bk}{{\bf k}}
\nc{\vrho}{\vec{\brho}} 
\nc{\vr}{\vec{\bf r}}
\nc{\bx}{{\bf x}}
\nc{\vx}{\vec{\bf x}}
\nc{\om}{\omega}
\nc{\brhoi}{\brho^{\cal I}}
\nc{\bzetai}{\bzeta^{\cal I}}
\nc{\vzeta}{\vec{\bzeta}}
\nc{\vzetai}{\vec{\bzeta}^{\cal I}}
\nc{\cI}{{\cal I}}
\nc{\cJ}{{\cal J}}
\nc{\cF}{{\cal F}}
\nc{\cW}{{\cal W}}
\nc{\cA}{{\cal A}}
\nc{\cL}{{\cal L}}
\nc{\cS}{{\cal S}}
\nc{\cC}{{\cal C}}
\nc{\cN}{{\cal N}}
\nc{\vrhoi}{\vec{\brho}^{\,\cal I}}
\nc{\xii}{\xi^{\cI}}
\nc{\etai}{\eta^{\cI}}
\nc{\la}{\lambda}
\nc{\de}{\delta}
\nc{\ep}{\varepsilon}
\nc{\vu}{\vec{\bf u}}
\nc{\bu}{{\bf u}}
\nc{\vui}{\vec{\bf u}^{\cI}}
\nc{\bui}{{\bf u}^{\cI}}
\nc{\bt}{{\bf t}}
\nc{\vt}{\vec{\bt}}
\nc{\bn}{{\bf n}}
\nc{\vn}{\vec{\bn}}
\nc{\bm}{{\bf m}}
\nc{\vm}{\vec{\bm}}
\nc{\vrp}{\vec{{\bf r}'}}
\nc{\vrc}{\vec{{\bf r}^c_p}}
\nc{\ts}{\tilde s}
\nc{\os}{\overline s}
\nc{\tom}{\tilde \om}
\nc{\tO}{\tilde \Omega}
\nc{\tS}{\tilde S}
\nc{\oS}{\overline S}
\nc{\vrhos}{\vrho_{\star}}
\nc{\vrhosi}{\vrho_{\star}^{\cI}}
\nc{\brhos}{\brho_{\star}}
\nc{\brhosi}{\brhos^{\cI}}
\nc{\vms}{\vm_{\star}}
\nc{\vmi}{\vm_{_{\cI}}}
\nc{\vM}{\vec{\bf M}}
\nc{\vMi}{\vM_{_{\cI}}}
\nc{\Ppi}{\Pp_{\cI}}
\nc{\vxi}{\vec{\bxi}}
\nc{\bK}{{\bf K}}
\nc{\bmi}{\bm_{_{\cI}}}
\nc{\Pppi}{\Ppi^p}
\nc{\be}{{\bf e}}
\nc{\bep}{{\bf e}^p}
\renewcommand{\hat}{\widehat}

\DeclareMathOperator{\Tr}{Tr}
\begin{document}
	\title{Synthetic aperture imaging and motion estimation using tensor methods}\author{Matan Leibovich\footnotemark[1], George Papanicolaou\footnotemark[2], and Chrysoula Tsogka\footnotemark[3]
	}\renewcommand{\thefootnote}{\fnsymbol{footnote}}
	\footnotetext[2]{Institute for Computational and Mathematical Engineering, Stanford University, Stanford, CA 94305. \\\   (matanle@stanford.edu)} 
	\footnotetext[2]{Department of
		Mathematics, Stanford University, Stanford, CA 94305.
		(papanicolalou@stanford.edu)}\footnotetext[3]{Department of Applied Mathematics,
  University of California, Merced, 5200 North Lake Road, Merced, CA
  95343 (ctsogka@ucmerced.edu)}\date{}
	\maketitle  \date{}
	\maketitle 
	
	
\begin{abstract}
		We consider a synthetic aperture imaging  configuration, such as synthetic aperture radar (SAR), where we want to first separate reflections from moving targets from those coming from a stationary background, and then to image separately the moving and the stationary reflectors. For this purpose, we introduce a representation of the data as a third order tensor formed from data coming from partially overlapping sub-apertures. We then apply a tensor robust principal component analysis (TRPCA) to the tensor data which separates them into the parts coming from the stationary and moving reflectors. Images are formed with the separated data sets. Our analysis shows a distinctly improved performance of TRPCA, compared to the usual matrix case. In particular, the tensor decomposition can identify motion features that are undetectable when using the conventional motion estimation methods, including matrix RPCA. We illustrate the performance of the method with numerical simulations in the X-band radar regime.
\end{abstract}
\section{Introduction}
	The problem of separating the echoes of moving targets from those of a stationary background in synthetic aperture radar (SAR) imaging is important because different imaging methods need to be employed in each case.  We show that the separation is improved considerably when the data is recast in tensor form, and the convex optimization problem of tensor robust principal component analysis (TRPCA) is used. TRPCA requires the generalization of matrix norms to tensor form, the challenge being to use a suitable tensor nuclear norm. We use a Fourier based tensor nuclear norm which is well suited for capturing the motion over multiple scales of variation. We first review briefly SAR imaging,  then the application of RPCA for motion detection, and follow with a statement of the main results of this paper. 
	
	\subsection{The synthetic aperture radar imaging problem}
	Synthetic aperture radar is used extensively in satellite and airborne imaging for many different applications \cite{moreira2013tutorial}. The main idea behind SAR is to combine coherently the information obtained with a single transmitter-receiver that is probing the medium from multiple locations. Thus, one can form a synthetic aperture and achieve high resolution images of reflectivity, even though the single receiver is incapable of resolving the scattered wavefronts. This is particularly relevant for airborne radar platforms.
	
	Data is collected by a moving platform, with a {\em slow-time} $s$ dependent position $\vr(s)$, emitting a sequence of {\em fast-time} broadband pulses $f(t)$, and recording the echoes corresponding to each pulse. The pulses have a limited time support, with a pulse repetition interval $\Delta s$, so that echoes from different pulses do not overlap.  The data collected are denoted by $D(s,t)$, the $t$- dependent series of echoes received from a pulse transmitted at $\vr(s)$. A schematic of a SAR imaging configuration is shown in Figure~\ref{fig:alpha}.
	
\begin{figure}[htbp]
	\centering
	\hspace{0em}\includegraphics[width=0.5\columnwidth]{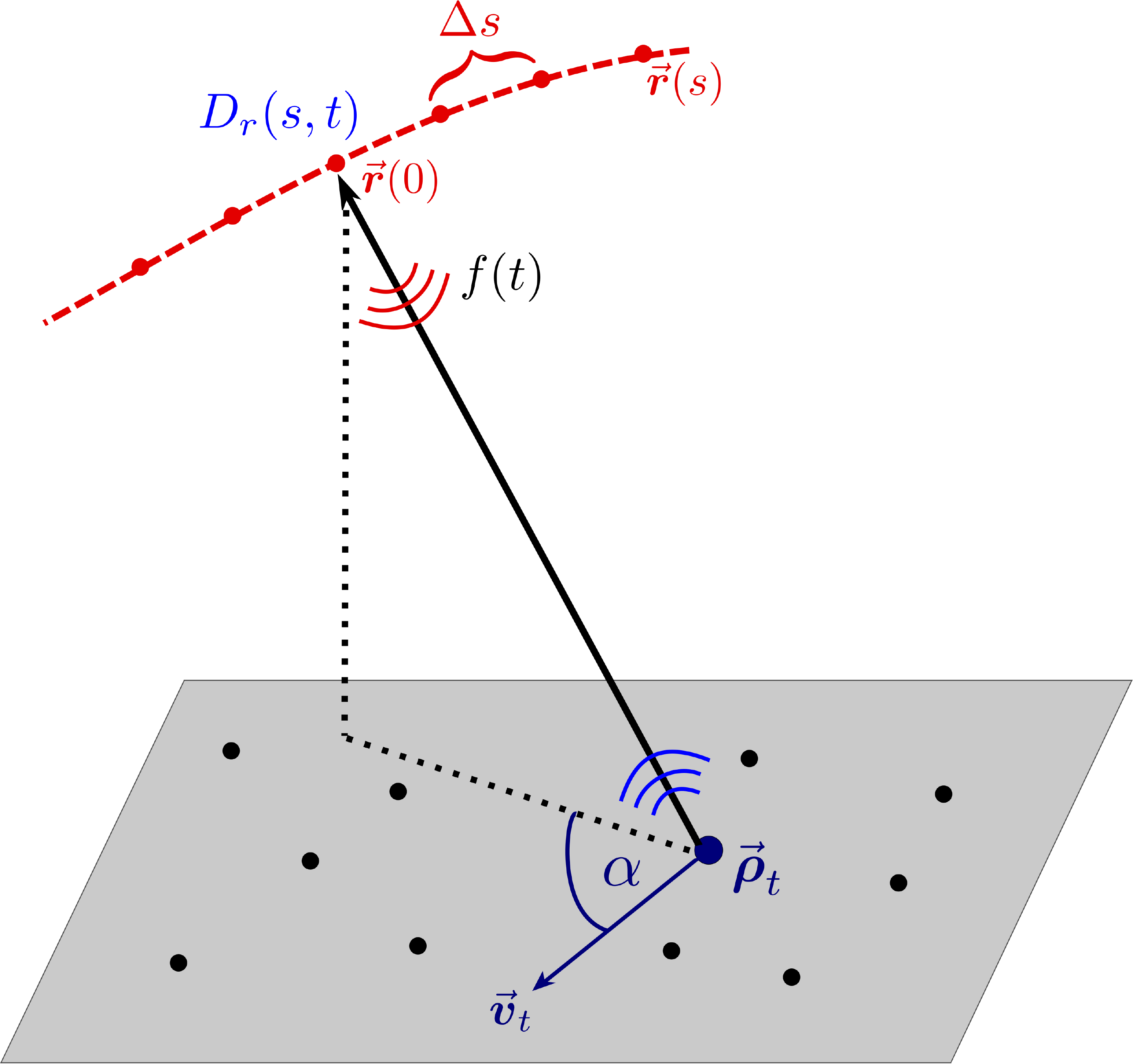}
	\caption{Synthetic aperture imaging configuration. At every time $s$ the airborne platform emits a pulse $f(t)$ and records the reflections $D(t,s)$ coming from the imaging region. There are 10 stationary point scatterers, and a single  moving target. The angle $\alpha$ in which the target is moving is relative to the vertical platform-target  plane at $s=0$, $\vr(0)$.}  	
	\label{fig:alpha}
\end{figure}
To maximize the power emitted, the probing pulses are long, of support $t_{c} \gg 1/B$ where $B$
denotes the bandwidth. They are linear frequency modulated chirps. To 
re-concentrate the energy of the reflected echoes to an interval of size $1/B$ 
they are convolved with the complex conjugate of the time-reversed 
emitted pulse. This is the pulse compression step.
Since the reflections relevant for imaging cover a limited area of support $\ll\Delta s$, a range compression is done as well, that is, 
we remove from the data the large phase $\omega \tau(s,\vrho_o)$ where $\vrho_o$ is a reference point.  
Pulse and range compression together give the down-ramped data, 
\begin{equation}
D_r(s,t) = \int \mathrm{d}t'
D\left(s,t-t'+\tau(s,\vrho_o)\right)
\overline{f(-t')} = \int \frac{\mathrm{d}\omega}{2 \pi} \overline{
	\widehat{f}(\omega)} \widehat{D}(s,\omega) e^{-\mathrm{i} \omega
	[t+\tau(s,\vrho_o)]}.
\label{eq:down-ramping}
\end{equation}
Here, $\tau(s,\vrho_o)$ is the round-trip travel time between the platform location at slow time $s$, $\vr(s)$ and
the reference point location $\vrho_o$,   
\begin{equation}
\dsp \tau(s,\vrho_o) = 2 \frac{\| \vr(s) - \vrho_o \|}{c},
\label{eq:tau}
\end{equation}
with $c$ the speed of light. 

The SAR data matrix $\cD \in {\mathbb{R}}^{(n+1)\times(m+1)}$ is actually obtained in discrete samples of $D_r(s,t)$. 
\begin{equation}
\cD_{il} =D_r\left(s_{i-\frac{n}{2}-1},
t_{l-1}\right), \quad i = 1, \ldots, n+1, ~ ~ l = 1, 
\ldots, m + 1,
\label{eq:MOD6}
\end{equation}
with slow times $s_j$ defined by
\begin{equation}
s_j = j \Delta s,\quad j= -n/2,\ldots, n/2 
\end{equation}
and fast times $t_l$ defined as
\begin{equation}
t_l = l \Delta t,\quad l=1, \ldots, m. 
\end{equation}
Here we assumed that the pulse repetition rate $\Delta s$ is an integer multiple of $\Delta t$ and set $m= \Delta s/ \Delta t$.
The SAR image is formed by summing coherently the down-ramped data $D_r(s_j,t)$ 
back-propagated to the imaging point $\vrho$ using the travel times differences 
$\tau(s_j,\vrho)-\tau(s_j,\vrho_o)$,
\begin{equation}
\dsp  I^{\text SAR}(\vrho) = \sum_{j=-n/2}^{n/2} 
D_r(s_j,\tau(s_j,\vrho)-\tau(s_j,\vrho_o)) .
\label{eq:SARfunctional}
\end{equation}

The SAR image processing \eqref{eq:SARfunctional} assumes that only reflections from stationary targets are contained in the down-ramped data $D_r(s,t)$. Consequently, if moving targets, as illustrated in Figure~\ref{fig:alpha},  are present in the region to be imaged their reflections are not correctly back-propagated and this results in blurred images affected by the reflectivity and the velocity of the moving targets. For a complex scene with many stationary and moving targets the image may be severely distorted and neither the stationary nor the moving targets may be imaged or tracked. To address this issue several motion estimation and separation strategies have been developed, which we now review briefly.

The oldest and most widely used approach is the Displaced Phase Center Antenna (DPCA) method in which two synchronized antennas are used, following the same trajectory with a small time delay. By subtracting the data traces collected at the two antennas the echoes due to the stationary background are essentially eliminated.  This approach does not aim at image formation as part of motion detection, but requires the necessary hardware to be in place to record the extra data. We refer to the classical handbook on SAR \cite{stimson} and a review report from the Lincoln Laboratory on this technique \cite{muehe2000displaced}, which has been used to improve the performance of moving-target-indicators radars since the 1950s. 

Other well known approaches are autofocus based algorithms \cite{fienup, barbarossa1998ambiguityautofocus,li2007} or more generally space-time adaptive processing algorithms \cite{ender1996, wang2006}.  For a recent review on sparsity driven techniques  for SAR imaging of scenes containing moving objects we refer to \cite{Willsky2014}. All these algorithms for moving target detection rely on forming a preliminary image first, and then detecting motion by sharpening features in the image. 

In this paper we take a different approach and rely on robust principal component analysis to solve the SAR data separation problem. Our approach exploits properties of the raw-data matrix so as to detect motion. It does not require the formation of the SAR image or the use of any special hardware. Only the standard monostatic single transmitter/receiver SAR data are used.  
Robust principal component analysis is traditionally used to separate signals from noise. Here we use to it to decompose a matrix into its low rank and sparse parts. The sparse part is not noise as in traditional RPCA \cite{candesRPCA, Zhou2010StablePC} but it is the signal corresponding to the moving targets echoes. The idea of using RPCA for SAR data separation was first proposed in \cite{borcea2013synthetic} and further developed and analyzed in \cite{leib2018RPCA} where optimal parameters were derived for achieving robust separation in SAR. We explain next how RPCA can be used in SAR.
	
\subsection{Robust principal component analysis for SAR data separation}

RPCA, or low rank plus sparse decomposition, was originally applied in video processing \cite{candesRPCA}. RPCA uses the fact that the moving targets and stationary background would generate data structures with different spectral properties. Indeed, the background data form a low rank matrix, while the moving objects echoes correspond to a sparse matrix. These {\em data structures} can be decomposed by solving the convex  optimization problem 	
\begin{equation}
\begin{split}
& \min_{L,S\in\mathbb C^{n_1\times n_2}}
\quad ||L||_*+ \eta ||S||_1\\
& \text{subject to} \quad L+S=D. 
\end{split}
\label{pca2}
\end{equation}
Here $||L||_*$ denotes the nuclear norm, that is the sum of the singular values of $L$, and $||S||_1$ is the matrix $\ell_1$-norm of $S$.  Assuming the matrix $D \in \mathbb{R}^{n_1 \times n_2}$ is the sum of a low rank matrix, $L_o$, and a sparse matrix, $S_o$, then, under some additional  conditions (cf. \cite{candesRPCA}), this optimization \eqref{pca2} recovers $L_o$ and $S_o$ exactly. 
	
RPCA has found applications in a variety of problems in imaging and image processing, such as denoising, feature extraction, and data recovery \cite{guo2016video,cetin2019joint,waters2011sparcs}. The main idea in the application of RPCA to the SAR problem is that one can identify the stationary background as the low rank component of the SAR data matrix, and the moving targets as the sparse component \cite{borcea2013motion}. 	
\begin{figure}[htbp!]
\begin{subfigure}[t]{0.325\textwidth}
	\includegraphics[width=1\columnwidth]{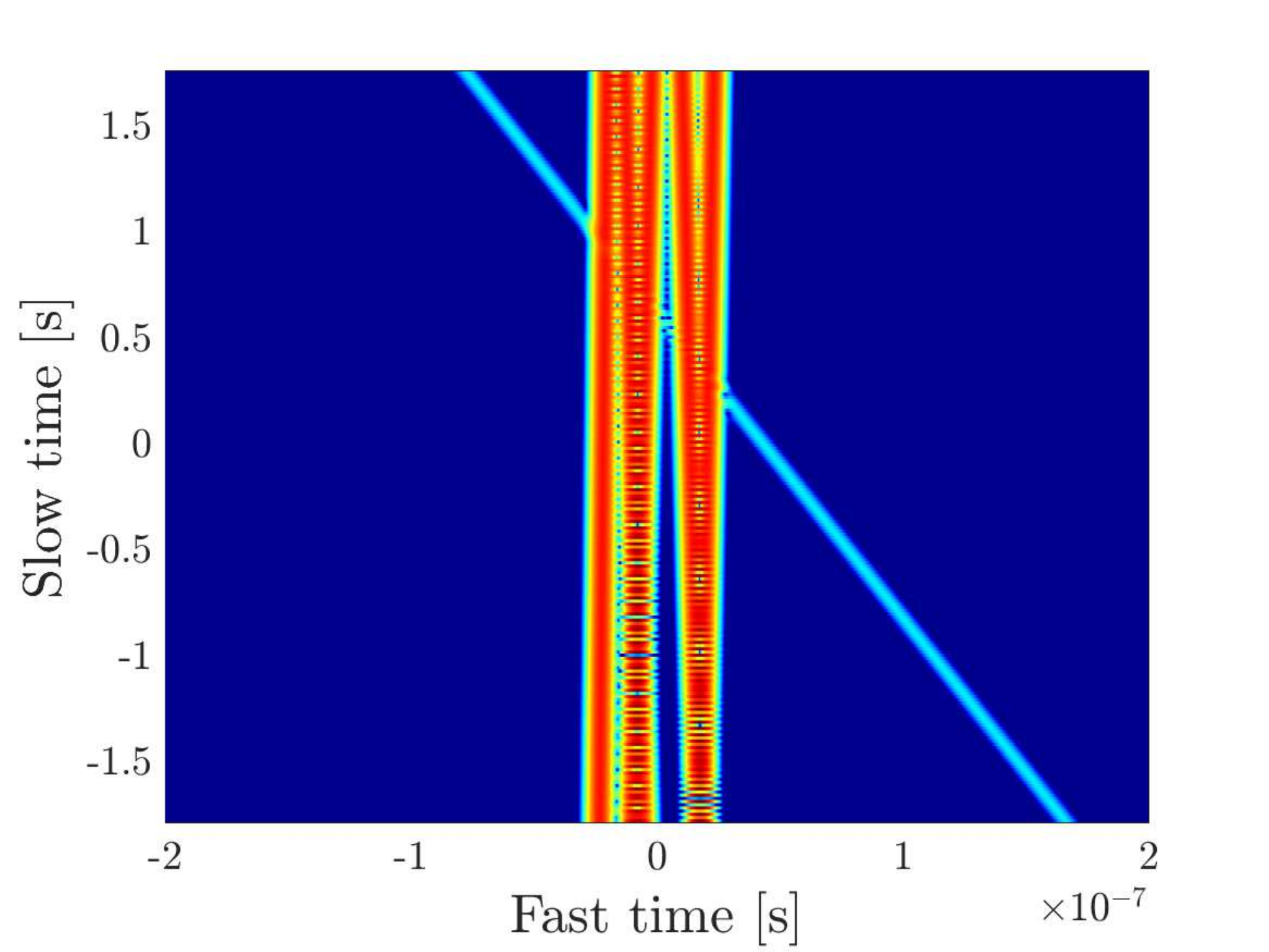}
	\caption{} 
\end{subfigure}
\begin{subfigure}[t]{0.325\textwidth}
	\includegraphics[width=1\columnwidth]{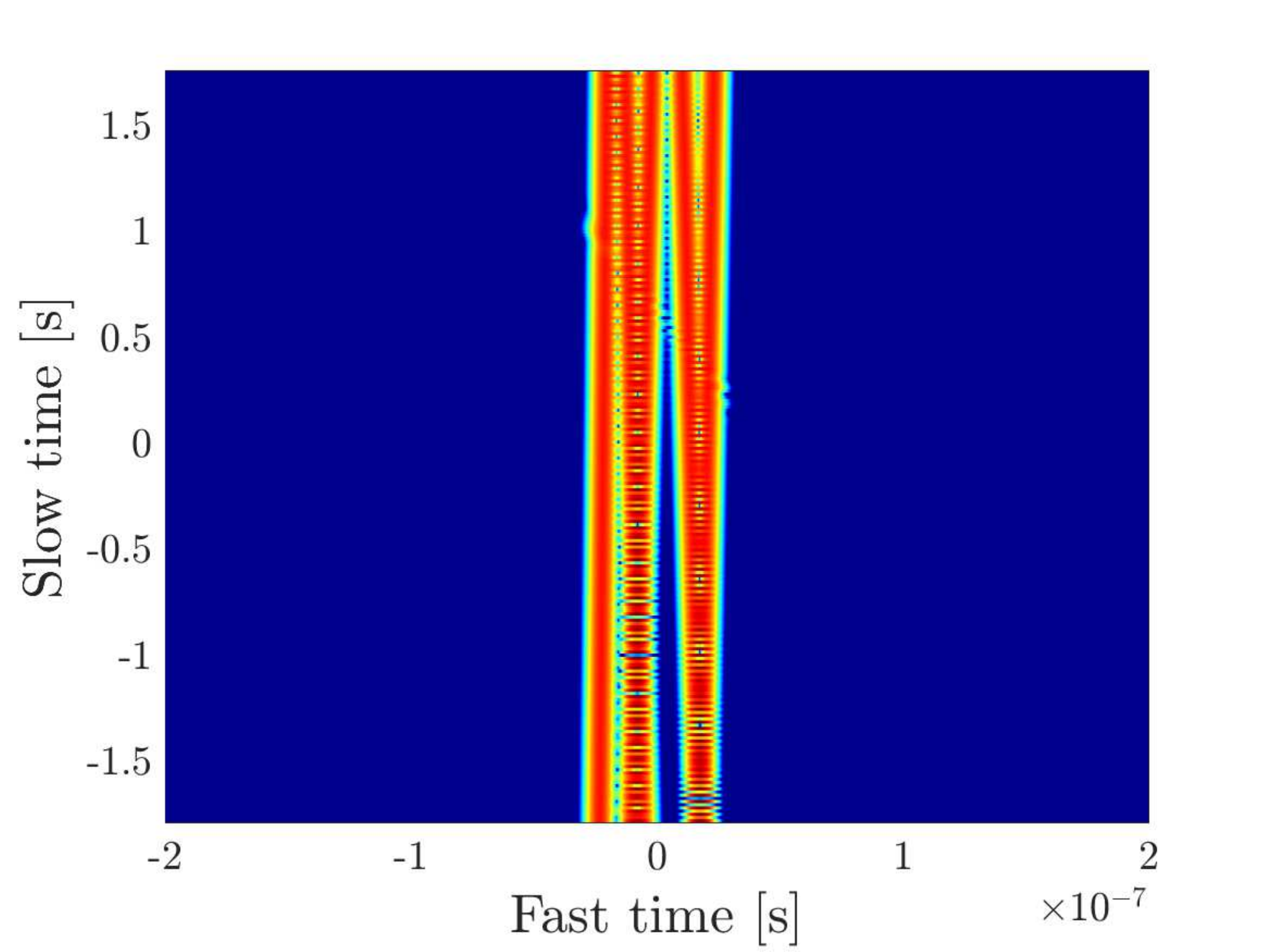}
	\caption{} 
\end{subfigure}
\begin{subfigure}[t]{0.325\textwidth}
	\includegraphics[width=1\columnwidth]{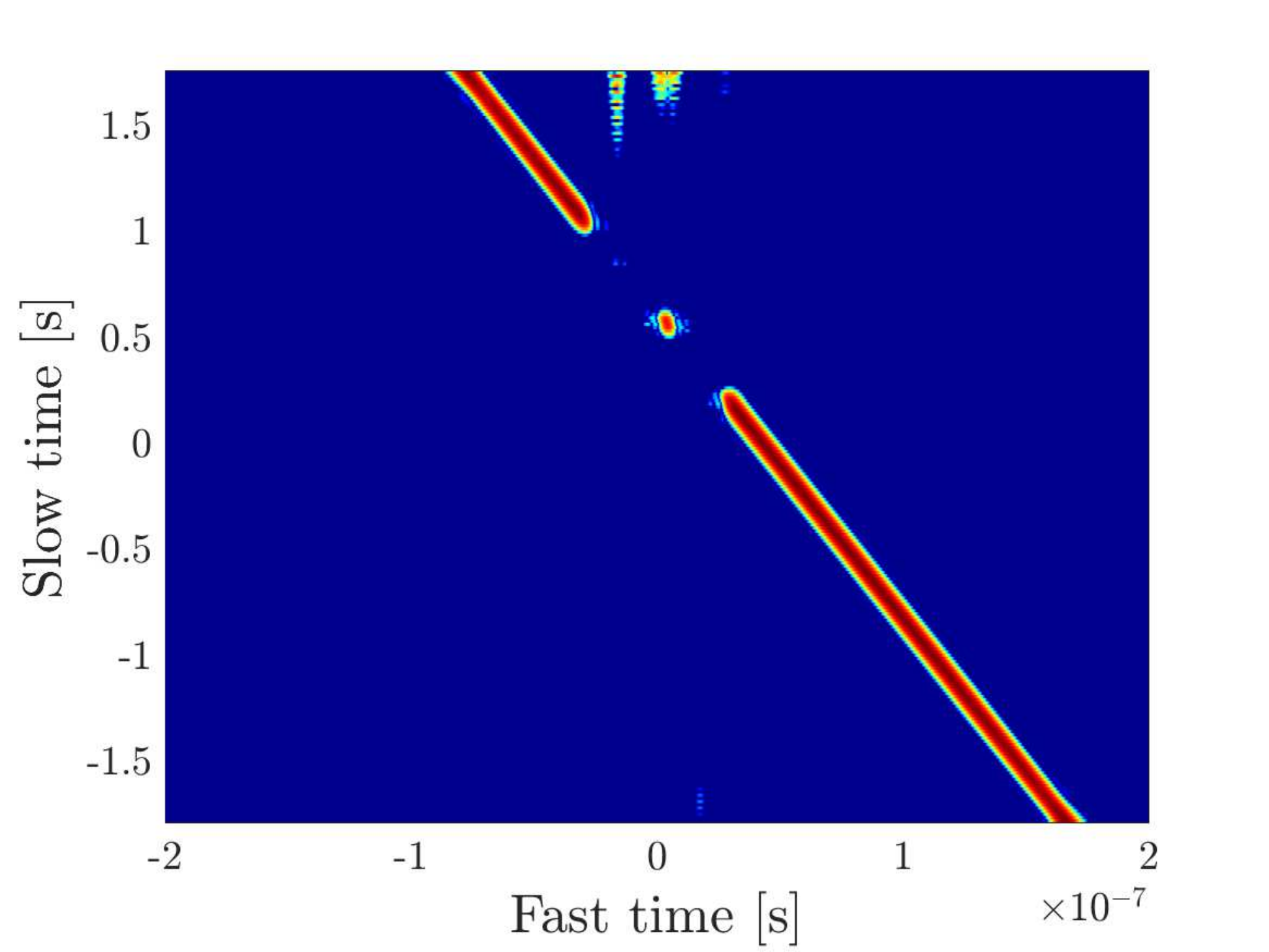}
	\caption{} 
\end{subfigure}
\caption{Example of RPCA applied to the data matrix $D_r(s,t)$. The moving target has velocity $\pmb{v}_t=15$m/s, $\alpha = 0$ and there are 10 stationary targets.
The moving target's reflectivity is 10\% of the reflectivity of the stationary scatterers. $(a)$ Original data matrix. The echoes corresponding to the moving target is weak but still visible in the data; $(b)$ RPCA: the low rank part for the stationary background; $(c)$ RPCA: the sparse part for the echoes from the moving target.  We observe that good separation is achieved with RPCA when using the optimal value for $\eta$ in \eqref{pca2}. This is the case for fast moving targets as in this example.}
\label{fig:DLS_opt}
\end{figure}

In \cite{leib2018RPCA}, we explored the performance of RPCA for the SAR data problem. We showed that one can determine optimal parameters for achieving robust separation. An example of SAR data separation achieved using RPCA  is illustrated in Figure \ref{fig:DLS_opt}. Here we use the simulation setup shown in Figure \ref{fig:alpha} with the moving target velocity $\pmb{v}_t=15$m/s and $\alpha = 0$. The moving target's reflectivity is 10\% of the reflectivity of the stationary scatterers. We observe that good separation is achieved with RPCA when using the optimal value for $\eta$ in \eqref{pca2}. This is the case for fast moving targets as in this example. 	
RPCA has proven to be an efficient way to detect and separate moving targets in SAR data. The nuclear norm is a good indicator of motion, since the data traces associated with moving targets are supported over a larger number of columns, compared to stationary ones. However, the algorithm has its limitations, which we discuss next. 
	
The column support of the target's echoes is determined by the possible values that the travel time difference, $\Delta\tau(s_j)=\tau(s_j,\vrho)-\tau(s_j,\vrho_o)$, takes in the aperture. 
In some cases, motion does not necessarily translate to an increased column support. 
More precisely, it was shown in 
\cite{leib2018RPCA} that the performance of RPCA depends on $N(\vec{\pmb{v}}_t)$, which is an estimate of the number of columns spanned by the target's echoes. 
To first order, we can approximate $N(\vec{\pmb{v}}_t)$ by
\begin{equation}
\begin{split}
N(\vec{\pmb{v}}_t)&\approx\frac{4S(a) }{\Delta t }\frac{\vr(0)-\vrho_o}{\|\vr(0)-\vrho_o\|}\cdot \frac{\vec{\pmb{v}}_t}{c}\\
&=\frac{4S(a) }{\Delta t }\frac{\pmb{v}_t}{c}\cos\alpha,
\end{split}
\label{eq:Nvt_1}
\end{equation}
where $S(a)$ is the total slow time aperture size. 
We observe that $N(\vec{\pmb{v}}_t)$ depends not only on the magnitude of the moving target's velocity, $\pmb{v}_t$, but also on the relative direction in which the target is moving, with respect to the platform, $\alpha$. The direction $\alpha=0$ is for targets moving parallel to the direction to the platform. Their traces will exhibit the largest variance in the value of $\Delta \tau(s)$. The direction $\alpha=\pi/2$ describes targets moving perpendicular to the direction to the platform, and their $\Delta \tau(s)$ will show much lass variance as a function of $s$. 

As can be seen in Figure~\ref{fig:D_alpha_ex}, targets moving at directions with $\alpha\neq0$ that exhibit smaller variations in the values $\Delta \tau(s)$ will be harder to detect using matrix RPCA. 
On the other hand, the associated data traces still behave differently than the ones of a stationary background. Specifically we can see that the traces are non linear.
\begin{figure}[htbp!]
	\centering
	\begin{subfigure}[t]{0.25\textwidth}
		\includegraphics[width=1\columnwidth]{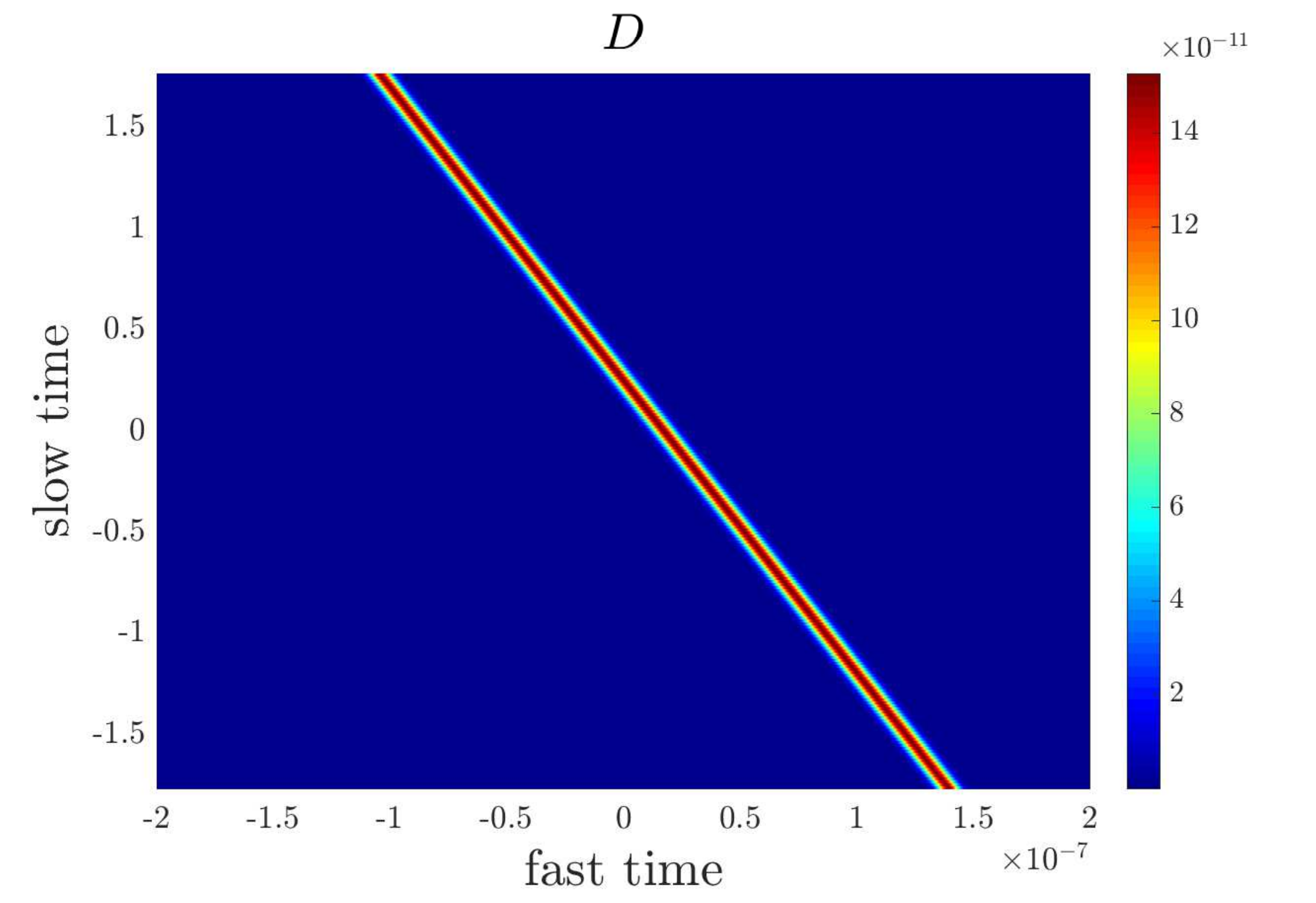}
		\caption{$|D_r(s,t)|$, $\alpha=0$}
		\label{fig:D_ex_alpha_0}
	\end{subfigure}
	\begin{subfigure}[t]{0.25\textwidth}
		\includegraphics[width=1\columnwidth]{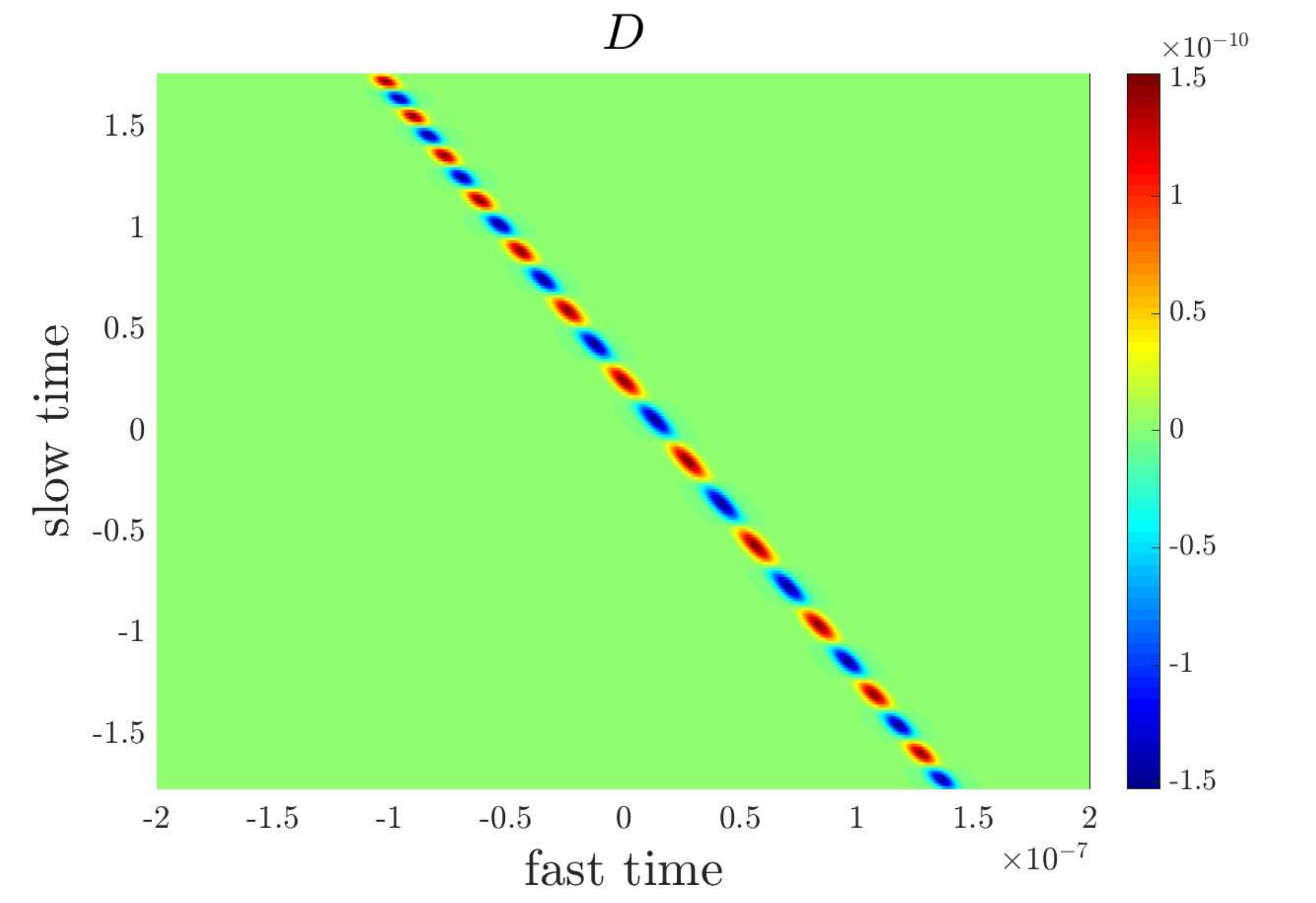}
		\caption{$\Re \{D_r(s,t)\}$, $\alpha=0$}
		\label{fig:D_ex_alpha_0_R}
	\end{subfigure}
	
	\begin{subfigure}[t]{0.25\textwidth}
		\includegraphics[width=1\columnwidth]{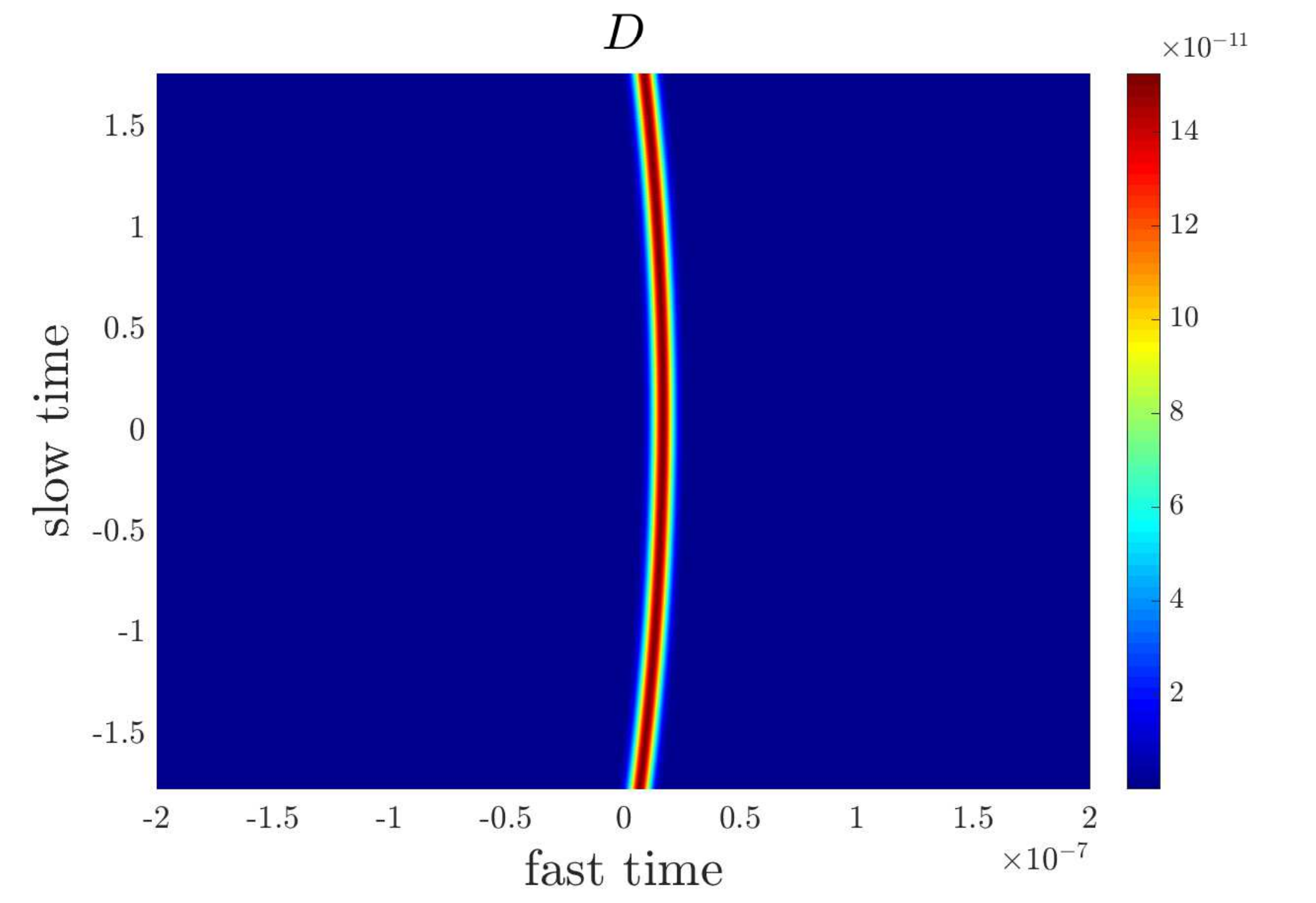}
		\caption{$|D_r(s,t)|$, $\alpha=\pi/2$}
		\label{fig:D_ex_alpha_pi_2}
	\end{subfigure}
	\begin{subfigure}[t]{0.25\textwidth}
		\includegraphics[width=1\columnwidth]{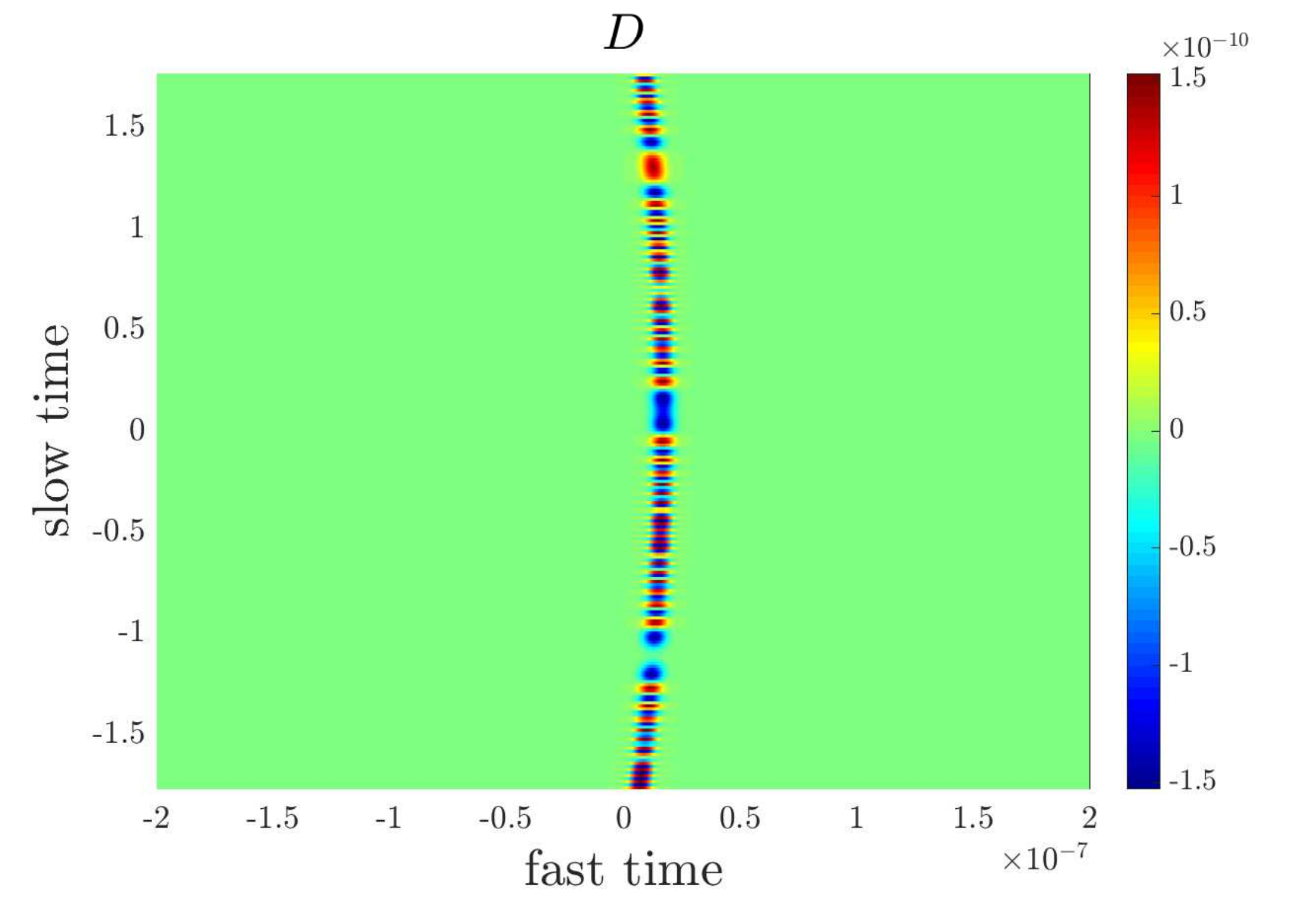}
		\caption{$\Re \{D_r(s,t)\}$, $\alpha=\pi/2$}
		\label{fig:D_ex_alpha_pi_2_R}
	\end{subfigure}
\caption{Effect of the direction of motion for a target moving at $\pmb{v}_t=15m/s$. When $\alpha=0$, top, the column support is large, and the phase ($\Delta\tau(s)$) is nearly linear. With $\alpha=\pi/2$, bottom, the column support is much smaller, but the phase is non linear.}
\label{fig:D_alpha_ex}
\end{figure}
	
\subsection{Main result of the paper: TRPCA for SAR data}
The limitations of RPCA for motion estimation in SAR motivates a new look into possible extensions and modifications of the algorithm. Tensor based methods have been of great interest because high dimensional data arise naturally as tensors in more and more applications \cite{cichocki2015tensor}. For the SAR problem, this is motivated by the possibility of detecting and estimating more complex moving target behavior, such as non-linearity in the phase, by representing the data in a higher dimension. This is achieved here by dividing the large synthetic aperture into smaller, overlapping sub-apertures.

The problem of imaging moving targets in SAR can also be viewed as an image registration process. That is, a process that provides a precise correspondence between two or more images of the same object captured from different locations, at different times, or using different sensors. From this perspective, the introduction of sub-apertures for data processing is a natural one.

In this paper we use the SAR data, and the choice of sub-aperture and overlap size will be motivated differently.
We recast the SAR data matrix as a third order tensor $\mathcal{A}$ by dividing the large synthetic aperture into smaller overlapping sub-apertures, indexed by $\ell$
    \begin{equation}
    \begin{split}
A^{(\ell)}(s,t)=D_r(s\ell\vartheta s_{\text{sub}},t),\quad s\in[0,s_{\text{sub}}],\\ A^{(\ell)}=\mathcal{A}(\cdot,\cdot,\ell)\in \mathbb{R}^{n_1\times n_2}, \quad \ell=0,...,n_3-1.
\end{split}
\label{eq:mat_to_tensor}
\end{equation}
Here $s_{\text{sub}}$ denotes the sub-aperture size, and $\vartheta$ is the overlap size, a number between 0 and 1. An illustration of tensor representation of the SAR data is given in Figure~\ref{fig:ar_tensor}. The value of the hyper-parameters that define the tensor representation, i.e. $s_{\text{sub}}$ and $\vartheta$, affect the performance of TRPCA. A detailed analysis that allows us to determine the optimal value for $s_{\text{sub}}$ and $\vartheta$ is carried out in Section~\ref{sec:3}.
\begin{figure}[htbp!]
		\centering
		\includegraphics[width=0.6\textwidth]{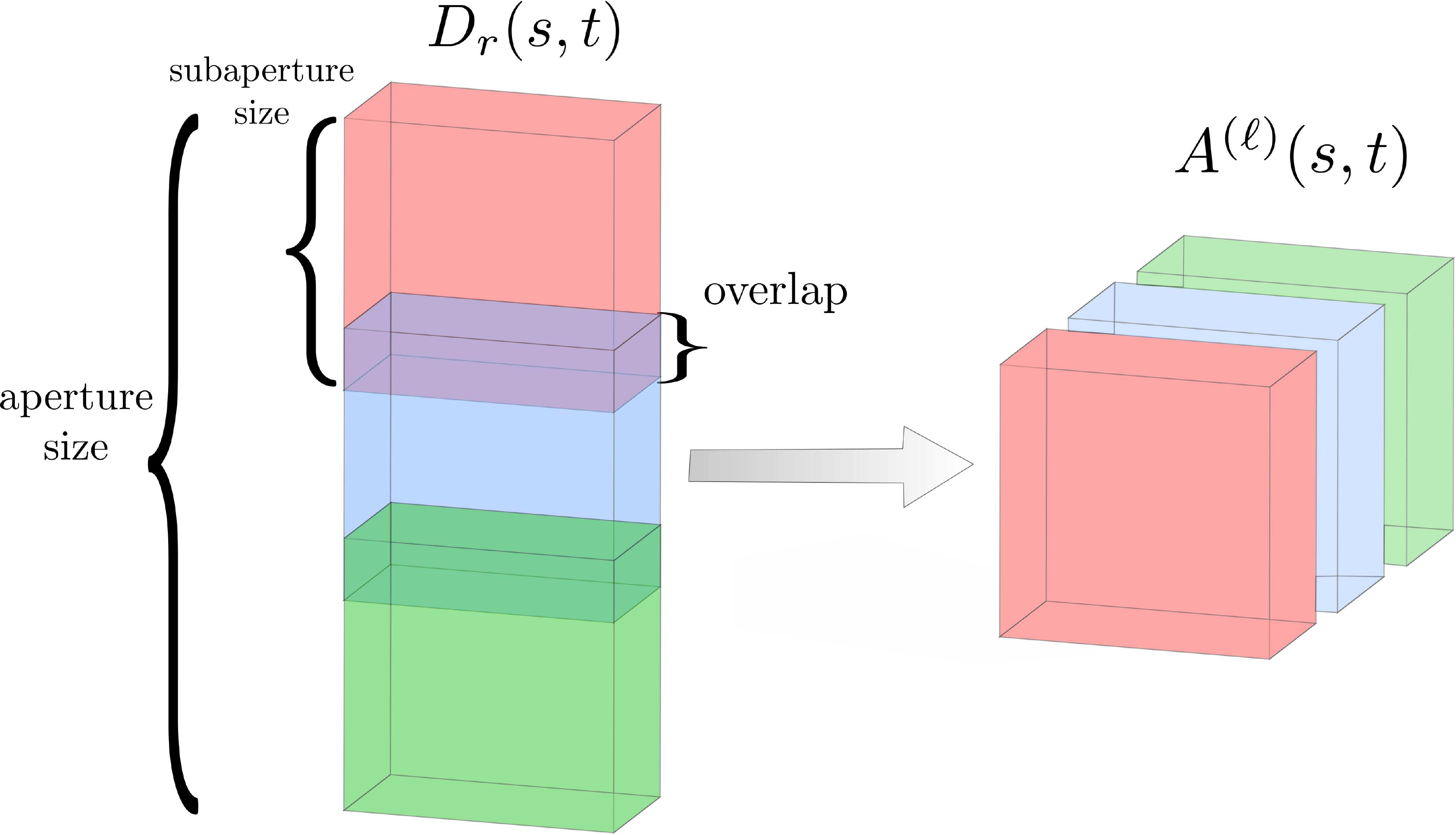}
	\caption{Schematic of SAR data tensor representation. For motion detection purposes, the large aperture data matrix $D_r(s,t)$ is converted to a 3rd order tensor $A^{(\ell)}(s,t)$. The tensor is composed of partially overlapping sub-apertures of the data $D_r(s,t)$.}
\label{fig:ar_tensor}
\end{figure}
We then define the TRPCA algorithm, using a specific extension of the nuclear norm for third-order tensors \cite{lu2016tensor}, and solve a tensor based RPCA optimization problem for complex valued third order tensors:
\begin{equation}
\begin{split}
& \min_{\cL,\cS\in\mathbb C^{n_1\times n_2\times n_3}}
\quad ||\cL||_{*,\mathcal F}+ \eta ||\cS||_1,\\
& \text{subject to} \quad \cL+\cS=\cA
\end{split}
\label{pca_ten}
\end{equation}
$\|\cdot\|_{*,\mathcal F}$, defined in \eqref{eq:tnn}, is a specific extension of the tensor nuclear norm which involves performing a Fourier Transform with respect to the sub-aperture index $\ell$. 

To evaluate the performance of the algorithm, we compare the separation achieved by matrix and tensor based RPCA. While the performance of TRPCA is not universally better than that of regular matrix RPCA, TRPCA performs significantly better in the cases where motion is hardest to detect as illustrated in the example considered in Figure \ref{fig:RPCAperf}. Here the setup shown in Figure \ref{fig:alpha} is considered with a slowly moving target with $\pmb{v}_t=1$m/s  and $\alpha = \pi/2$. The target's reflectivity is 10\% of the reflectivity of the stationary scatterers and its echoes are barely detectable in the original data (see Figure \ref{fig:RPCAperf}-(a)). As we see from the results, i.e., figures  \ref{fig:RPCAperf}-(b) and \ref{fig:RPCAperf}-(c) almost perfect separation is achieved between the stationary and the moving targets echoes. There is some noise at the edges of the slow time window in the sparse component but this does not really affect the imaging results. 	
\begin{figure}[htbp!]
	\begin{subfigure}[t]{0.325\textwidth}
	\includegraphics[width=1\columnwidth]{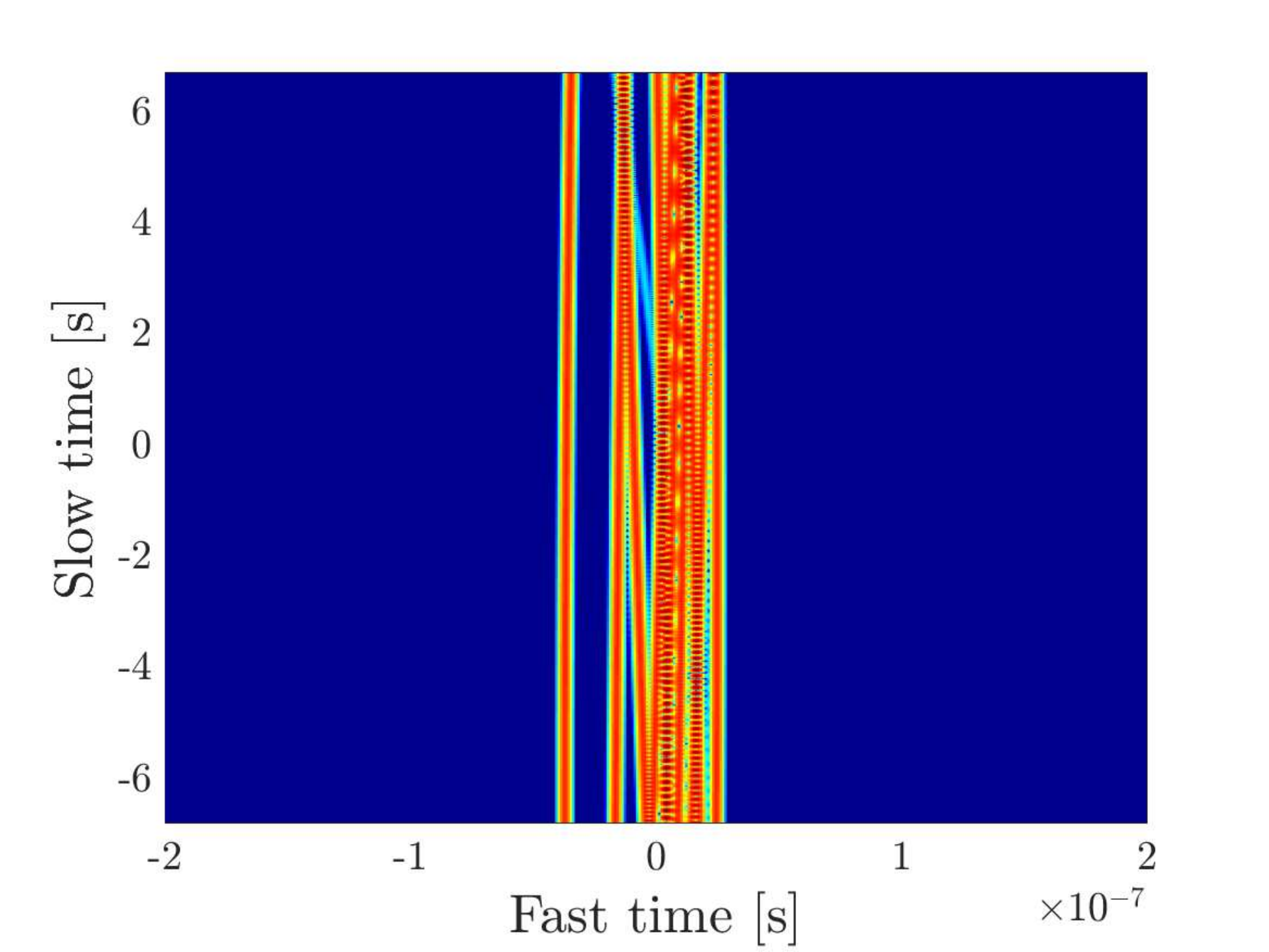}
	\caption{} 
\end{subfigure}
\begin{subfigure}[t]{0.325\textwidth}
	\includegraphics[width=1\columnwidth]{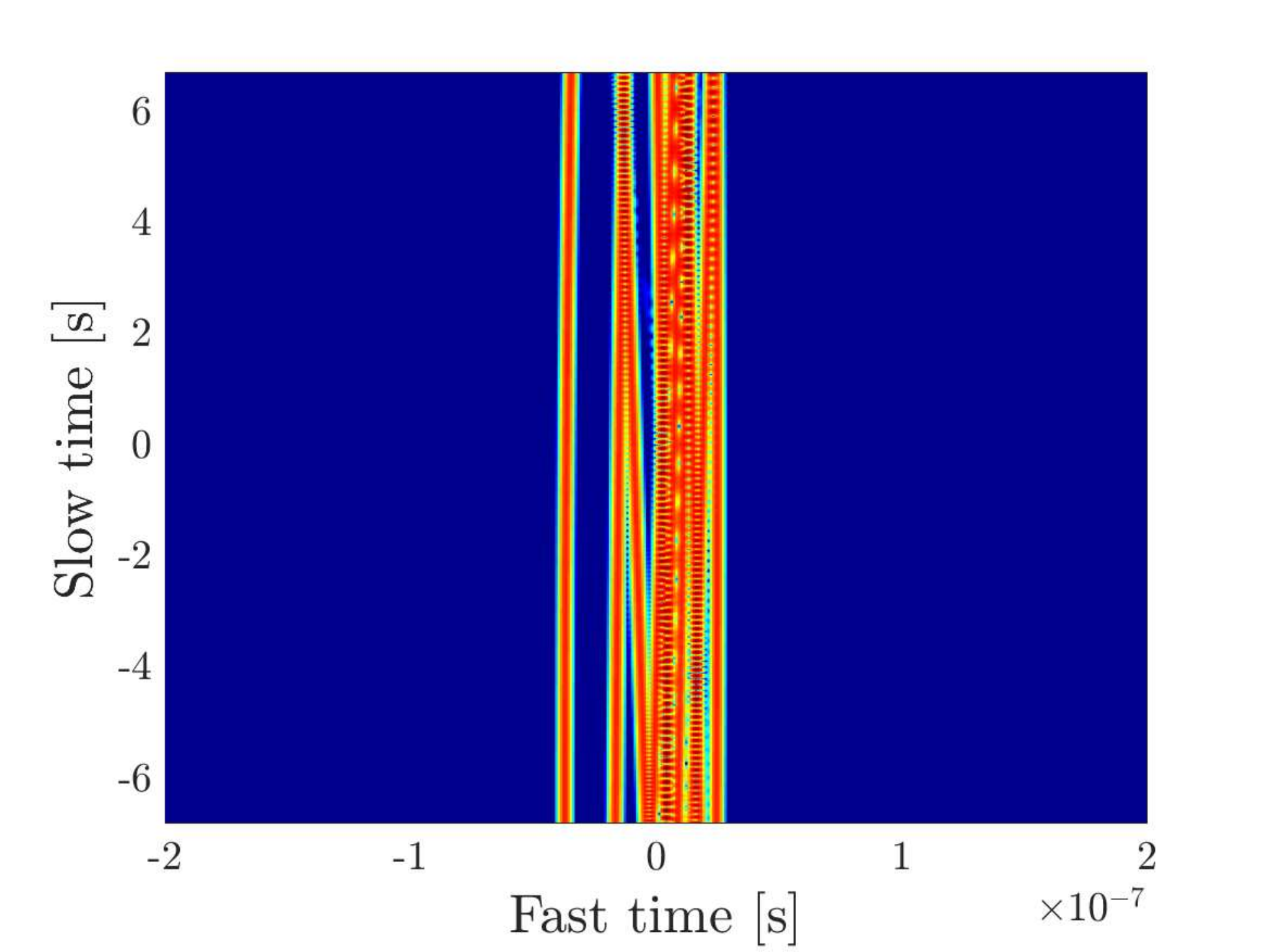}
	\caption{} 
\end{subfigure}
\begin{subfigure}[t]{0.325\textwidth}
	\includegraphics[width=1\columnwidth]{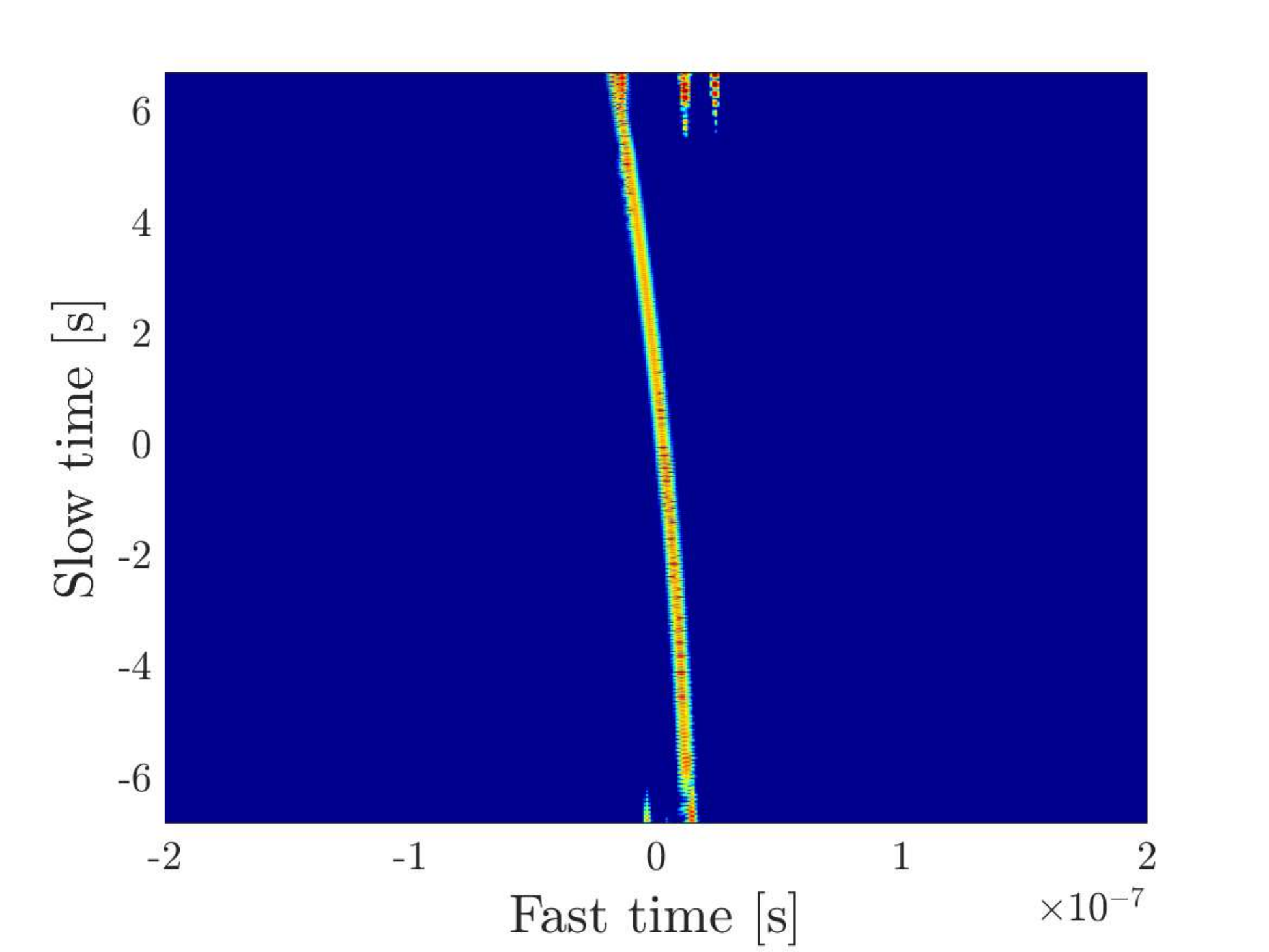}
	\caption{} 
\end{subfigure}
\caption{Example of TRPCA applied to SAR data corresponding to a slow moving target. The simulation setup shown in Figure \ref{fig:alpha} is used with the moving target velocity $\pmb{v}_t=1$m/s  and $\alpha = \pi/2$. The target's reflectivity is 10\% of the reflectivity of the stationary scatterers. $(a)$ Original data matrix. The echoes of the weak and slow moving target are barely discernible in the data; $(b)$ TRPCA result: the low rank part corresponding to the stationary background; $(c)$ RPCA result: the sparse part corresponding to the moving target; We observe that TRPCA achieves good data separation in a very challenging setting.}
\label{fig:RPCAperf}
\end{figure}

These observations are further explained in the paper where (i) we provide analysis of the nuclear norms for SAR data tensors corresponding to stationary and moving targets; (ii) we present lower and upper bounds on the values of the tensor nuclear norm, and show how these bounds can be used to explain the observed phenomena; (iii)~we prove that the bounds are achieved for limiting edge cases, and that the stationary background and the moving target tend towards those cases.

To summarize, the results of this paper demonstrate that SAR data naturally admit a tensor representation, and provide an example to the added benefit of using the tensor decomposition. In an operational system TRPCA could be combined with other methods such as DPCA or matrix RPCA to detect motion in cases where other methods underperform. Further research will go into other possible extensions of RPCA to tensor form, application of tensor representation to other imaging problems such as iSAR and satellite imaging, and optimization of the algorithm's parameters for robust performance.
	
	The rest of the paper is structured as follows: In Section~\ref{sec:2} we present the TRPCA algorithm for the SAR data problem. 
	In Section~\ref{sec:3} we  investigate the effect of the hyper-parameters on the performance of TRPCA, both numerically, and through analysis of idealized cases. We prove bounds on the tensor nuclear norm, and study how these bounds are related to the observed performance. 
	In Section~\ref{sec:4} we present numerical examples of data separation and the corresponding imaging results which illustrate the advantages of tensor over matrix RPCA for SAR problems. 
	We end in Section~\ref{sec:5} with our conclusions. 
\section{TRPCA for SAR data} \label{sec:2}
In this section we describe the TRPCA algorithm for the SAR data problem, starting with discussion of the tensor nuclear norm.
\subsection{TRPCA formulation}
 \label{sec:mat_recast}
Using the representation of the SAR data as the tensor $\mathcal{A}$ defined in \eqref{eq:mat_to_tensor}, 
we want to solve the following tensor RPCA optimization 
\begin{equation}
\min\limits_{\cL+\cS=\mathcal{A}} \|\cL\|_{*}+\eta\|\cS\|_1,
\label{eq:ten_rpca_1}
\end{equation}
where $\mathcal{L}$ and $\mathcal{S}$ are third order tensors, and $\|\cdot\|_{*}$ is some generalization of the nuclear norm to third order tensors. 
In Appendix~\ref{app:tensor_decomp}, we review common tensor decomposition methods, and the extension of the nuclear norm to tensor form as a relaxed rank estimate.
	 
A natural definition of the nuclear norm \cite{friedland2018nuclear}  that extends to higher dimensions is 
\begin{equation}
\|\mathcal{A}\|_{*,\mathcal T}=\inf\left\{\sum\limits_{i=1}^r|\lambda_i| \hspace{0.5em}\Big| \hspace{0.5em}\mathcal{A}=\sum\limits_{i=1}^r \lambda_i u^1_i \otimes u^2_i\otimes\cdots \otimes u_i^d,\|u_i^j\|=1,r\in\mathbb{N}\right\}. 
\label{eq:nuc_norm_matrix_alt_0}
\end{equation}
Note that the $u_i^j$ do not need to be orthogonal.
As outlined in \cite{friedland2018nuclear}, \eqref{eq:nuc_norm_matrix_alt_0} is in general an NP hard problem to compute. 
Therefore we look for alternative, more tractable definitions of the nuclear norm.
	
The simplest alternative would be to take the matrix panels with respect to a specific dimension and compute the matrix SVD on every panel separately so that, for example,
\begin{eqnarray}
&\mathcal{A}\in\mathbb{R}^{n_1\times n_2 \times n_3}, \quad A^{(\ell)}=A(\cdot,\cdot,\ell)=U^{(\ell)}\Sigma^{(\ell)}V^{(\ell)H},\quad \ell=1,\cdots,n_3,\\
&\|\mathcal{A}\|_{*,\mathcal D}=\sum\limits_{\ell=1}^{n_3}\|A^{(\ell)}\|_*=\sum\limits_{\ell=1}^{n_3} \sum\limits_{j=1}^{r_\ell} \sigma^{(\ell)}_j.
\label{eq:nuc_dec}
\end{eqnarray}
This definition retains some of the properties of singular values, for example the Hilbert-Schmidt norm must equal the sum of singular values,
\begin{equation}
\|\mathcal{A}\|_F^2=\sum\limits_{i_1=1}^{n_1}\sum\limits_{i_2=1}^{n_2}\sum\limits_{i_3=1}^{n_3}|a_{i_1,i_2,i_3}|^2=\sum\limits_{i=1}^r \sigma_i^2=\sum\limits_{\ell=1}^{n_3}\sum\limits_{j=1}^{r_\ell}\left[\sigma^{(\ell)}_j\right]^2.
\end{equation}
We later show that \eqref{eq:nuc_dec} serves as an upper bound to \eqref{eq:nuc_norm_matrix_alt_0}. 
	The decomposition in \eqref{eq:nuc_dec} is {\em decoupled} in a sense (hence the superscript $\mathcal D$), since the singular values of different panels are computed independently. Using this definition in the RPCA algorithm, would result in the regular matrix RPCA on every sub-aperture. 
	
	Another possible extension was introduced in \cite{braman2010third,kilmer2011factorization,kilmer2013third}, and used for RPCA in \cite{zhang2014novel,lu2016tensor}. This is based on an extension of matrix multiplication to incorporate a circular convolution with respect to the third dimension, 
\begin{equation}
C=\mathcal{A}*\mathcal{B}, \quad C^{(\ell)}=\frac{1}{n_3}\sum\limits_{p=0}^{n_3-1}A^{(p)}B^{(\ell-p)|n_3}.
\end{equation}
This is equivalent to representing the third order tensor as a block-circulant matrix, with the tensor multiplication homeomorphic to the regular matrix multiplication
\begin{equation}
\mathcal{A}\in\mathbb{R}^{n_1\times n_2 \times n_3}\rightarrow\frac{1}{\sqrt{n_3}} \begin{pmatrix} A^{(0)}&A^{(1)}&\cdots& A^{(n_3-1)}\\A^{(n_3-1)}&A^{(0)}&\cdots&A^{(n_3-2)}\\&&\ddots&\\A^{(1)}&A^{(2)}&\cdots&A^{(0)}\end{pmatrix}\in\mathbb{C}^{n_1n_3\times n_2n_3}.
\label{eq:block_circ}
\end{equation}
Block-circulant matrices can be block-diagonalized by a Discrete Fourier Transform (DFT) matrix, 
\begin{equation}
\frac{1}{\sqrt{n_3}} \begin{pmatrix} A^{(0)}&A^{(1)}&\cdots& A^{(n_3-1)}\\A^{(n_3-1)}&A^{(0)}&\cdots&A^{(n_3-2)}\\&&\ddots&\\A^{(1)}&A^{(2)}&\cdots&A^{(0)}\end{pmatrix}\overset{\mathcal{F}}{\longrightarrow }
\begin{pmatrix} \hat{A}^{(0)}&&&\\&\hat{A}^{(1)}&&\\&&\ddots&\\&&&\hat{A}^{(n_3-1)}\end{pmatrix},
\label{eq:block_circ_2}
\end{equation}
which is equivalent to
\begin{equation}
\begin{split}
&\hat{\mathcal{A}}=\mathcal{F}_3\mathcal{A}, \quad \hat{A}^{(k)}=\frac{1}{\sqrt{n_3}}\sum\limits_{\ell=0}^{n_3-1}\omega_{n_3}^{\ell k}A^{(\ell)},\quad \omega_{n_3}=e^{\frac{i2\pi}{n_3}},\quad\hat{A}^{(k)}\in \mathbb{R}^{n_1\times n_2}, k=0,\dots,n_3-1.
\end{split}
\end{equation}
Thus, another estimate of the nuclear norm is given by
\begin{equation}
\|\mathcal{A}\|_{*,\mathcal F}=\sum\limits_{k=0}^{n_3-1}\|\hat{A}^{(k)}\|_{*}.
\label{eq:tnn}
\end{equation}
This last definition of the tensor nuclear norm given by {\eqref{eq:tnn} proves to be well suited for the SAR motion detection problem.  We show this by a performance analysis, carried out in Section~\ref{sec:3}.
\subsection{TRPCA algorithm}
We can now recast \eqref{eq:ten_rpca_1} as
\begin{equation}
	\min\limits_{\cL+\cS=\mathcal{A}} \hspace{1em} \|\cL\|_{*,\mathcal F}+\eta_\mathcal{F} \|\cS\|_1=	\min\limits_{\cL+\cS=\mathcal{A}}
	\sum\limits_{\ell=0}^{n_3-1}\|\hat{L}^{(\ell)}\|_*+\eta\|S^{(\ell)}\|_1, 
\label{eq:ten_rpca_2}
\end{equation}
where $L^{(\ell)}$ and $S^{(\ell)}$ are the sub-apertures of $\cL$ and $\cS$, respectively.
As for the matrix case \cite{lin2010augmented}, we can solve the constrained optimization problem by an augmented Lagrangian
\begin{equation}
L(\cL,\cS,\mathcal{Y},\mu)=\|\cL\|_{*,\mathcal F}+\eta_\mathcal{F} \|\cS\|_1+\langle \mathcal{Y}, \mathcal{A}-\cL -\cS \rangle + \frac{\mu}{2}\|\mathcal{A}-\cL -\cS\|_F^2,
\end{equation}
where the Hilbert-Schmidt norm and inner product are the natural element-wise product extensions of the matrix case. For more details see Appendix \eqref{eq:ten_inner}.
We can again solve this iteratively, using the Alternating Direction Method of Multipliers (ADMM), noting that the Hilbert-Schmidt penalty term is, by Parseval's theorem, invariant under DFT,
\begin{equation}
\|\mathcal{A}-\cL -\cS\|_F^2=\sum\limits_{\ell=0}^{n_3-1}\|A^{(\ell)}-L^{(\ell)} -S^{(\ell)}\|_F^2=\sum\limits_{p=0}^{n_3-1}\|\hat{A}^{(p)}-\hat{L}^{(p)} -\hat{S}^{(p)}\|_F^2.
\end{equation}
	 Thus, solving by Alternating Direction Method of Multipliers (ADMM), the $\cL_k$ minimization step involves singular value thresholding in the Fourier domain, and the $\cS_{k}$ minimization uses element wise thresholding in the real domain. Thresholding is done via the  operator $\Theta_\lambda$
 \begin{equation}
 \Theta_{\lambda}(a)=e^{i \arg a} \max(|a|-\lambda,0).
 \end{equation}
 
The algorithm is outlined in Algorithm~\ref{alg:inexact_alm}, and is a modification of the Inexact ALMM RPCA algorithm introduced in \cite{lin2010augmented}.
\begin{algorithm}[htbp]
		\caption{$\cD=\cL+\cS$ Inexact ALM method }
		\begin{algorithmic}[1]
			\setstretch{1}
			\STATE{\textbf{Input}: Observation tensor $\cA\in \mathbb{C}^{n_1\times n_2\times n_3}$}
			\STATE{$\mu_0=\max\limits_{\ell}\|A^{(\ell)}\|_2,\hspace{1em}\rho=1.4$}
			\WHILE{not converged}
			\STATE{$\bar{\cA}=\mathcal{F}_3(\cA-\cS_{k}-\mu_k^{-1}\mathcal{Y}_k)$}
			\FOR{$\ell=0,...,n_3-1$}
			\STATE{$\left[\bar{U}^{(\ell)}_{k+1},\hspace{0.2em}\bar{\Sigma}^{(\ell)}_{k+1},\hspace{0.2em}\bar{V}^{(\ell)}_{k+1}\right]=\text{svd}\left(\hat{A}^{(\ell)}\right)$}
			\STATE{$\hat{L}^{(\ell)}_{k+1}=\bar{U}^{(\ell)}_{k+1}\hspace{0.3em}\Theta_{\mu_k^{-1}}\hspace{-0.4em}\left[\bar{\Sigma}^{(\ell)}_{k+1}\right]\hspace{0.3em}\bar{V}^{(\ell)*}_{k+1}$}
			\ENDFOR
			\STATE{${\cL}_{k+1}=\mathcal{F}^{-1}_3(\bar{\cL}_{k+1})$}
			\STATE{$\cS_{k+1}=\Theta_{\eta\mu_k^{-1}}\hspace{-0.4em}\left[\cA-\cL_{k+1}+\mu_k^{-1}\mathcal Y_k\right]$}
			\STATE{$\mathcal{Y}_{k+1}=\mathcal{Y}_k+\mu_k\left(\cA-\cL_{k+1}-\cS_{k+1}\right)$}
			\STATE{ $\mu_{k+1}=\rho \mu_k$}
			\STATE{$k\rightarrow k+1$}
			\ENDWHILE
			\RETURN $\cL_{k+1},\cS_{k+1}$
		\end{algorithmic}
		\label{alg:inexact_alm}
\end{algorithm}
\section{TRPCA performance analysis}  \label{sec:3}

We follow here the same approach as in \cite{leib2018RPCA} in order to analyze the performance of TRPCA. The key idea is that there is a finite range of values for $\eta_\mathcal{F}$ in the objective \eqref{eq:ten_rpca_2} that are admissible $  \eta_{\min,\mathcal{F}} \le \eta_{\mathcal{F}} \le \eta_{\max,\mathcal{F}}$. 
Indeed, if $\eta_\mathcal{F}$ is too large then the nuclear norm term might be small even for moving targets, i.e. $\|\mathcal{S}\|_*\le \eta \|\cS\|_1$. 	
	If, on the ohter hand, $\eta_\mathcal{F}$ is too small the $\ell_1$ term might be small even for the stationary background $\|\mathcal{L}\|_*\le \eta \|\cL\|_1$. 
Thus, we can estimate the quantities, $\eta_{\max,\mathcal{F}}$ and $\eta_{\min,\mathcal{F}}$  by
	\begin{equation}
	\begin{split} 
	\eta_{\max,\mathcal{F}}=\sup_{  \mathcal{S} \ {\rm moving\ targets}} \frac{\|\mathcal{S} \|_{*,\mathcal F}}{\|\mathcal{S} \|_1},\\
	\eta_{\min,\mathcal{F}}=\inf_{ \mathcal{L} \ {\rm stationary\ background}} \frac{\|\mathcal{L} \|_{*,\mathcal F}}{\|\mathcal{L}\|_1}.
	\label{eq:eta_minmax1}
	\end{split}
	\end{equation}
The classes of moving target and low rank data structures can be defined in several ways. Following \cite{leib2018RPCA}, we use our data model, defined in \ref{app:data_model}, and choose representatives of each to use in simulation.

	We wish to choose $\eta$ small enough so that the $\ell_1$ term is favorable for moving targets and large enough such that the nuclear term is favorable for the stationary background. We can define an objective which balances both requirements
		\begin{equation}
	F(\eta)=\frac{\eta_{\min,\mathcal{F}}}{\eta}+\frac{\eta}{\eta_{\max,\mathcal{F}}}
	\end{equation}
	With the optimal value
\begin{equation}
 \eta^*_{\mathcal{F}}=\sqrt{\eta_{\max,\mathcal{F}}\eta_{\min,\mathcal{F}}}.
	\end{equation}
Moreover, we expect that the larger the ratio $\eta_{\max,\mathcal{F}}/\eta_{\min,\mathcal{F}}$ the better the achieved separation, as the objective would have a wider range of admissible $\eta$'s.

We would like to use TRPCA under settings that increase this ratio, i.e.,  get the smallest possible nuclear norm for the stationary background and the largest possible nuclear norm for the moving target. Our objective in this section is to use this ratio so as to determine optimal values for the hyper-parameters, i.e. the sub-aperture and overlap sizes. To compute $\eta_{\max,\mathcal{F}}$ and $\eta_{\min,\mathcal{F}}$ as defined in  \eqref{eq:eta_minmax1} we would need to consider all possible scenarios of stationary and moving targets. To get first an insight for how the hyper-parameters affect $\eta_{\max,\mathcal{F}}$ and $\eta_{\min,\mathcal{F}}$, we define them for any specific example of a SAR data tensor $ \mathcal{A} = \mathcal{A}_L+\mathcal{A}_S$ as
\begin{equation}
	\eta_{\max,\mathcal{F}}=\frac{\|\mathcal{A}_S\|_{*,\mathcal F}}{\|\mathcal{A}_S\|_1},\\\quad \eta_{\min,\mathcal{F}}=\frac{\|\mathcal{A}_L\|_{*,\mathcal F}}{\|\mathcal{A}_L\|_1}.
	\label{eq:eta_minmax}
	\end{equation}
For reference, we also investigate these quantities for the decoupled case, that is,
	 	\begin{equation}
	 \eta_{\max,\mathcal{D}}=\frac{\|\mathcal{A}_S\|_{*,\mathcal D}}{\|\mathcal{A}_S\|_1},\\\quad \eta_{\min,\mathcal{D}}=\frac{\|\mathcal{A}_L\|_{*,\mathcal D}}{\|\mathcal{A}_L\|_1},
	 \label{eq:eta_minmax_dec}
	 \end{equation}
and compare their behavior to \eqref{eq:eta_minmax}.

In what follows we first consider in Section \ref{sec:hyper_sim} specific SAR data scenarios and observe how the quantities in \eqref{eq:eta_minmax} and  \eqref{eq:eta_minmax_dec} depend on the hyper-parameters and the moving target's trajectory. Observing that the main quantity that determines these ratios is the tensor nuclear norm, we introduce in Section \ref{sec:tnn_analysis} tensor norm inequalities that help us analyze the performance of TRPCA  for general SAR data. 
  
\subsection{TRPCA performance analysis for specific SAR data scenarios}
\label{sec:hyper_sim}
	We consider a stationary background with $10$ point scatterers and a single point moving target, with $\pmb{v}_t=1$m/s. The moving target's reflectivity is 10\% of the reflectivity of the other, stationary targets. We vary the target's trajectory angle $\alpha$ with respect to $\hat{x}$ in the 2D plane, that is,
\begin{equation}
\vec{\pmb{v}}_t=\pmb{v}_t[\cos\alpha,\sin\alpha,0],
\end{equation} 
between $\alpha=0$ and $\alpha=\pi/2$ using a step size  $\Delta \alpha=\pi/16$. A schematic of the simulation setting in given in Figure~\ref{fig:alpha}. The parameters of the simulation are as follows: The total aperture size $s_{\text{tot}}$ is fixed at $11.5$s and the platform is moving at $200$m/s, so that the effective aperture size is $2,300$m. This yields two data matrices $D_{L}$ and $D_S$ associated with the stationary background and moving target respectively. 
	
	We next let the sub-aperture size $s_{\text{sub}}$ take the values $[0.005s_{\text{tot}},0.01s_{\text{tot}},0.02s_{\text{tot}},0.03s_{\text{tot}},\dots,0.3s_{\text{tot}}]$. For each sub-aperture size, we change the overlap $\vartheta$, as a fraction of the sub-aperture size, to be $0.1,0.2,\dots,0.9$. For each of these configurations we create the tensor data structures $\mathcal{A}_L,\mathcal{A}_S$, out of $D_L,D_S$, according to \eqref{eq:mat_to_tensor}. 
	The number of subapertures $n_3$ is determined by the other parameters through the following formula
	\begin{equation}
	n_3=1+\left\lceil \frac{s_{\text{tot}}-s_{\text{sub}}}{(1-\vartheta)s_{\text{sub}}}\right\rceil.
	\label{eq:n3}
	\end{equation}
	In Figure~\ref{fig:n3_ex} we illustrate how $n_3$ varies as a function of the overlap $\vartheta$ for different sub-aperture sizes.	
	\begin{figure}[htbp!]
		\centering
		\hspace{0em}\includegraphics[width=0.6\columnwidth]{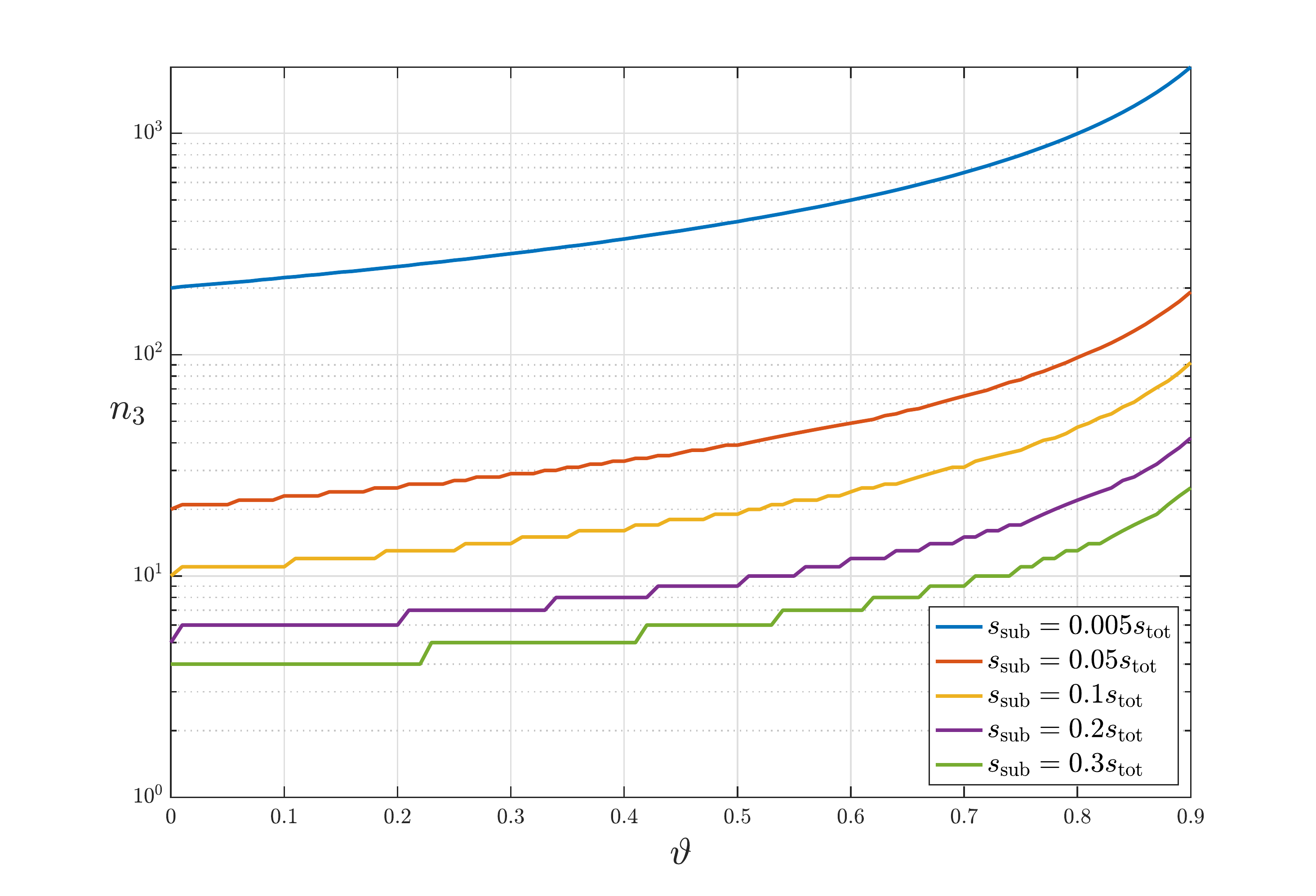}
		\caption{Illustration of how $n_3$ varies as a function of $\vartheta$ in \eqref{eq:n3} 
		for a fixed sub-aperture size and total aperture. 
		$n_3$ increases as the size of the sub-aperture decreases and as $\vartheta$ increases.}
		\label{fig:n3_ex}
	\end{figure}	
	For each configuration, we compute the $\ell_1$ norm and the nuclear norm, for both the decoupled and the tensor forms and plot the ratio of the quantities in  \eqref{eq:eta_minmax_dec} and \eqref{eq:eta_minmax} as function of $\alpha$ and the tensor hyper-parameters in figures~\ref{fig:eta_plot_dec} and ~\ref{fig:eta_plot} respectively. 	
	\begin{figure}[htbp!]
		\centering
		\begin{subfigure}[t]{0.22\textwidth}
			\includegraphics[width=1\columnwidth]{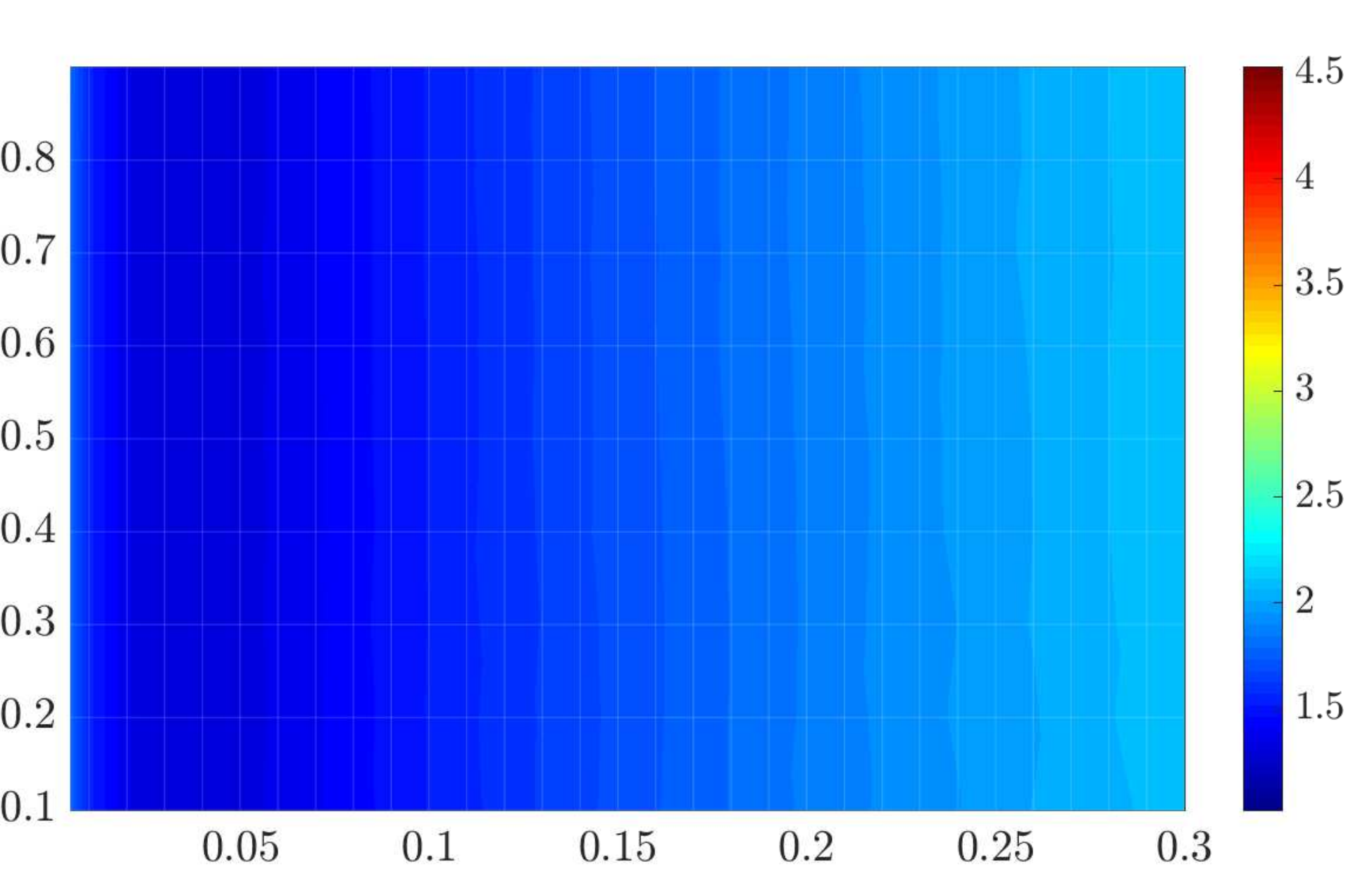}
			\caption*{$\alpha=0$} 
			\label{}
		\end{subfigure}
		\begin{subfigure}[t]{0.22\textwidth}
			\centering
			\includegraphics[width=1\columnwidth]{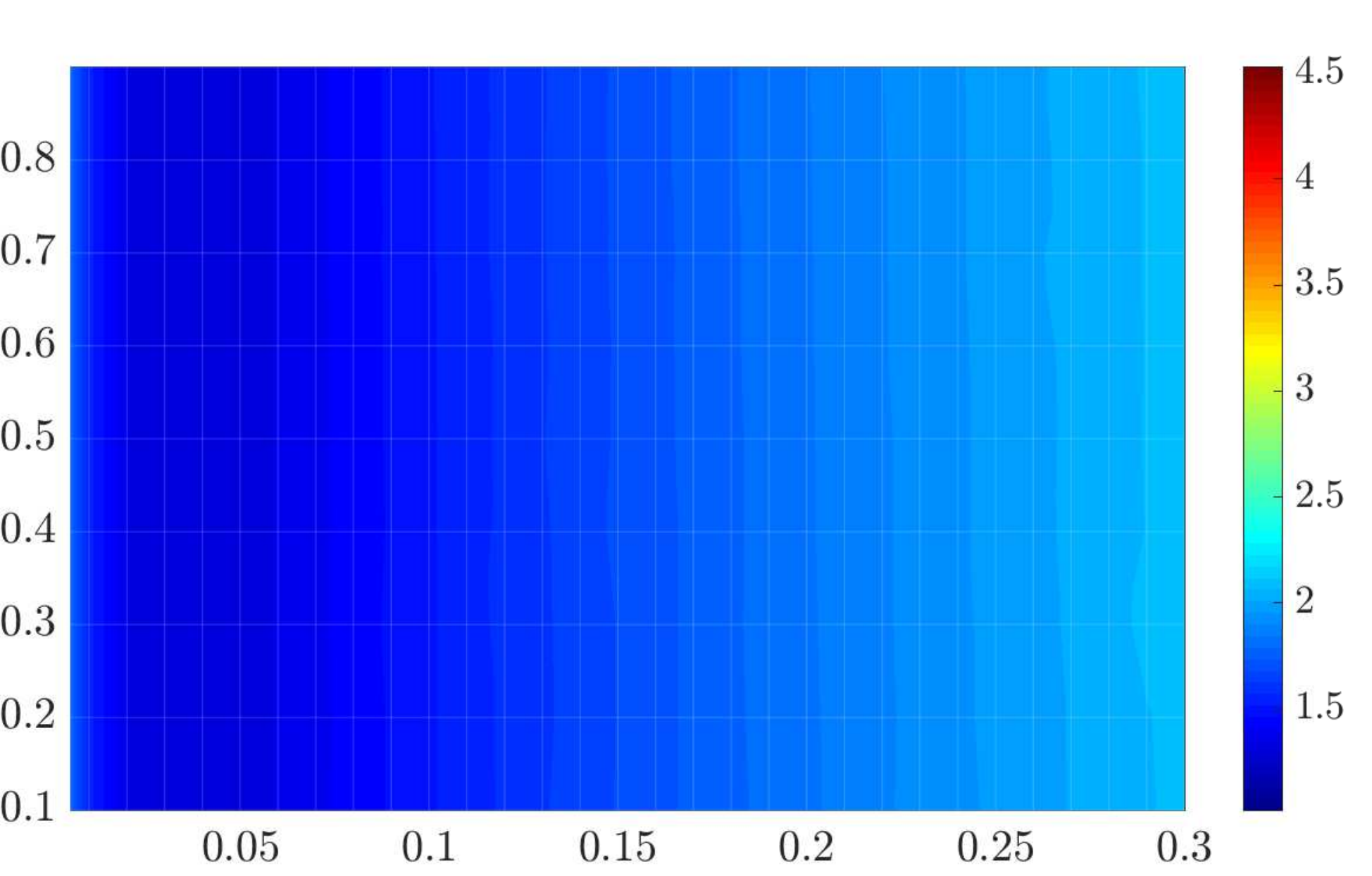}
				\caption*{$\alpha=\pi/16$} 
			\label{}
		\end{subfigure}
		\begin{subfigure}[t]{0.22\textwidth}
			\includegraphics[width=1\columnwidth]{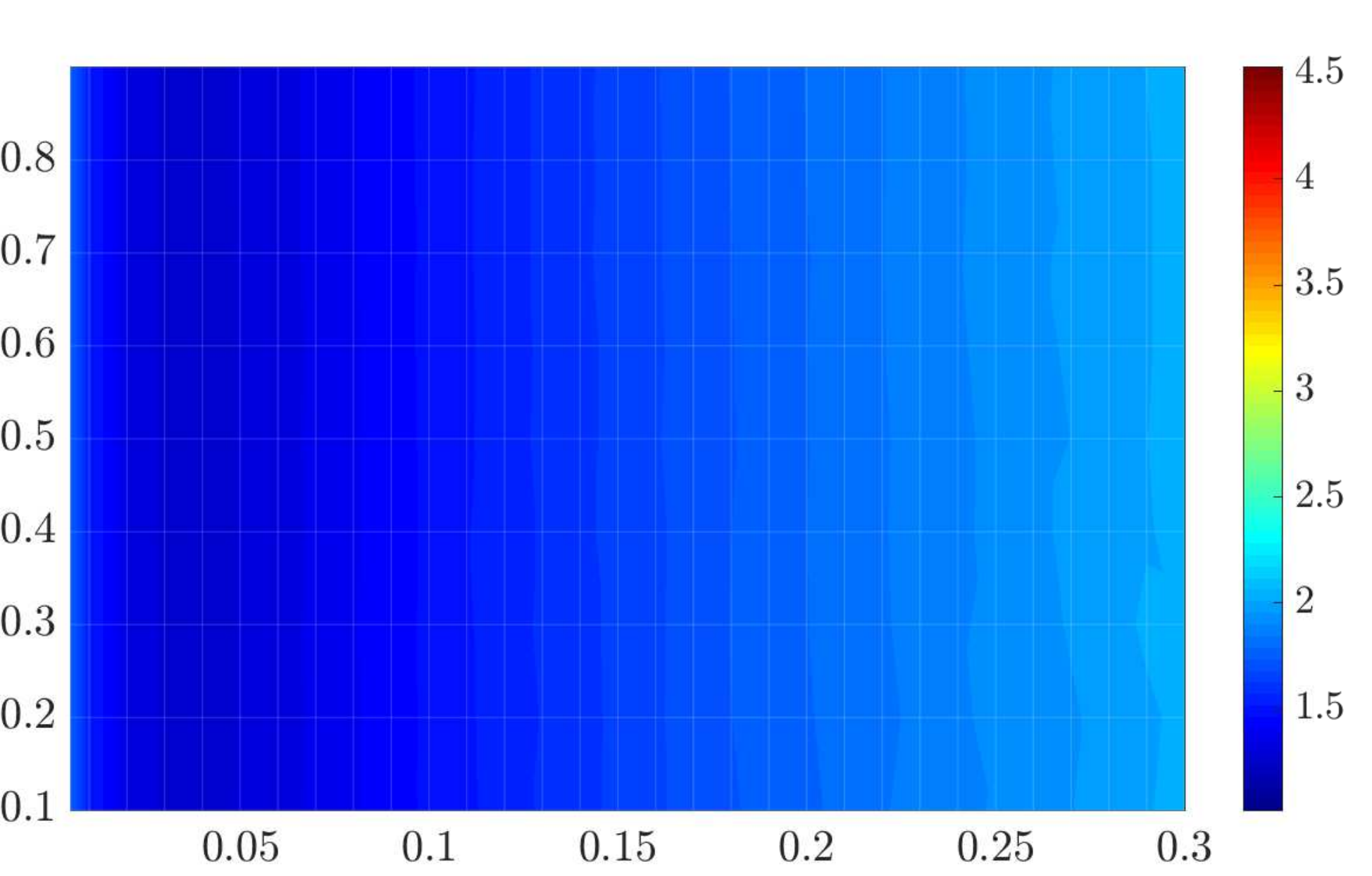}
				\caption*{$\alpha=\pi/8$} 
			\label{}
		\end{subfigure} 
		\begin{subfigure}[t]{0.22\textwidth}
			\includegraphics[width=1\columnwidth]{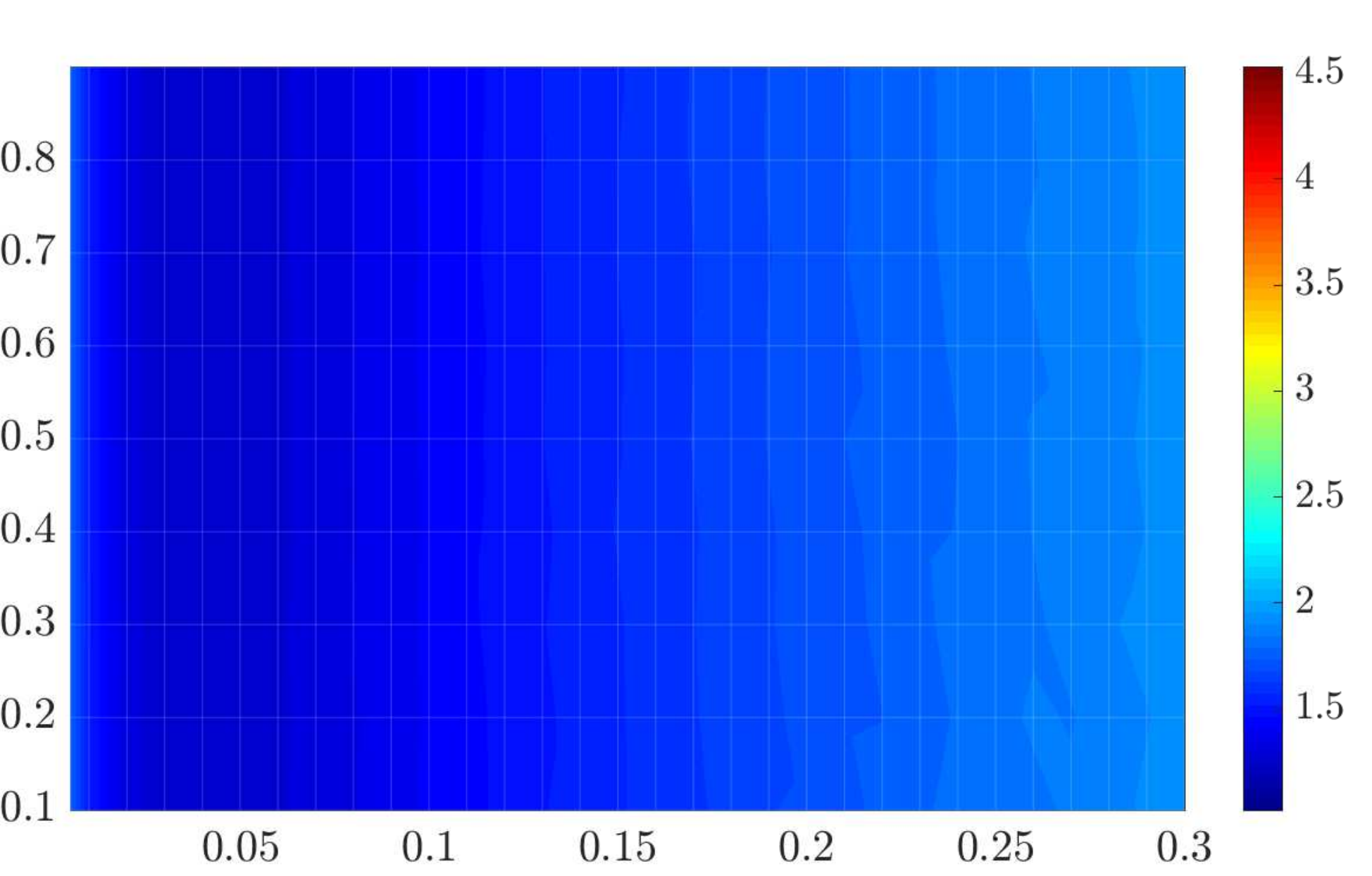}
				\caption*{$\alpha=3\pi/16$} 
			\label{}
		\end{subfigure}
		\begin{subfigure}[t]{0.22\textwidth}
			\centering
			\includegraphics[width=1\columnwidth]{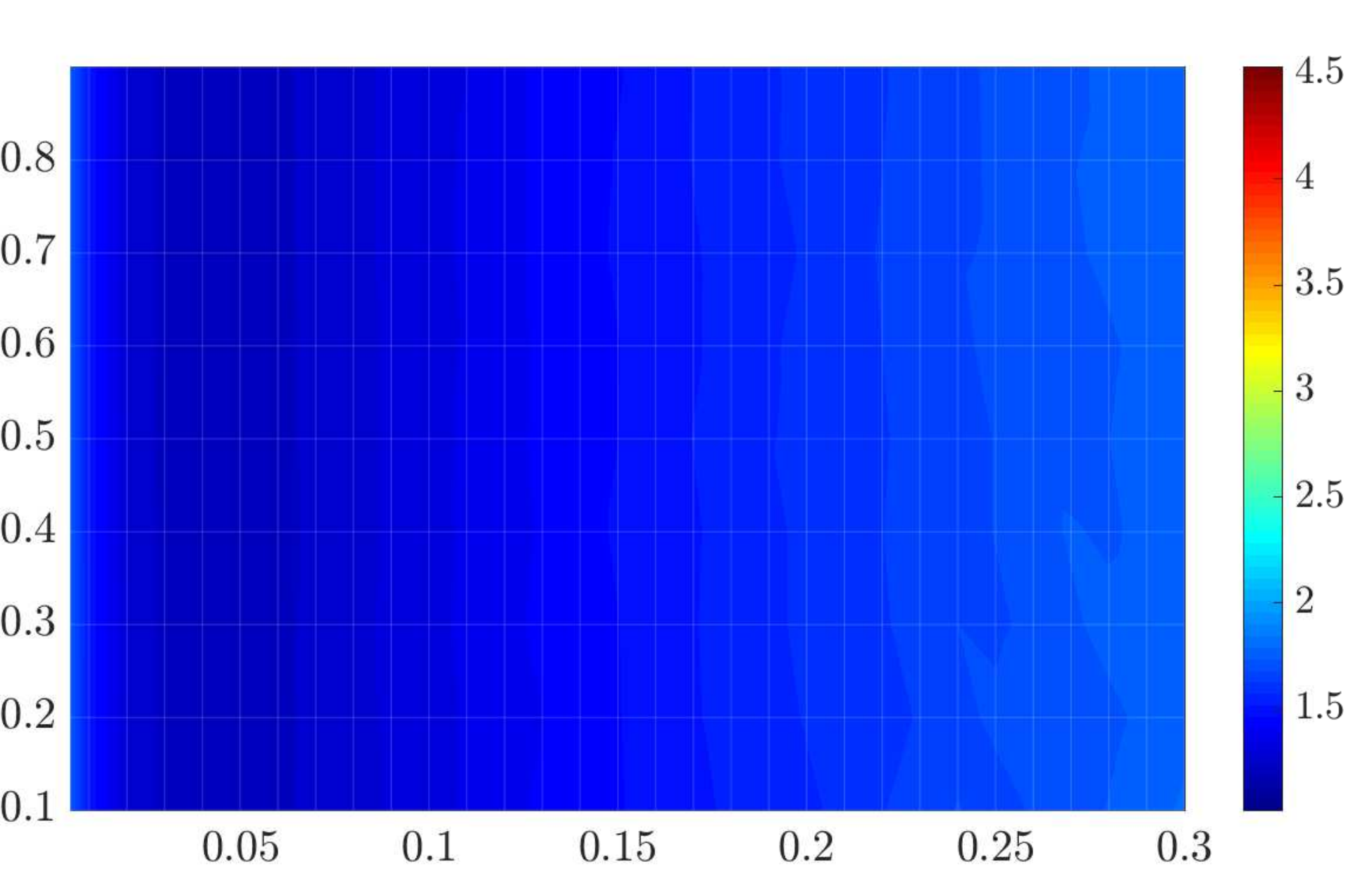}
				\caption*{$\alpha=\pi/4$} 
			\label{}
		\end{subfigure}
		\begin{subfigure}[t]{0.22\textwidth}
			\includegraphics[width=1\columnwidth]{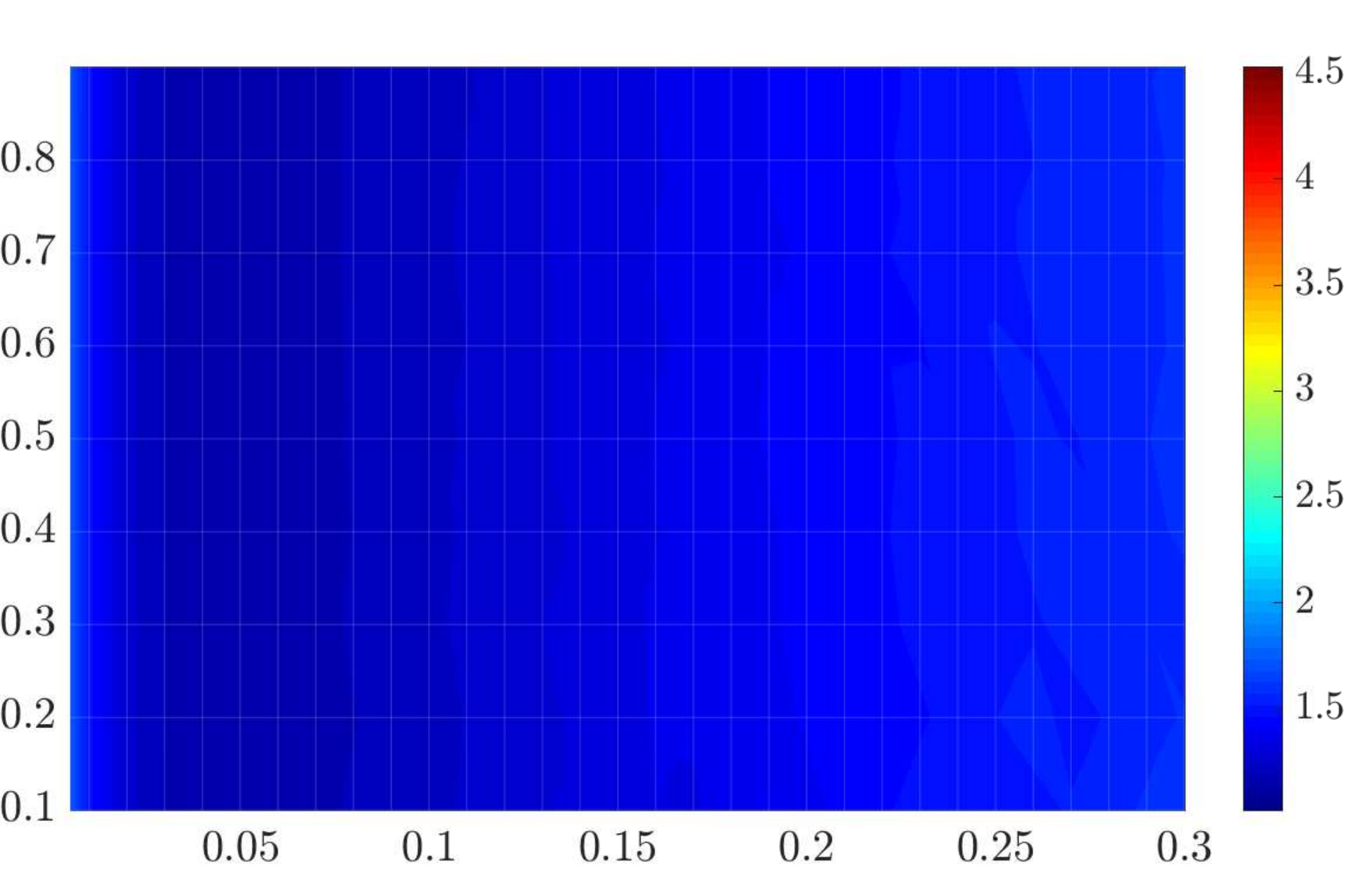}
				\caption*{$\alpha=5\pi/16$} 
			\label{}
		\end{subfigure} 
		\begin{subfigure}[t]{0.22\textwidth}
			\includegraphics[width=1\columnwidth]{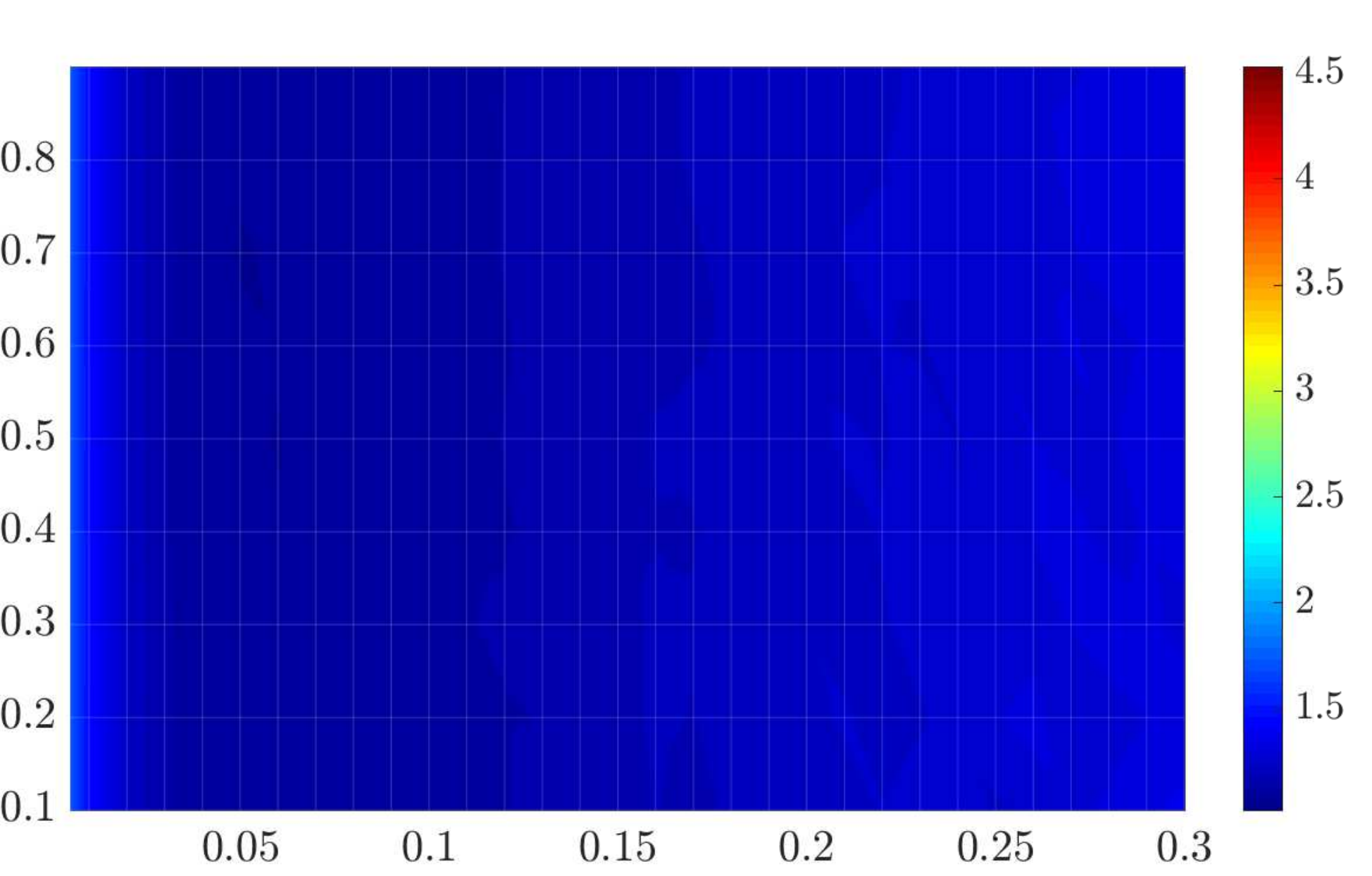}
			\caption*{$\alpha=3\pi/8$} 
				\end{subfigure}
		\begin{subfigure}[t]{0.22\textwidth}
			\centering
			\includegraphics[width=1\columnwidth]{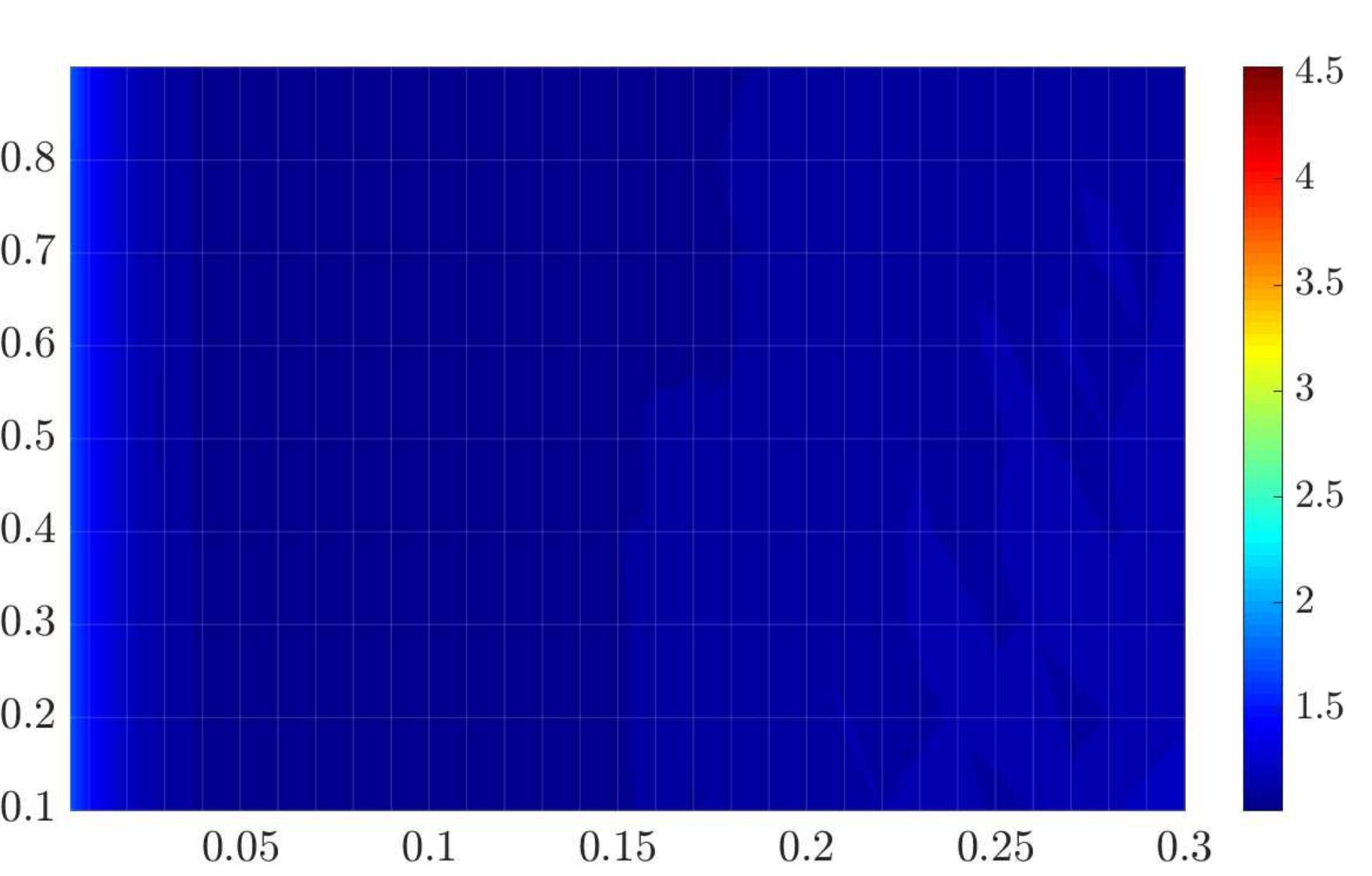}
			\caption*{$\alpha=7\pi/16$} 
			\label{}
		\end{subfigure}
		\caption{Ratio of decoupled nuclear norm and $\ell_1$ norm as a function of data hyper-parameters. We plot $\eta_{\max,\mathcal{D}}/\eta_{\min,\mathcal{D}}$ defined in \eqref{eq:eta_minmax_dec},  for varying hyper-parameters. The $x$-axis is the sub-aperture size and the $y$-axis is the overlap size. We observe that the ratio is small and favors using larger apertures, with almost no dependence on the sub-aperture overlap. We can see a slight decrease in value as $\alpha$ increases.}
		\label{fig:eta_plot_dec}
	\end{figure}	
For the decoupled form (see Figure~\ref{fig:eta_plot_dec}) we observe a weak dependence on the overlap, while the ratio tends to grow with the sub-aperture size. The maximal value for the ratio is achieved for smaller angles, in consistency with results of regular matrix RPCA. 
	
	The tensor case illustrated in Figure~\ref{fig:eta_plot}, presents strong angular dependence, favoring different hyper-parameter configurations at different angles: smaller angles tend to achieve optimal ratio for large apertures with low overlap. However as $\alpha$ increases, the optimal parameter configuration tends towards smaller sub-apertures with higher overlap. The color scale is the same in figures ~\ref{fig:eta_plot_dec} and \ref{fig:eta_plot} indicating that smaller values are obtained in the decoupled case.
		\begin{figure}[htbp!]
		\centering
		\begin{subfigure}[t]{0.22\textwidth}
			\includegraphics[width=1\columnwidth]{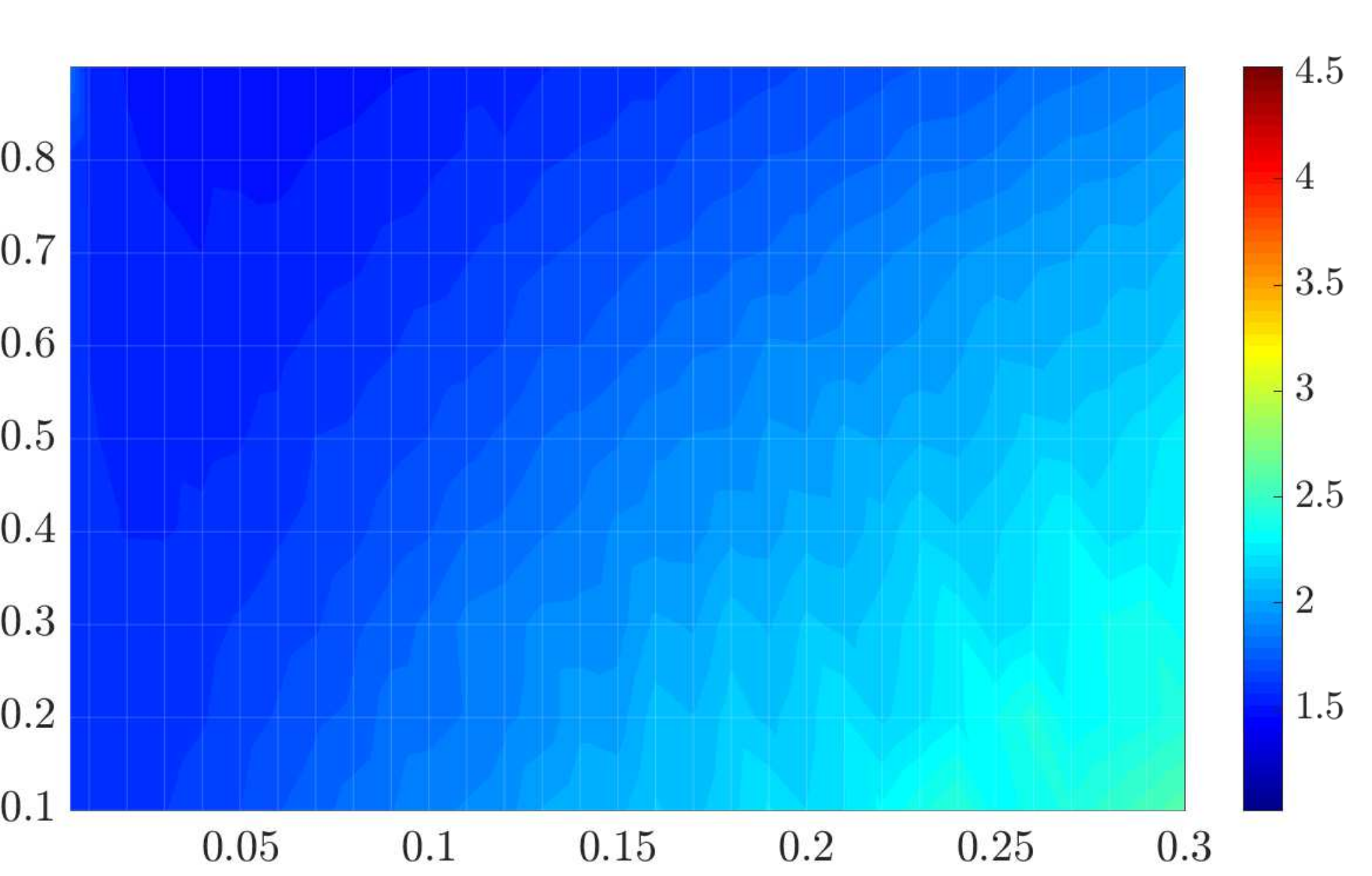}
			\caption*{$\alpha=0$} 
			\label{}
		\end{subfigure}
		\begin{subfigure}[t]{0.22\textwidth}
			\centering
			\includegraphics[width=1\columnwidth]{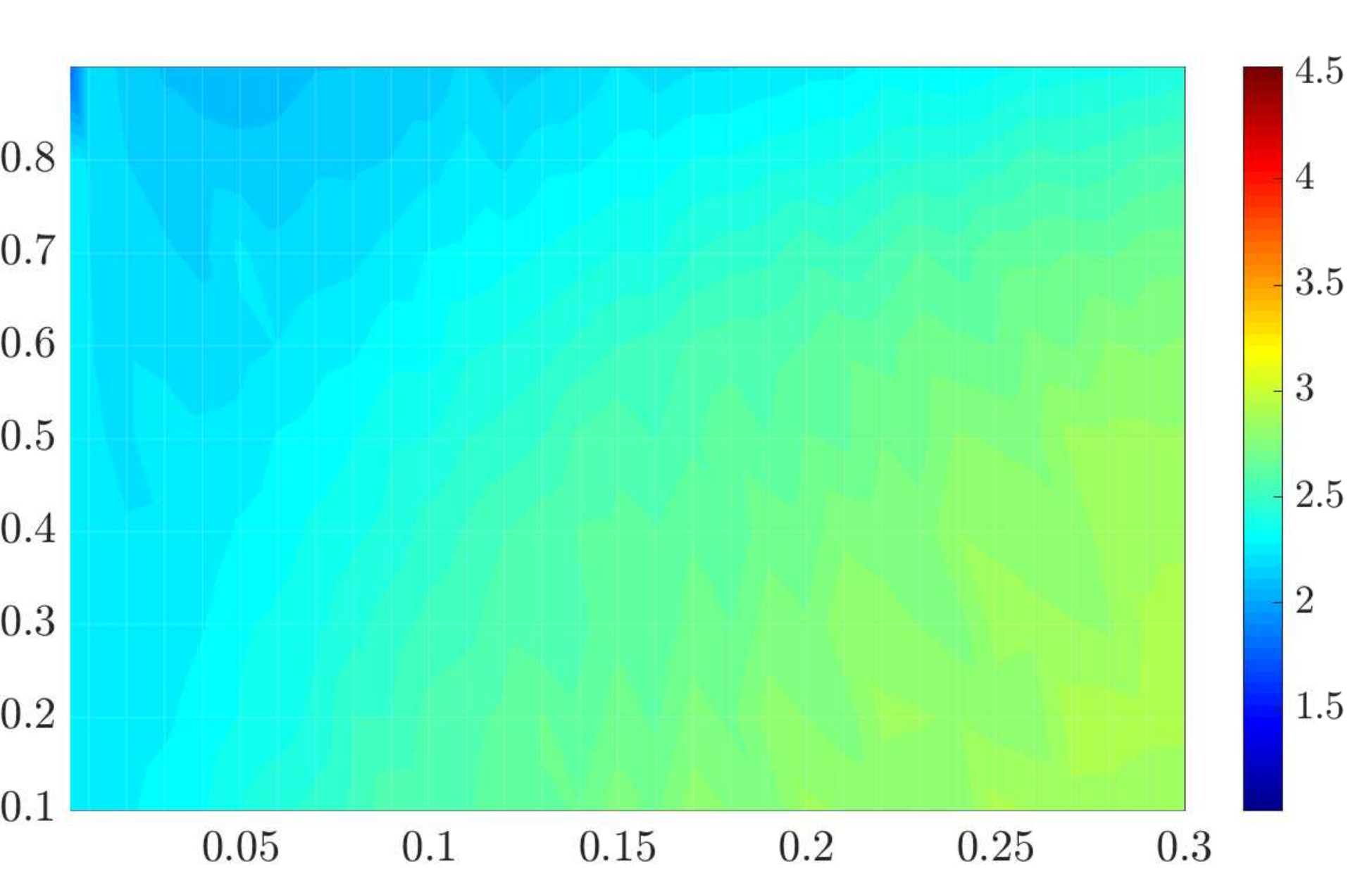}
			\caption*{$\alpha=\pi/16$} 
			\label{}
		\end{subfigure}
		\begin{subfigure}[t]{0.22\textwidth}
			\includegraphics[width=1\columnwidth]{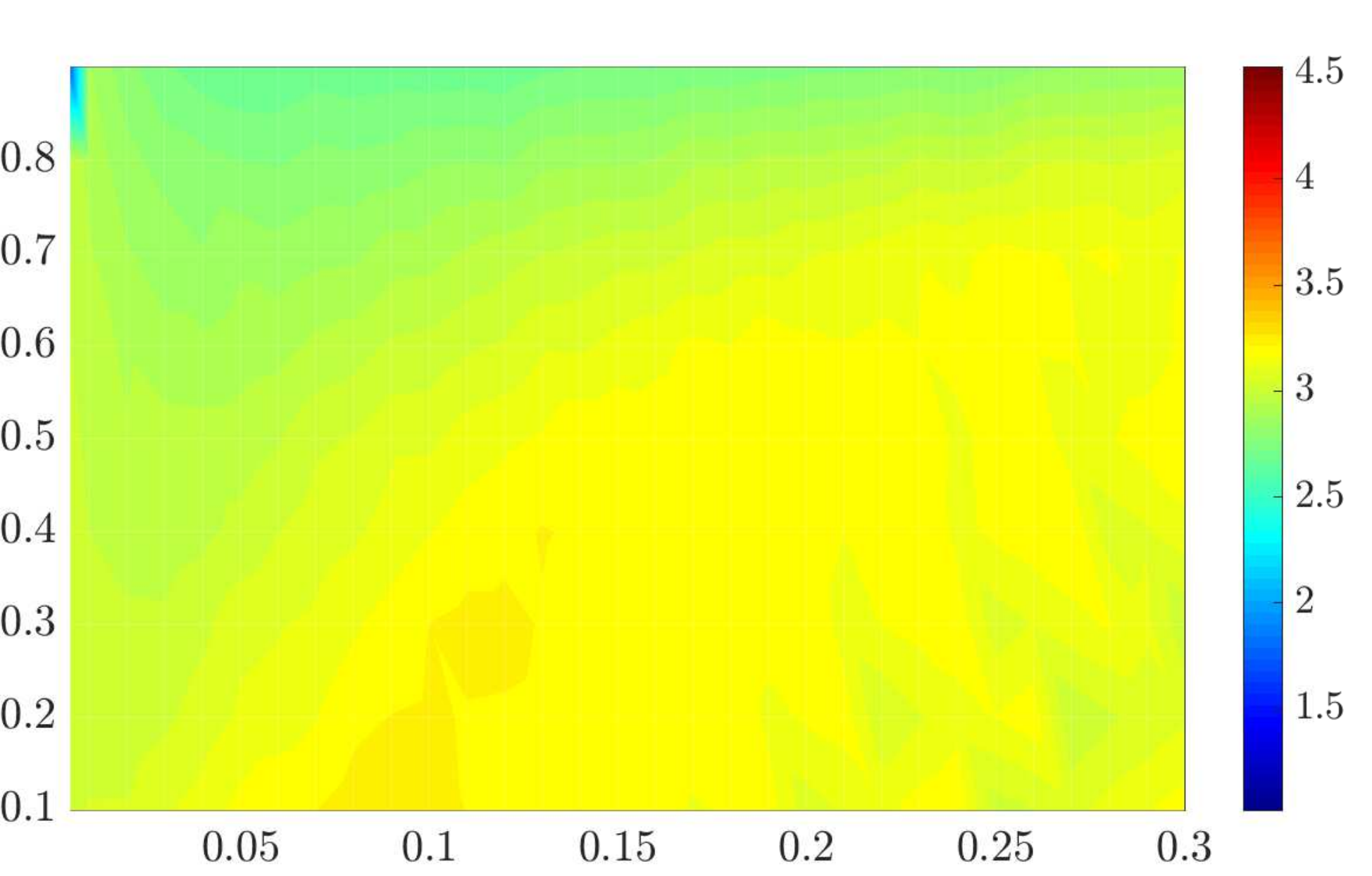}
			\caption*{$\alpha=\pi/8$} 
			\label{}
		\end{subfigure} 
		\begin{subfigure}[t]{0.22\textwidth}
			\includegraphics[width=1\columnwidth]{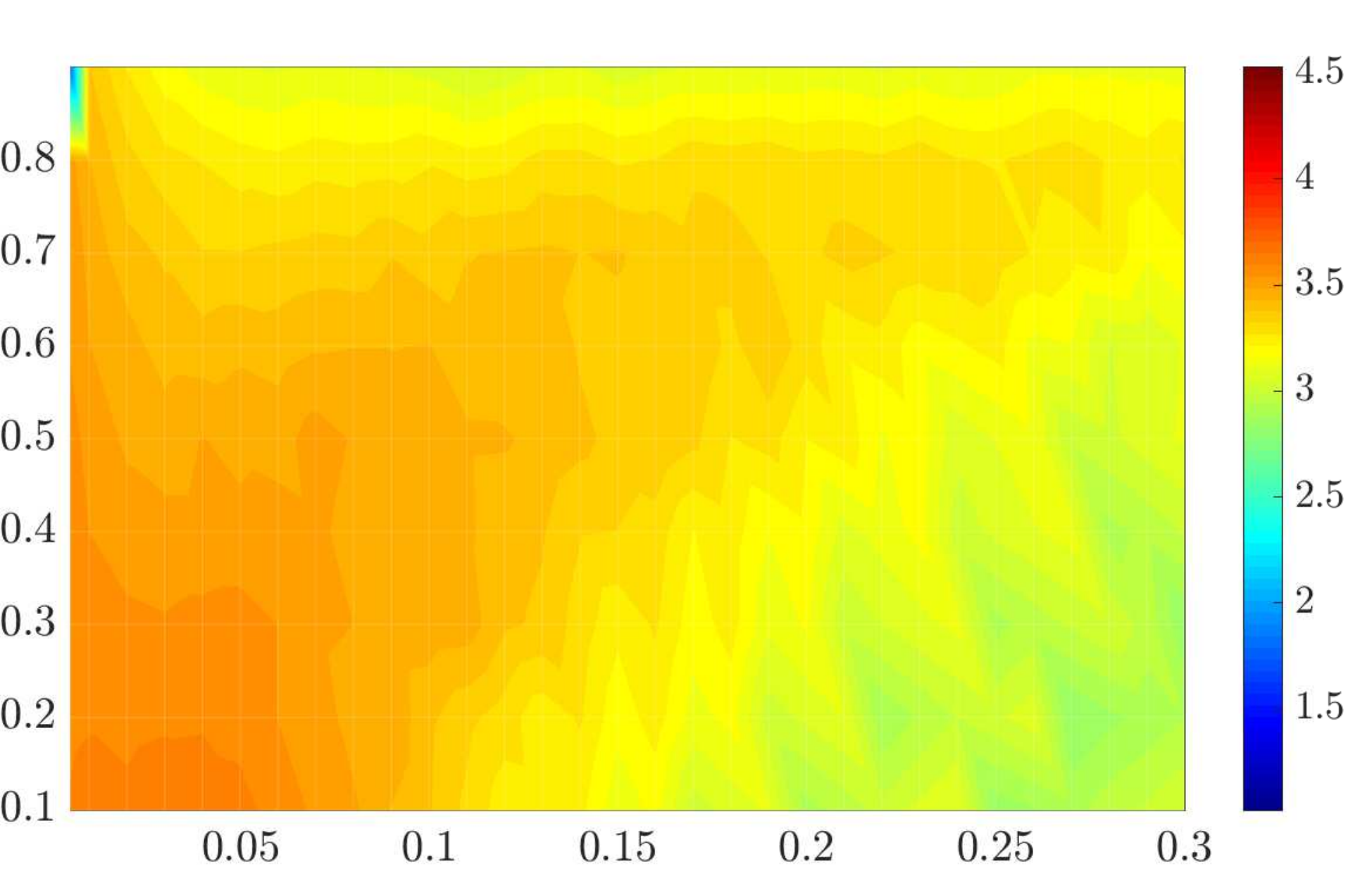}
			\caption*{$\alpha=3\pi/16$} 
			\label{}
		\end{subfigure}
		\begin{subfigure}[t]{0.22\textwidth}
			\centering
			\includegraphics[width=1\columnwidth]{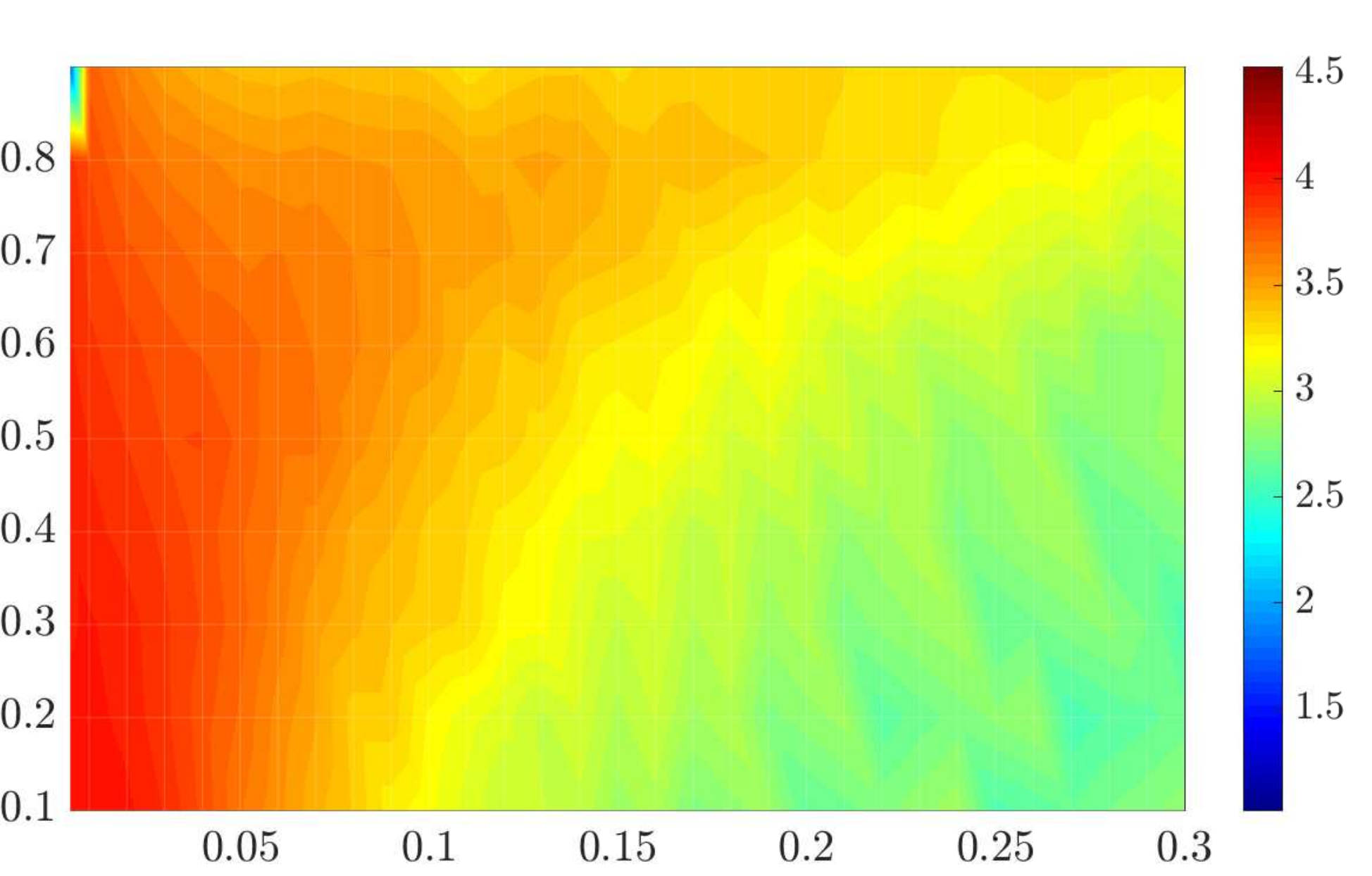}
			\caption*{$\alpha=\pi/4$} 
			\label{}
		\end{subfigure}
		\begin{subfigure}[t]{0.22\textwidth}
			\includegraphics[width=1\columnwidth]{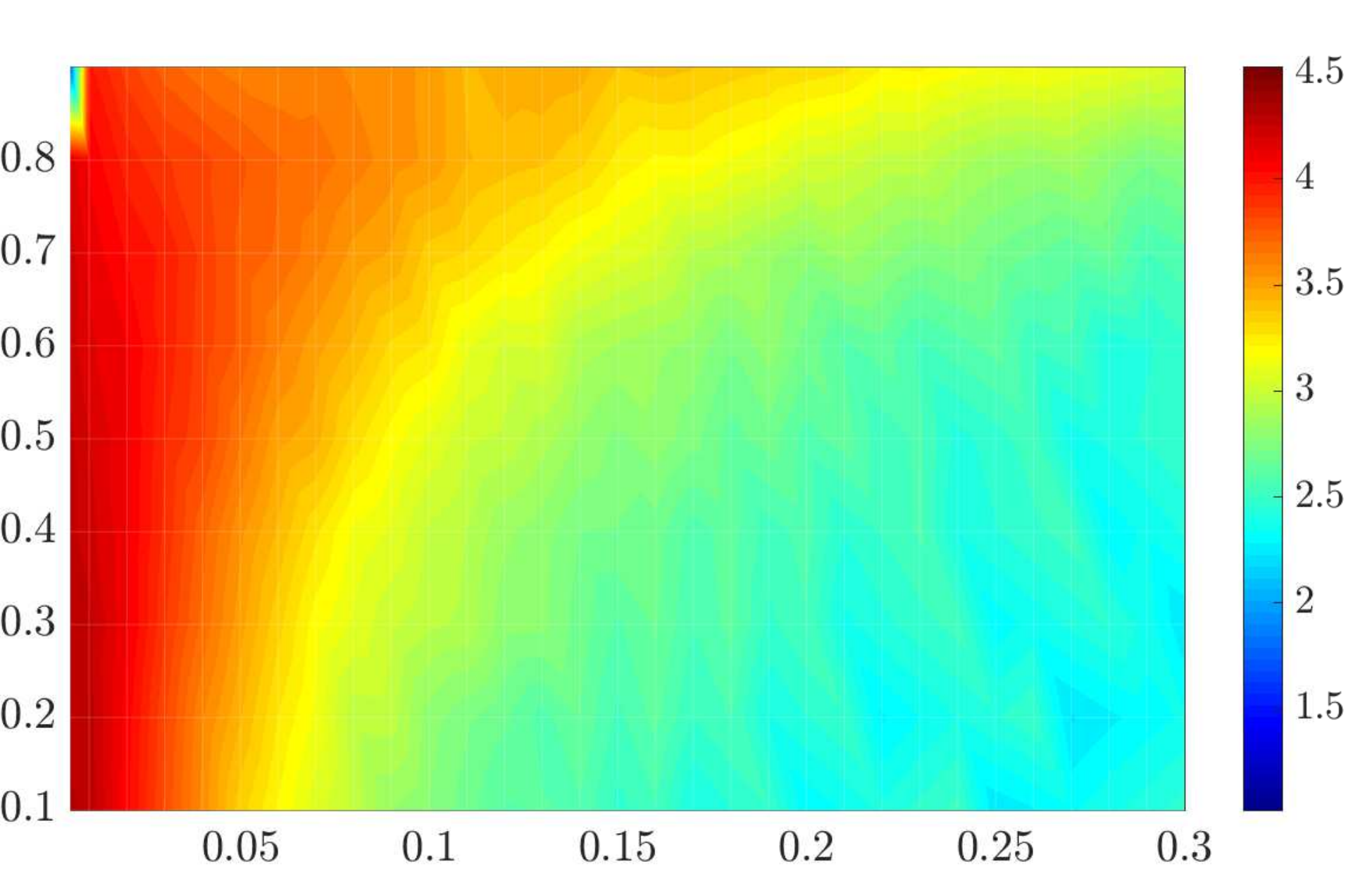}
			\caption*{$\alpha=5\pi/16$} 
			\label{}
		\end{subfigure} 
		\begin{subfigure}[t]{0.22\textwidth}
			\includegraphics[width=1\columnwidth]{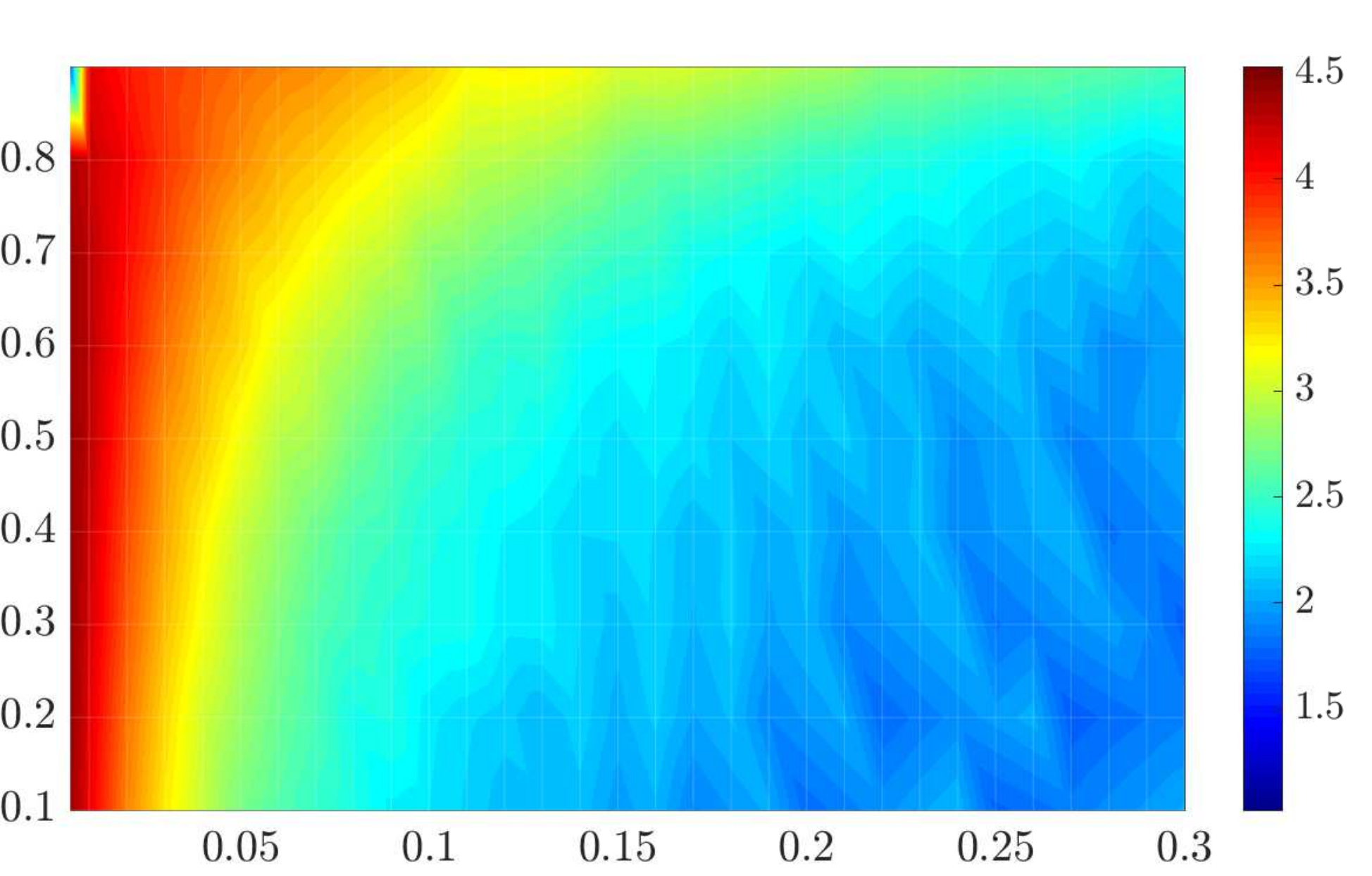}
			\caption*{$\alpha=3\pi/8$} 
		\end{subfigure}
		\begin{subfigure}[t]{0.22\textwidth}
			\centering
			\includegraphics[width=1\columnwidth]{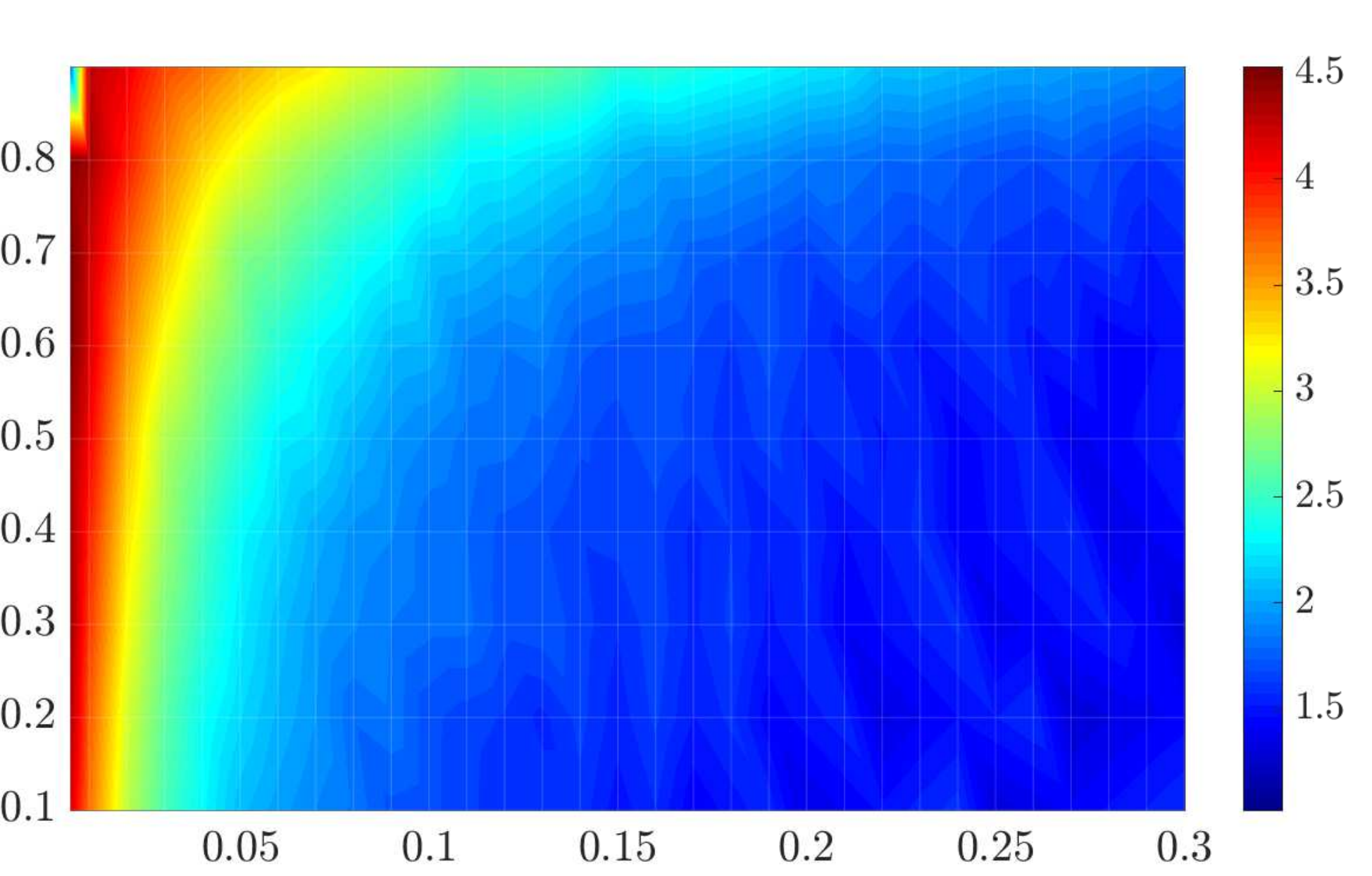}
			\caption*{$\alpha=7\pi/16$} 
			\label{}
		\end{subfigure}
		\caption{Ratio of tensor nuclear norm and $\ell_1$ norm as a function of data hyoerparameters. We plot $\eta_{\max,\mathcal{F}}/\eta_{\min,\mathcal{F}}$, defined in \eqref{eq:eta_minmax}, for varying hyper-parameters. The $x$ axis and $y$ axis are as in Figure~\ref{fig:eta_plot_dec}. We observe that the dependence on the hyper-parameter changes with the target's direction. For small angles, larger values are achieved for large sub-apertures with low overlap. For larger angles, the optimal parameters shift towards smaller sub-apertures with higher overlap. The most robust behavior is achieved around $\pi/4$. An interpretation of these results is given in Section~\ref{sec:interpret}}
		\label{fig:eta_plot}
	\end{figure}

In order to gain more insight, we break down the terms in \eqref{eq:eta_minmax}, namely the $\ell_1$ and nuclear norms of each term, both for the decoupled and tensor forms, and observe their variation in the configuration space.
	Let us first consider the ratio of the background and moving target's $\ell_1$ norm, as illustrated in Figure~\ref{fig:L1_ratio}. We can see that  the $\ell_1$ norm is insensitive to the direction the target is moving as well as to the hyper-parameters. 
	\begin{figure}[htbp!]
	\centering
	\begin{subfigure}[t]{0.22\textwidth}
		\includegraphics[width=1\columnwidth]{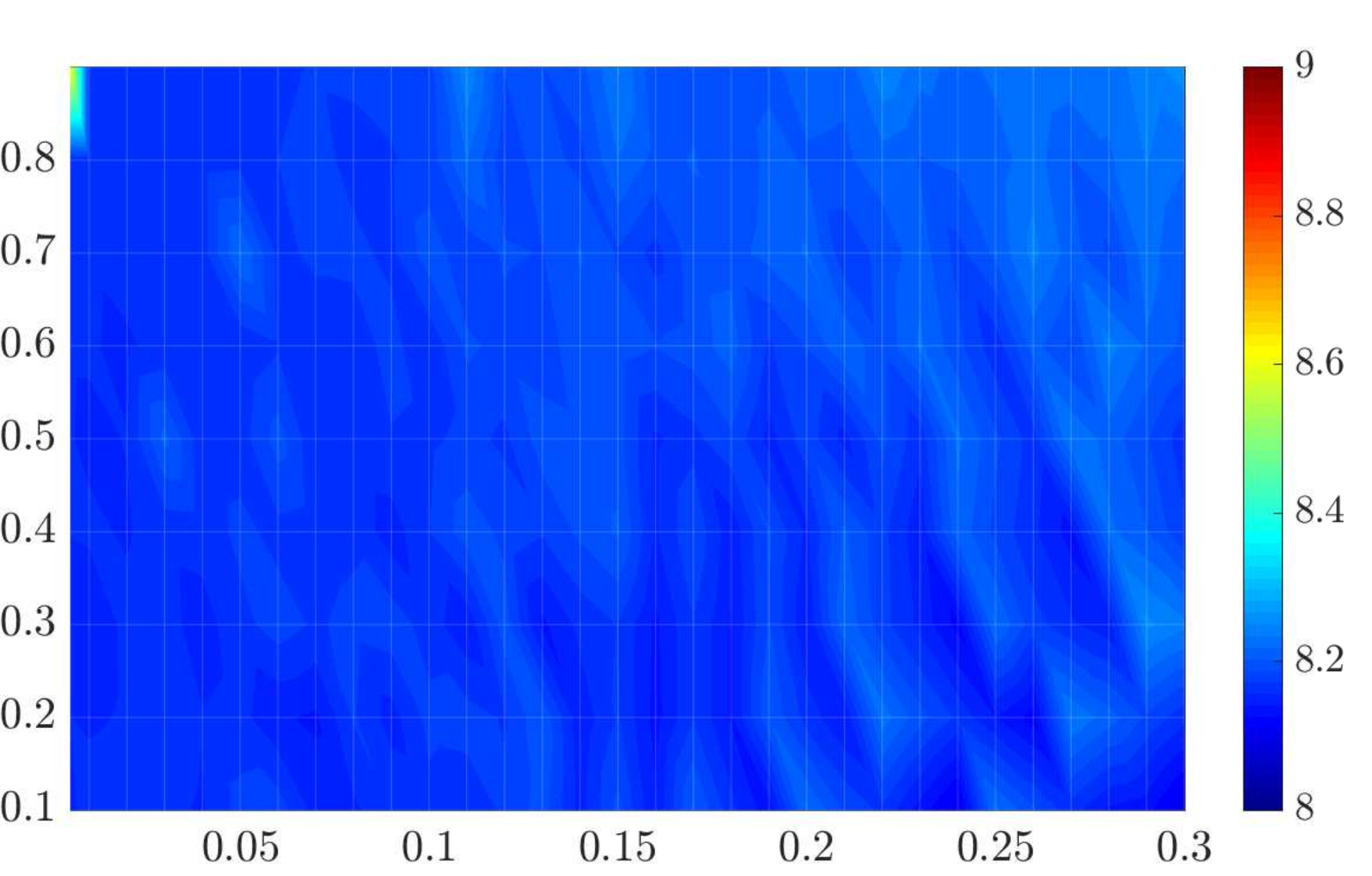}
		\caption*{$\alpha=0$} 
		\label{}
	\end{subfigure}
	\begin{subfigure}[t]{0.22\textwidth}
		\centering
		\includegraphics[width=1\columnwidth]{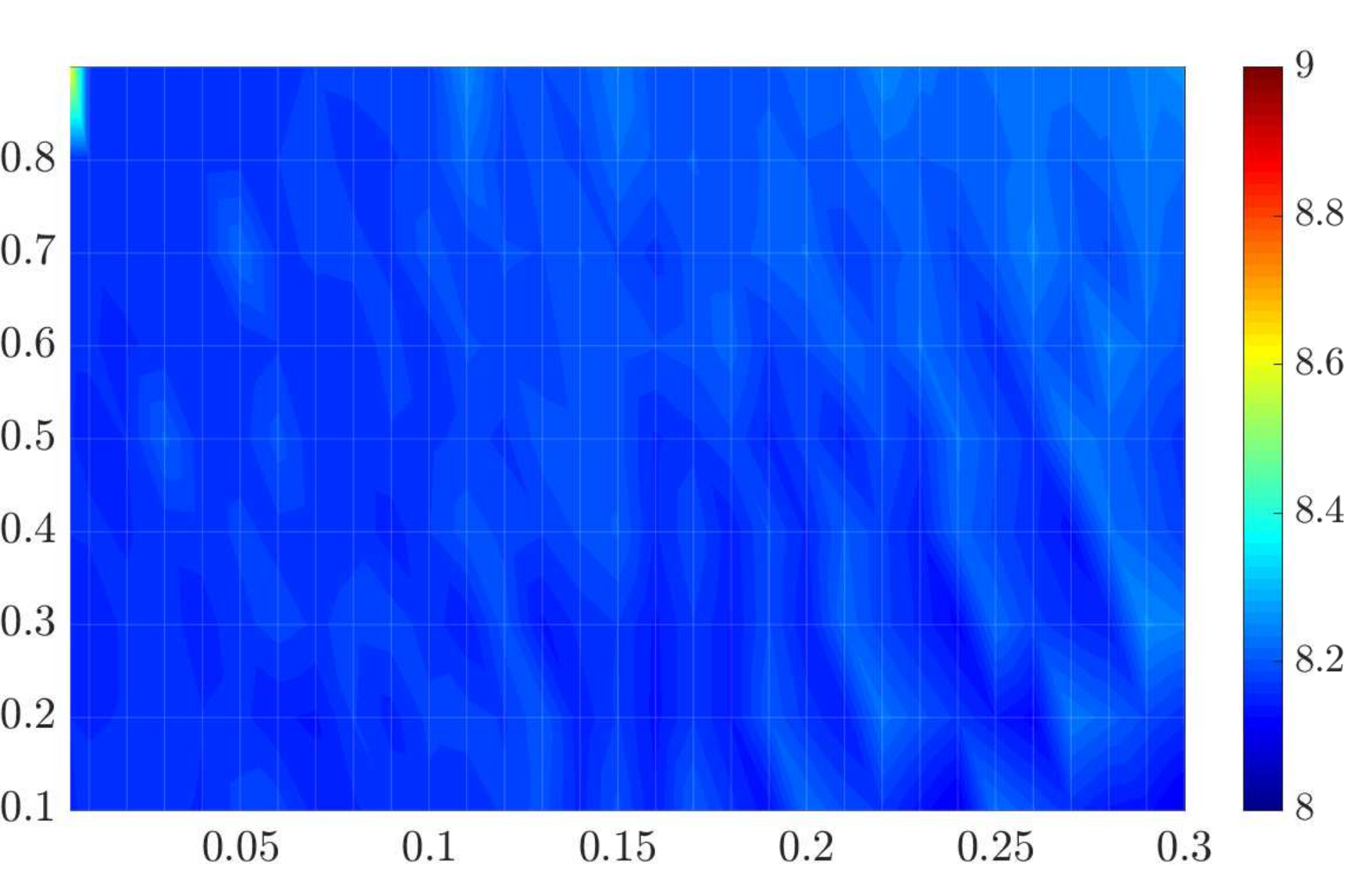}
		\caption*{$\alpha=\pi/16$} 
		\label{}
	\end{subfigure}
	\begin{subfigure}[t]{0.22\textwidth}
		\includegraphics[width=1\columnwidth]{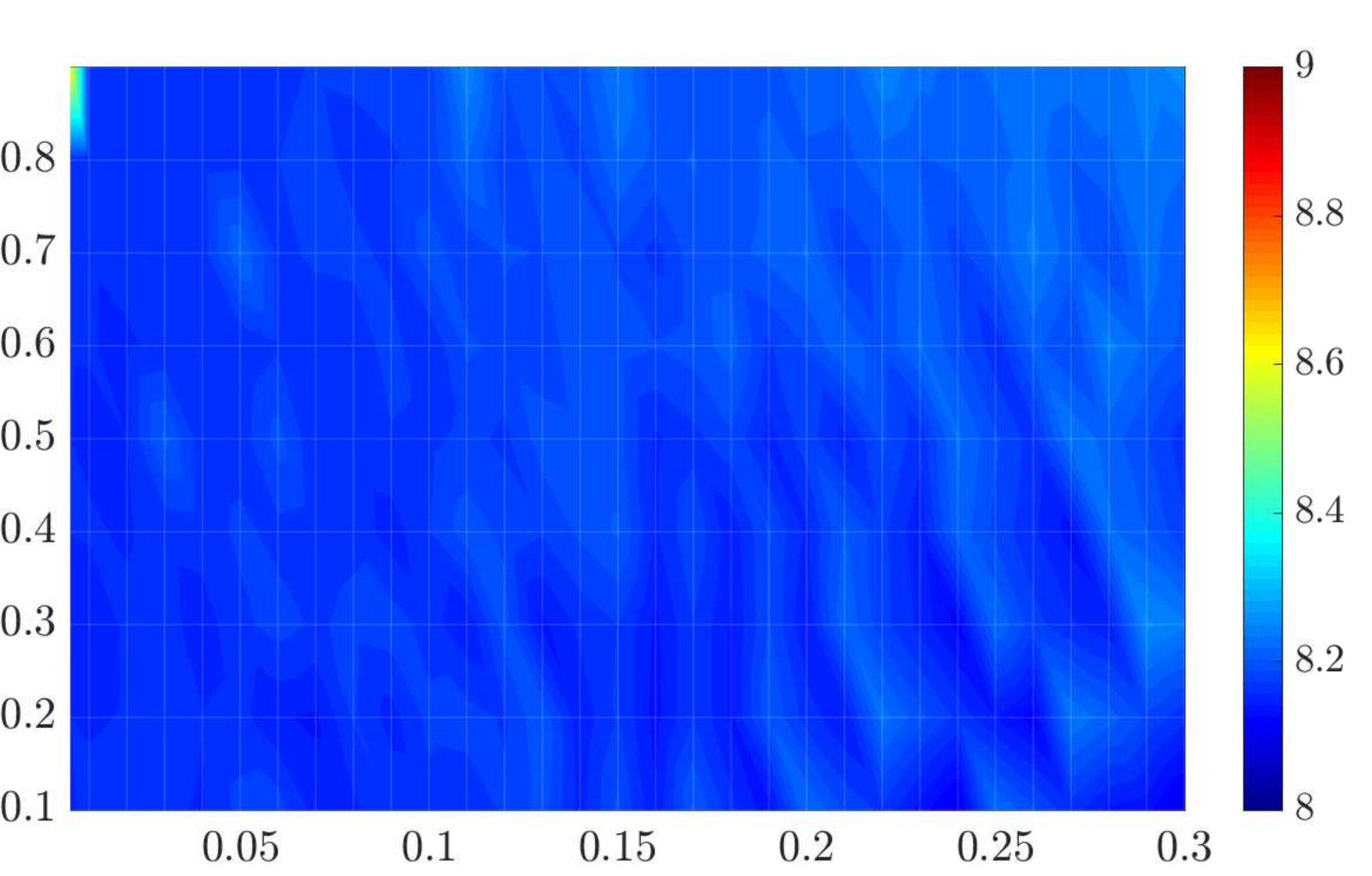}
		\caption*{$\alpha=\pi/8$} 
		\label{}
	\end{subfigure}
	\begin{subfigure}[t]{0.22\textwidth}
		\includegraphics[width=1\columnwidth]{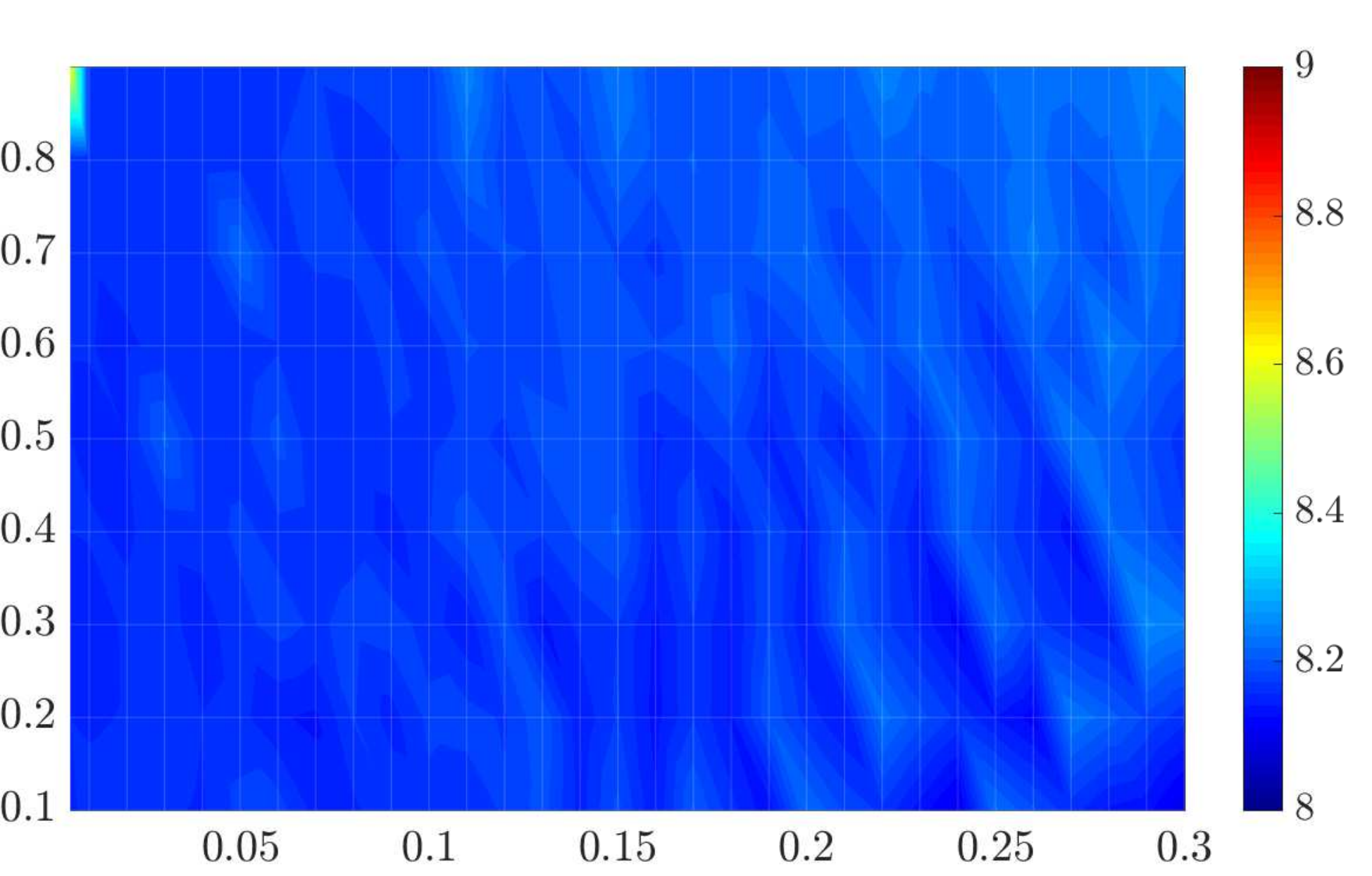}
		\caption*{$\alpha=3\pi/16$} 
		\label{}
	\end{subfigure}
	\begin{subfigure}[t]{0.22\textwidth}
		\centering
		\includegraphics[width=1\columnwidth]{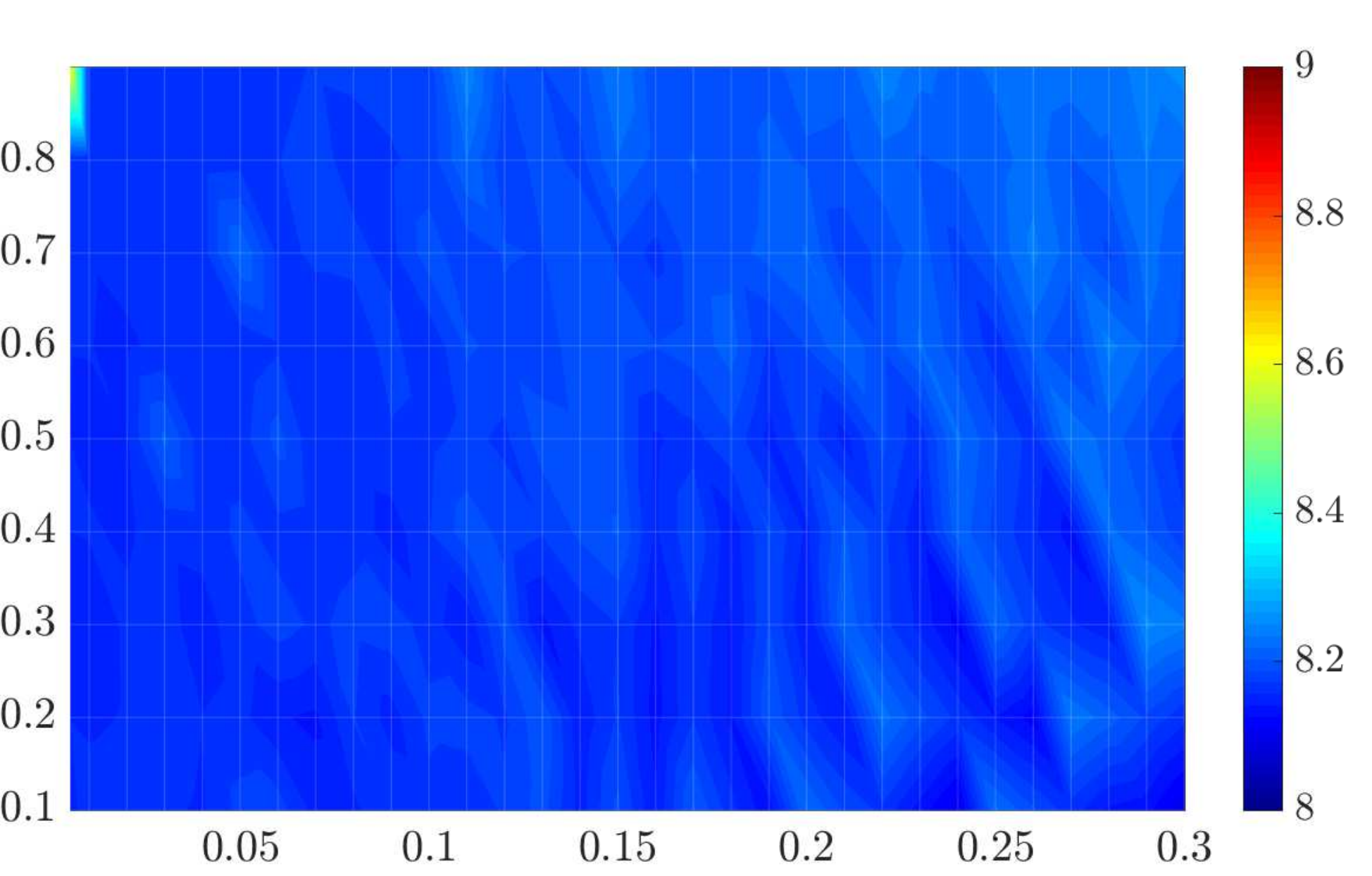}
		\caption*{$\alpha=\pi/4$} 
		\label{}
	\end{subfigure}
	\begin{subfigure}[t]{0.22\textwidth}
		\includegraphics[width=1\columnwidth]{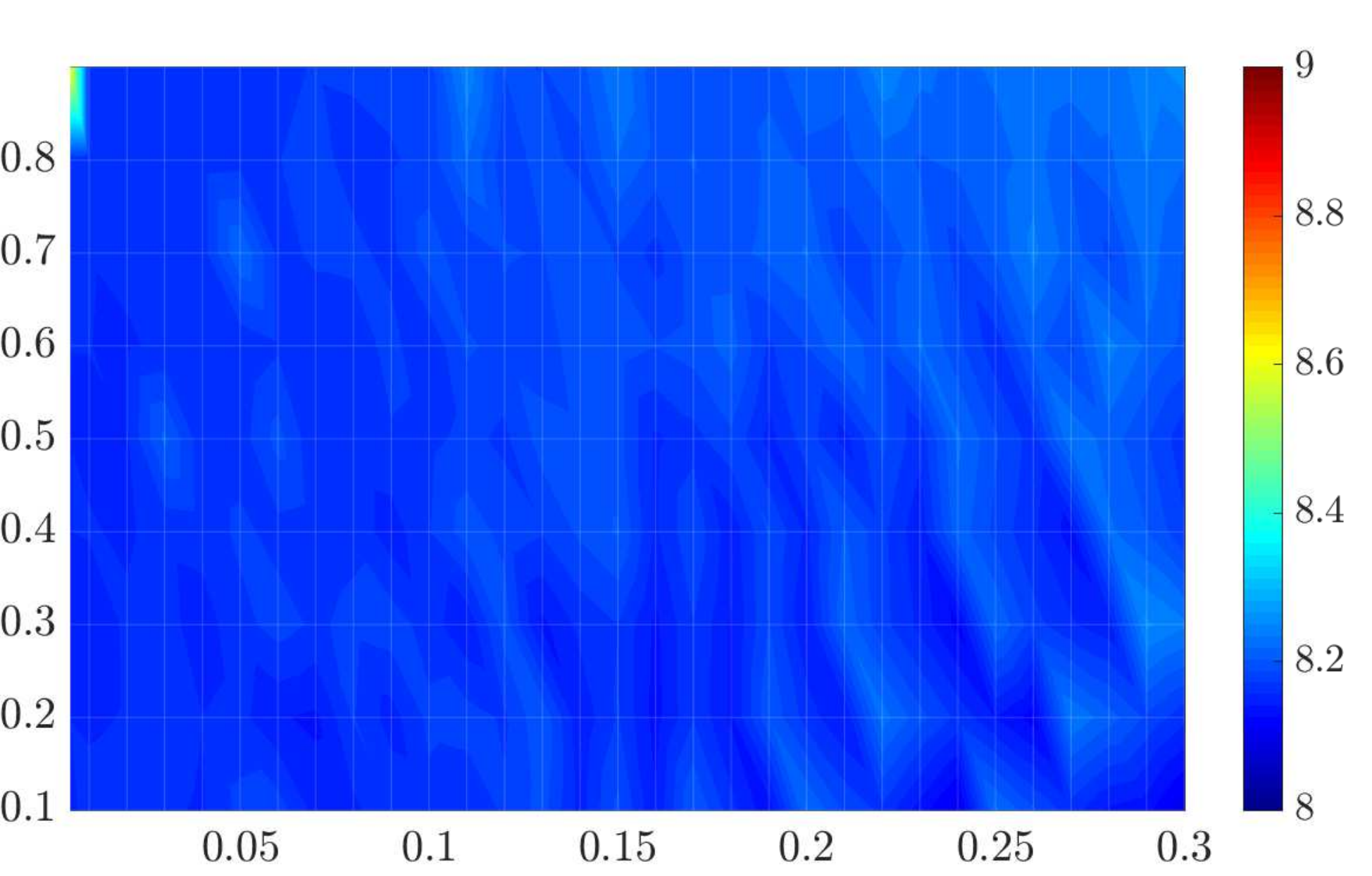}
		\caption*{$\alpha=5\pi/16$} 
		\label{}
	\end{subfigure} 
	\begin{subfigure}[t]{0.22\textwidth}
		\includegraphics[width=1\columnwidth]{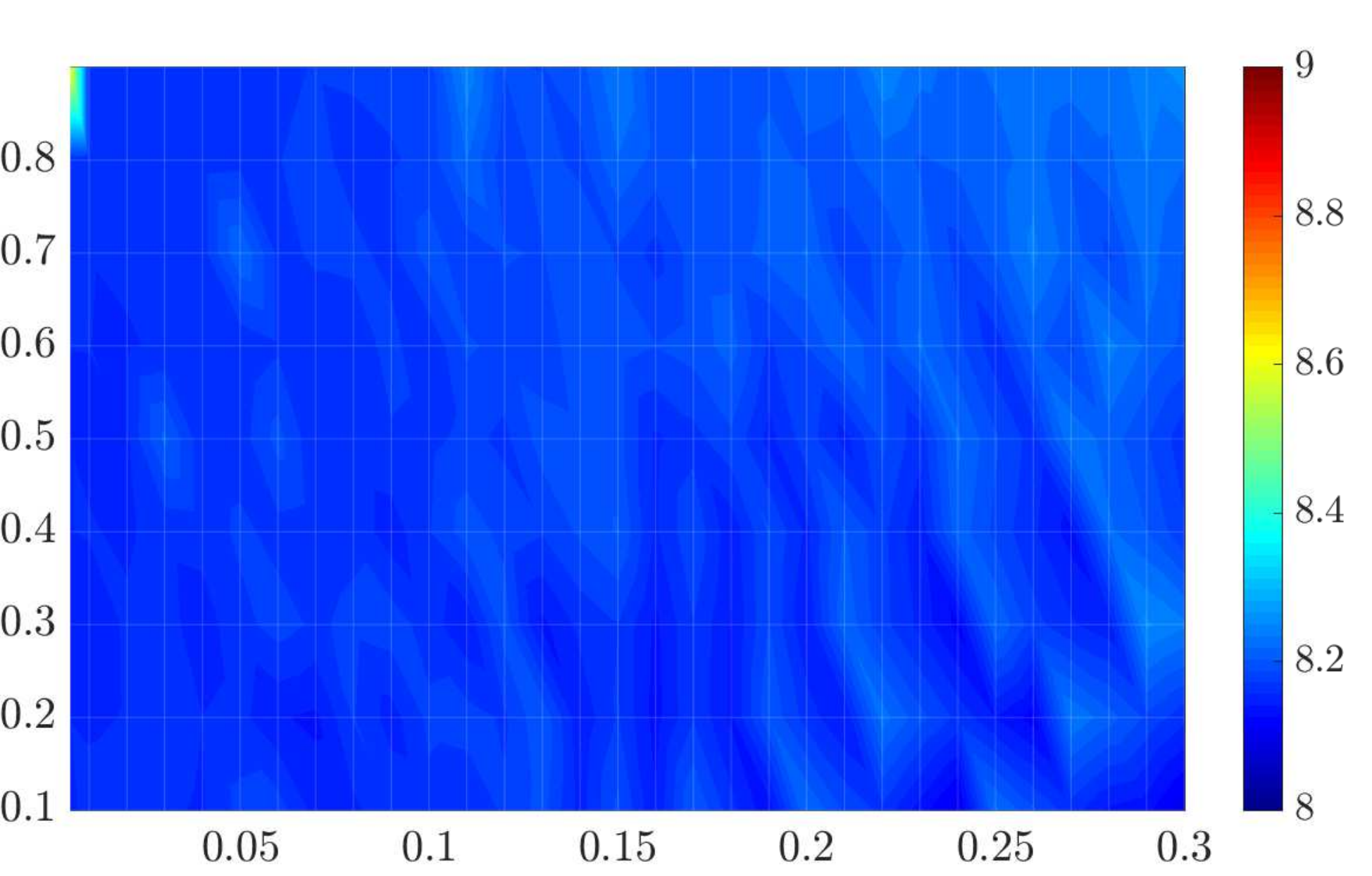}
		\caption*{$\alpha=3\pi/8$} 
	\end{subfigure}
	\begin{subfigure}[t]{0.22\textwidth}
		\centering
		\includegraphics[width=1\columnwidth]{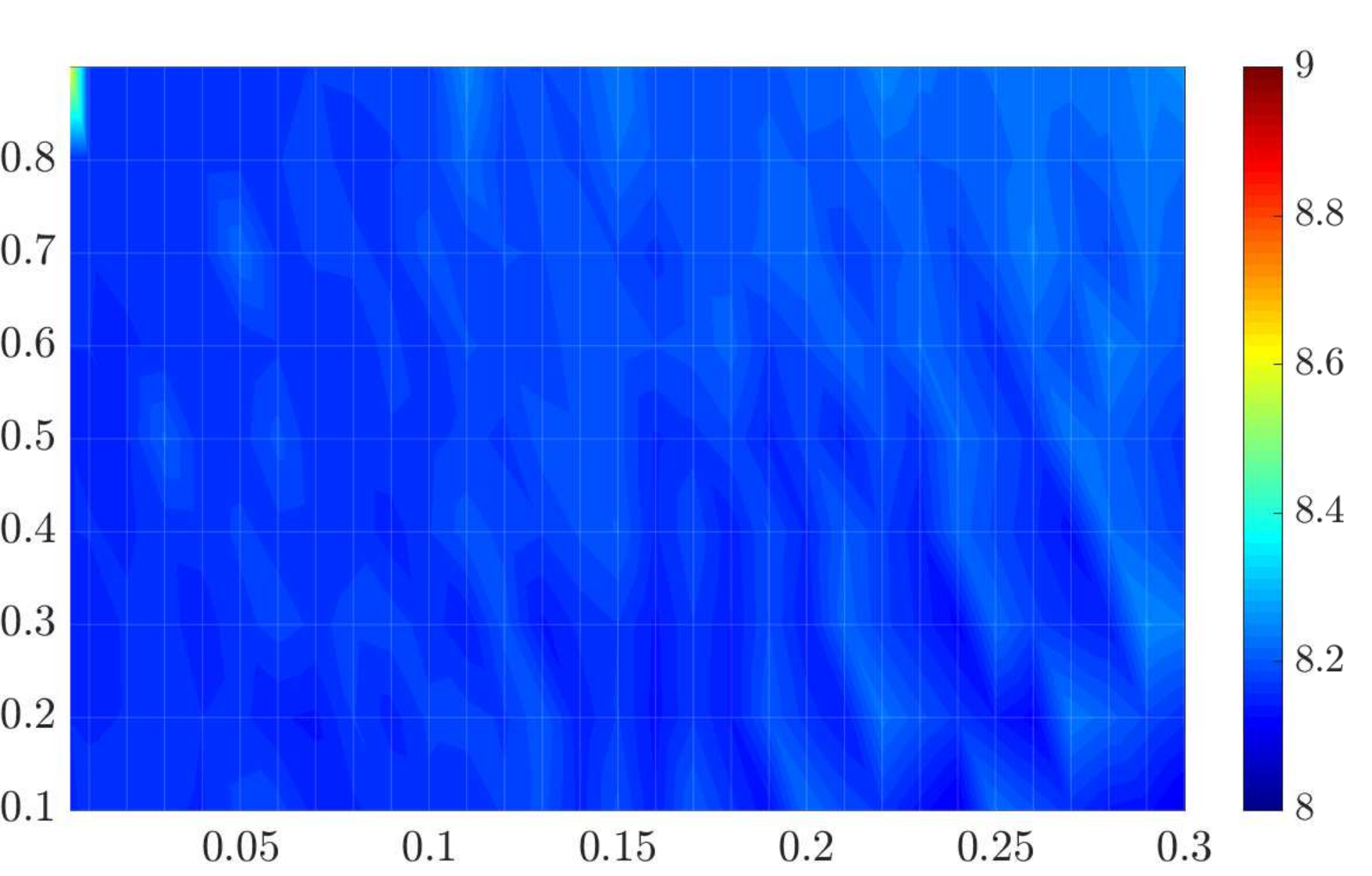}
		\caption*{$\alpha=7\pi/16$} 
		\label{}
	\end{subfigure}
		\caption{$\ell_1$ norm ratio for moving and stationary target as a function of hyper-parameters. We plot the ratio of the background's and moving target's $\ell_1$ norm for different directions, for varying hyper-parameters. The $x$ axis and $y$ axis are as in Figure~\ref{fig:eta_plot_dec}. As we can see, there is little variation and the $\ell_1$ norm values remain similar for all directions.}
\label{fig:L1_ratio}
\end{figure}
    
   We next wish to observe the effect of the different parameters on the nuclear norm. To better understand the effect on the tensor nuclear norm, we look at the ratio between the tensor nuclear norm and the decoupled nuclear norm. Rather than looking at the norms separately, this ratio indicates what is the added benefit of the tensor nuclear norm, factoring out single aperture effects.  
   We plot the ratio of nuclear and decoupled norms for the stationary background as function of the hyper-parameters in Figure~\ref{fig:nucS_ratio_0}. Here we observe strong dependence.
 The tensor nuclear norm is actually smaller than the decoupled norm for small sub-apertures with high overlap.
	\begin{figure}[htbp!]
		\centering
		\includegraphics[width=0.3\columnwidth]{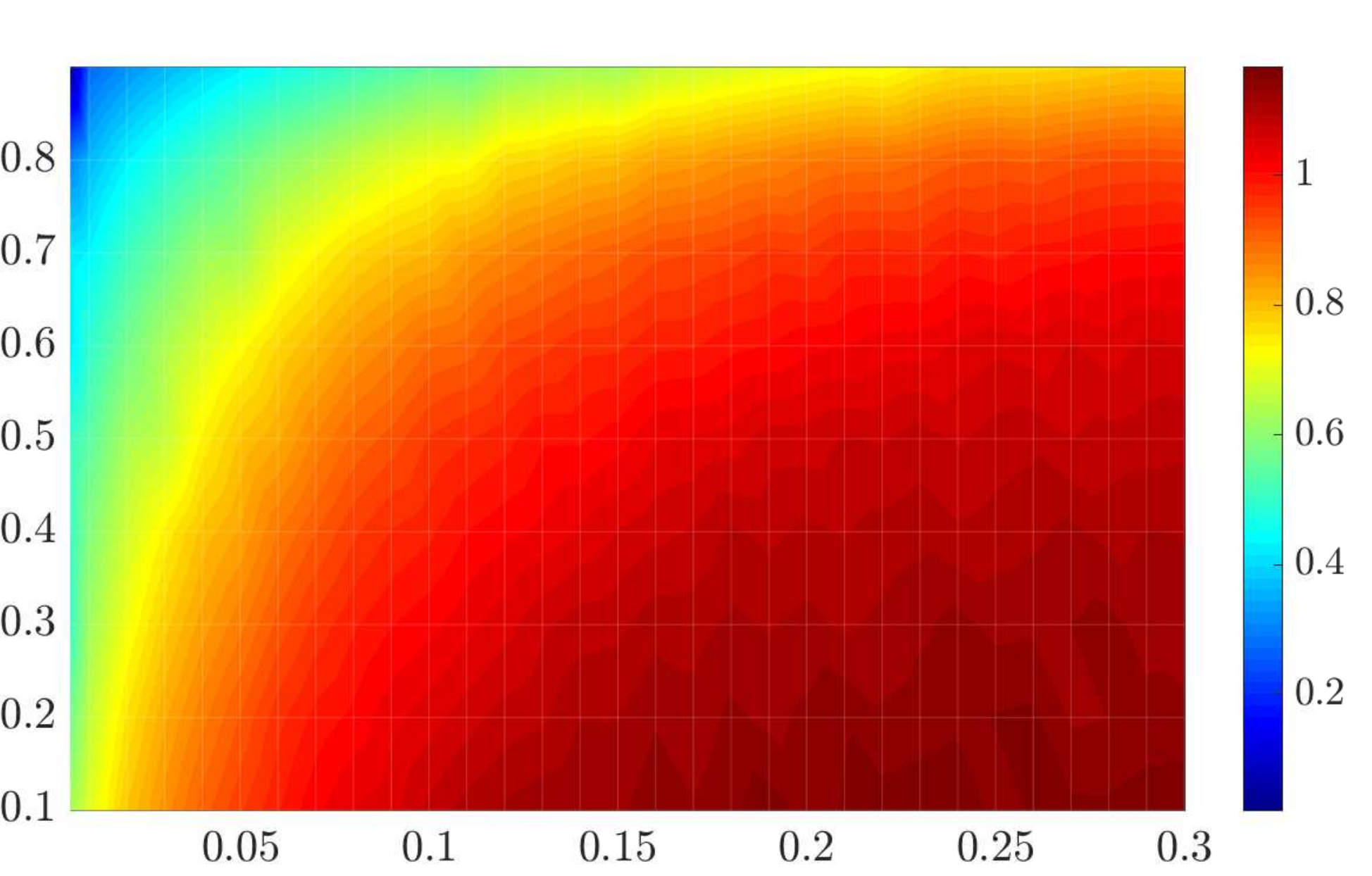}
		\caption{Tensor vs. decoupled nuclear norm for a stationary background, as a function of hyper-parameters. We plot the ratio of tensor and decoupled nuclear norm $\|\mathcal A\|_{*,\mathcal{F}}/\|\mathcal A\|_{*,\mathcal{D}}$ for a stationary background containing 10 targets, for varying hyper-parameters. The $x$ axis and $y$ axis are as in Figure~\ref{fig:eta_plot_dec}. We see that for every sub-aperture size, the ratio decreases with the overlap. An interpretation of this result is given in Section~\ref{sec:interpret} }
		\label{fig:nucS_ratio_0}
	\end{figure}
    
 We next plot the ratio of the nuclear tensor and decoupled norms for the moving target in Figure~\ref{fig:nucV_ratio}. In this case we observe strong dependence on both the direction and the hyper-parameters. The highest ratio is achieved for small sub-apertures with low overlap, but the value, and rate of variation show strong angular dependence. 
	\begin{figure}[htbp!]
	\centering
	\begin{subfigure}[t]{0.22\textwidth}
		\includegraphics[width=1\columnwidth]{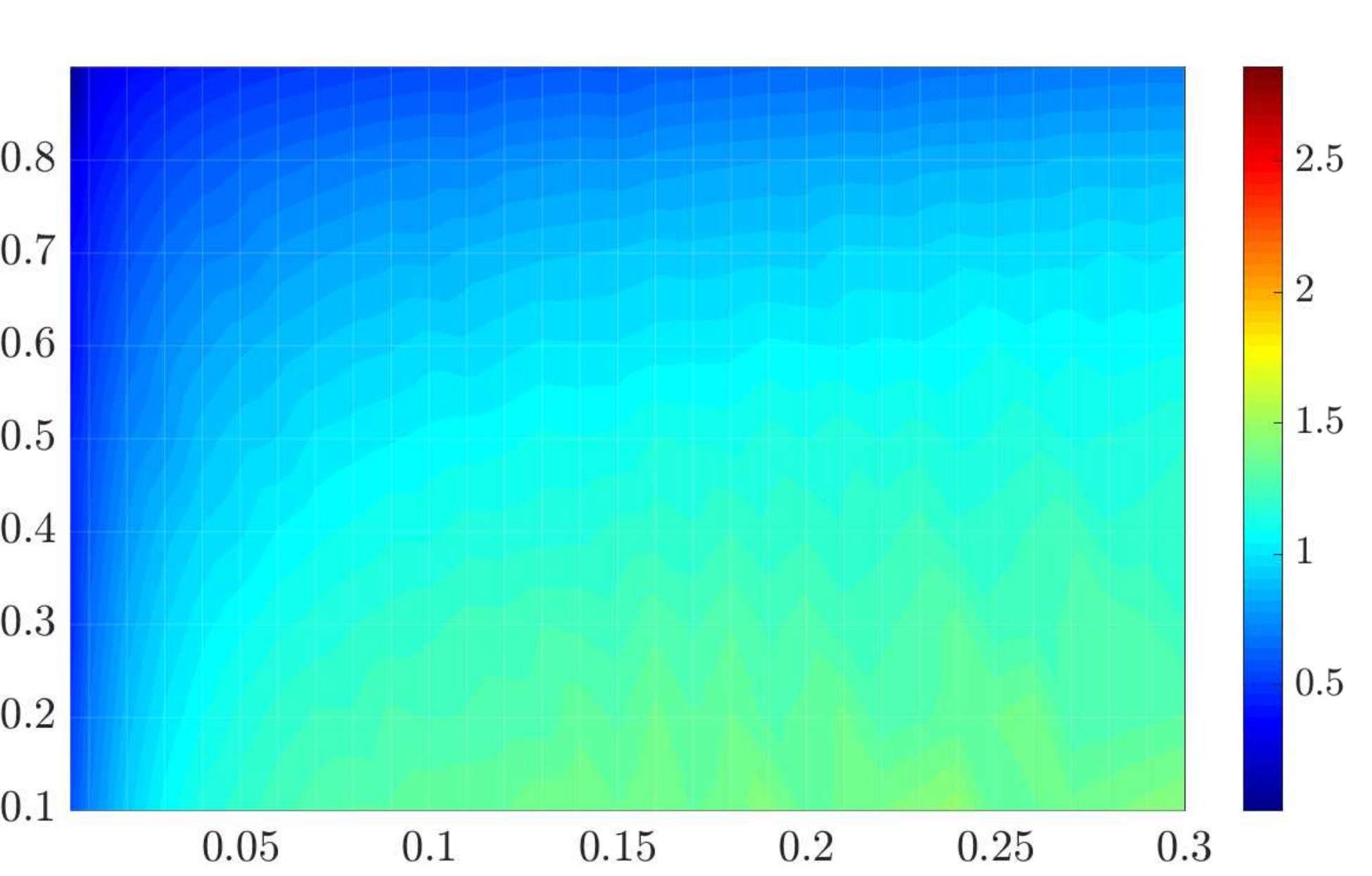}
		\caption*{$\alpha=0$} 
		\label{}
	\end{subfigure}
	\begin{subfigure}[t]{0.22\textwidth}
		\centering
		\includegraphics[width=1\columnwidth]{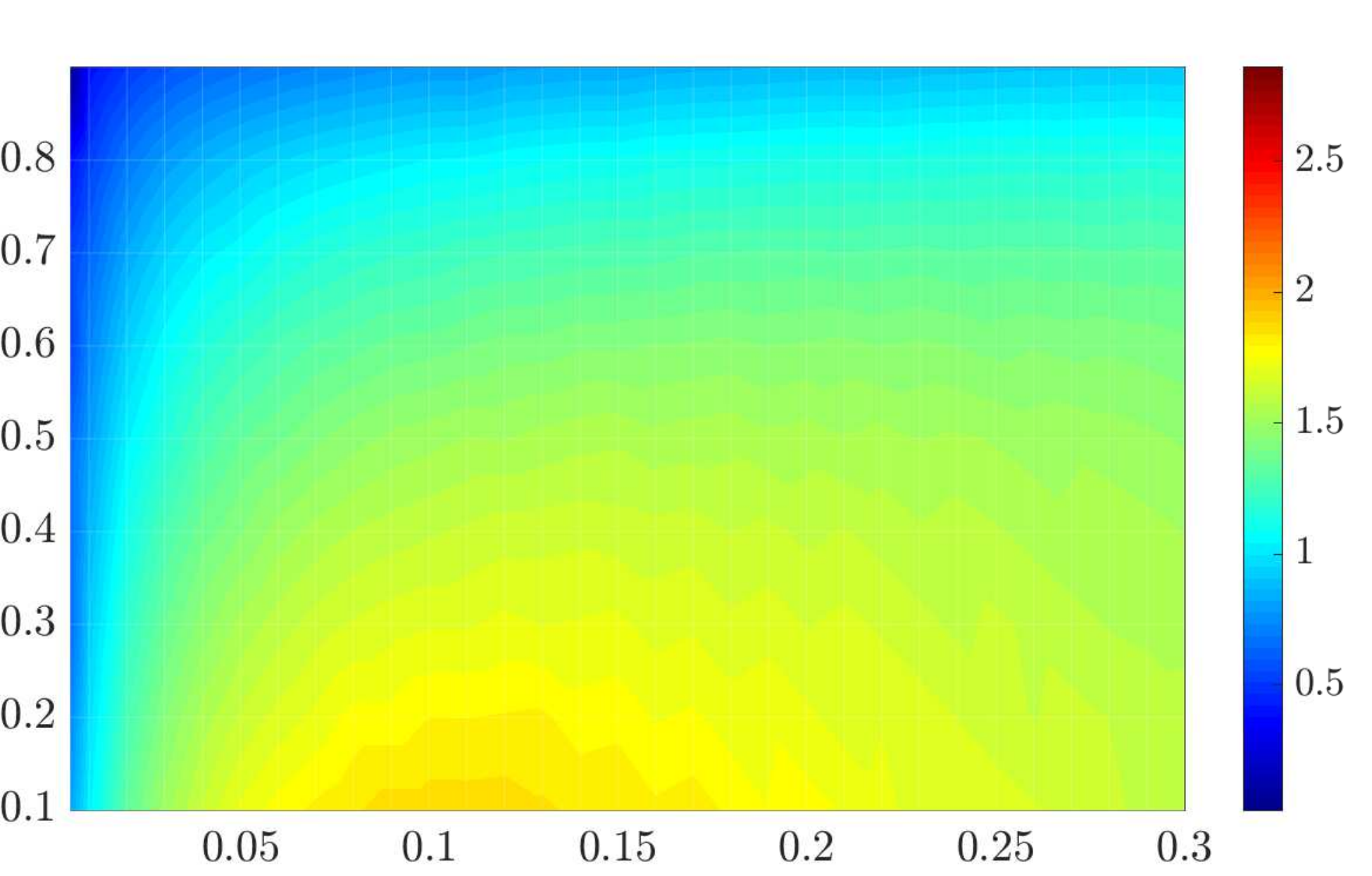}
		\caption*{$\alpha=\pi/16$} 
		\label{}
	\end{subfigure}
	\begin{subfigure}[t]{0.22\textwidth}
		\includegraphics[width=1\columnwidth]{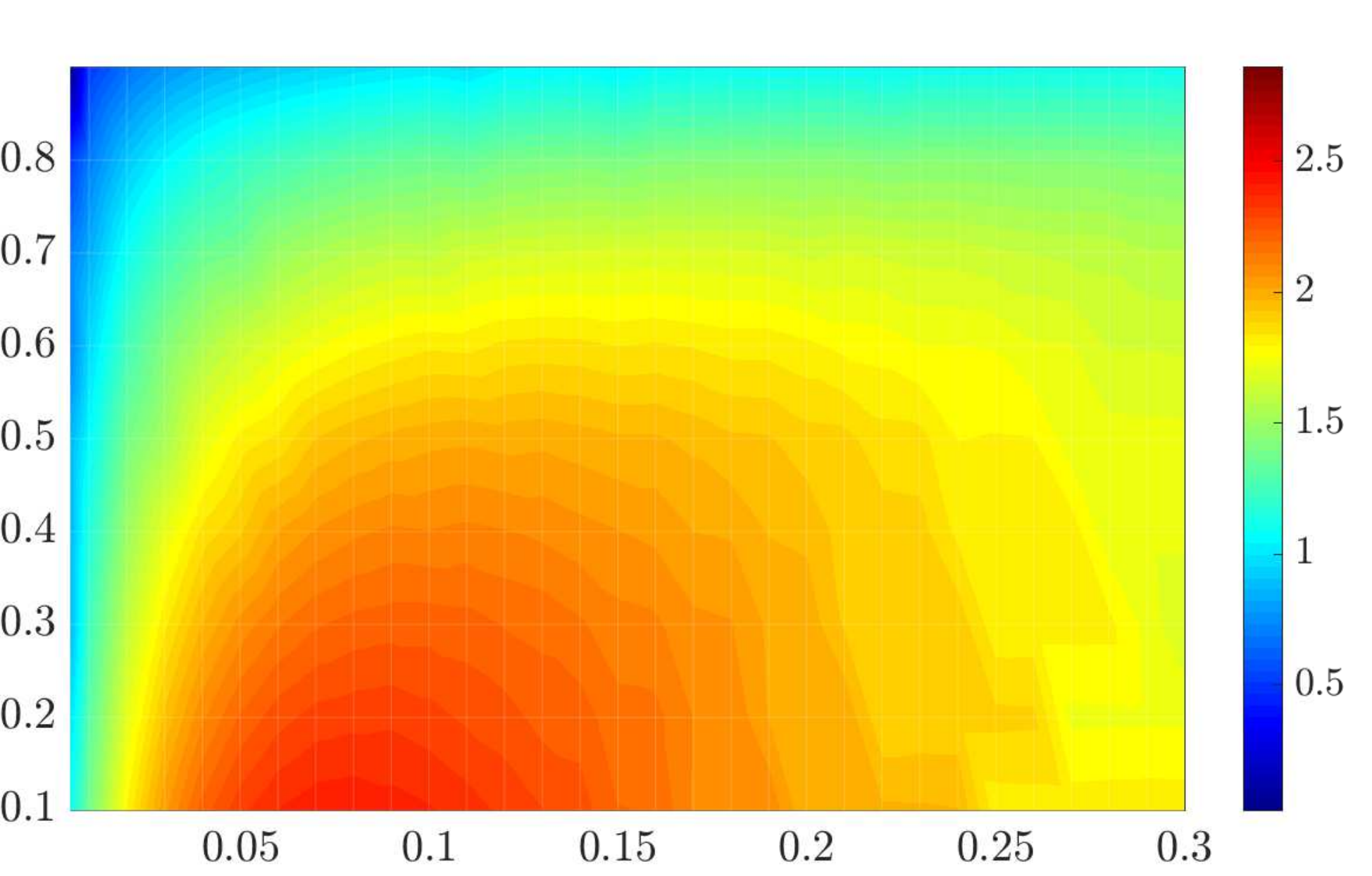}
		\caption*{$\alpha=\pi/8$} 
		\label{}
	\end{subfigure} 
		\begin{subfigure}[t]{0.22\textwidth}
		\includegraphics[width=1\columnwidth]{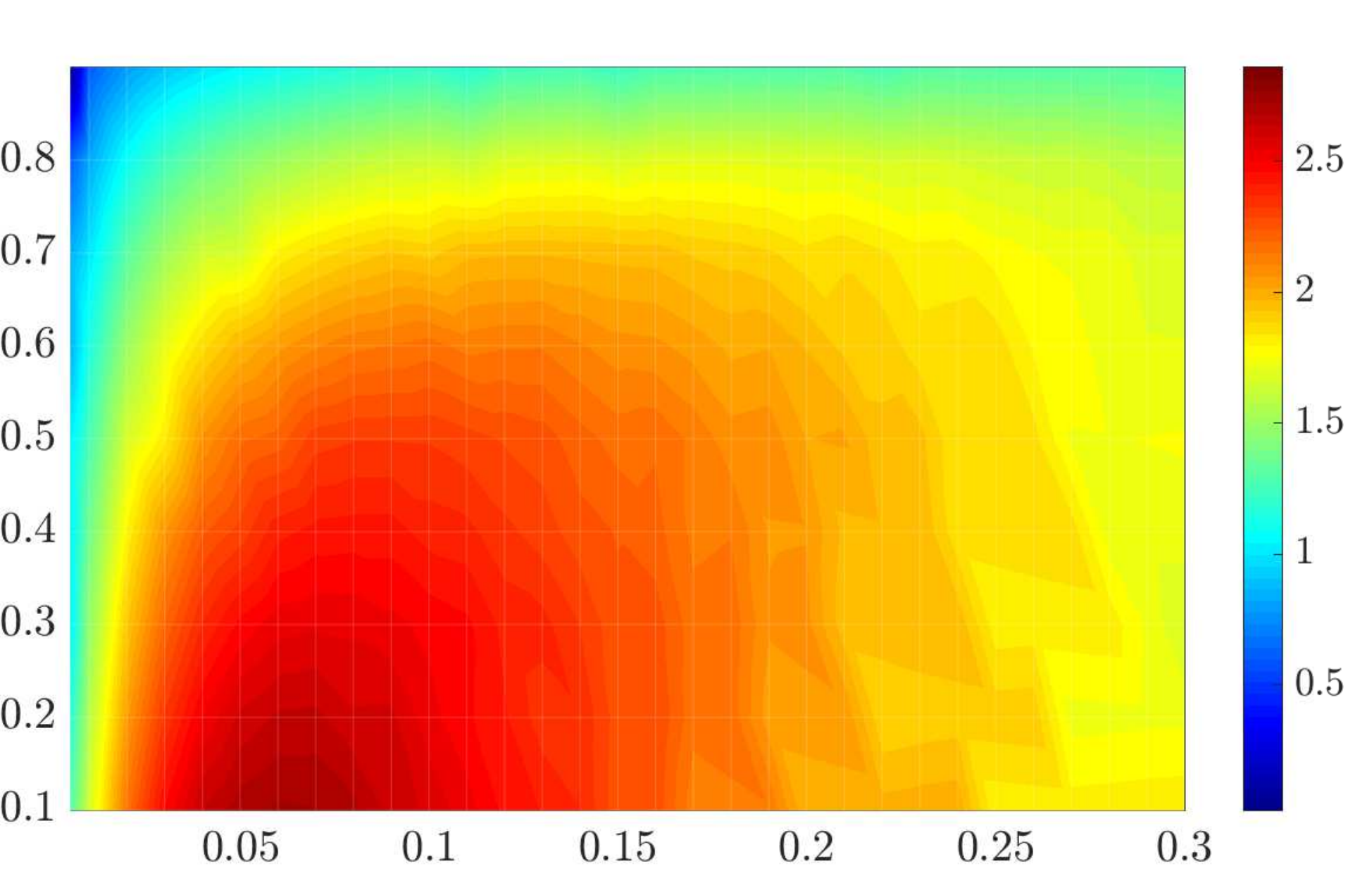}
		\caption*{$\alpha=3\pi/16$} 
		\label{}
	\end{subfigure}
	\begin{subfigure}[t]{0.22\textwidth}
		\centering
		\includegraphics[width=1\columnwidth]{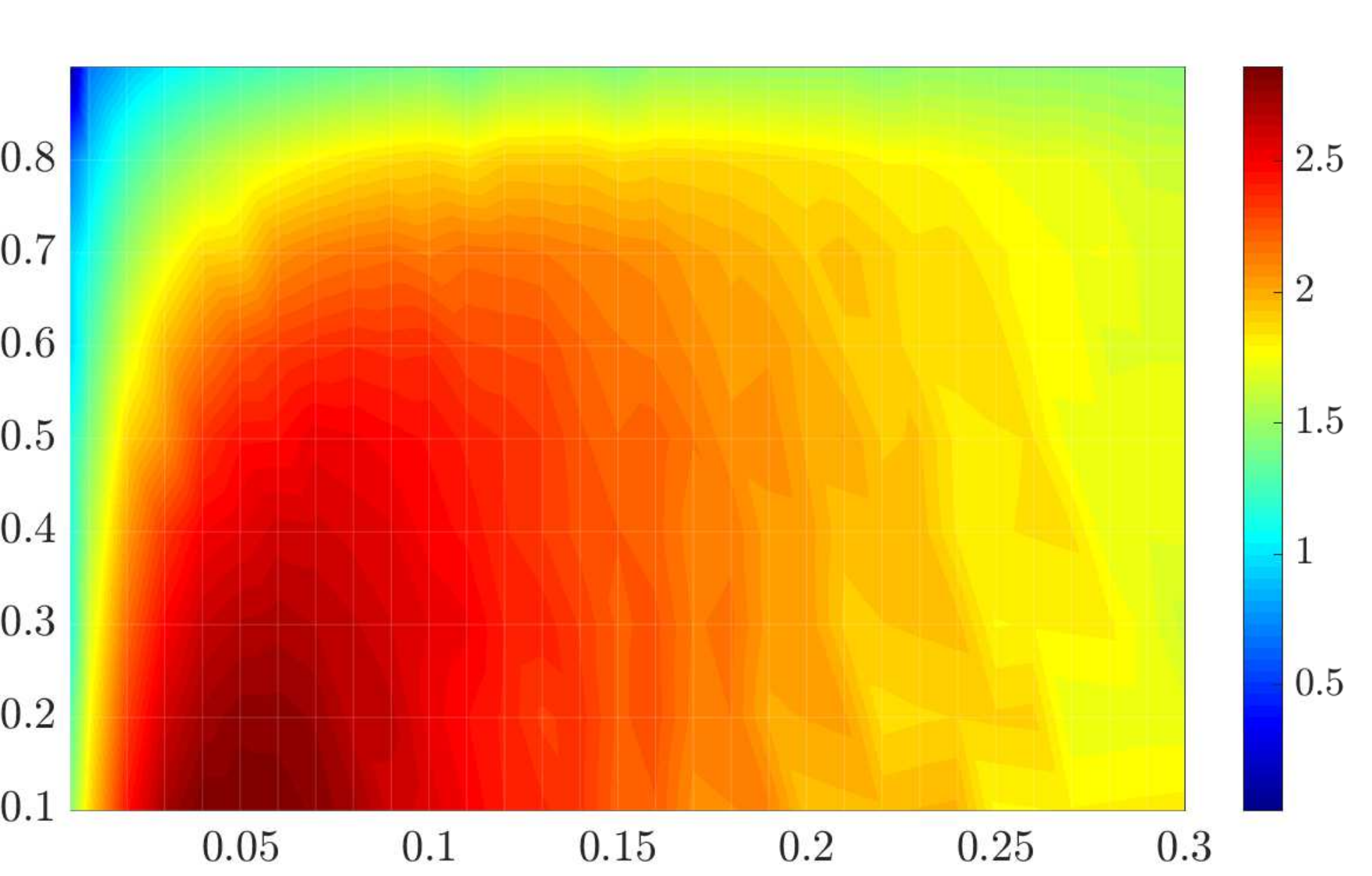}
		\caption*{$\alpha=\pi/4$} 
		\label{}
	\end{subfigure}
	\begin{subfigure}[t]{0.22\textwidth}
		\includegraphics[width=1\columnwidth]{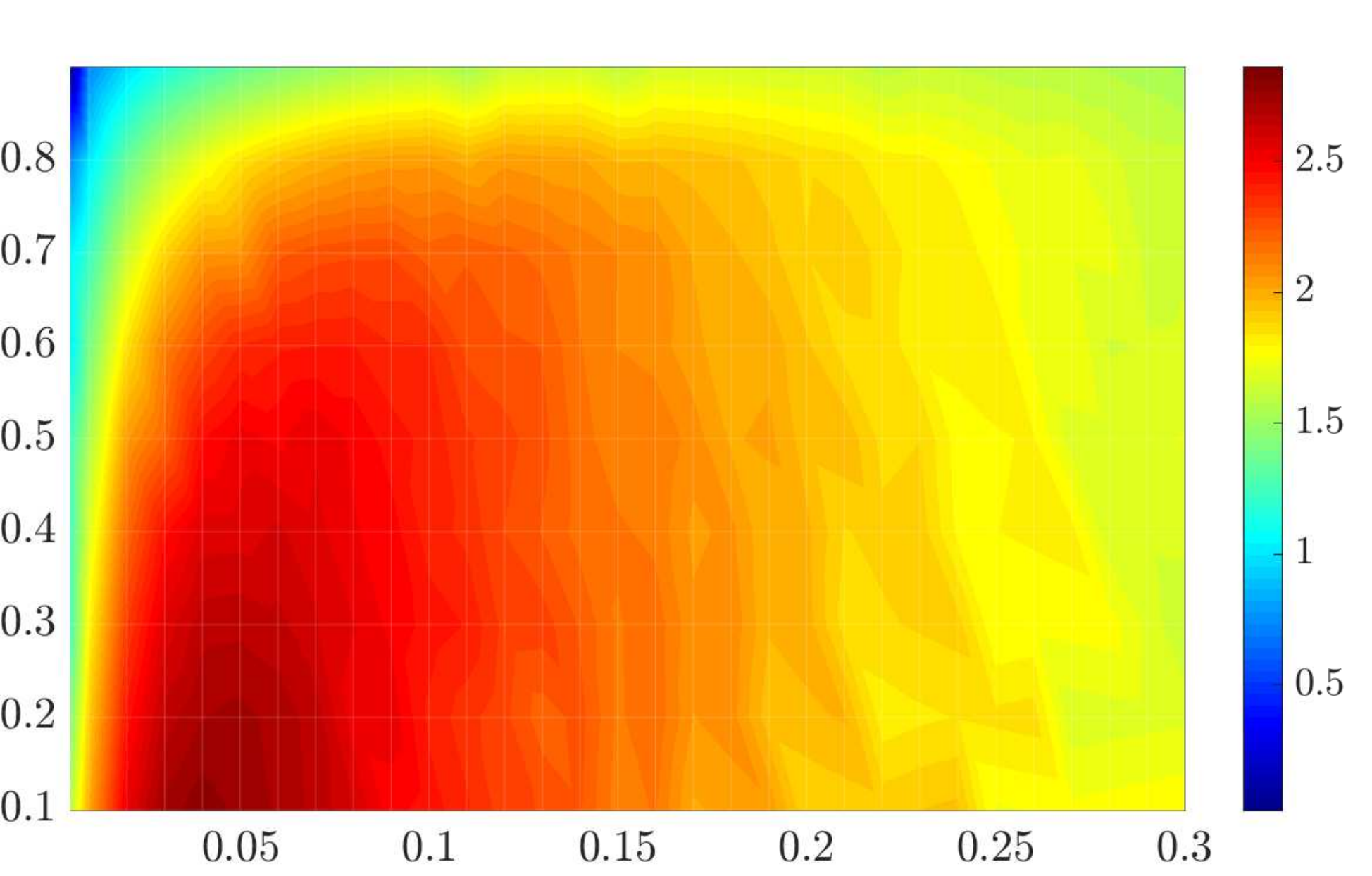}
		\caption*{$\alpha=5\pi/16$} 
		\label{}
	\end{subfigure} 
		\begin{subfigure}[t]{0.22\textwidth}
		\includegraphics[width=1\columnwidth]{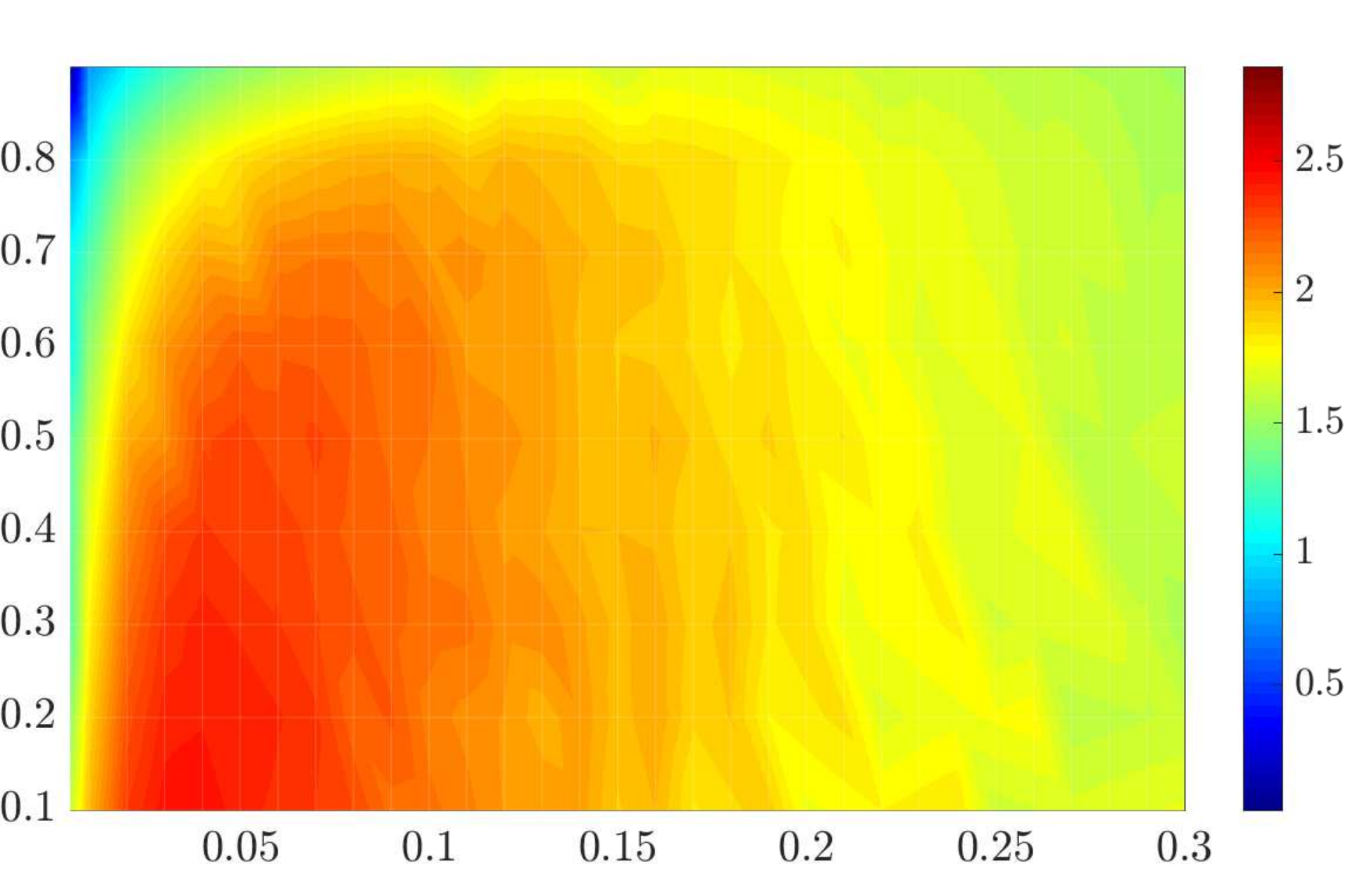}
		\caption*{$\alpha=3\pi/8$} 
		\end{subfigure}
	\begin{subfigure}[t]{0.22\textwidth}
		\centering
		\includegraphics[width=1\columnwidth]{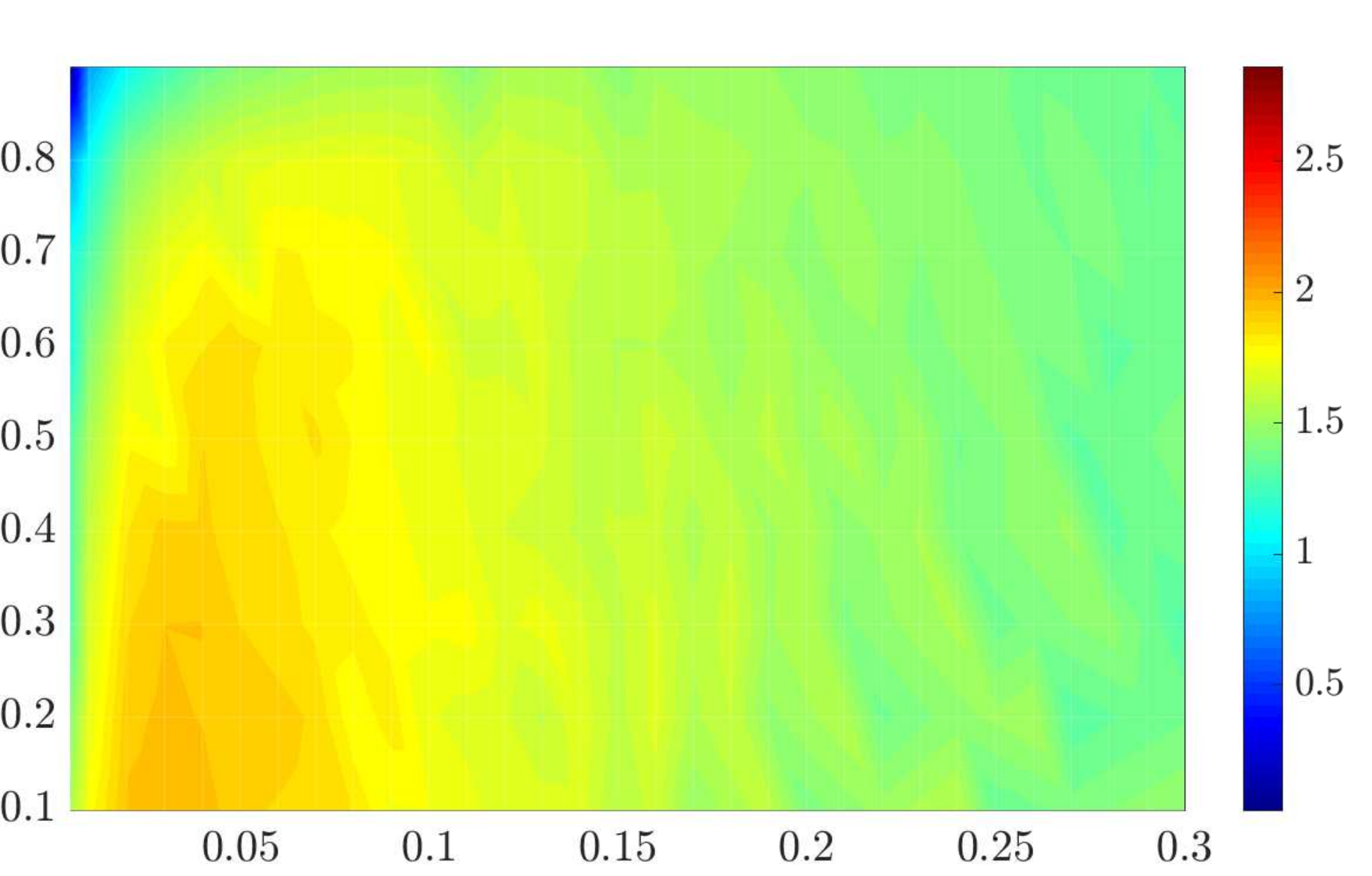}
		\caption*{$\alpha=7\pi/16$} 
		\label{}
	\end{subfigure}
				\caption{Tensor vs. decoupled nuclear norm for a moving target, as a function of hyper-parameters. We plot the ratio of tensor and decoupled nuclear norm $\|\mathcal A\|_{*,\mathcal{F}}/\|\mathcal A\|_{*,\mathcal{D}}$ for a moving target for varying hyper-parameters. The $x$ axis and $y$ axis are as in Figure~\ref{fig:eta_plot_dec}. We see that the highest gain is achieved around $\pi/4$, with a low overlap. Smaller angles favor larger sub-apertures while larger angles favor smaller sub-apertures. An interpretation of these results is given in Section~\ref{sec:interpret}. }
	\label{fig:nucV_ratio}
	\end{figure}

These results present distinct patterns and raise the following questions: 
\begin{itemize} 
	\item[Q1.] Why, as illustrated in Figure~\ref{fig:nucS_ratio_0} high overlap decreases the ratio of tensor to decoupled nuclear norm for the stationary target? 
   	\item[Q2.] What affects the differences in the nuclear norm between moving vs. stationary targets?
	\item[Q3.] Why, as illustrated in Figure~\ref{fig:eta_plot}, is high $\eta_{\max}/\eta_{\min}$ ratio, least sensitive to hyper-parameters, achieved for $\alpha=\pi/4$?  
\end{itemize}
To answer these questions we need to explain how the values of the hyper-parameters and the angle $\alpha$ affect the tensor nuclear norm. This requires a better understanding of the tensor data structure for stationary and moving targets. 
  	
To get a better insight on the role of $\alpha$, we plot in Figure~\ref{fig:ex_demo} the data traces for a moving target for $\alpha=0,\ \pi/4$ and $\pi/2$.
 The sub-aperture size and the overlap are fixed to $0.1$.  In the left column of Figure~\ref{fig:ex_demo} we plot the raw time-domain data for five sub-apertures while in the second column we plot the data after performing the Fourier transform along the sub-aperture index. In the third column we show the row inner product matrix $\hat{A}^{(k)H}\hat{A}^{(k)}$ for a single panel after Fourier.  

Our observations are the following: In the raw data, we notice very little variation in slopes for $\alpha=0$ and $\alpha=\pi/4$. For $\alpha=\pi/2$ variation is noticeable. After performing Fourier, for $\alpha=0$ there is little variation over the entire aperture, i.e., the slope is not changing rapidly between sub-apertures. For $\alpha=\pi/4$ the column support is similar to $\alpha=0$ while for $\alpha=\pi/2$ variation over the total aperture is apparent. The column range is also a lot smaller in this case. The third column of Figure~\ref{fig:ex_demo} illustrates that the orthogonality between columns of matrix $\hat{A}^{(k)}$ increases as the angle $\alpha$ increases. Indeed, this is manifested by the suppression of the diagonal elements in $\hat{A}^{(k)H}\hat{A}^{(k)}$. We also see that the column range is suppressed for $\pi/2$ (bottom right plot).		

\begin{figure}[htbp!]
	\centering
	{\vskip 17mm\hskip -18.2cm$\alpha=0$\vskip -17mm}
	\begin{subfigure}[t]{0.325\textwidth}
		\centering
		\includegraphics[width=1\columnwidth]{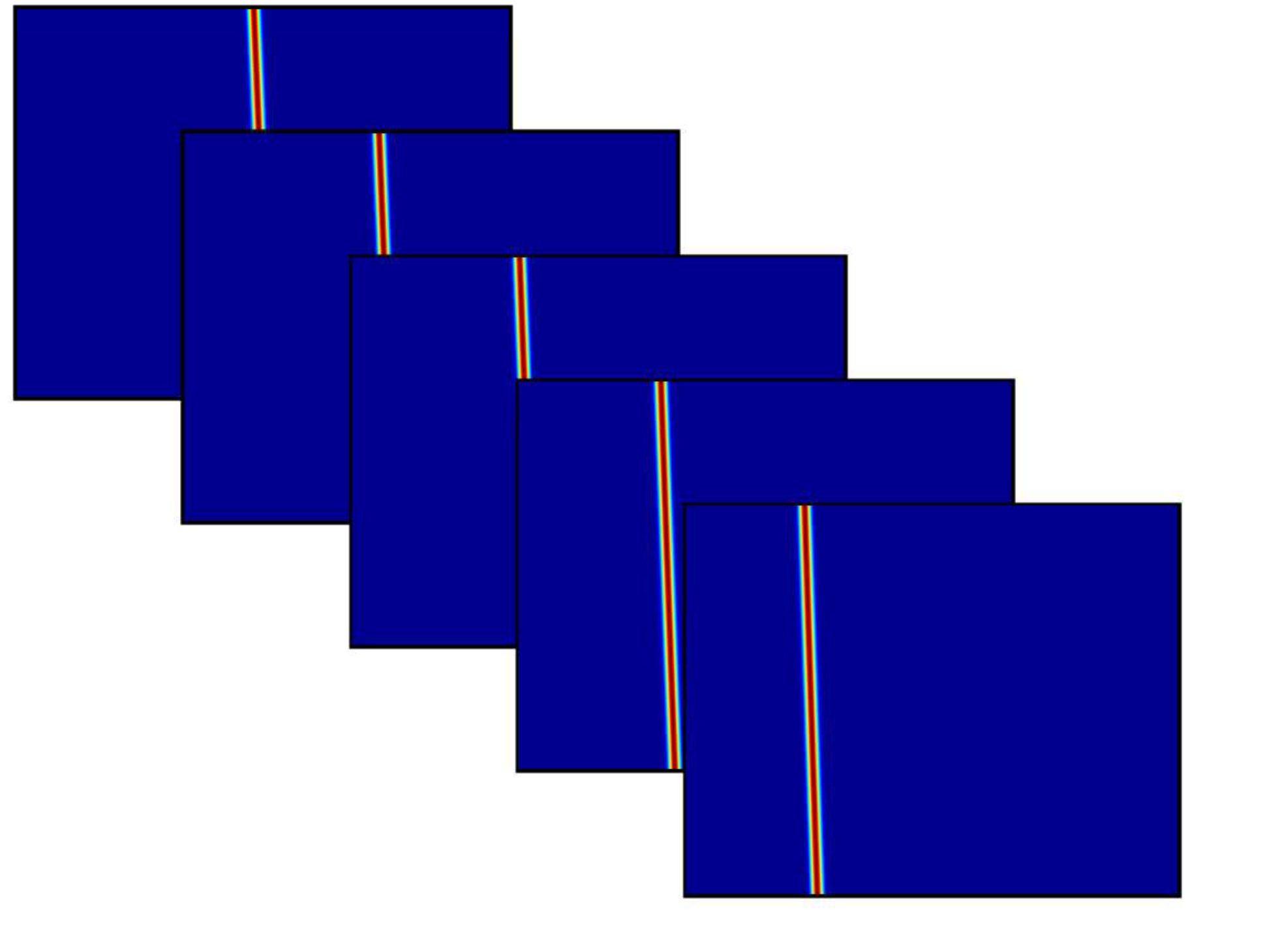}
		\label{fig:DT_0}
	\end{subfigure}
	\begin{subfigure}[t]{0.325\textwidth}
		\centering
		\includegraphics[width=1\columnwidth]{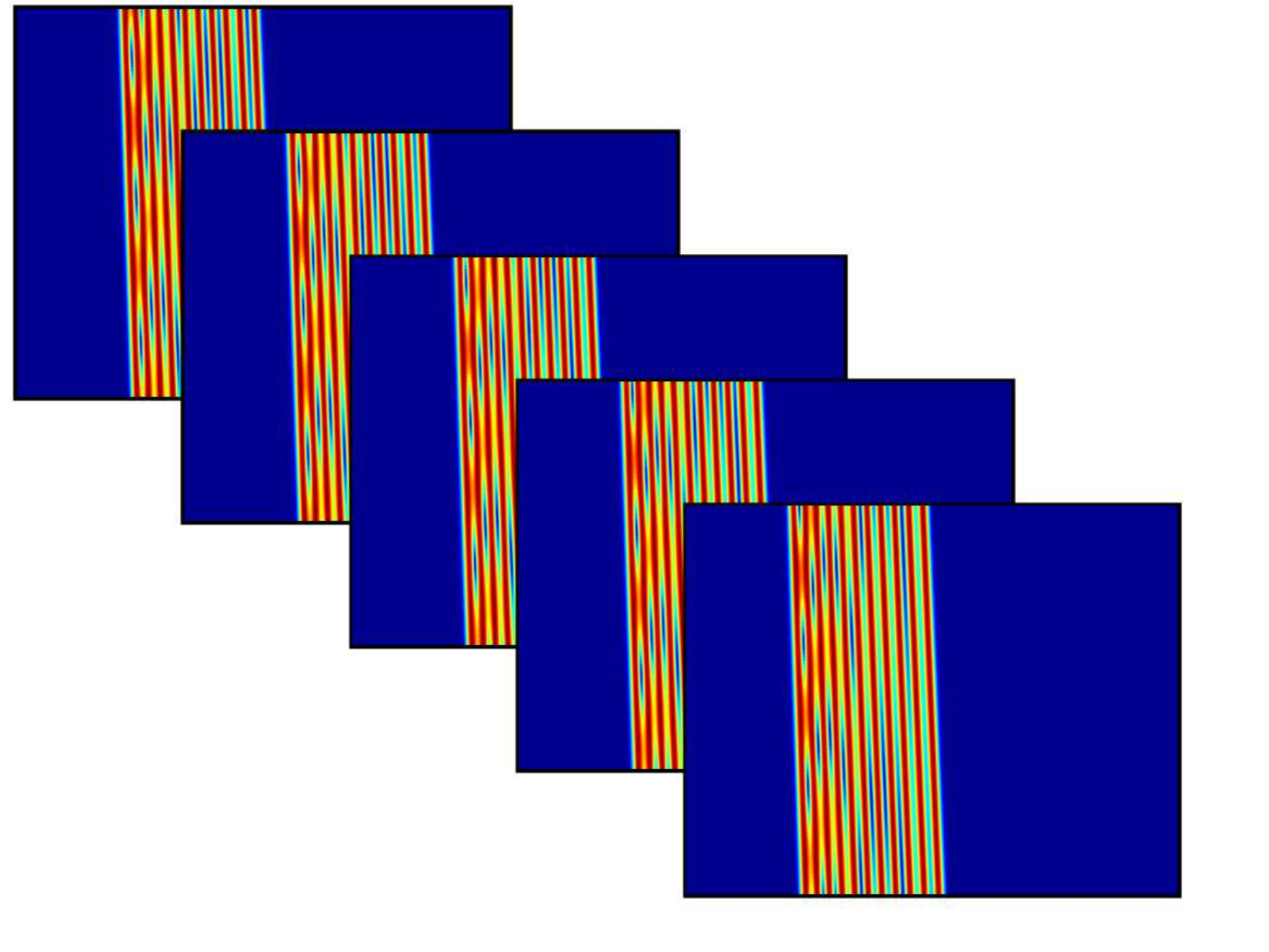}
		\label{fig:DTf_0}
	\end{subfigure}
	\begin{subfigure}[t]{0.325\textwidth}
		\centering
		\includegraphics[width=1\columnwidth]{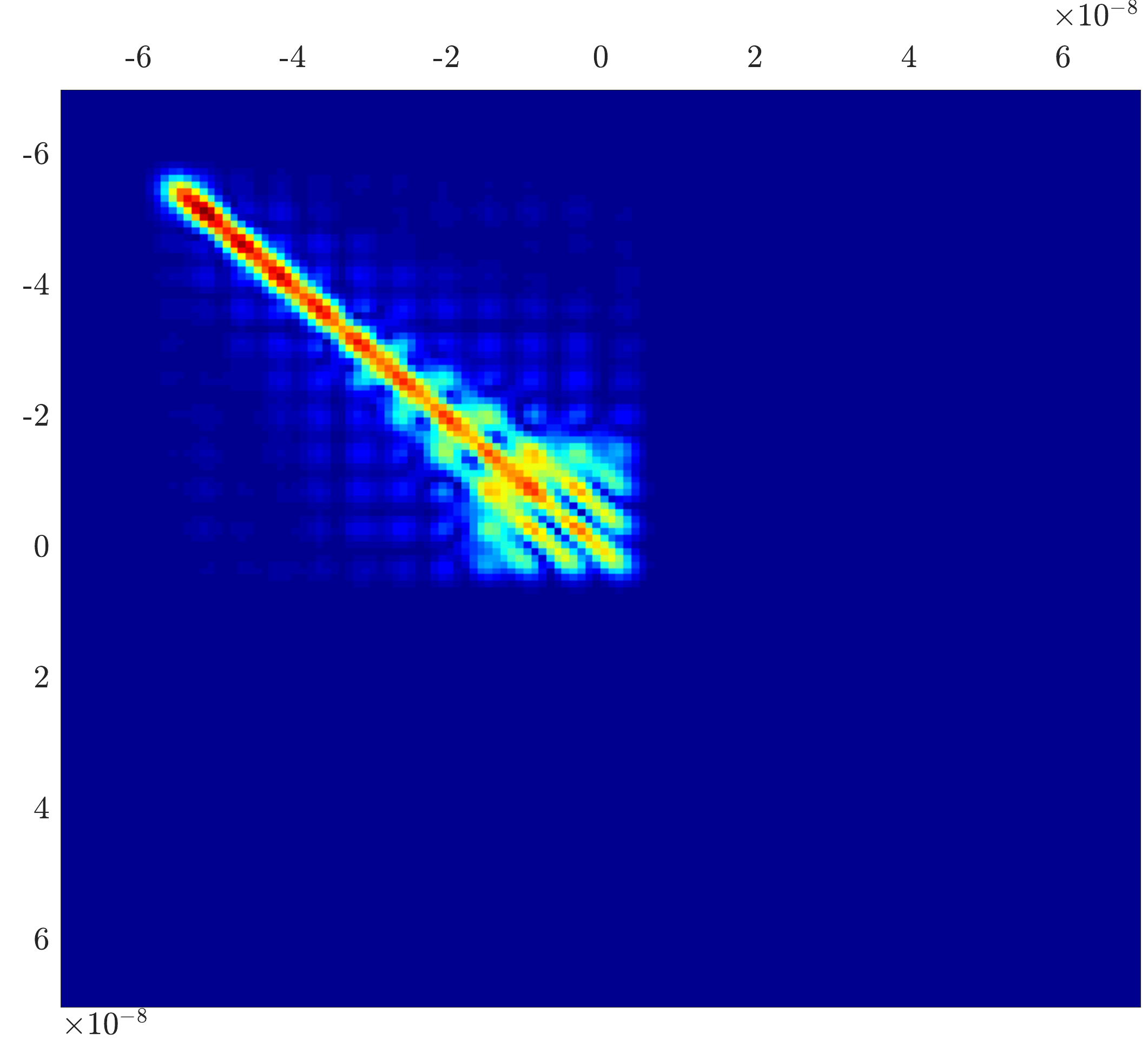}
		\label{fig:DDT_0}
	\end{subfigure}
	{\vskip 17mm\hskip -18.2cm$\alpha=\pi/4$\vskip -17mm}
	\begin{subfigure}[t]{0.325\textwidth}
		\centering
		\includegraphics[width=1\columnwidth]{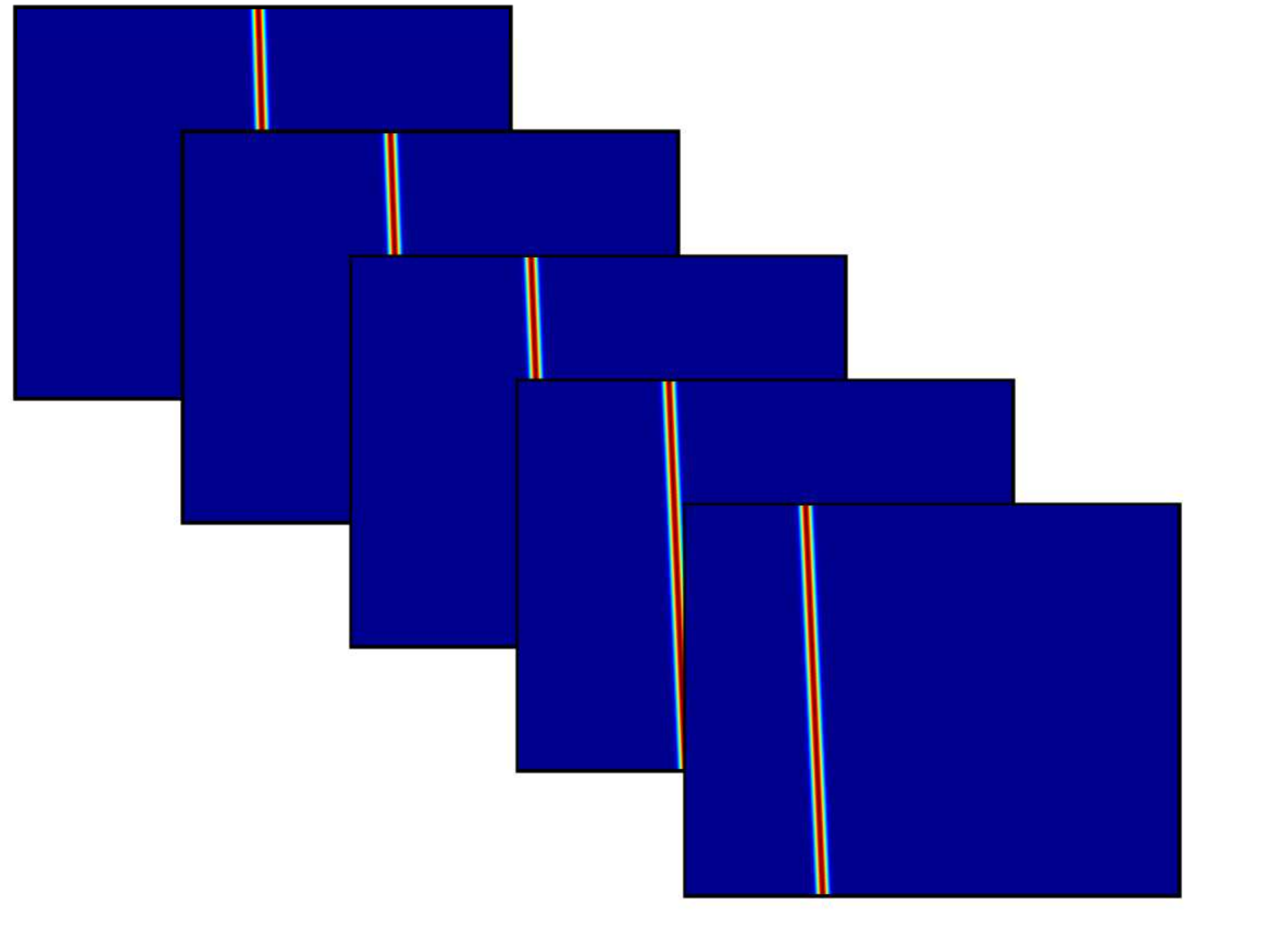}
		\label{fig:DT_pi_4}
	\end{subfigure}
	\begin{subfigure}[t]{0.325\textwidth}
		\centering
		\includegraphics[width=1\columnwidth]{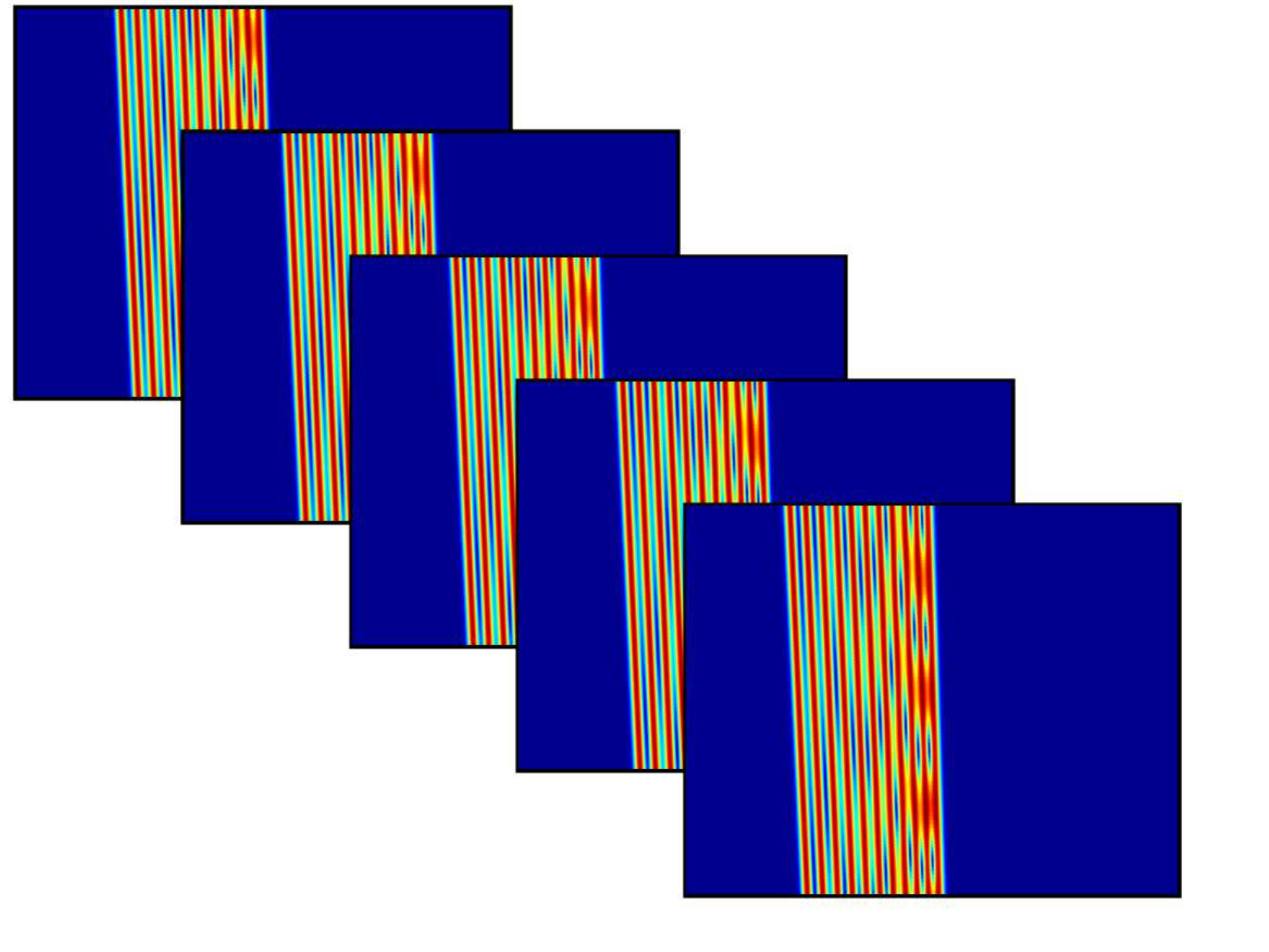}
		\label{fig:DTf_pi_4}
	\end{subfigure}
	\begin{subfigure}[t]{0.325\textwidth}
		\centering
		\includegraphics[width=1\columnwidth]{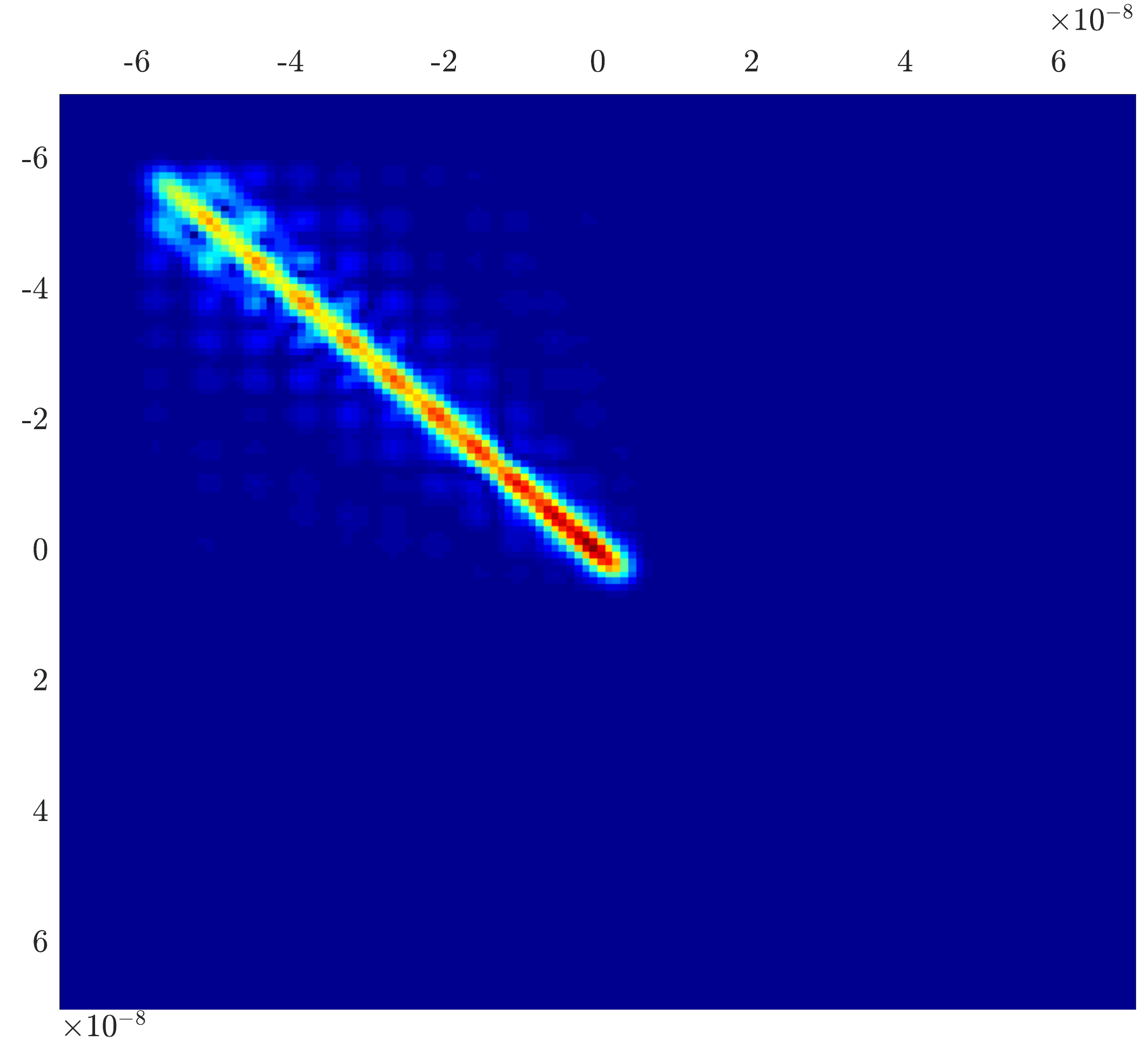}
		\label{fig:DDT_pi_4}
	\end{subfigure}
	{\vskip 17mm\hskip -18.2cm$\alpha=\pi/2$\vskip -17mm}
	\begin{subfigure}[t]{0.325\textwidth}
		\includegraphics[width=1\columnwidth]{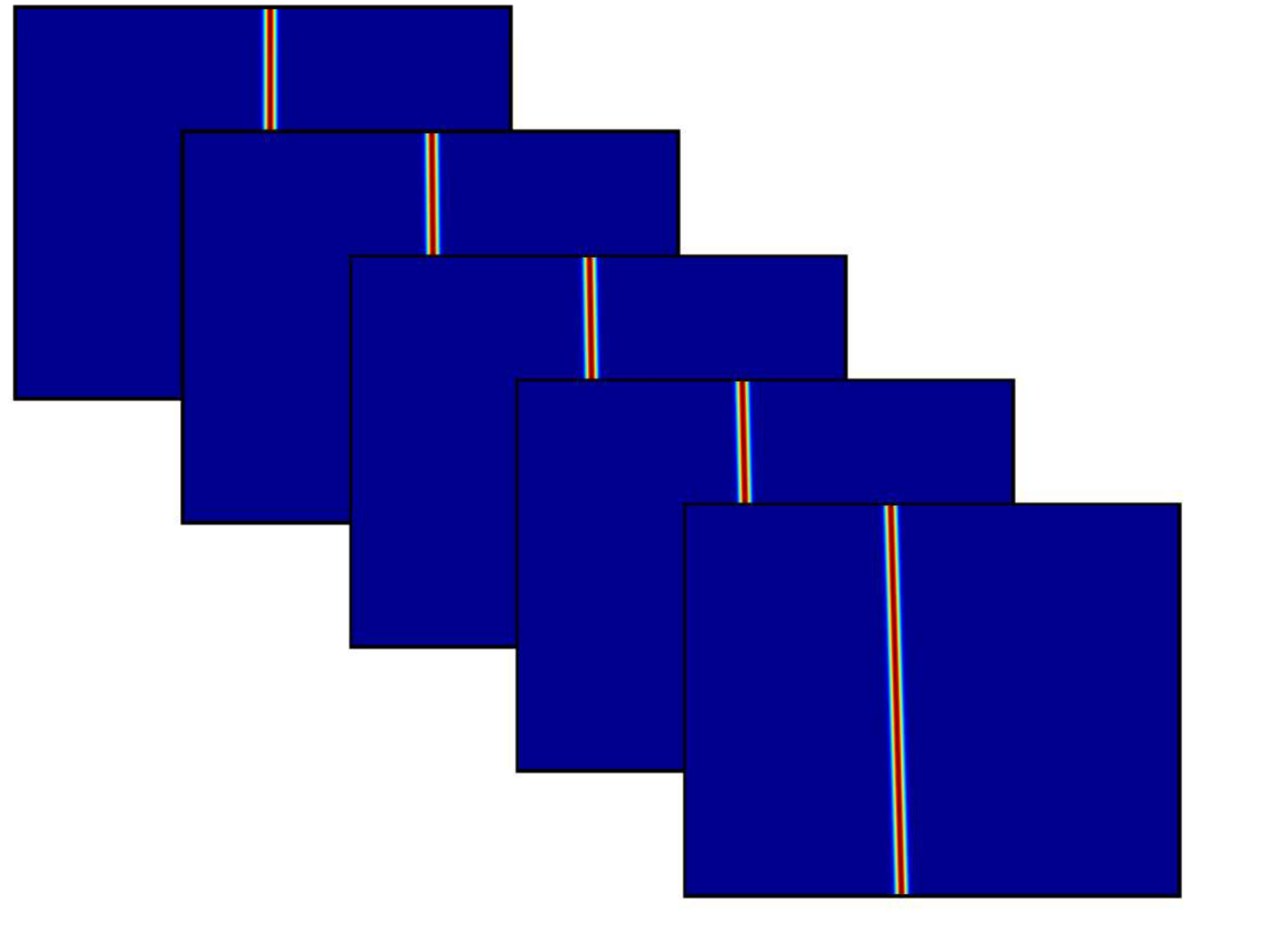}
		\caption{} 
		\label{fig:DT_pi_2}
	\end{subfigure}
	\begin{subfigure}[t]{0.325\textwidth}
		\includegraphics[width=1\columnwidth]{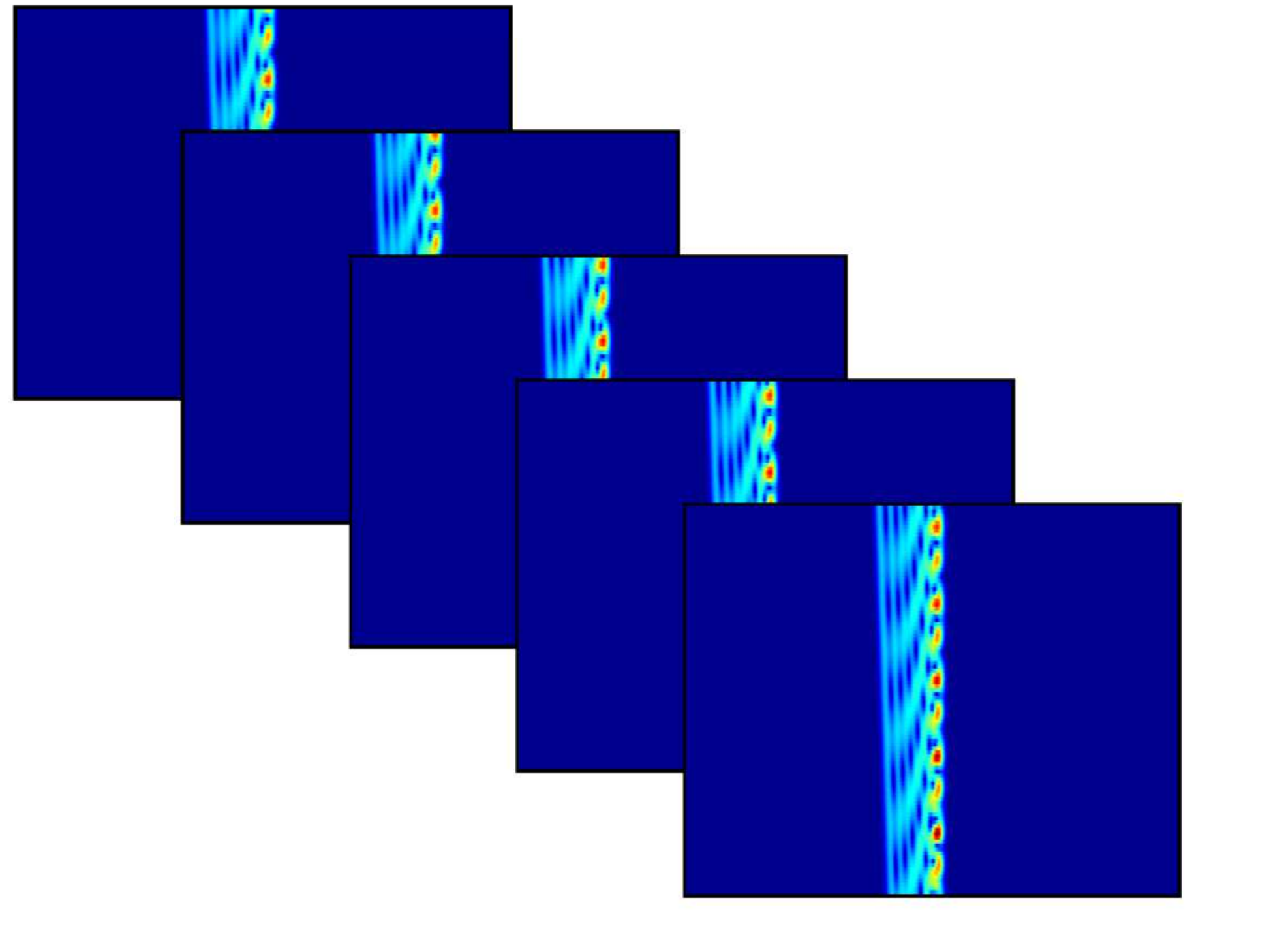}
		\caption{} 
		\label{fig:DTf_pi_2}
	\end{subfigure}
	\begin{subfigure}[t]{0.325\textwidth}
		\centering
		\includegraphics[width=1\columnwidth]{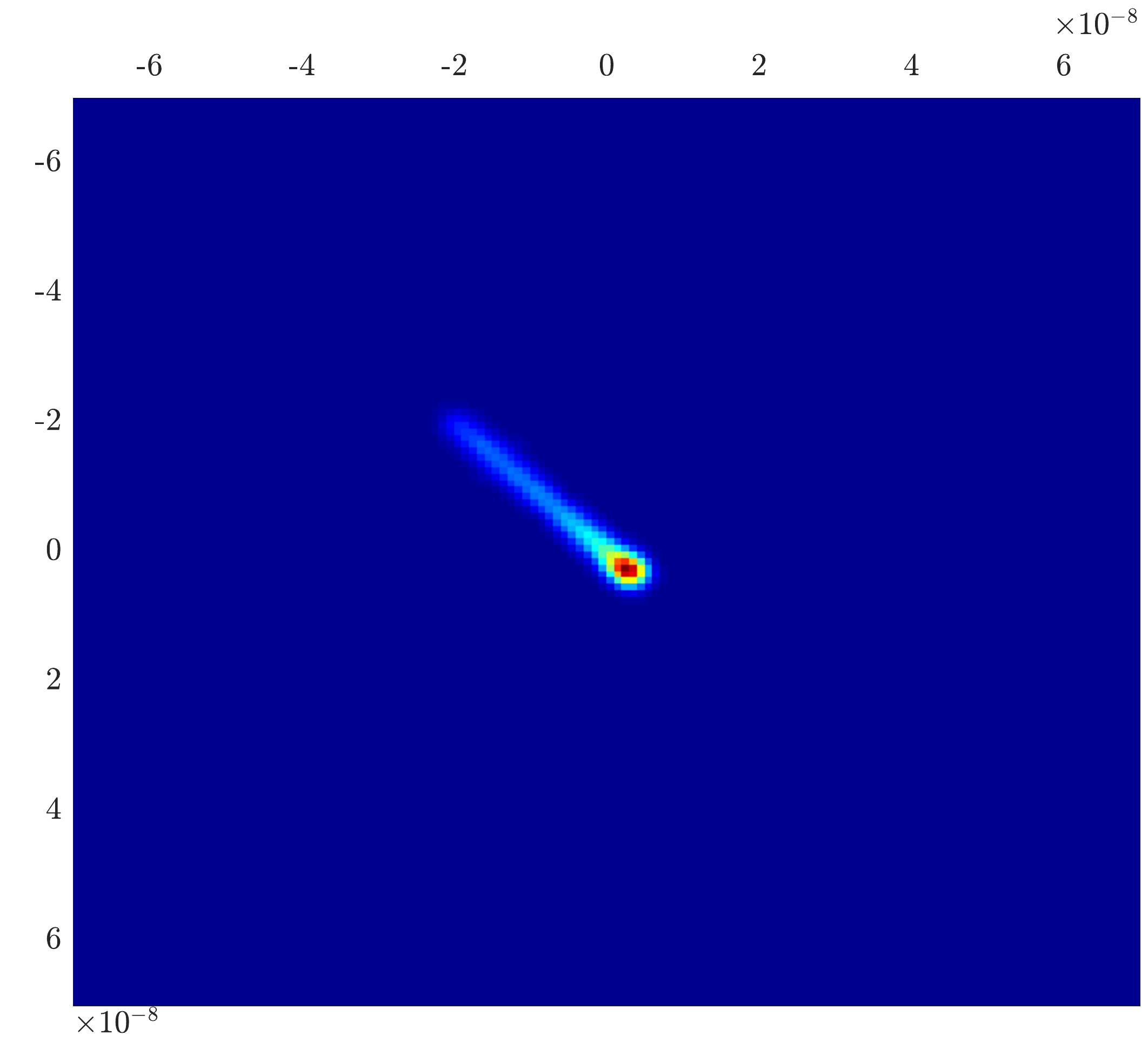}
		\caption{}
		\label{fig:DDT_pi_2}
	\end{subfigure}
	\caption{
		Behavior of SAR data of a moving target for different directions;
		$(a)$ Raw data on five sub-apertures for a moving target. Notice very little variation in slopes for $\alpha=0$ and $\alpha=\pi/4$. For $\alpha=\pi/2$ variation is noticeable even before performing Fourier; $(b)$ Data after performing Fourier with respect to the sub-aperture index. For $\alpha=0$ the slope is not changing rapidly between sub-apertures. For $\alpha=\pi/4$ the column support is similar to $\alpha=0$. For $\alpha=\pi/2$ variation over the total aperture is apparent. The column range is also a lot smaller; $(c)$ $\hat{A}^{(k)H}\hat{A}^{(k)}$ for a single panel after Fourier. For $\alpha=0$, the column range is large, but there are a lot off-block diagonal terms (the inner product of different sub-apertures). For $\alpha=\pi/4$ the column range is similar to $\alpha=0$, but off-block diagonal terms are suppressed, suggesting the sub-apertures are closer to being orthogonal. For $\alpha=\pi/2$, off-block diagonal terms are greatly suppressed. Notice that the total column range is also reduced. We give qualitative explanation to the effect each of this observations has on the tensor nuclear norm in Section~\ref{sec:interpret}. } 
		\label{fig:ex_demo}
\end{figure}

These empirical observations require further analysis, provided in the following section where we also answer the questions Q1 to Q3.
\subsection{TRPCA analysis for general SAR data}
\label{sec:tnn_analysis}
In this section we first give an abstract result for the upper and lower bounds of the tensor nuclear norm. Then we show that the lower bound is obtained for an ideally stationary background with no variation between ${A}^{(k)}$ for the sub-apertures while the upper bound is attained for an ideal moving target with exact orthogonality in the columns of $\hat{A}^{(k)}$. We close this section by discussing why the Fourier based tensor nuclear norm is optimal for the SAR data separation problem. 

\paragraph{A general result on the bounds of the tensor nuclear norm $\| \cdot\|_{*,\mathcal{F}}$.}
We have the following Proposition. 
\begin{proposition}
	The tensor nuclear norm, as defined in \eqref{eq:tnn}, has the following lower and upper bounds
	\begin{equation}
	\left(\sum_{i=0}^{n_3-1} \|A^{(i)}\|_*^2\right)^{1/2}\le \|\mathcal A\|_{*,\mathcal{F}}\le \sqrt{n_3} \sum\limits_{i=0}^{n_3-1}\|A^{(i)}\|_*
	\label{eq:nuc_fourier_bounds}
	\end{equation}
When all matrices $A^{(i)}$ are the same then
	\begin{equation}
	\|\mathcal{A}\|_{*,\mathcal{F}} = \left(\sum_{i=0}^{n_3-1} \|A^{(i)}\|_*^2\right)^{1/2}
	\label{eq:L_behavior}
	\end{equation}
	while when the columns of the different matrices $A^{(i)}$ are orthogonal then
	\begin{equation}
	\|\mathcal{A}\|_{*\mathcal{F}} = \sqrt{n_3} \sum\limits_{i=0}^{n_3-1}\|A^{(i)}\|_* .
	\label{eq:S_behavior}
	\end{equation}
	\label{prop:prop_1}
\end{proposition}
To prove Proposition~\ref{prop:prop_1} we need the following Theorem and Corollary. 
	\begin{theorem}
		 For matrices $A_1,A_2,\dots,A_k,\quad A_i\in\mathbb{C}^{m\times n_i}$, the following holds for the matrix  
			\begin{equation}
			\pmb{A}=[A_1,A_2,\dots,A_k]\in \mathbb{C}^{m\times N},\quad N=\sum\limits_{i=1}^k n_i
			\label{eq:nuc_norm_mat}
			\end{equation}
			\begin{equation}
			\left(\sum\limits_{i=1}^k\|A_i\|_*^2\right)^{1/2}\le	\|\pmb{A}\|_*\le \sum\limits_{i=1}^k\|A_i\|_*.
			\label{eq:nuc_norm_bounds}
			\end{equation}
			\label{theorem_block}
		\end{theorem}
	\begin{proof}
	The proof is given in Appendix~\ref{app:cor_proof}
	\end{proof}
\begin{corollary}
	For $A_1,\dots,A_k$ in \eqref{eq:nuc_norm_mat}
	\begin{enumerate}[(a)]
		\item If all the matrices are mutually orthogonal  
		\begin{equation}
		\sum\limits_{i=1}^k \text{rank}(A_i)\le m, \quad A_i^HA_j=0 , \quad \forall i\neq j,
		\end{equation}
		the upper bound of \eqref{eq:nuc_norm_bounds} is attained.
		\label{col:cor_1_b}
		\item If $A_i=\beta_iA, \hspace{0.5em} \beta_i\in\mathbb{C}$, then the lower bound of \eqref{eq:nuc_norm_bounds} is attained and
		\begin{equation}
		\|\pmb{A}\|_*=\|\pmb{\beta}\|_2\|A\|_*,\quad \|\pmb{\beta}\|_2=\left(\sum\limits_{i=1}^k|\beta_i|^2\right)^{1/2}.
		\end{equation}	
		\label{col:cor_1_c}
	\end{enumerate}
	\label{col:cor_1}
\end{corollary}
	\begin{proof}
	The proof is given in Appendix~\ref{app:cor_proof}
\end{proof}

We use these general results to prove \eqref{eq:nuc_fourier_bounds}. 
\begin{proof}[Proof of \eqref{eq:nuc_fourier_bounds}]	
Write 
\begin{equation}
\frac{1}{\sqrt{n_3}} \begin{pmatrix} A^{(0)}&A^{(1)}&\cdots& A^{(n_3-1)}\\A^{(n_3-1)}&A^{(0)}&\cdots&A^{(n_3-2)}\\&&\ddots&\\A^{(1)}&A^{(2)}&\cdots&A^{(0)}\end{pmatrix}
=\Big[C_1,C_2,\dots,C_{n_3}\Big], \quad C_i=\frac{1}{\sqrt{n_3}}\begin{pmatrix} A^{(i-1)}\\A^{(i-2)}\\\vdots\\A^{(i)}\end{pmatrix}.
\end{equation}
We now apply Theorem~\ref{theorem_block} twice.
First, for $\pmb{A}=C_i^T$,$\|A^T\|_*=\|A\|_*$, we get
\begin{equation}
\frac{1}{\sqrt{n_3}}\left(\sum\limits_{i=0}^{n_3-1}\|A^{(i)}\|_*^2\right)^{1/2}\le \|C_i\|_*\le \frac{1}{\sqrt{n_3}}\sum\limits_{i=0}^{n_3-1}\|A^{(i)}\|_*.
\label{eq:block_proof_1}
\end{equation}
Then, for $\pmb{A}=\Big[C_1,C_2,\dots,C_{n_3}\Big]$, we get
\begin{equation}
\left(\sum\limits_{i=1}^{n_3}\|C_i\|_*^2\right)^{1/2}\le \|\mathcal{A}\|_{*,\mathcal F}\le \sum\limits_{i=1}^{n_3}\|C_i\|_* .\label{eq:block_proof_2}
\end{equation}
Plugging the lower and upper bounds of \eqref{eq:block_proof_1} in to the lower and upper bound of \eqref{eq:block_proof_2}, gives the result.
\end{proof}

We next show how are \eqref{eq:L_behavior} and \eqref{eq:S_behavior} obtained by estimating the nuclear norm $\|\cdot\|_{*,\mathcal F}$ for stationary and moving targets. We show that, for high overlap, the $\|\cdot\|_{*,\mathcal F}$  norm decreases for stationary targets, while low overlap, increases the norm for moving targets.
	\paragraph{$\|\mathcal A\|_{*,\mathcal F}$ for stationary targets.}
	\label{sec:stat_anal}
Let us consider a stationary background, whose returns do not change between pulses. This means that the data traces are the same in every sub-aperture
\begin{equation}
A^{(i)}=A, \quad i=0,\dots ,n_3-1.
\end{equation}
The total decoupled nuclear norm is
\begin{equation}
\|\mathcal{A}\|_{*,\mathcal D}=\sum\limits_{i=0}^{n_3-1}\|A^{(i)}\|_*=n_3\|A\|_*.
\end{equation}
On the other hand for $\|\mathcal A\|_{*,\mathcal F}$, $\hat{A}^{(i)}$ will have a distinct form
since, with respect to the third dimension, we are performing DFT over a constant vector,
\begin{equation}
\hat{A}^{(k)}_{ij}=\frac{1}{\sqrt{n_3}}\sum\limits_{\ell=0}^{n_3-1}\omega_{n_3}^{\ell k}A^{(\ell)}_{ij}=A_{ij}\frac{1}{\sqrt{n_3}}\sum\limits_{\ell=0}^{n_3-1}\omega_{n_3}^{\ell k}=\sqrt{n_3}A_{ij}\delta_{k0}.
\end{equation}
i.e., after performing the DFT in the third dimension, the sub-apertures have the form
\begin{equation}
\hat{A}^{(k)}=\begin{cases}\sqrt{n_3}A,&k=0\\0, &k\neq 0\end{cases}.
\end{equation}
 Thus, the total nuclear norm becomes
\begin{equation}
\|\mathcal{A}\|_{*,\mathcal F}=\sum\limits_{k=0}^{n_3-1}\|\hat{A}^{(k)}\|_*=\sqrt{n_3}\|A\|_*=\left(\sum\limits_{i=0}^{n_3-1}\|A^{(i)}\|_*^2\right)^{1/2}.
\end{equation}
This is indeed the lower bound of \eqref{eq:nuc_fourier_bounds}. Remark also that since the decoupled nuclear norm is always $\sum\limits_{i=0}^{n_3-1}\|A^{(i)}\|_*$, 
we get an effective decrease of the nuclear norm by a factor of $\sqrt{n_3}$, with respect to the decoupled norm.
	
	We assumed here a completely stationary background. In reality the echoes from stationary targets do vary along the aperture. However, the variation rate along different sub-apertures, both in support and in phase, is much slower for stationary data traces. Hence, larger overlap would be beneficial, as it guarantees that when performing DFT, the sub-apertures are slowly varying, thus suppressing possible amplitude variation and phase decoherence, that will {\em smear out} the energy over multiple panels after performing DFT. This analysis answers Q1 and is in agreement with the results in Figure~\ref{fig:nucS_ratio_0} where we observed that as the overlap of the sub-apertures increases, $\|\mathcal A\|_{*\mathcal F}/\|\mathcal A\|_{*\mathcal D}$ decreases.

\paragraph{$\|\mathcal A\|_{*,\mathcal F}$ for moving targets.}
\label{sec:mov_anal}
To answer Q2 and Q3, we consider the case of a single moving target, whose data traces are approximately linear. We further assume no overlap between the sub-apertures, so that there is no column support overlap between different sub-apertures. This is a good approximation for targets moving linearly in parallel to the projection of  $\vr(0)-\vrho_o$  on the 2D plane ($\alpha=0$, see Figure~\ref{fig:D_ex_alpha_0}).  In this case, every panel would be a translation of the same matrix, without any column overlap. Therefore, we can write
\begin{equation}
A^{(\ell)}=\left[\underbrace{0,\dots,0}_{{\ell-1} \text{ times}},A,0\dots,0 \right]
\end{equation}
i.e., the sub-apertures are copies of the same matrix, supported on disjoint subsets of columns at every sub-aperture.
After performing DFT, the $k$th  panel would be
\begin{equation}
\hat{A}^{(k)}=\frac{1}{\sqrt{n_3}}\left[A, \omega_{n_3}^k A,\omega_{n_3}^{2k}A, \dots,\omega_{n_3}^{(n_3-1)k} A,\right],\quad \omega_{n_3}=e^{i2\pi /n_3}.
\label{eq:form_moving_tensor}
\end{equation}
Notice that $\hat{A}^{(k)}$ has the same structure as $\pmb{A}$ in Corollary~\ref{col:cor_1}.\ref{col:cor_1_c}, with $\beta_j=\frac{1}{\sqrt{n_3}}\omega_{n_3}^{jk}$. Since $\|\pmb{\beta}\|_2=1$ we have $\|\hat{A}^{(k)}\|_*=\|A\|_*$, and the values of the tensor nuclear norm and decoupled nuclear norm are the same. 

This seems in accordance with the results of Figure~\ref{fig:nucV_ratio}, where we see very little improvement of the nuclear norm for the tensor over the decoupled norm for most hyper-parameter choices when $\alpha=0$. 

But then, why is the performance improving for $\alpha>0$ ? We explain this in the following section.

\subsection{Interpretation of the nuclear norm angle dependence in SAR data}
\label{sec:interpret}
As we demonstrated in Figure~\ref{fig:D_alpha_ex}, and show in greater detail  in Appendix~\ref{app:cross_terms}, in general the data traces are not linear, and the inner product between different sub-apertures is small. However, their deviation from linearity greatly depends on the direction in which the target is moving. Targets that are moving at angles different than $0$ result in higher variation in the slope of the data traces between different sub-apertures, which leads to effective orthogonality between data traces of different sub-apertures. This leads to an increase in the nuclear norm, achieving the upper bound as in Corollary \ref{col:cor_1}\ref{col:cor_1_b}
\begin{equation}
\|\hat{A}^{(k)}\|_*=\frac{1}{\sqrt{n_3}}\sum\limits_{i=0}^{n_3-1}\|A^{(i)}\|_*
\end{equation}
This explains why $\alpha\neq 0$ improves the nuclear norm ratio, since when the upper bound is achieved the tensor nuclear norm would see an improvement by a factor of $\sqrt{n_3}$ over the decoupled form, since in the decoupled form 
\begin{equation}
\|A\|_{*,\mathcal{D}}=\sum\limits_{i=0}^{n_3-1} \|A^{(i)}\|_*.
\end{equation}
While in the tensor form it will be
\begin{equation}
\|A\|_{*,\mathcal{F}}=\sum\limits_{k=0}^{n_3-1}\|\hat{A}^{(k)}\|_*=\sum\limits_{k=0}^{n_3-1} \sum\limits_{i=0}^{n_3-1}\frac{1}{\sqrt{n_3}}\|A^{(i)}\|_*=\sqrt{n_3}\sum\limits_{i=0}^{n_3-1} \|A^{(i)}\|_*,
\end{equation}
and indeed the upper bound of \eqref{eq:nuc_fourier_bounds}.

 However, our numerical results suggest that the improvement that is least sensitive to the choice of the hyper-parameters is around $\pi/4$. To explain this, let us recall that in our analysis, we assumed that after performing DFT, we can partition the data traces originating from different sub-apertures, i.e., that there is no column support overlap between the data of different sub-apertures. However, in practice there is always a column support overlap, since the sub-apertures have non zero overlap. 
 In the case where all data traces have the same linear slope, the sub-aperture overlap (i.e., the number of overlapping rows between successive sub-apertures) determines the fixed column overlap, proportional  to the number of overlapping rows by the slope. However, as the data traces become more curved, the effective slope varies across different sub-apertures and might even change sign, leading to an increase in the column overlap, no longer determined by the sub-aperture overlap. 
 This reduces the possible increase in nuclear norm, as the effective column range does not grow at the same rate with increasing number of sub-apertures. Thus, there is a trade off between the non-linearity of the traces, which implies orthogonality, and their column range support, leading to the most robust performance of the TRPCA around $\alpha=\pi/4$, as is observed in Figure~\ref{fig:nucV_ratio}. This is also abserved in Figure~\ref{fig:ex_demo}, where we can see that for $\alpha=\pi/2$ the column range is limited and, when looking at the panels of $\hat{A}^{(k)}$ after performing DFT, there is significant variation in the overlap, compared to $\alpha=0$ and $\alpha=\pi/4$.

\paragraph{Optimality of $\|\cdot\|_{*,\mathcal F}$ for SAR data separation.}
\label{subsec:yier_nuclear_effective}
The performance of RPCA in the SAR context improves when the ratio between the nuclear norm of the background and the moving target increases. The main motivation in seeking a tensor based representation is to enhance the low-rank and sparse/full-rank structure of the background and moving targets respectively. 

The results of Corollary~\ref{col:cor_1} are a particular case of a more general result, proven in \cite{li2016bounds}. That result states that the nuclear norm of a tensor $\mathcal{T}$, can be bound by the nuclear norm of any regular partition of $\mathcal T$ (defined in  \cite{li2016bounds}) $\mathcal{T}_1,\mathcal{T}_2,\dots\mathcal{T}_k$ by
\begin{equation}
\left(\sum\limits_{i=1}^k \|\mathcal{T}_i\|_*^2\right)^{1/2}\le \|\mathcal{T}\|_{*,\mathcal T}\le \sum\limits_{i=1}^k \|\mathcal{T}_i\|_*
\end{equation} 
Specifically, in our case the general statement applies when we choose the partition of $\mathcal A$ to be the decomposition of the data into different sub-apertures  $A^{(i)}$, giving
\begin{equation}
\left(\sum\limits_{i=0}^{n_3-1}\|A^{(i)}\|_*^2\right)^{1/2}\le	\|\mathcal{A}\|_{*,\mathcal T}\le \sum\limits_{i=0}^{n_3-1}\|A^{(i)}\|_*.
\label{eq:nuc_norm_bounds_ten}
\end{equation}
Comparing \eqref{eq:nuc_norm_bounds_ten} to \eqref{eq:nuc_fourier_bounds} we note that the Fourier based nuclear norm achieves the lower bound for a completely stationary background, while it achieves a factor of $\sqrt{n_3}$ times the upper bound for non linear targets, moving rapidly. It is in this sense, that the Fourier method is optimal for the SAR data separation problem. 
\section{Numerical results}  \label{sec:4}
In this section we use the same setup as before (cf. section~\ref{sec:hyper_sim}) to generate synthetic SAR data for different configurations. We then use TRPCA and compare its performance with two other methods: matrix RPCA over the entire data, and decoupled RPCA over the different sub-apertures. We first show the data separation results in Section \ref{sec:4.1} and then the corresponding images in Section \ref{sec:4.2}.
\subsection{Data separation results}
\label{sec:4.1}
We consider a stationary background with 15 point scatterers and a target moving at a slow velocity of $1m/s$. We construct the data matrix $D_L$ and data tensor $\mathcal{A}_L$ for the background and do the same with $D_S$ and $\mathcal{A}_S$ for the moving target. The hyper-parameters are the same as in the previous section. We vary the target's direction $\alpha$ between $0$, $\pi/4$ and $\pi/2$ and perform RPCA in three forms:
\begin{enumerate}[(a)]
	\item TRPCA with $\eta$ set to its optimal value:
	\begin{equation}
	\eta^*_{\mathcal F}=\sqrt{\eta_{\max,\mathcal{F}}\eta_{\min,\mathcal{F}}},\quad \eta_{\max,\mathcal{F}}=\frac{\|\mathcal{A}_S\|_{*,\mathcal F}}{\|\mathcal{A}_S\|_1},\eta_{\min,\mathcal{F}}=\frac{\|\mathcal{A}_L\|_{*,\mathcal F}}{\|\mathcal{A}_L\|_1}.
	\end{equation}	
	\item `Decoupled': We perform matrix RPCA one sub-aperture at a time. Here $\eta$ is also set to its optimal value following \cite{leib2018RPCA}:
	\begin{equation}
	\eta^{(i)*}=\sqrt{\eta^{(i)}_{\max}\eta^{(i)}_{\min}},\quad \eta^{(i)}_{\max}=\frac{\|A^{(i)}_S\|_*}{\|A^{(i)}_S\|_1},\eta^{(i)}_{\min}=\frac{\|A^{(i)}_L\|_*}{\|A^{(i)}_L\|_1}. 
	\end{equation}
	\item Matrix RPCA over the entire data matrix with $\eta$ set to its optimal value:
	\begin{equation}
	\eta^*=\sqrt{\eta_{\max}\eta_{\min}},\quad \eta_{\max}=\frac{\|D_S\|_*}{\|D_S\|_1},\eta_{\min}=\frac{\|D_L\|_*}{\|D_L\|_1}
	\end{equation}
\end{enumerate}
For methods $(a)$ and $(b)$, we need to reconstruct the separated data on the entire aperture $D_L$ and $D_S$ from the separated data on the overlapping sub-apertures. 
From \eqref{eq:mat_to_tensor}, one can reconstruct the original data matrix $D_r(s,t)$ by
\begin{equation}
D_r(s,t)=\sum\limits_{\ell\in X}\theta_\ell A^{(\ell)}(s-\ell \vartheta s_{\text{sub}}),\  \sum \limits_{\ell}\theta_\ell=1,\quad X=\{\ell \text{ s.t. } s-\ell \vartheta s_{\text{sub}}\in [0,s_{\text{sub}}]\}.
\label{eq:tensor_to_mat}
\end{equation}
We can use any linear combination of the appropriate data entries in the sub-apertures provided that the weights $\theta_l$ sum up to one. The 
choice used is the `innermost' aperture (i.e. farthest from the edges), which proves to yield stable results, 
\begin{equation}
\begin{split}
&{D}_L(s,t)=L^{(\ell^*)}(s-\ell^*\vartheta,t)\\
&\ell^*=\arg\min_\ell \ell^2+(\ell-n_3+1)^2,\quad s-\ell \vartheta \in [0,s_{\text{sub}}].
\end{split}
\label{eq:opt_l}
\end{equation}
The same is done for ${D}_S$.
	
The separation results are illustrated in Figure~\ref{fig:RPCA_results}. We observe that the performance of TRPCA improves as the angle $\alpha$ increases. TRPCA outperforms the other two methods for $\alpha\ge \pi/4$, while for small angles, all methods are struggling. 

\begin{figure}[htbp!]
	\centering
	{\vskip 17mm\hskip -18.2cm$\alpha=0$\vskip -17mm}
	\begin{subfigure}[t]{0.245\textwidth}
		\includegraphics[width=1\columnwidth]{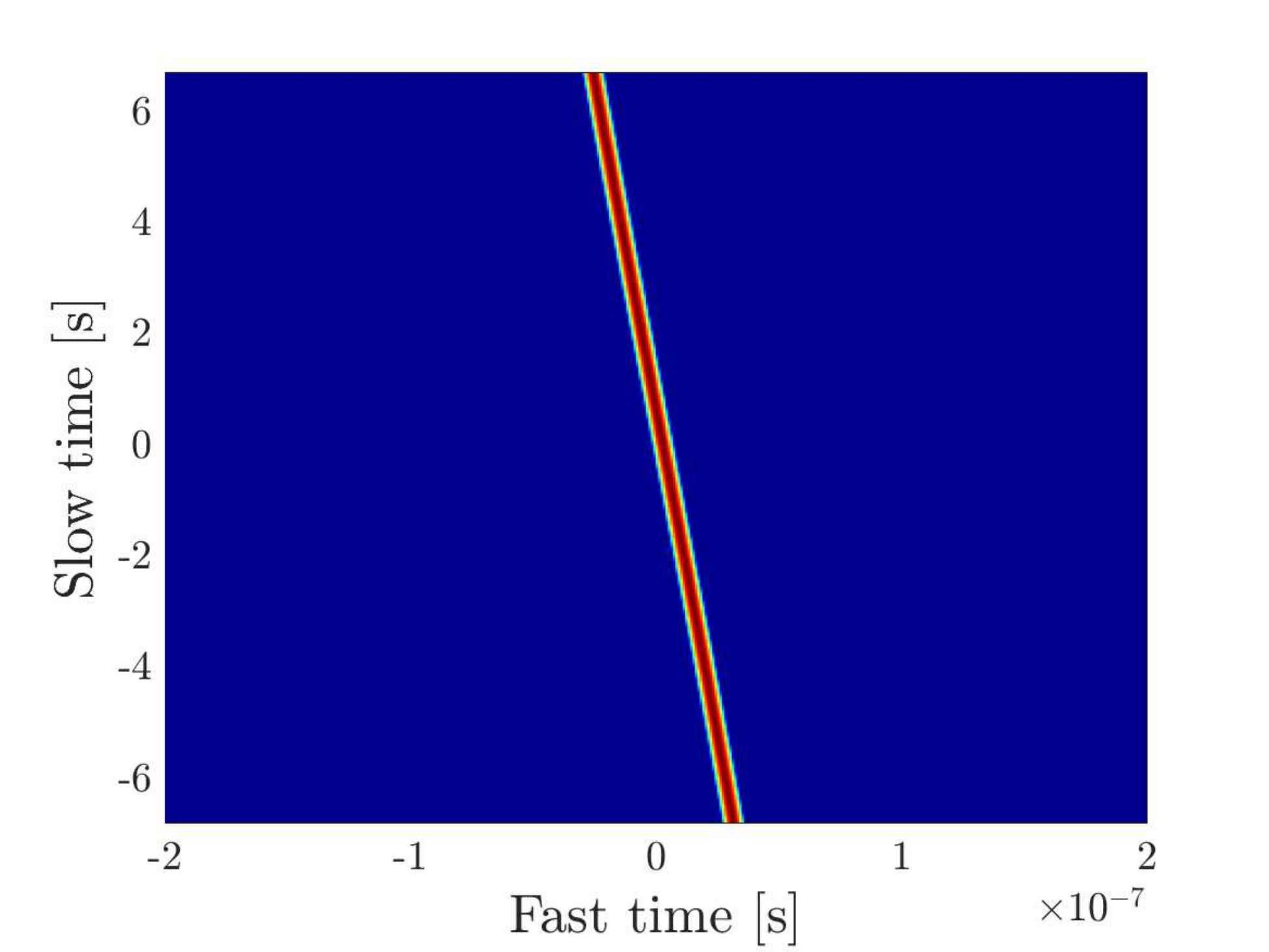}
		\label{fig:S_orig_0}
	\end{subfigure}
	\begin{subfigure}[t]{0.245\textwidth}
	\includegraphics[width=1\columnwidth]{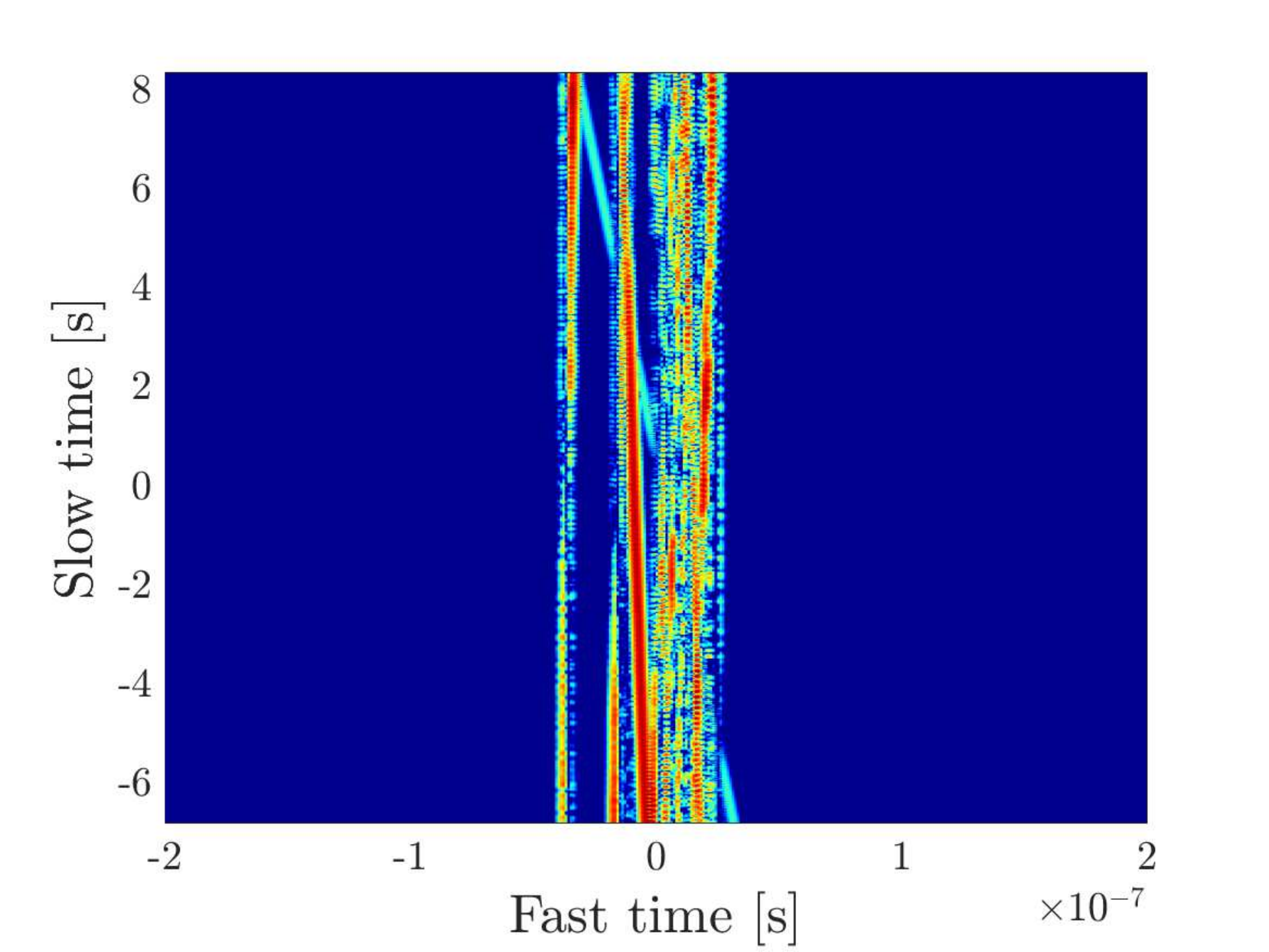}
	\label{fig:S_ten_0}
	\end{subfigure}
		\begin{subfigure}[t]{0.245\textwidth}
		\includegraphics[width=1\columnwidth]{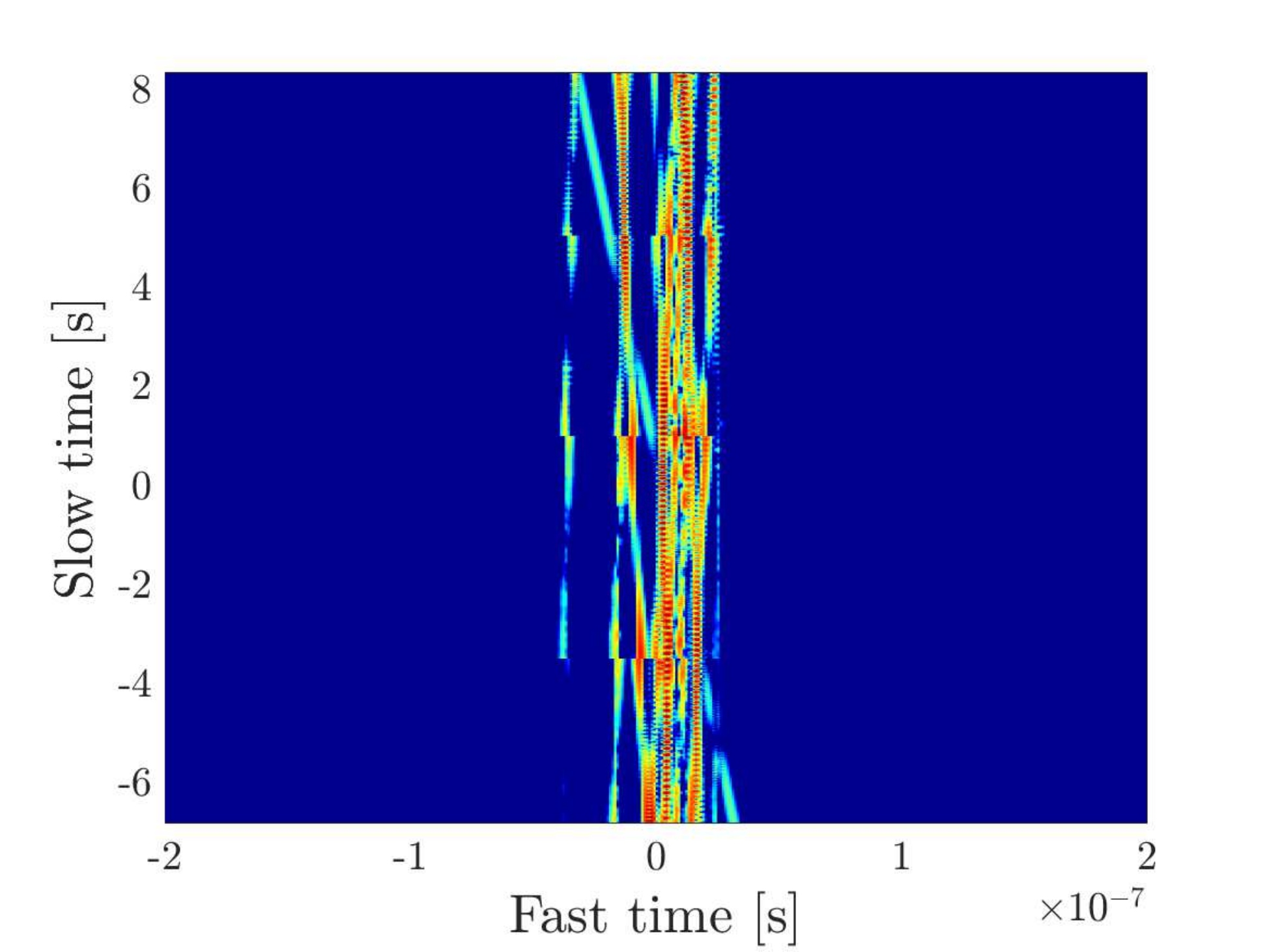}
		\label{fig:S_dec_0}
	\end{subfigure}
		\begin{subfigure}[t]{0.245\textwidth}
		\includegraphics[width=1\columnwidth]{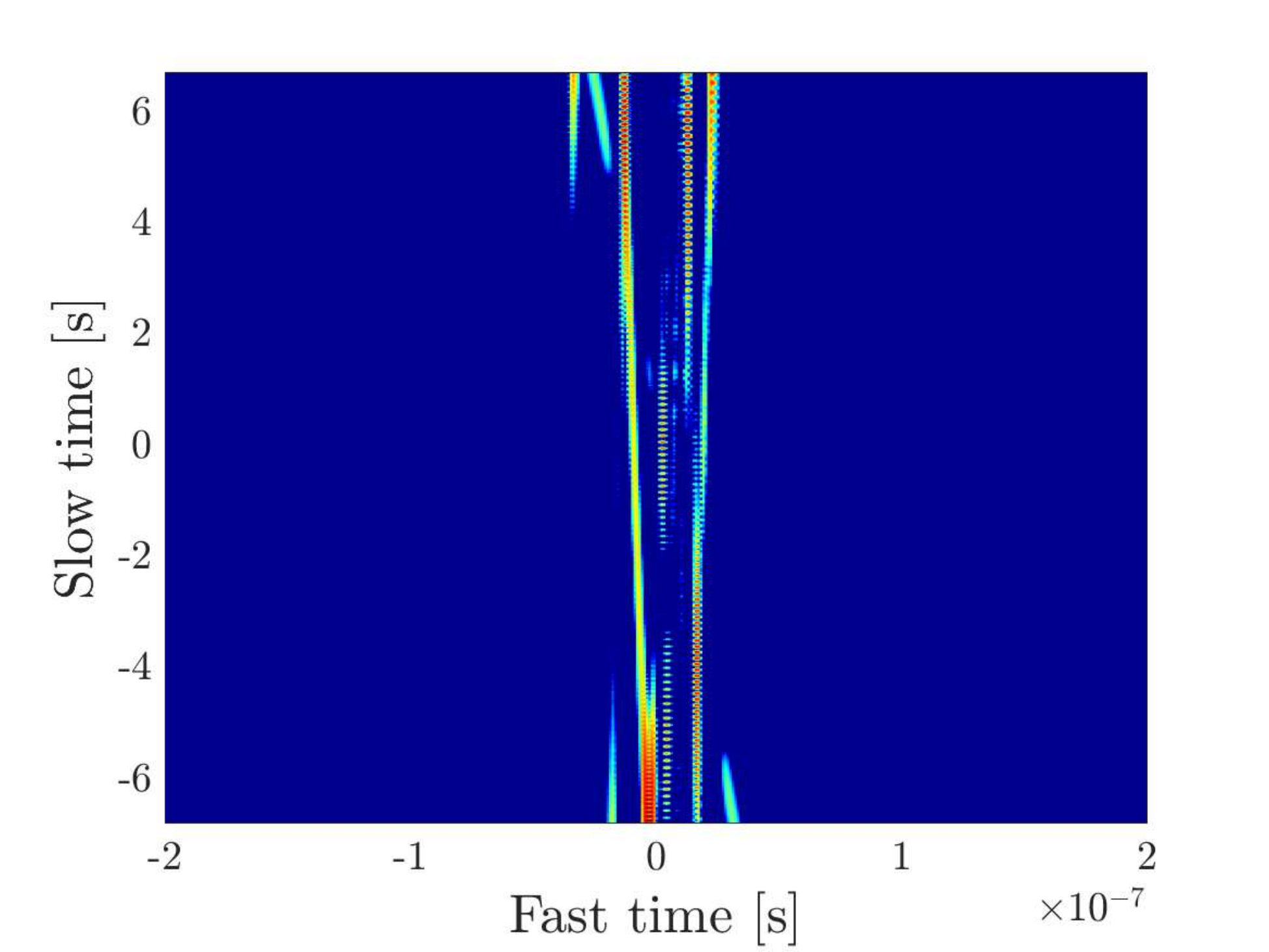}
		\label{fig:S_tot_0}
\end{subfigure}

	{\vskip 17mm\hskip -18.2cm$\alpha=\pi/8$\vskip -17mm}
	\begin{subfigure}[t]{0.245\textwidth}
	\includegraphics[width=1\columnwidth]{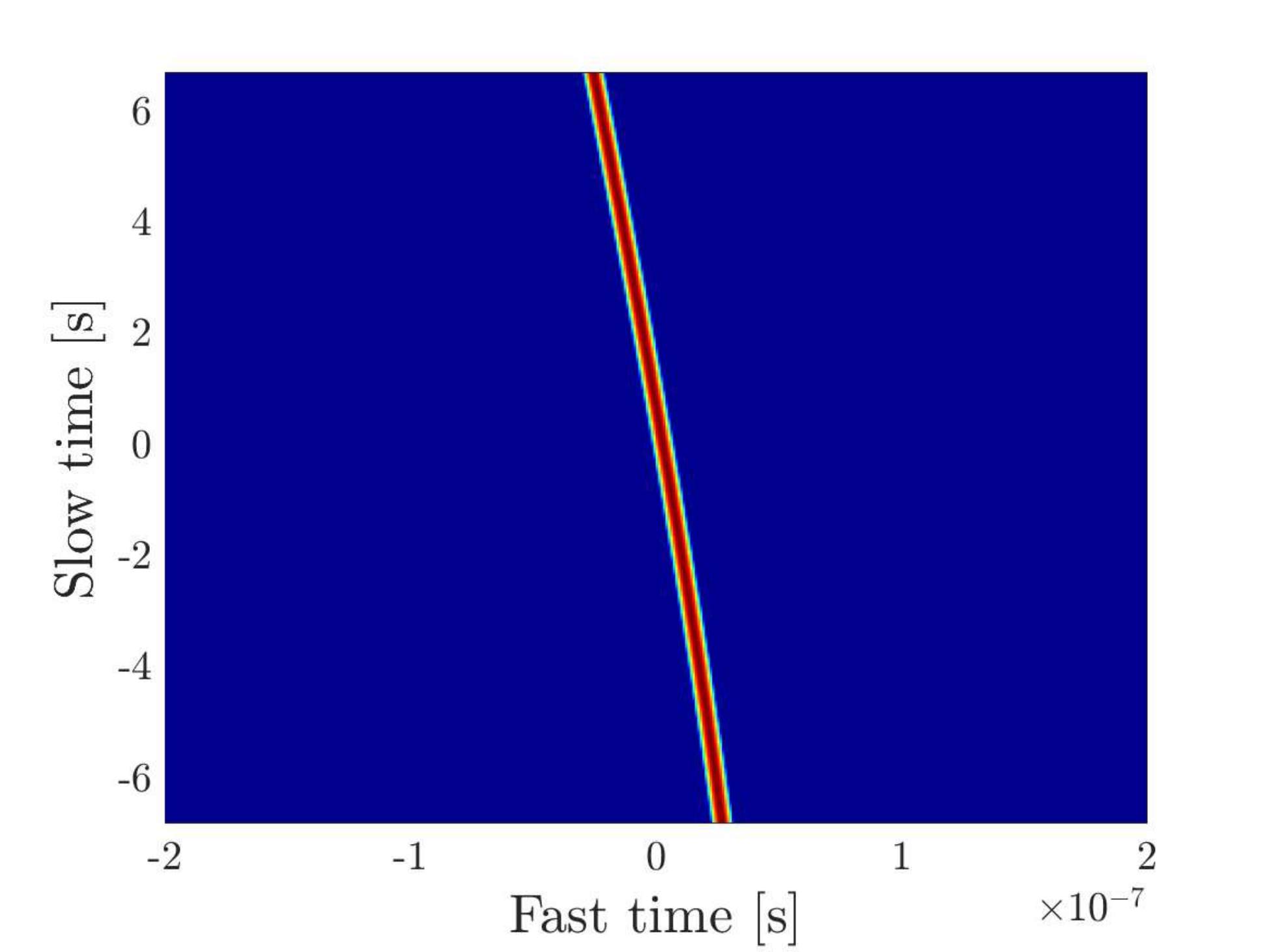}
	\label{fig:S_orig_pi_8}
\end{subfigure}
\begin{subfigure}[t]{0.245\textwidth}
	\includegraphics[width=1\columnwidth]{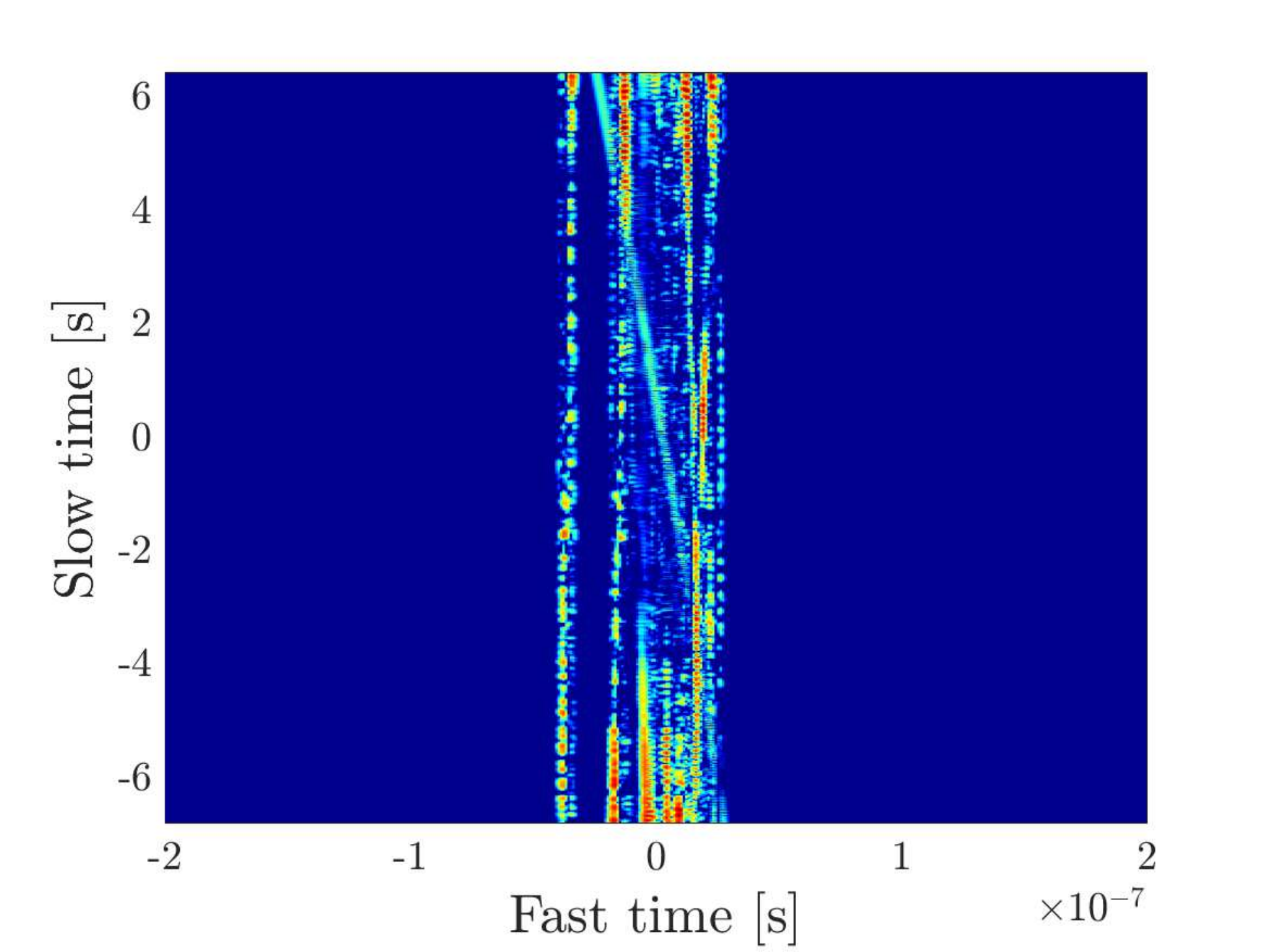}
	\label{fig:S_ten_pi_8}
\end{subfigure}
\begin{subfigure}[t]{0.245\textwidth}
	\includegraphics[width=1\columnwidth]{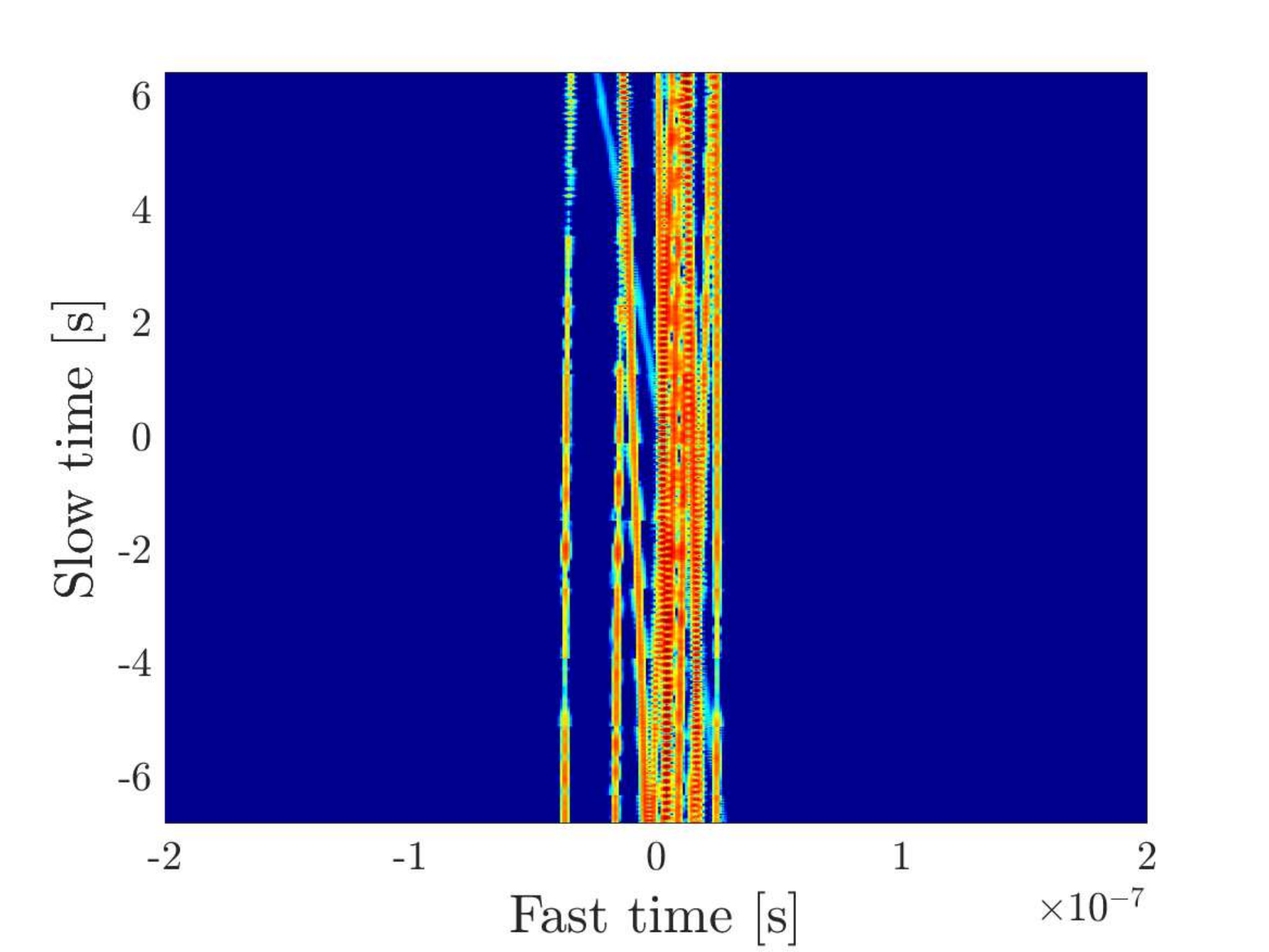}
	\label{fig:S_dec_pi_8}
\end{subfigure}
\begin{subfigure}[t]{0.245\textwidth}
	\includegraphics[width=1\columnwidth]{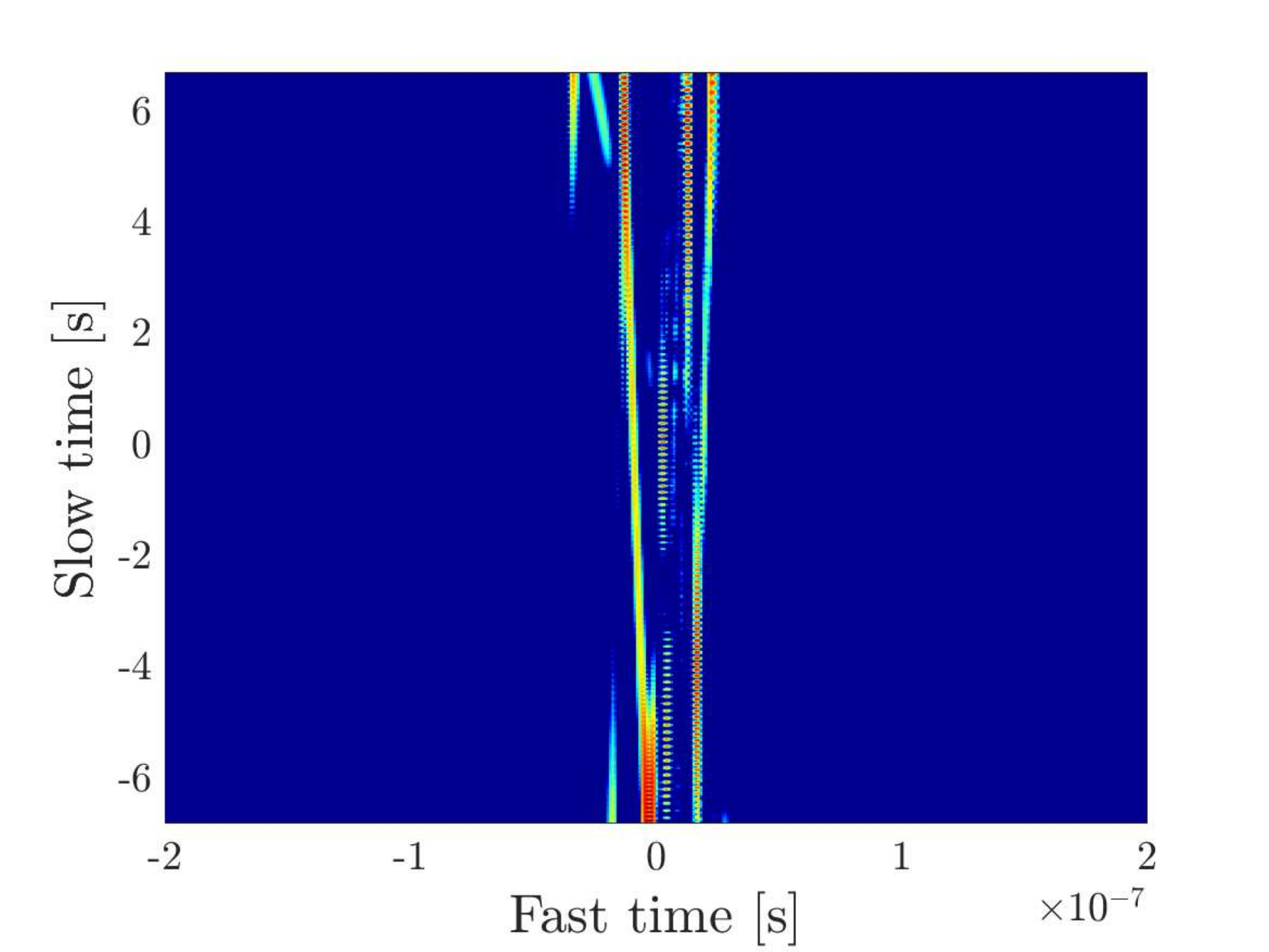}
	\label{fig:S_tot_pi_8}
\end{subfigure}
{\vskip 17mm\hskip -18.2cm$\alpha=\pi/4$\vskip -17mm}
	\begin{subfigure}[t]{0.245\textwidth}
	\includegraphics[width=1\columnwidth]{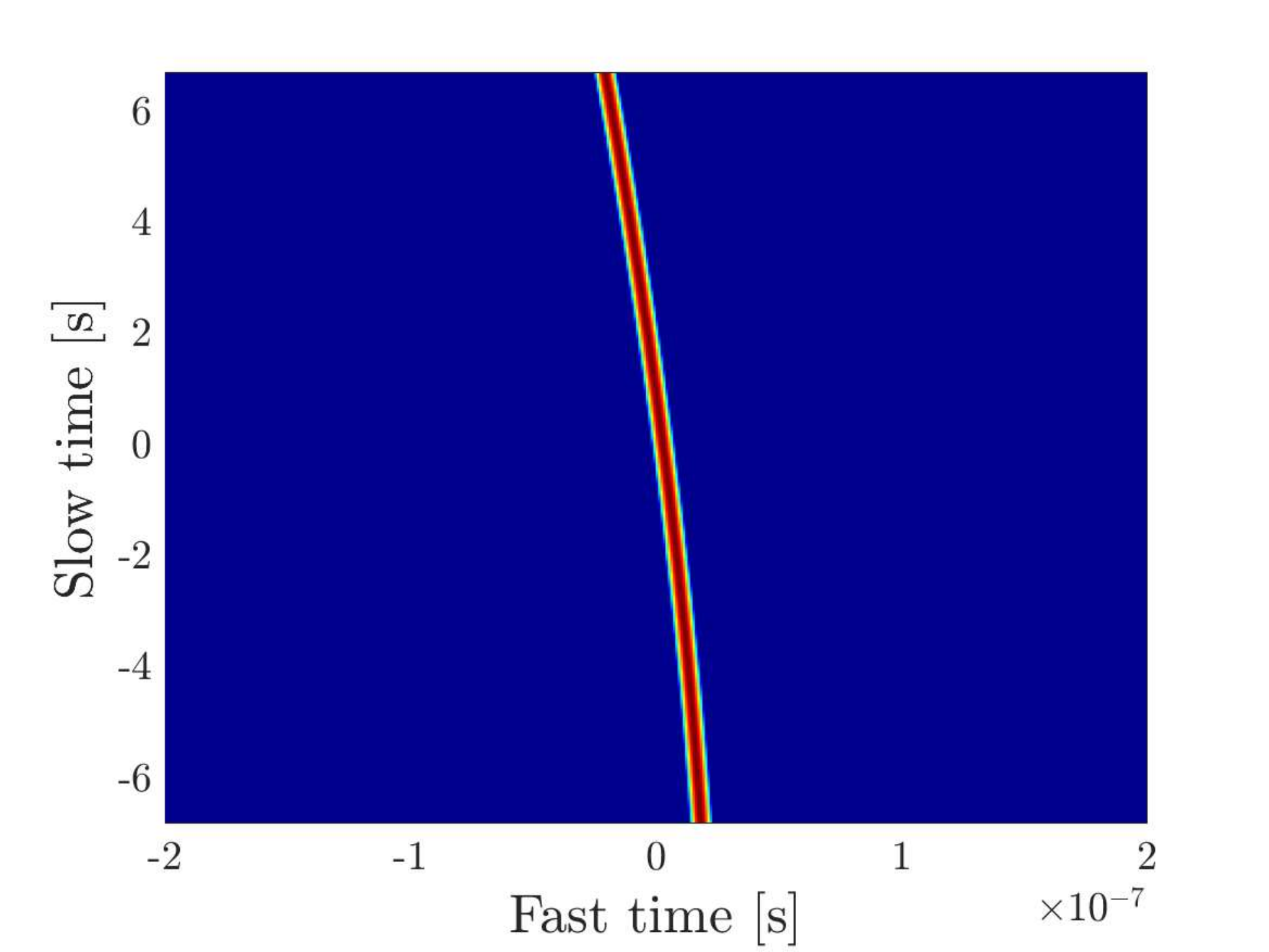}
	\caption{} 
	\label{fig:S_orig_pi_4}
\end{subfigure}
\begin{subfigure}[t]{0.245\textwidth}
	\includegraphics[width=1\columnwidth]{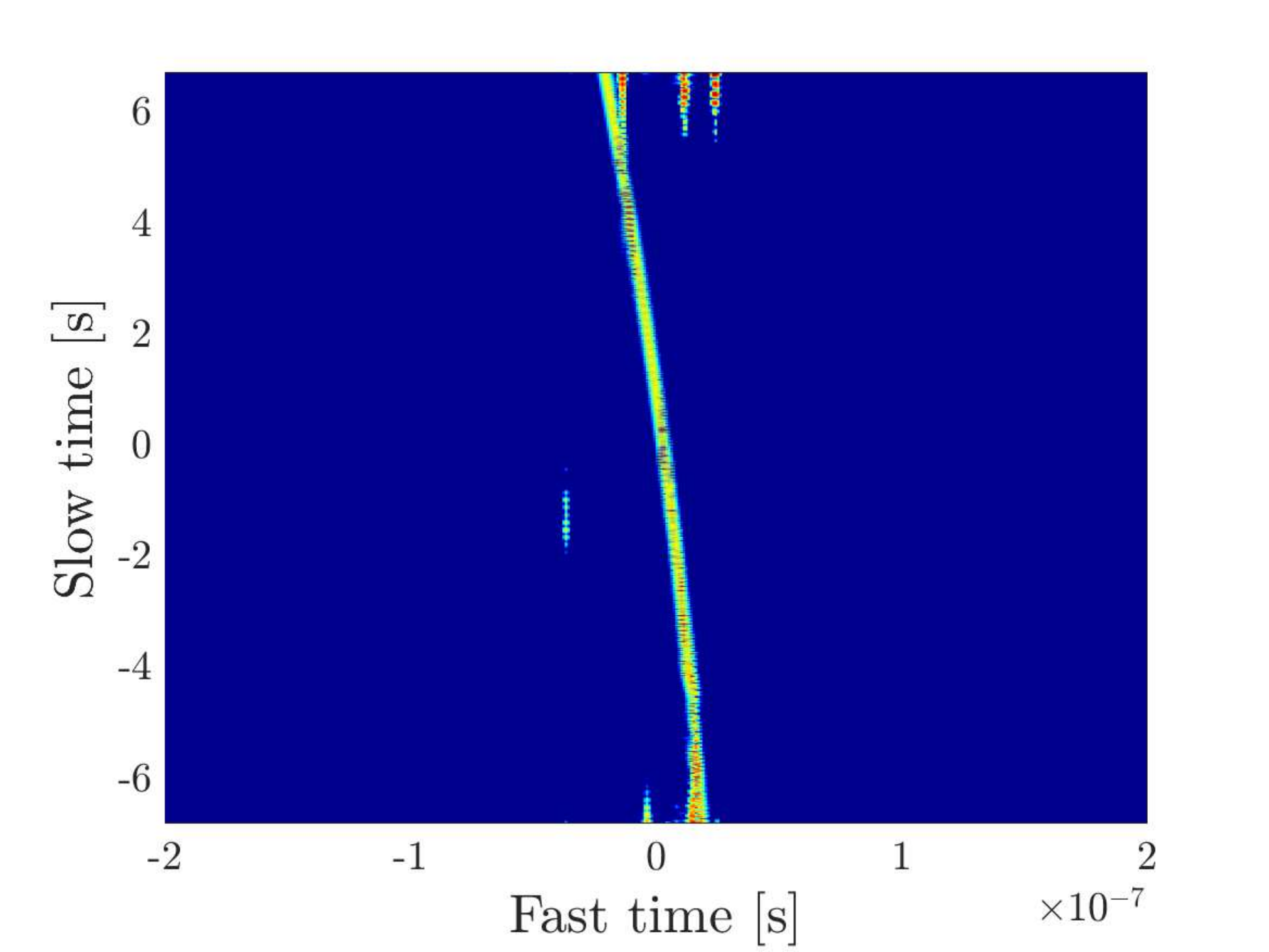}
	\caption{} 
	\label{fig:S_ten_pi_4}
\end{subfigure}
\begin{subfigure}[t]{0.245\textwidth}
	\includegraphics[width=1\columnwidth]{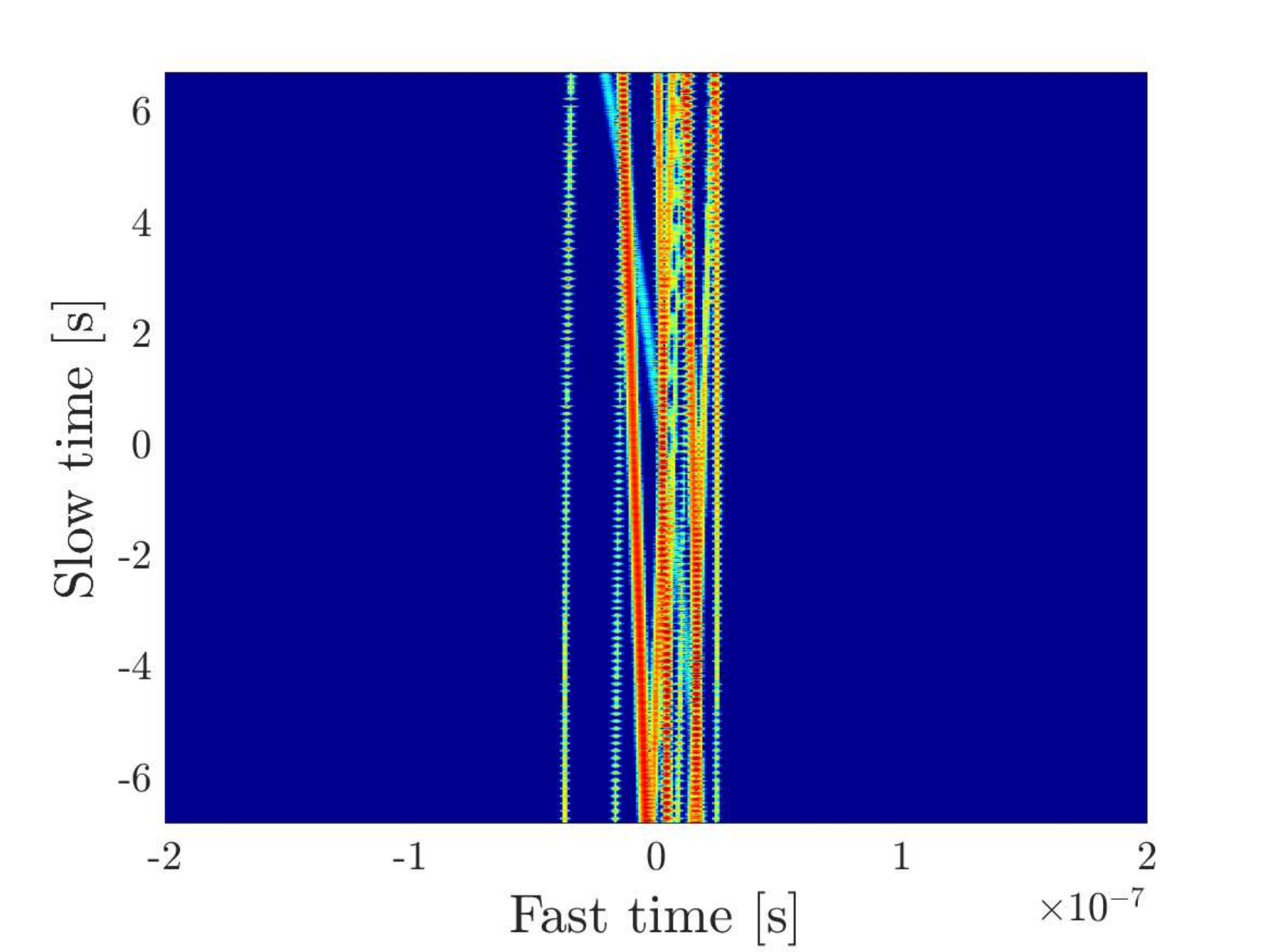}
	\caption{} 
	\label{fig:S_dec_pi_4}
\end{subfigure}
\begin{subfigure}[t]{0.245\textwidth}
	\includegraphics[width=1\columnwidth]{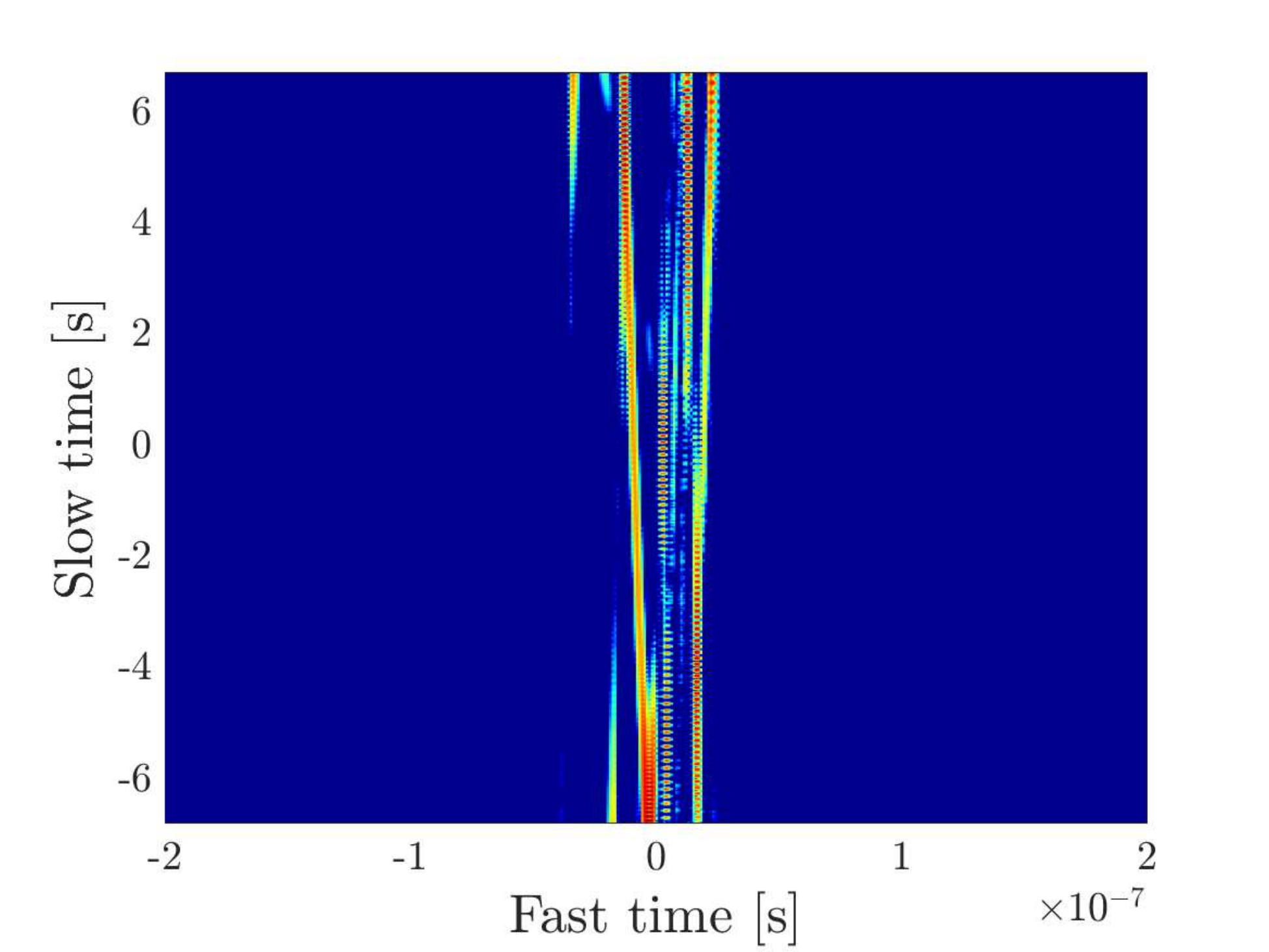}
	\caption{} 
	\label{fig:S_tot_pi_4}
\end{subfigure}
	{\vskip 17mm\hskip -18.2cm$\alpha=3\pi/8$\vskip -17mm}
	\begin{subfigure}[t]{0.245\textwidth}
		\includegraphics[width=1\columnwidth]{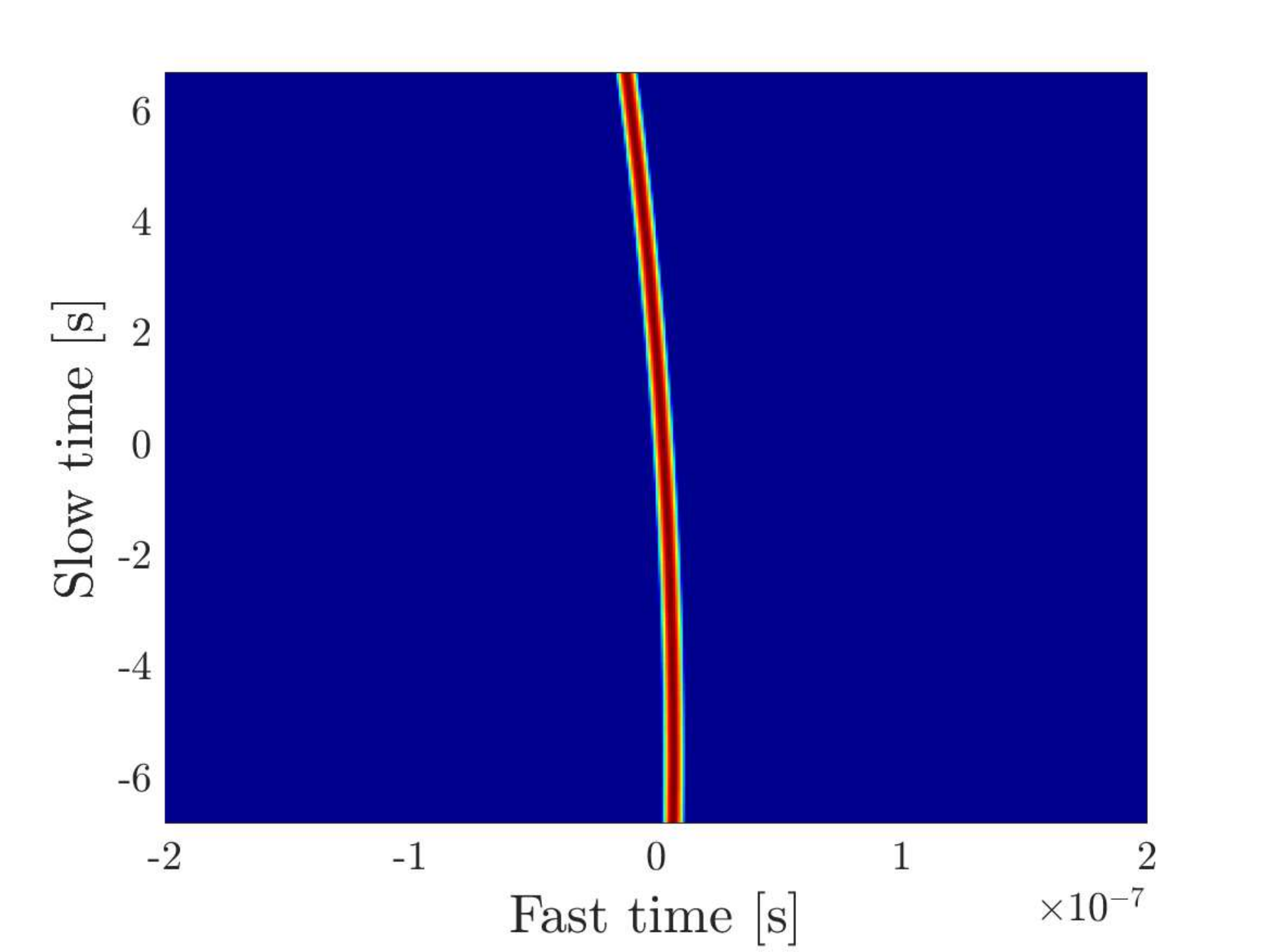}
		\label{fig:S_orig_3pi_8}
	\end{subfigure}
	\begin{subfigure}[t]{0.245\textwidth}
		\includegraphics[width=1\columnwidth]{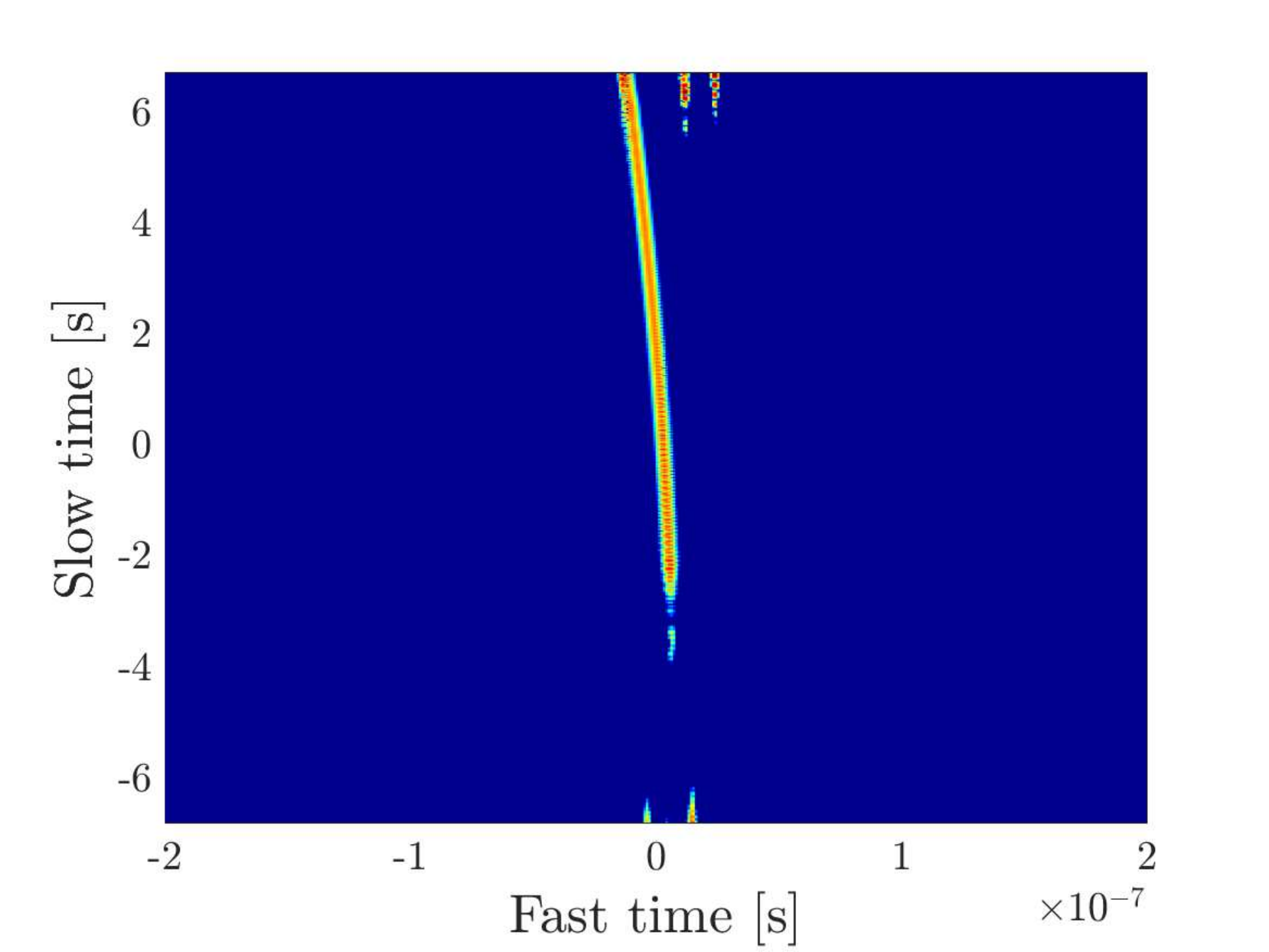}
		\label{fig:S_ten_3pi_8}
	\end{subfigure}
	\begin{subfigure}[t]{0.245\textwidth}
		\includegraphics[width=1\columnwidth]{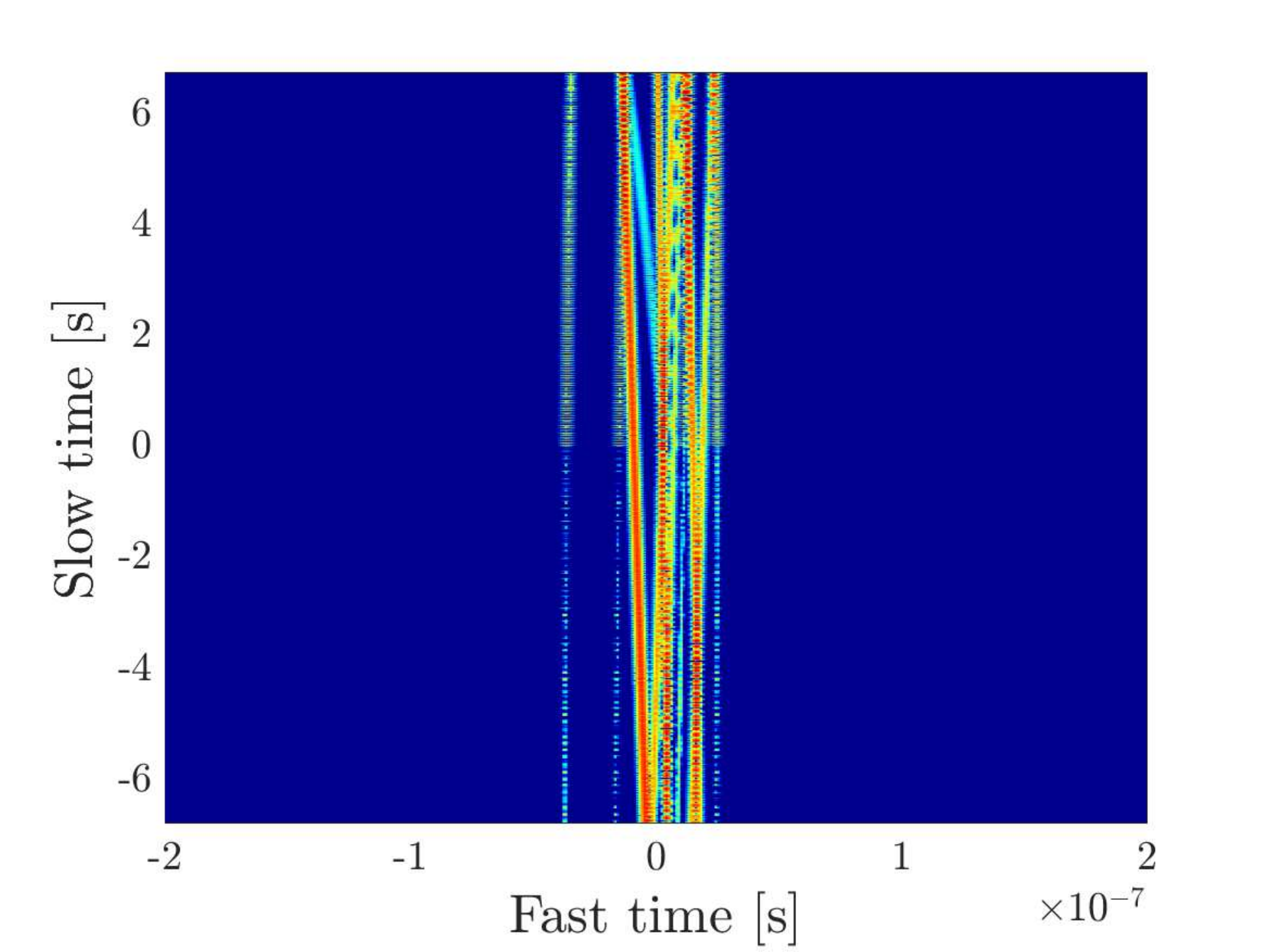}
		\label{fig:S_dec_3pi_8}
	\end{subfigure}
	\begin{subfigure}[t]{0.245\textwidth}
		\includegraphics[width=1\columnwidth]{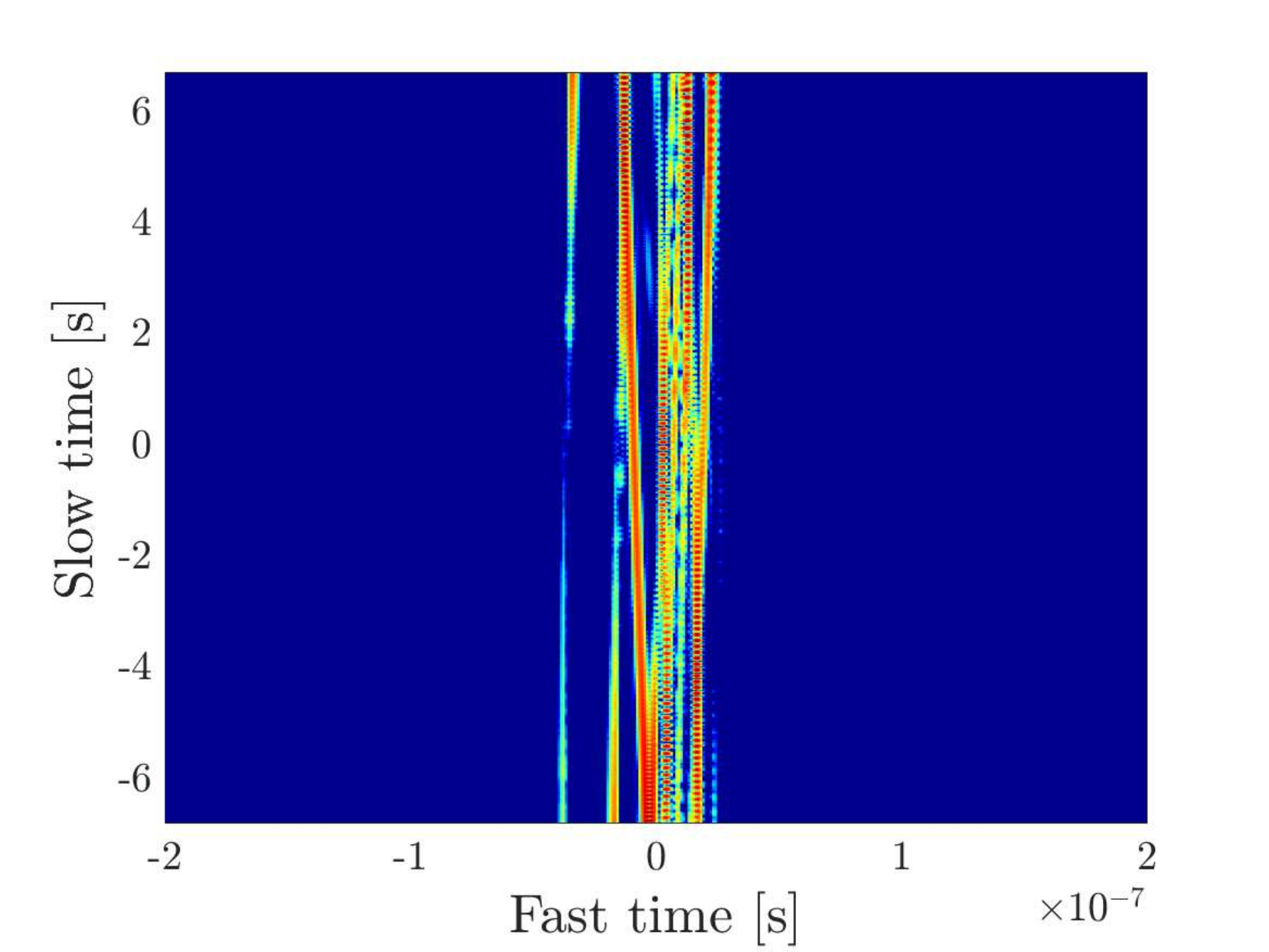}
		\label{fig:S_tot_3pi_8}
	\end{subfigure}
	{\vskip 17mm\hskip -18.2cm$\alpha=\pi/2$\vskip -17mm}
	\begin{subfigure}[t]{0.245\textwidth}
		\includegraphics[width=1\columnwidth]{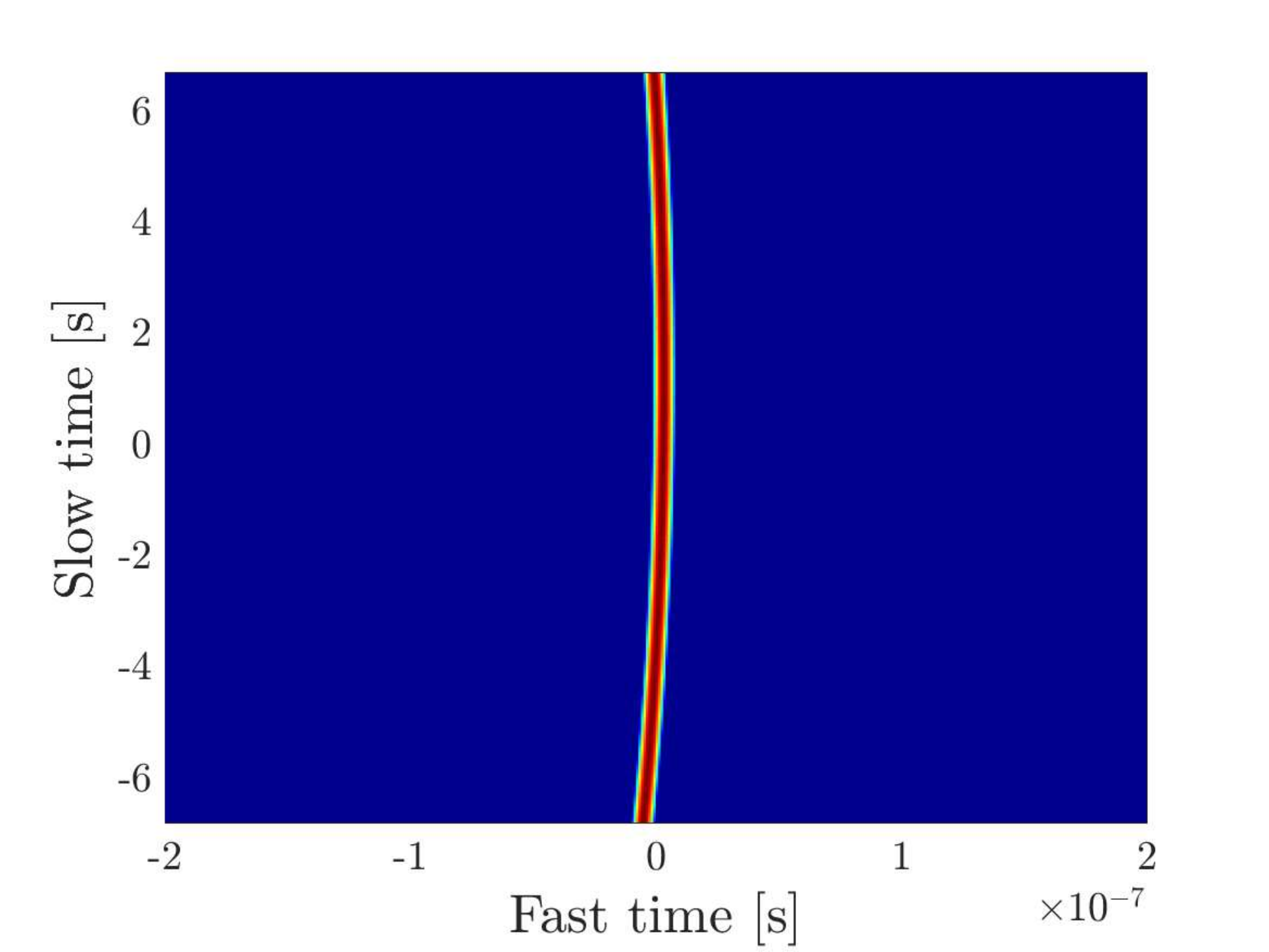}
		\caption{} 
		\label{fig:S_orig_pi_2}
	\end{subfigure}
	\begin{subfigure}[t]{0.245\textwidth}
		\includegraphics[width=1\columnwidth]{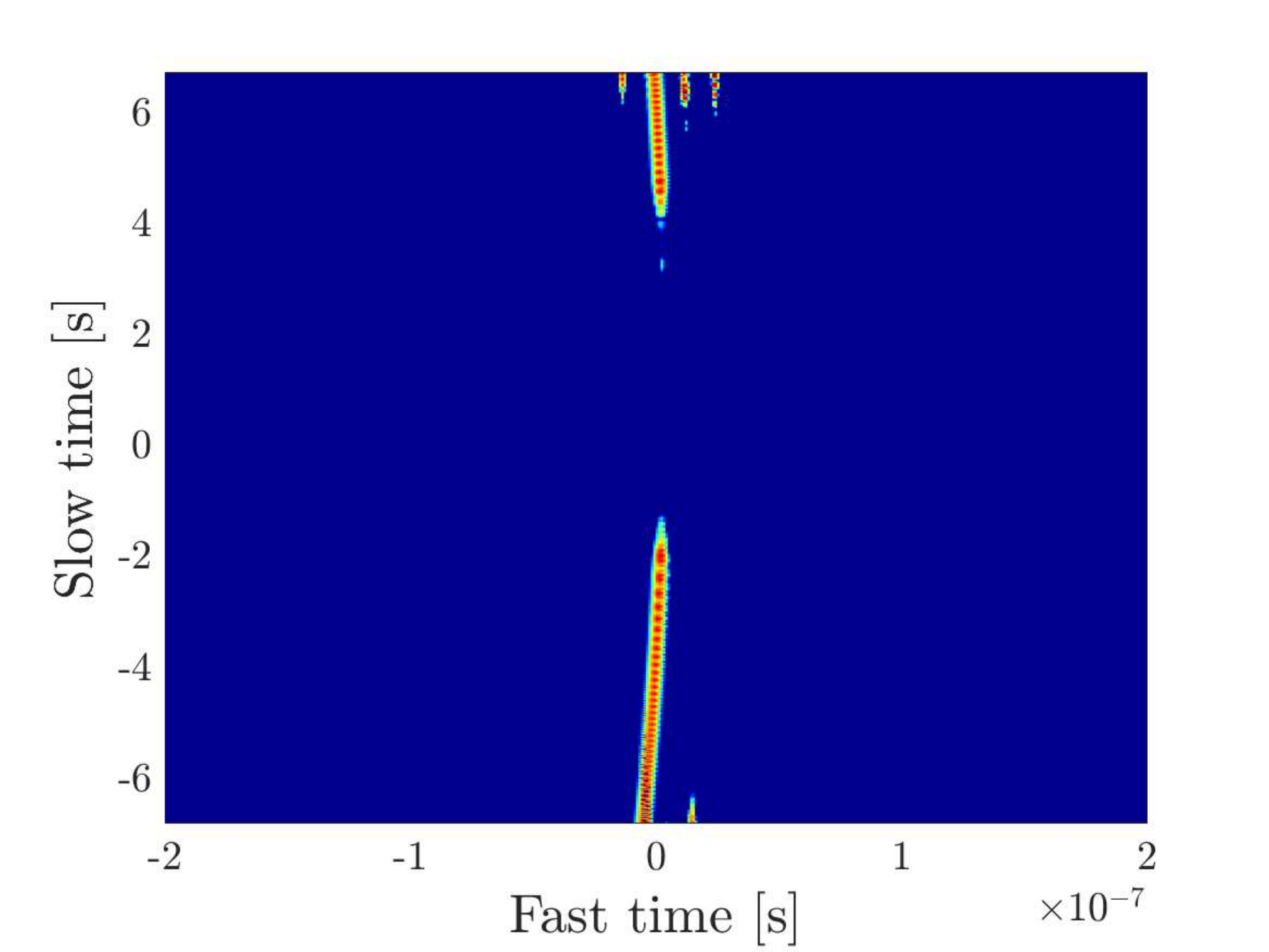}
		\caption{} 
		\label{fig:S_ten_3pi_16}
	\end{subfigure}
	\begin{subfigure}[t]{0.245\textwidth}
		\includegraphics[width=1\columnwidth]{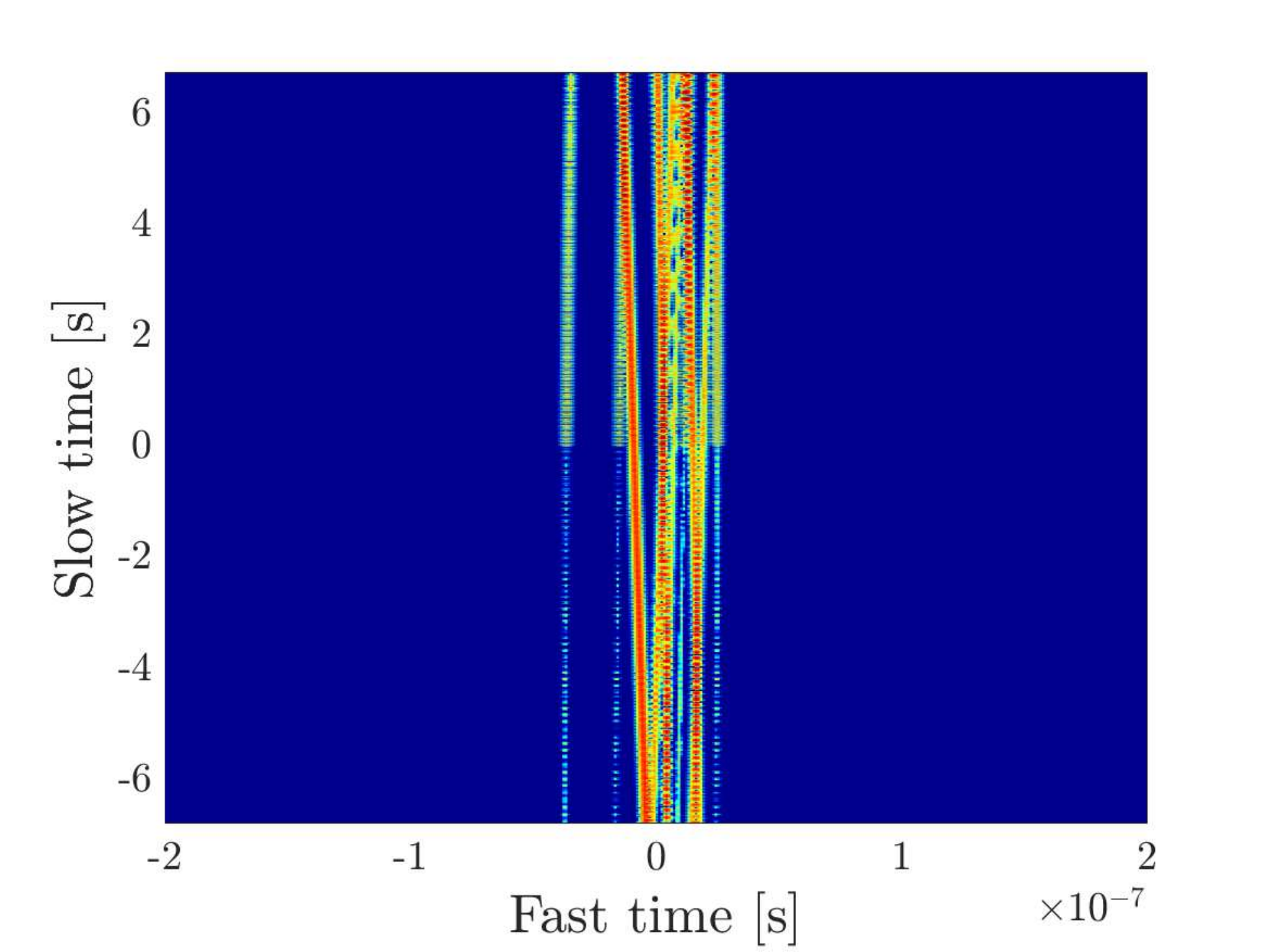}
		\caption{} 
		\label{fig:S_dec_pi_2}
	\end{subfigure}
	\begin{subfigure}[t]{0.245\textwidth}
		\includegraphics[width=1\columnwidth]{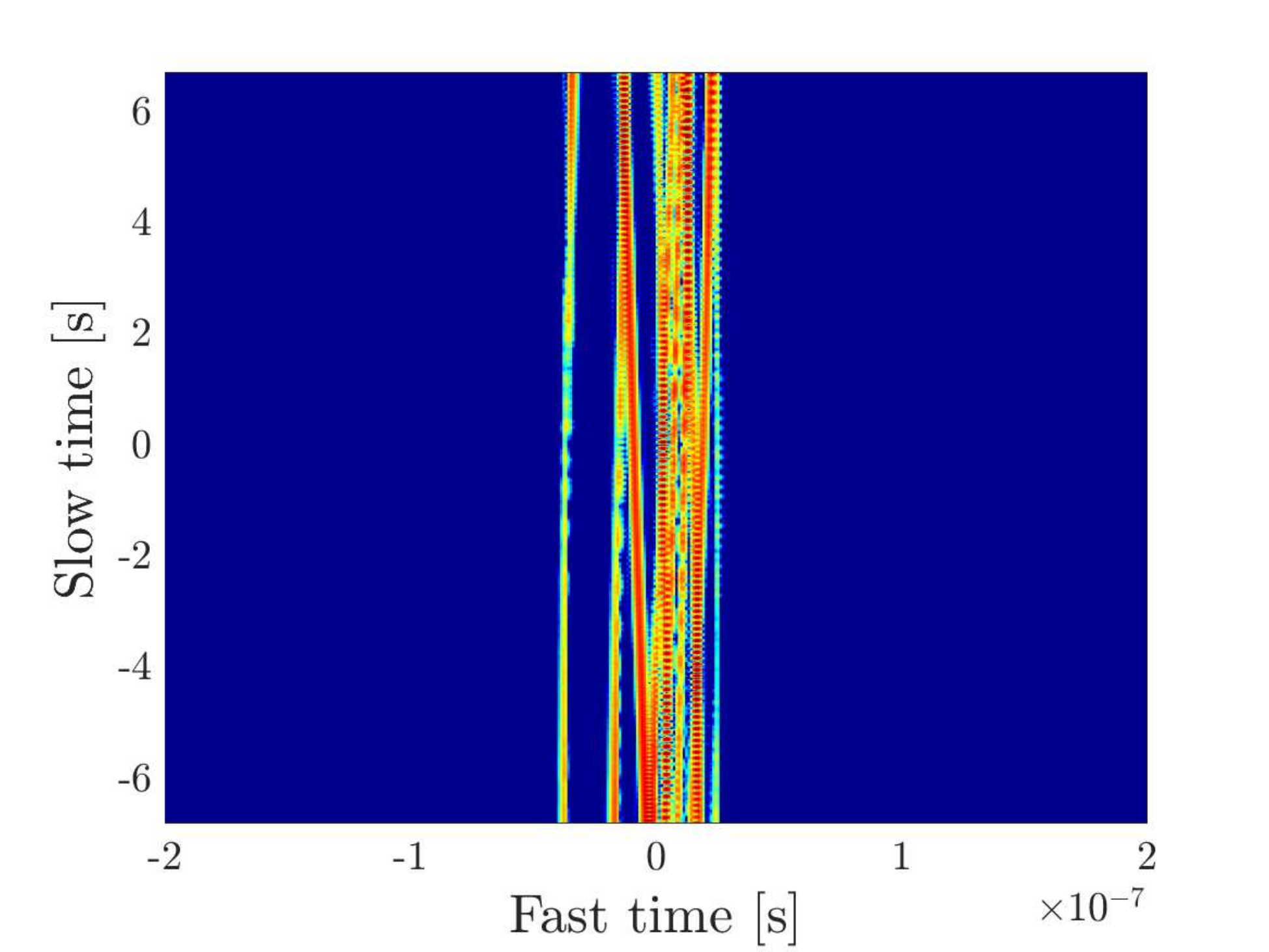}
		\caption{} 
		\label{fig:S_tot_pi_2}
	\end{subfigure}
\caption{RPCA performance. $(a)$ Original trace of moving target; $(b)$ TRPCA; $(c)$ `Decoupled' RPCA; $(d)$ Matrix RPCA. We can see that for $\alpha\ge \pi/4$ TRPCA outperforms other methods, with its performance improving with larger angles. For small angles, all methods are struggling. Some of the signal is lost around the points where the gradient with respect to the slow time is zero, which are stationary phase points.}
\label{fig:RPCA_results}
\end{figure}

\subsection{Imaging results }
\label{sec:4.2}
We present here the imaging results for the angles of Figure~\ref{fig:RPCA_results} for which a clear separation was possible, i.e., for $\alpha\ge \pi/4$. We do not present imaging results when good data separation is not achieved, since the extraction of motion parameters is challenging and prone to errors. To form an image for the moving target, we need to compensate for the target's velocity when evaluating the SAR functional of \eqref{eq:SARfunctional}. We explain below how this is done.

\paragraph{Motion Estimation}
We assume that the sparse part, after performing TRPCA, is composed of a single target. Hence, we can extract $y(s)=\Delta\tau(s)$ from stable peak locations, and compare it to a candidate one, depending on  trial target position and velocity
\begin{equation}
f_{\vrho,\vec{v}}(s)=\tau(s,\vrho+\vec{v}s)-\tau(s,\vrho_o).
\end{equation}
We can then extract $\vrho,\vec{v}$, from solving the following minimum loss problem
\begin{equation}
{\vrho}^*,{\vec{v}}^*=\arg\min \limits_{\vrho,\vec{v}}\sum\limits_s\mathcal{L}_\delta(y(s)-f_{\vrho,\vec{v}}(s)).
\end{equation}
With $\mathcal{L}_\delta$ a Huber loss:
\begin{equation}
\mathcal{L}_\delta(x)=\begin{cases} \frac{1}{2}x^2 ,& |x|\le \delta ,\\\delta(|x|-\frac{1}{2}\delta),&|x|>\delta,\end{cases} \mbox{ with }\quad \delta=10\Delta t.
\end{equation}
We only need the velocity parameters $\vec{v}$ for the SAR functional. We use MATLAB's \texttt{fmincon} routine to solve the optimization problem and obtain a robust estimate. An illustration of the deviation of the extracted trace from the exact one is given in Figure~\ref{fig:mov_param_extraction}, for $\alpha=\pi/4$.

Since the other two methods do not provide good separation results, it is impossible to extract the velocity parameters from the $S$ part of the data. 
Therefore to image the moving target, we form an exhaustive 4D imaging functional, where we vary both location and velocity parameters. 

\begin{equation}
I^{\text{SAR}}(\vrho,\vec{\pmb{v}})= \sum_{j=-n/2}^{n/2} 
S(s_j,\tau(s_j,\vrho+s_j \vec{\pmb{v}})-\tau(s_j,\vrho_o)) .
\end{equation}
We present the result of backpropagation for the other two methods for $I^{\text{SAR}}(\vrho, \vec{\pmb{v}}_t)$, for $\vec{\pmb{v}}=\vec{\pmb{v}}_t$ the target's actual velocity vector. 	
We present the imaging output of the three RPCA methods in Figure~\ref{fig:imaging_results}. One can see that TRPCA provides a clear image of the moving target, while the other two methods fail to do so. 
\begin{figure}[htbp!]
	\centering
	\includegraphics[width=0.5\columnwidth]{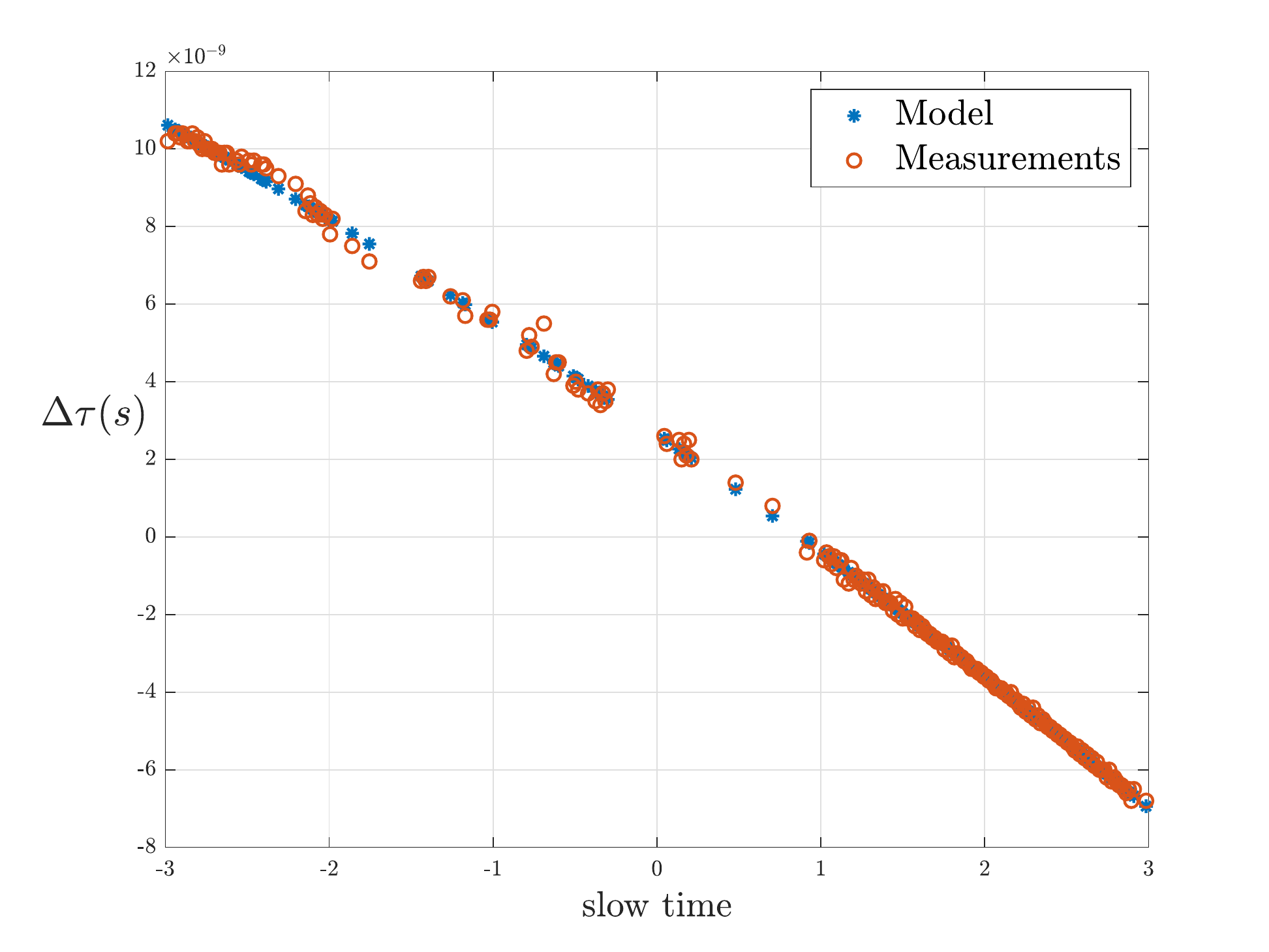}
	\caption{Parameter extraction. Comparison of measured $\Delta\tau(s)$ and analytical one, using the estimated parameters. The sparse part of the data as provided by the TRPCA algorithm is used to determine the velocity parameters, that we subsequently use in the imaging process.} 
	\label{fig:mov_param_extraction}
\end{figure}

\begin{figure}[htbp!]
	\centering
	\begin{subfigure}[t]{0.3275\textwidth}
		\includegraphics[width=1\columnwidth]{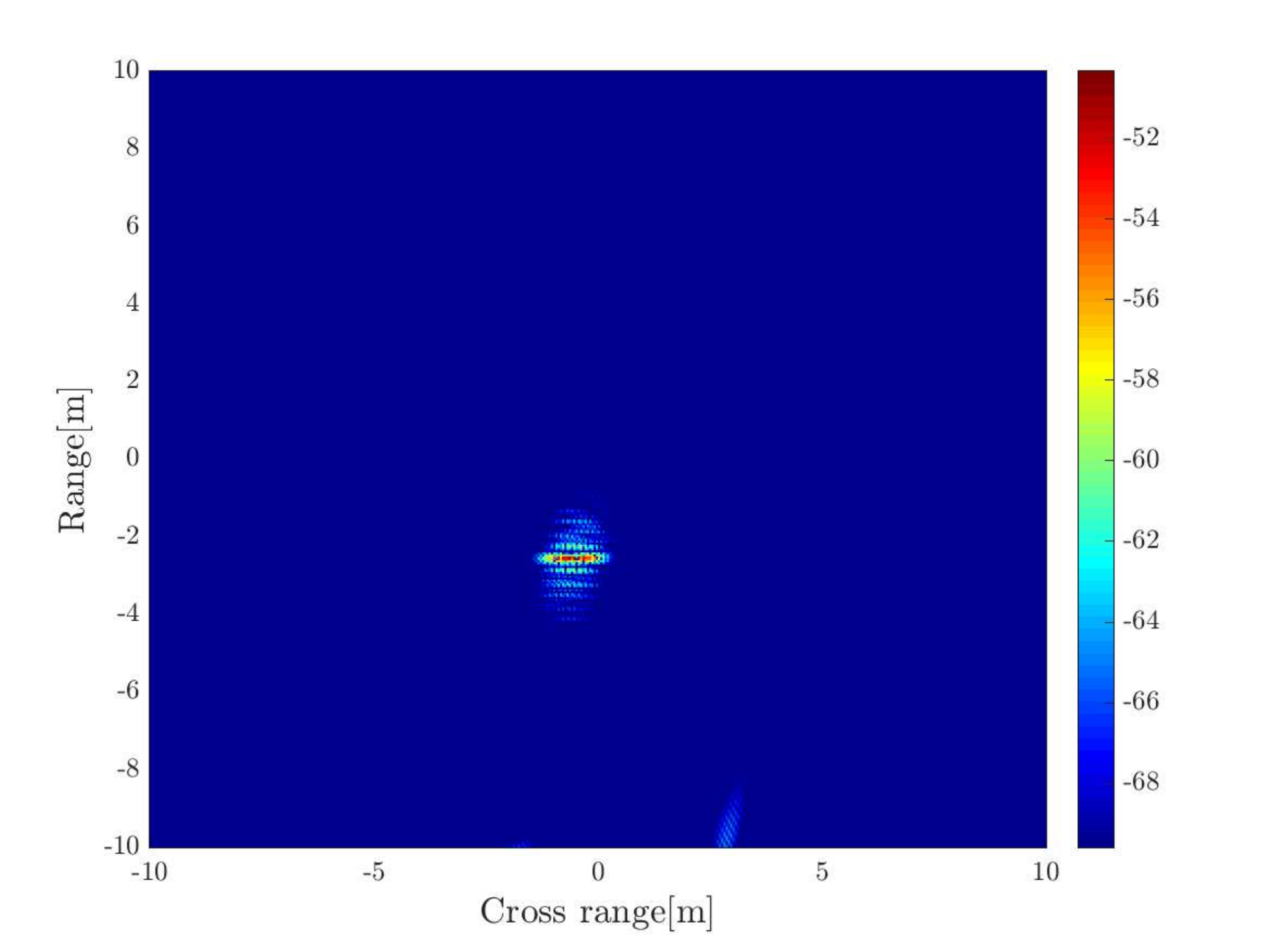}
		\label{}
	\end{subfigure}
	\begin{subfigure}[t]{0.3275\textwidth}
		\centering
		\includegraphics[width=1\columnwidth]{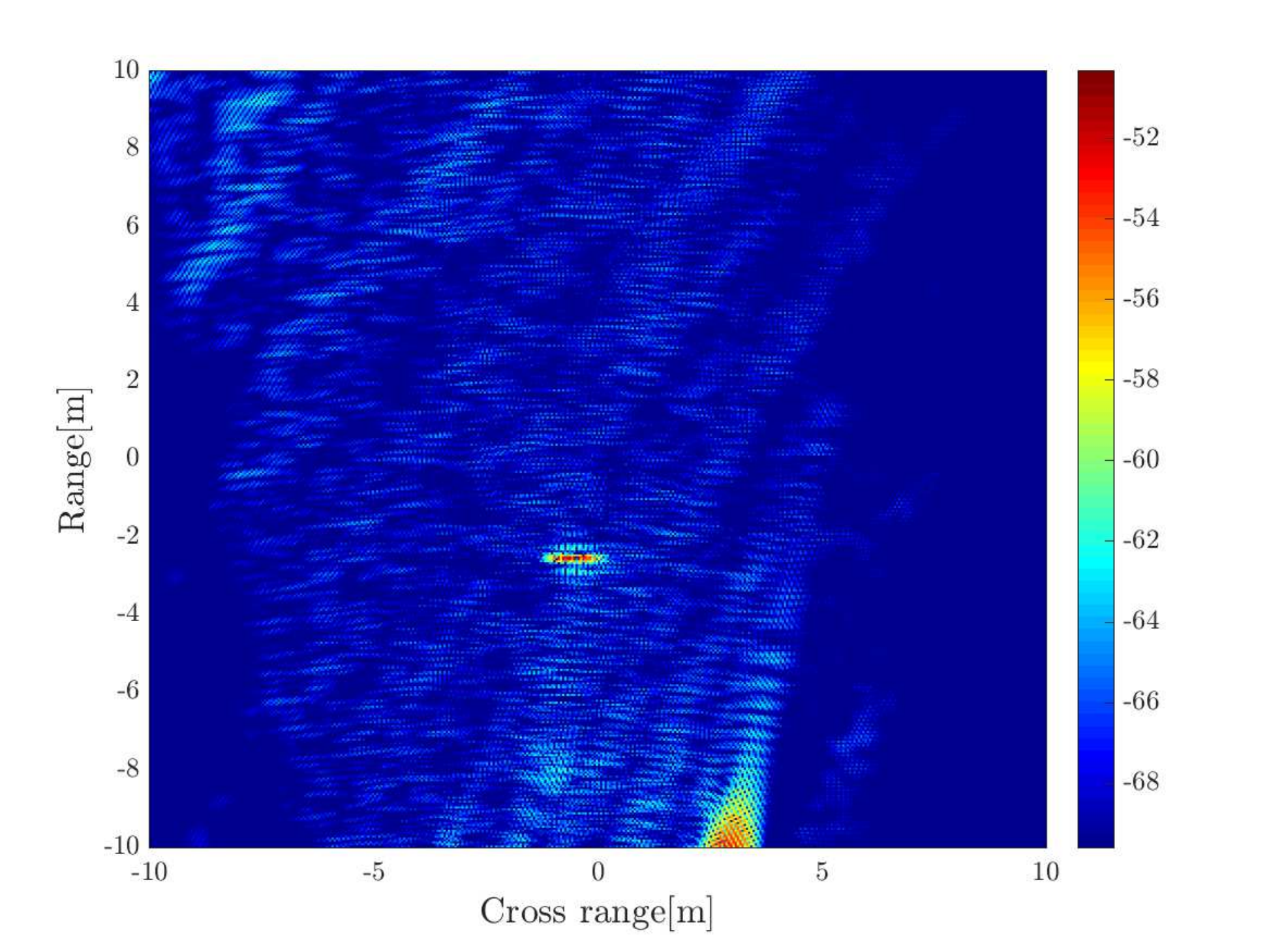}
		\caption*{$\alpha=\pi/4$}
		\label{}
	\end{subfigure}
	\begin{subfigure}[t]{0.3275\textwidth}
		\includegraphics[width=1\columnwidth]{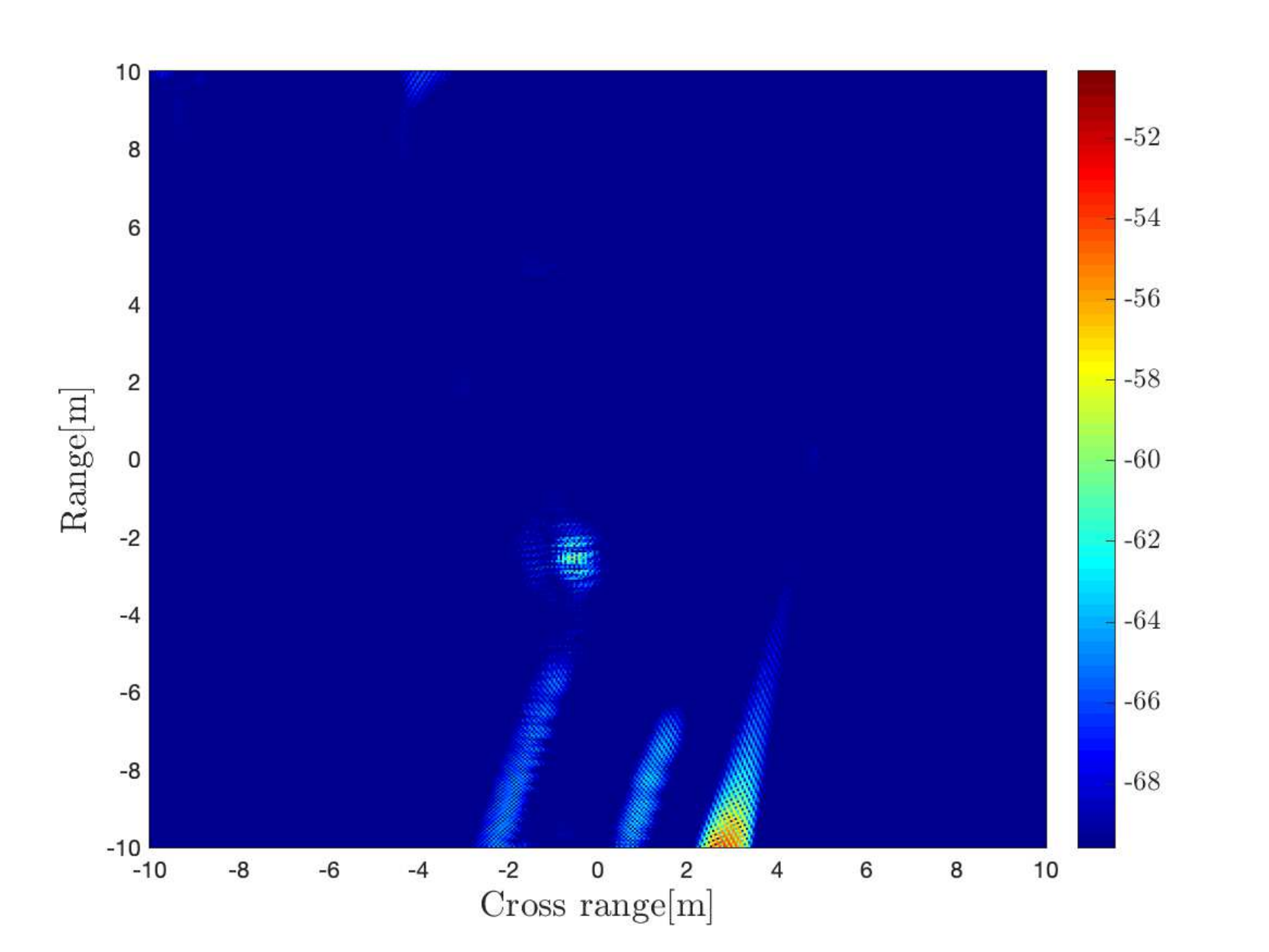}
		\label{}
	\end{subfigure}
	\begin{subfigure}[t]{0.3275\textwidth}
		\includegraphics[width=1\columnwidth]{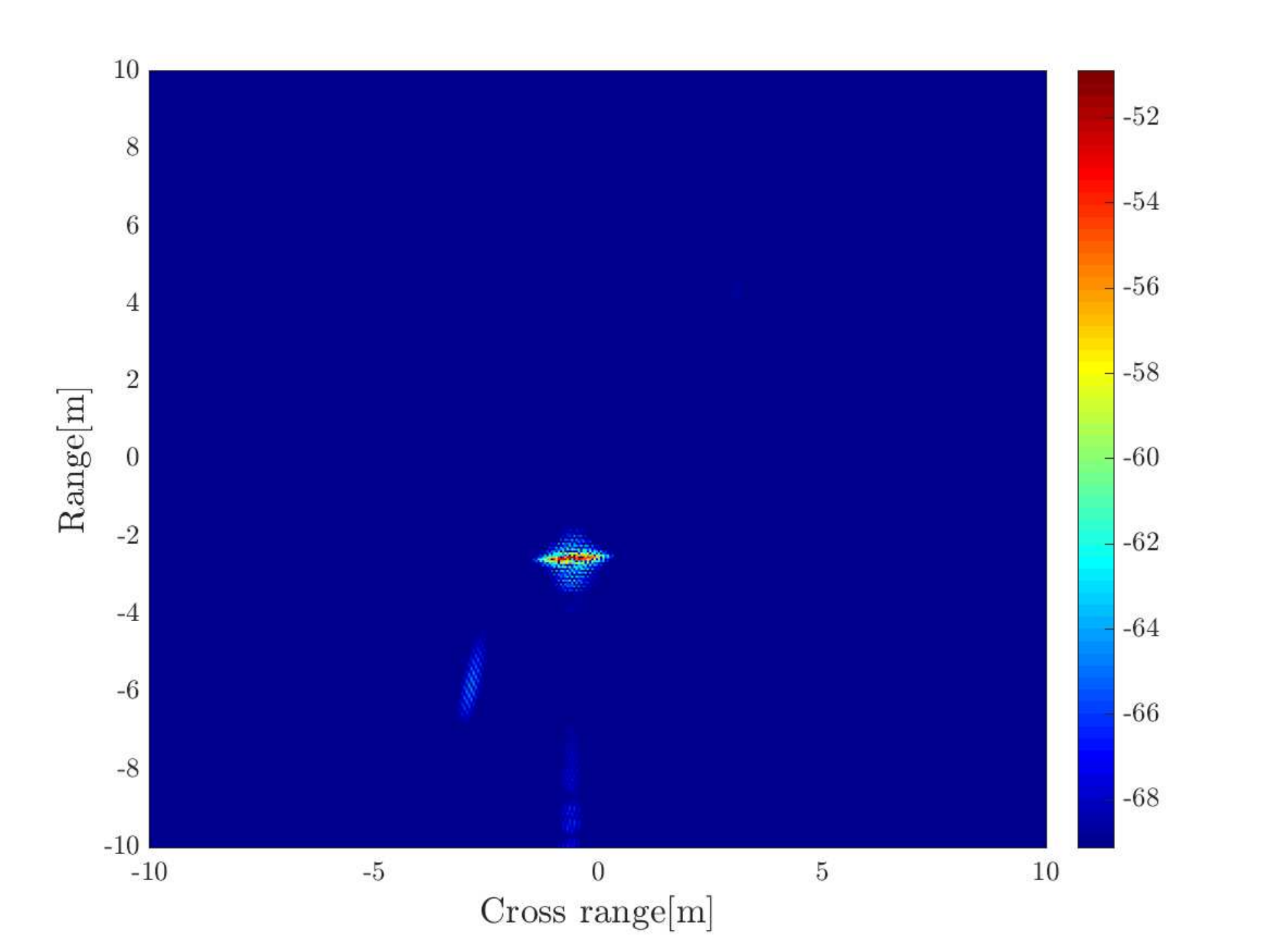}
		\label{}
	\end{subfigure}
	\begin{subfigure}[t]{0.3275\textwidth}
		\centering
		\includegraphics[width=1\columnwidth]{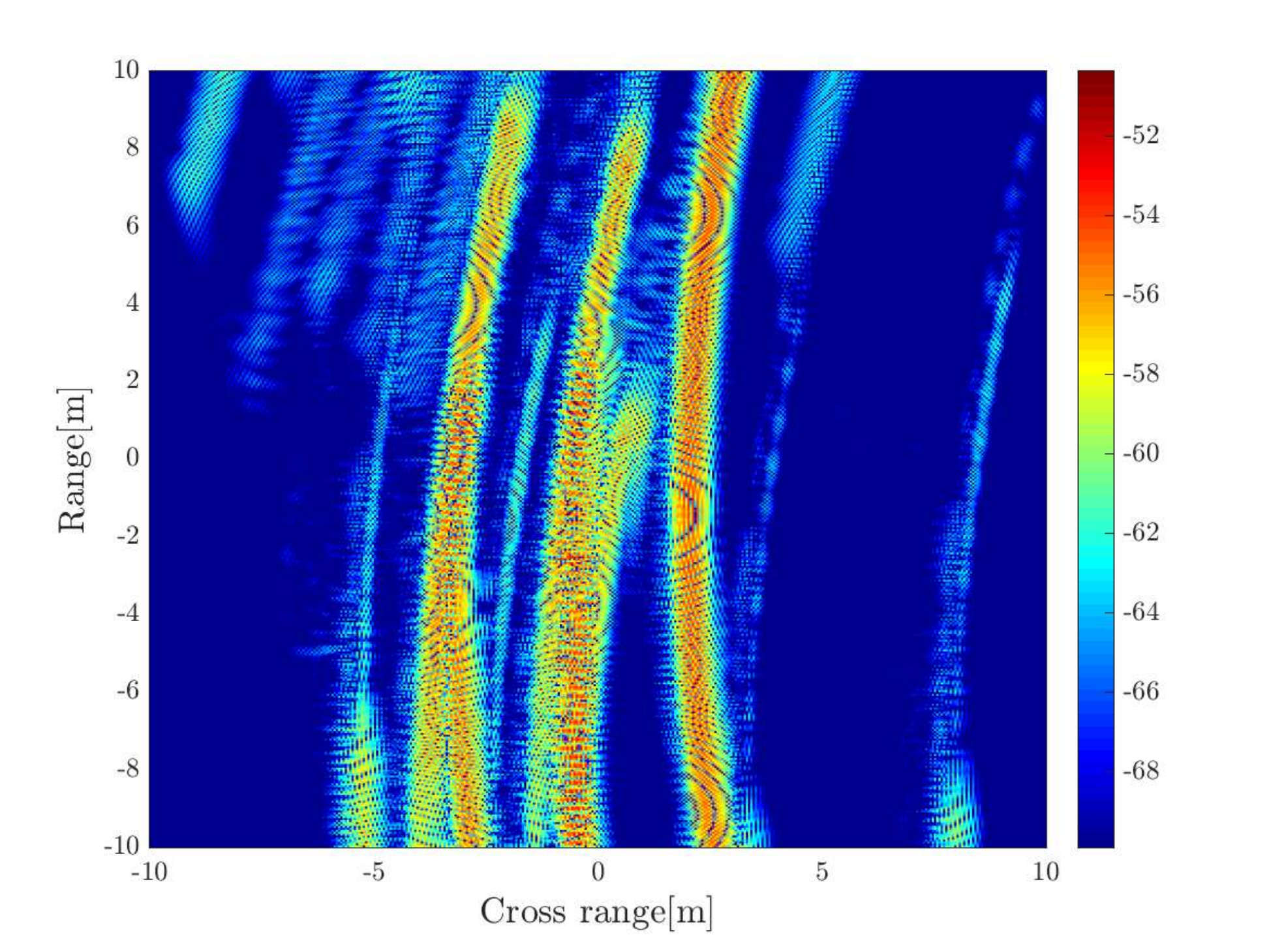}
		\caption*{$\alpha=3\pi/8$}
		\label{}
	\end{subfigure}
	\begin{subfigure}[t]{0.3275\textwidth}
		\includegraphics[width=1\columnwidth]{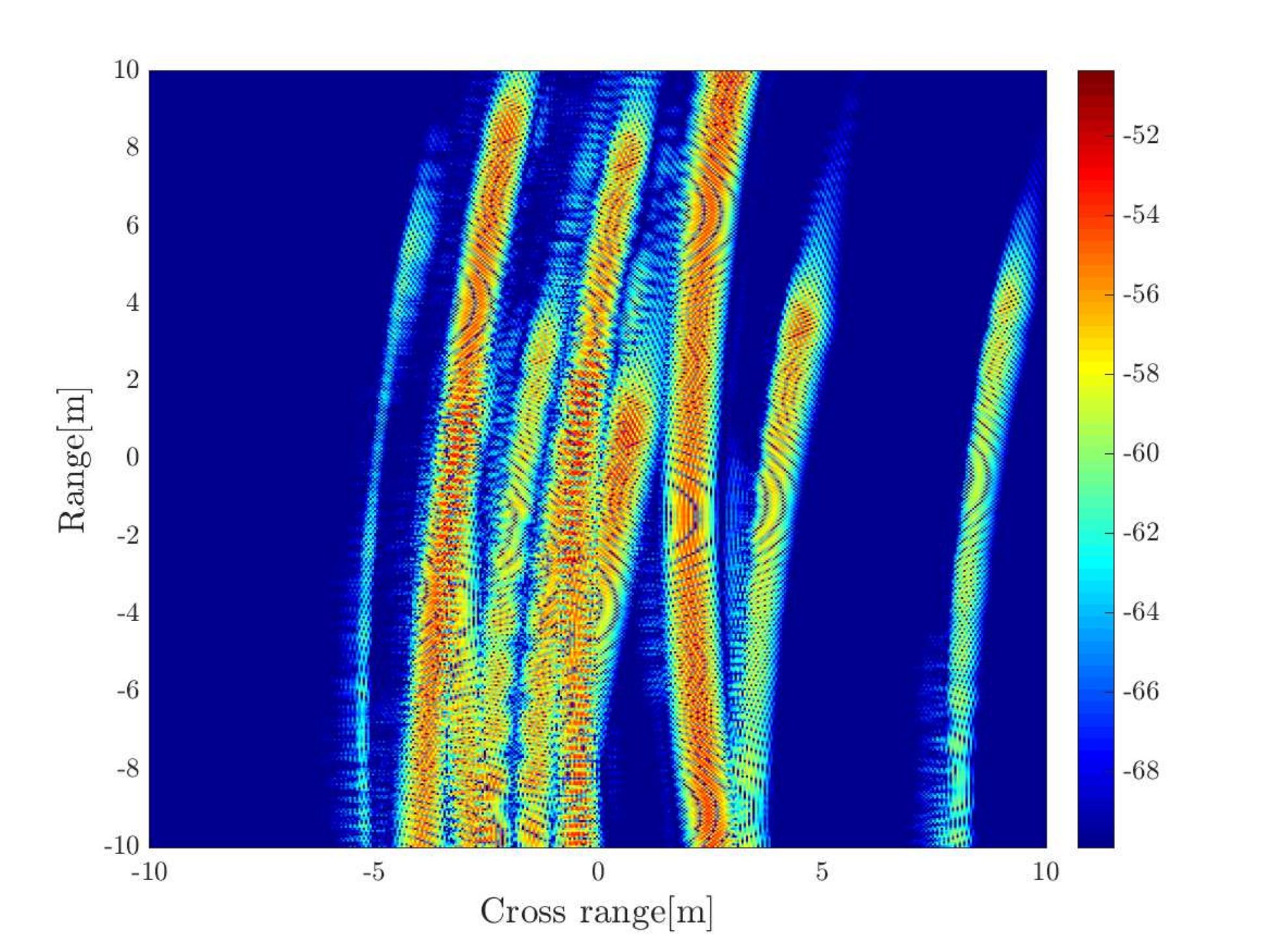}
		\label{}
	\end{subfigure}
	\begin{subfigure}[t]{0.3275\textwidth}
		\includegraphics[width=1\columnwidth]{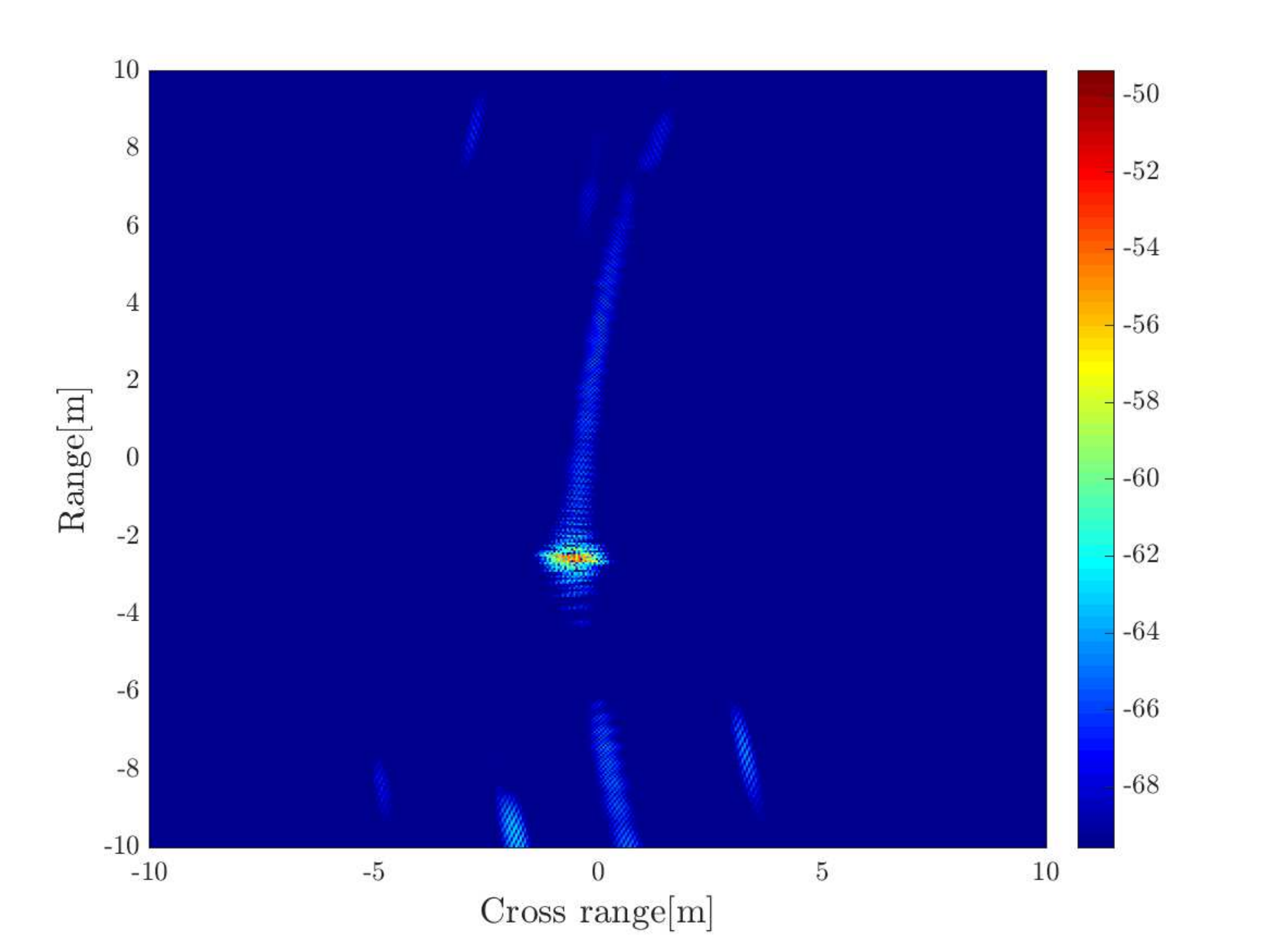}
	\end{subfigure}
	\begin{subfigure}[t]{0.3275\textwidth}
		\centering
		\includegraphics[width=1\columnwidth]{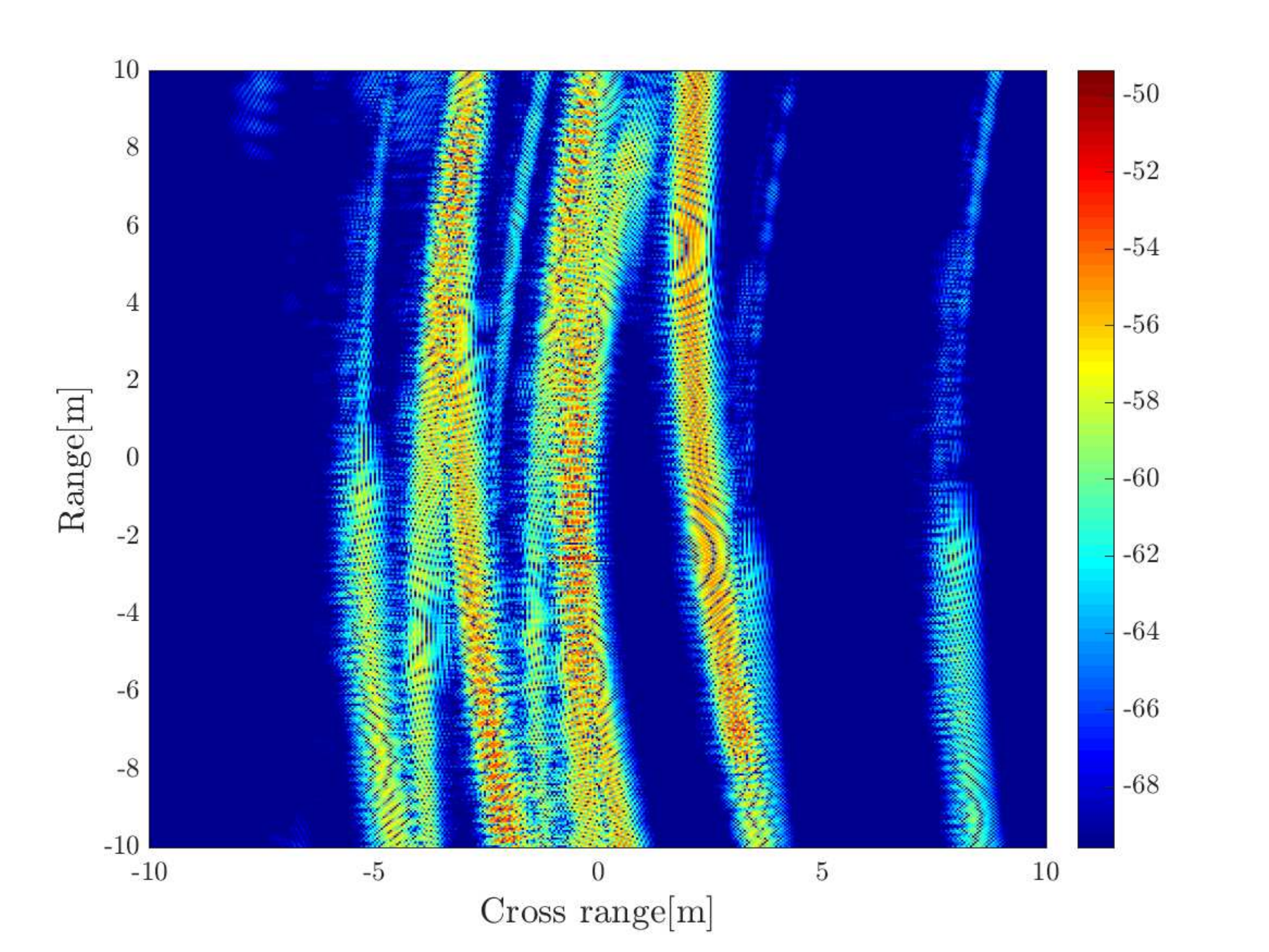}
		\caption*{$\alpha=\pi/2$}
		\label{}
	\end{subfigure}
	\begin{subfigure}[t]{0.3275\textwidth}
		\includegraphics[width=\textwidth]{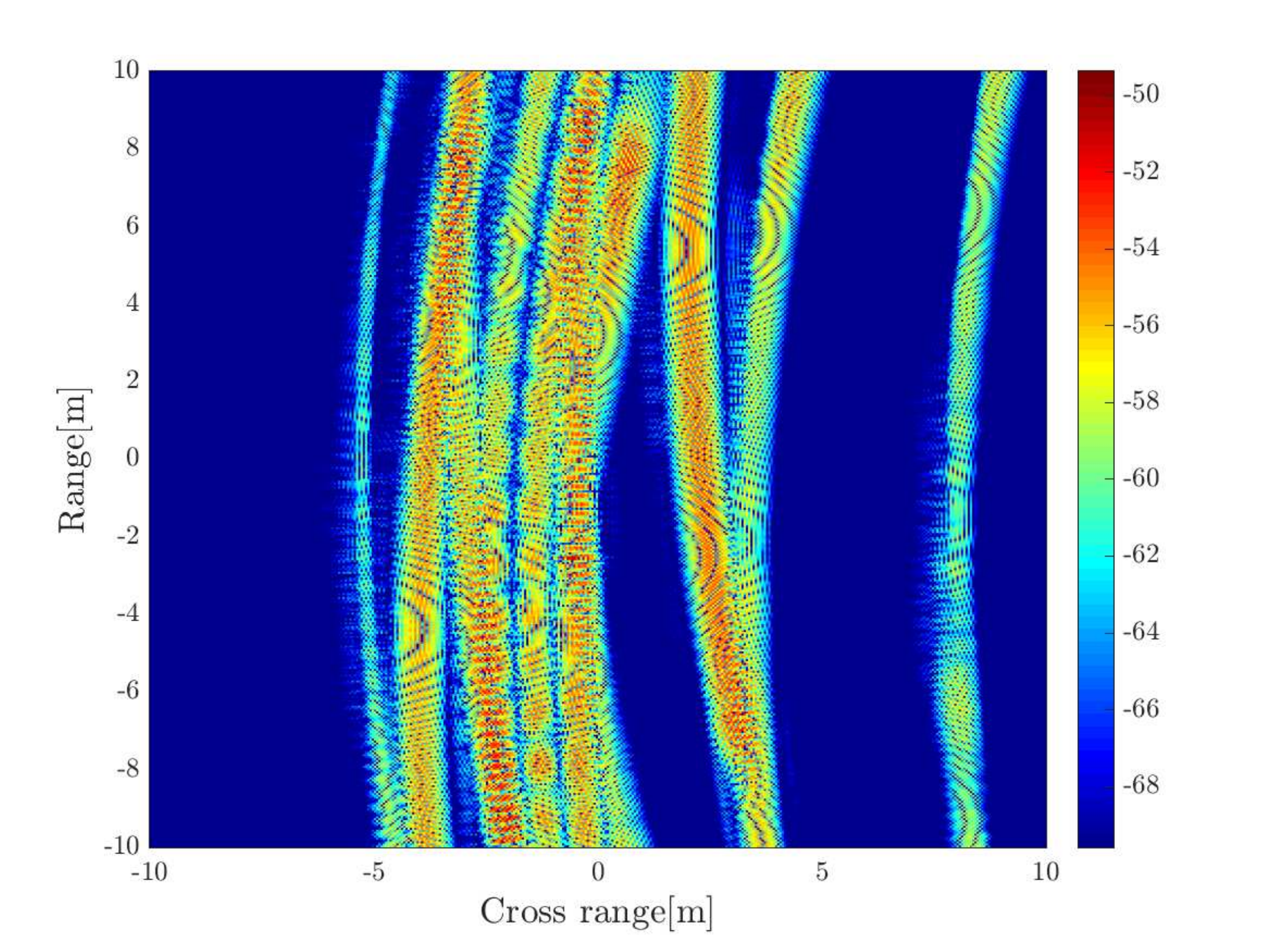}
		\label{}
	\end{subfigure}
\vskip 2mm
\begin{subfigure}[t]{0.3275\textwidth}
	\caption{} 
\end{subfigure}
\begin{subfigure}[t]{0.3275\textwidth}
	\caption{}
\end{subfigure}
\begin{subfigure}[t]{0.3275\textwidth}
	\caption{} 
\end{subfigure}
   \caption{Imaging results for the moving target for $\alpha=0,\pi/4,\pi/2$. $(a)$ using TRPCA ; $(b)$ decoupled RPCA $(c)$ Large aperture matrix RPCA. We observe that TRPCA 
   provides superior imaging results even for large angles ($\alpha=\pi/2$), whereas the other methods fail to suppress the strongly scattering background.}\label{fig:imaging_results}
\end{figure}
\section{Conclusions} \label{sec:5}
In this paper, we considered the problem of motion detection in Synthetic Aperture Radar (SAR) systems. Our main focus was on separating the data originating from moving targets from those of the stationary background. This separation is necessary since SAR's fundamental assumption is that the objects to be imaged do not vary between successive acquisitions. 

We built upon the previously introduced Robust Principal Component Analysis (RPCA) algorithm, separating the data matrix $D$ into low-rank, $L$,  and sparse parts, $S$, associated with the background and the moving target, respectively.  Separation is achieved by solving a convex optimization problem that minimizes the objective $\|L\|_*+\eta\|S\|_1$. RPCA has been shown effective in detecting moving targets. However, there are limitations in terms of the speed and the direction of the moving target. In particular, matrix RPCA fails for targets that are slowly moving, or moving in a direction that is perpendicular to the direction of the SAR platform.

In order to achieve better separation, we have recasted the data matrix as a third order tensor $\mathcal{A}$, made of partially overlapping sub-apertures of the original aperture. We employed a specific extension of the nuclear norm, using the Fourier transform, to define a tensor version of RPCA.  We show with analysis and numerical simulations that this specific extension of the tensor nuclear norm, is well suited for the purpose of SAR motion detection, picking on the non linearity of the phase for moving targets at certain directions, while reducing the norm of the stationary background. Numerical simulations in the X-band SAR surveillance regime demonstrated the performance of TRPCA, especially in challenging cases, where motion separation is unattainable using matrix RPCA. 

This work provides a compelling example of a 2D problem that benefits from its reformulation in tensor form. We have shown that there is a natural specific extension of the nuclear norm for the tensor SAR data, which provides optimal results for the motion detection problem, picking up on features such as curvature, which are lost in the matrix representation.

\section{Acknowledgements}
The work of M. Leibovich and G. Papanicolaou was partially supported by AFOSR FA9550-18-1-0519. The work of C. Tsogka was partially supported by AFOSR FA9550-17-1-0238 and AFOSR FA9550-18-1-0519.  

\newpage

\appendix
\setcounter{equation}{0}
\numberwithin{equation}{section}
\renewcommand{\theequation}{\Alph{section}.\arabic{equation}}
\section{The SAR data matrix}
\label{app:data_model}
We describe here our model for the SAR data matrix for a scene comprised of $N$ small point-like targets. We denote $\sigma_i$ the reflectivity of the  $i$th target and $\vrho_i(s)$ its location at slow time $s$.  Considering that the SAR platform emits a pulse $f(t)$, the down-ramped data are obtained by convolving the received echoes with $\overline{f(-t)}$. It is therefore as if the antenna emitted the pulse $f_p(t)$ defined as
\begin{equation}
f_p(t) = \int dt' f(t') \overline{f(t'-t)} .
\label{eq:f_p}
\end{equation}
Considering the free-space causal Green's function for the scalar wave equation
\begin{equation}
G(\vr,t,\vr',t')=\frac{\delta(\|\vr-\vr'\|-c(t-t'))}{4\pi\|\vr-\vr'\|}\mathbbm{1}_{t>t'},
\end{equation}
and assuming 
\begin{equation}
\|\vrho_i-\vrho_j\|\ll\|\vr(s)-\vrho_i\|,\forall i, s
\end{equation}
we can model the down-ramped SAR data matrix as 
\begin{equation} \label{eq:MOD1}
D_r(s,t)\propto\sum\limits_{i=1}^N\sigma_i  f_p(t - \Delta \tau_i(s)), 
\end{equation}
i.e., the data are proportional to a superposition of pulses $f_p$ shifted by $\Delta \tau_i(s)$ and multiplied by $\sigma_i$, the reflectivity of each target.
 
The difference travel time $\Delta \tau_i(s)$ is the round trip travel time from the antenna location $\vr(s)$ to the target $\vrho_i(s)$ from which the round trip travel time from the antenna location $\vr(s)$ to the reference location $\vrho_o$ is subtracted to account for the range compression step of down-ramping,
\begin{equation}
\Delta \tau_i(s)=\frac{2(\|\vr(s)-\vrho_i(s)\|-\|\vr(s)-\vrho_o\|)}{c}.
\label{eq:Delta_tau}
\end{equation}
We used here the start-stop approximation which neglects the targets' displacement during the round trip travel time. This is justified in radar because the electromagnetic waves 
travel at the speed of light, which is many orders of magnitude larger than the speed of the targets and the platform.

It is common to assume that $f_p(t)$ consists of a base-band waveform
$f_B(t)$ modulated by a carrier frequency $\nu_o = \om_o/(2 \pi)$,
\begin{equation}
f_p(t) =  \cos(\om_o t) f_B(t).
\label{eq:1.1}
\end{equation}
Its Fourier transform is
\begin{eqnarray}
\hat f_p(\om) = \int d t \, f_p(t) e^{i \om t}
= \frac{1}{2} \left[\hat f_B(\om + \om_o) +
\hat f_B(\om - \om_o)\right].
\end{eqnarray}
Here $\hat f_B(\om)$ is supported in the interval $[-\pi B,\pi B]$, where
$B$ is the bandwidth while  $\hat f(\om)$ is supported in $[-\omega_o-\pi B,-\omega_o+\pi B] \cup [\omega_o-\pi B,\omega_o+\pi B]$.

The down-ramped SAR data matrix takes the form
\begin{equation}
D_r(s,t)=\sum\limits_{i=1}^N\sigma_i \cos(\omega_o(t-\Delta \tau_i(s)))f_B(t-\Delta \tau_i(s)),
\label{eq:Dr_time}
\end{equation}
and its Fourier transform is 
\begin{equation}
\hat{D}_r(s,\omega)=\sum\limits_{i=1}^N\frac{e^{-i\omega\Delta \tau_i(s)}}{2}\left[\hat{f}_B(\omega+\omega_o)+\hat{f}_B(\omega-\omega_o)\right].
\label{eq:Dr_freq}
\end{equation}

A lossless baseband transformation, introduced in \cite{leib2018RPCA}, is applied to the data, removing the carrier frequency $\omega_o$

\begin{equation}
\hat{D}_{r {\text {\tiny B}}}(s,\omega)=\hat{h}_{LP}(\omega)\hat{D}_o(s,\omega)= e^{-i(\omega-\omega_o) \Delta\tau(s)}\hat{f}_B(\omega).
\end{equation}

\begin{equation}
\hat{h}_{LP}(\omega)=\begin{cases} 2& |\omega|\le 2B\\ 0& |\omega|>\alpha B\end{cases},\quad{ \alpha B\ll \omega_o}.
\end{equation}
\subsection{Matrix robust principal component analysis}
\label{subsec:Matrix_rpca}
To produce good imaging results, one needs to detect and separate the echoes corresponding to moving targets from the ones of the complex background.

We follow here the approach proposed in \cite{borcea2013synthetic}, where the robust principal component analysis was used for separating the data matrix $\cD$ as in \eqref{eq:MOD6} in two subsets: the echoes due to stationary targets that form the low rank part of the data matrix and the moving targets echoes which constitute the sparse part. RPCA consists of solving the following convex optimization problem

\begin{equation}
\begin{split}
& \min_{L,S\in\mathbb C^{n_1\times n_2}}
\quad ||L||_*+ \eta ||S||_1\\
& \text{subject to} \quad L+S=D. 
\end{split}
\label{pca2_2}
\end{equation}
The idea is that the rank of a matrix, which is a non convex objective, can be relaxed to the nuclear norm, that is the sum of singular values, while the element wise $\ell_1$ norm, promotes sparsity. 
The parameter $\eta$ balances the ratio between the nuclear norm of $L$ and the $\ell_1$-norm of $S$. The recommended value for $\eta$ proposed in \cite{candesRPCA} and used for SAR data in \cite{borcea2013synthetic} is  
\begin{equation}
\eta = \frac{1}{\sqrt{\max
		\{n_1, n_2\}}}.
\label{eq:eta}
\end{equation}
It was shown in \cite{borcea2013synthetic} that with this choice for $\eta$ the RPCA algorithm is sensitive to the window size of the data.

As was shown in \cite{leib2018RPCA}, an optimal choice of $\eta$ exists for which the separation is robust for the SAR problem. By evaluating the different norms for stationary and moving objects, one can find a range of admissible values for $\eta$, and an optimal one can be derived for a target moving at a velocity $v_t$,
\begin{equation}
\eta^*=\sqrt{\frac{\Delta s B\Delta t}{4 S(a)\sqrt{\pi}}}   \sqrt{\frac{ 1}{\sqrt{\frac{ N(\vec{\pmb{v}}_t)B\Delta t}{2\sqrt{\pi}}+\frac{1}{2}}} \frac{\frac{\sqrt{2}N(\vec{\pmb{v}}_t)B \Delta t}{\pi}+1}{2}},
\label{eq:eta_opt_mat}
\end{equation}
where, $N(\vec{\pmb{v}}_t)$ is the column support of the moving target, to leading order,
\begin{equation}
N(\vec{\pmb{v}}_t)\approx\frac{4S(a) }{\Delta t }\frac{1}{\|\vr(0)-\vrho_o(0)\|}(\vr(0)-\vrho_o(0))\cdot \frac{v_t}{c},
\label{eq:Nvt}
\end{equation}
and $S(a)$ is the total slow-time aperture size.
In \cite{leib2018RPCA}, it is also shown, that the separation improves as $N(\vec{\pmb{v}}_t)$ increases.
   \section{Tensor representation and decomposition}
   \label{app:tensor_decomp}
Matrix decomposition methods such as SVD, Principal Component Analysis (PCA), Non negative Matrix Factorization (NMF), have enjoyed tremendous success in data science, and have proved to be an essential tool in processing high volume data, from dimensionality reduction to classification. 

Tensors are extension of vectors and matrices to higher orders. For example, a $p$-order tensor $\mathcal{A}$ is defined as
\begin{equation}
\mathcal{A} \in \mathbb{C}^{n_1\times n_2 \cdots\times n_p}, \qquad a_{i_1,i_2,\cdots,i_p}\in \mathbb{C},\quad  i_k\in\{1,\cdots,n_k\}
\label{eq:ten_def}
\end{equation}

The use of tensors is appealing, since it generalizes the properties of matrix data structures, and allows for the representation of more complex patterns. 
\subsection{Tensor Decomposition methods}
Indeed, tensor representations and decomposition of data have been of great interest for several decades. They arise naturally in applications involving high dimensional data, historically in psychometrics \cite{tucker1966some}, chemometrics \cite{appellof1981strategies}, and more recently in statistics \cite{shashua2005non}, bioinformatics \cite{hore2016tensor}, finance \cite{jondeau2018moment}, signal processing \cite{cichocki2015tensor} and many more. 

A tensor decomposition is defined as a representation of the tensor as a sum of other tensors. For example, a rank-1 SVD decomposition would be of the form
\begin{equation}
a_{i_1,i_2,\cdots,i_p}=\sum\limits_{j=1}^r \sigma_j u^{(j,1)}_{i_1}\otimes u^{(j,2)}_{i_2}\cdots\otimes u^{(j,p)}_{i_p}, \quad u^{(j,k)}\in \mathbb{R}^{n_k}, \|u^{(j,k)}\|_2=1,
\label{eq:ten_dec}
\end{equation}
where $\{u^{(\cdot,k)}\}$ are an orthogonal basis of $\mathbb{R}^{n_k}$.

We can see that the regular SVD is a particular case of this decomposition
\begin{equation}
a_{i_1,i_2}=\sum\limits_{j=1}^r \sigma_j u^{(j)}_{i_1}v^{(j)*}_{i_2},
\end{equation}
identifying $u^{(j)}=u^{(j,1)},v^{(j)}=u^{(j,2)*}$.

However, some of the properties of matrices do not carry on  to higher dimensional objects.
Decompositions of the form of \eqref{eq:ten_dec} do not exist for general higher order tensors (there are multiple counter examples). Notice that the orthogonality requirement requires $r\le \min\limits_k n_k$, which does not hold in general. Thus, a need arises for other extensions of matrix decomposition. 

Some extensions of matrix decomposition such as Tucker Higher Order Singular Value Decomposition (HOSVD), Canonical Polyadic (CPD), or Tensor Train (TT) have enjoyed applicability \cite{kolda2009tensor}, but they involve a more complex structure than \eqref{eq:ten_dec}. For example, in HOSVD \cite{de2000multilinear}, the singular values are replaced by a core tensor so that the representation is
\begin{equation}
\begin{split}
&a_{i_1,i_2,\cdots,i_p}=\sum\limits_{j_1=1}^{r_1}\cdots \sum\limits_{j_p=1}^{r_p} \sigma_{\mathbf{J}} \hspace{0.3em}u^{(j_1,1)}_{i_1}\otimes u^{(j_2,2)}_{i_2}\cdots\otimes u^{(j_p,p)}_{i_p}\\
& u^{(j_k,k)}\in \mathbb{R}^{n_k}, \quad \|u^{(j_k,k)}\|_2=1,\quad \mathbf{J}=\{j_1,\cdots,j_p\}
\label{eq:hosvd}
\end{split}
\end{equation}

Specifically, for our purposes, the notion of a rank of tensor, in the sense of a decomposition such as \ref{eq:ten_dec}, is not well understood, and its computation has been shown to be NP hard \cite{hillar2013most}.  

\subsection{Tensor nuclear norm  and its estimates}
RPCA relies on the nuclear norm as a relaxed rank estimate. Recall that for a matrix $A$  the nuclear norm is defined as
\begin{equation}
\|A\|_*=\sum\limits_{i=1}^r\sigma_i(A)=\max\limits_{\|X\|_\sigma\le1}\langle A,X\rangle,
\label{eq:nucnorm_def}
\end{equation} 
where
\begin{equation}
\langle A,X\rangle=\Tr(A^HX),
\end{equation}
is the regular matrix inner product and 
\begin{equation} \|X\|_\sigma= \max\limits_{\|u\|_2\le 1}\|Xu\|_2=\max\limits_{\|u\|_2,\|v\|_2\le1}\langle X,u\otimes v\rangle .
\end{equation}
is the matrix spectral norm (or two-norm) \cite{golub2012matrix}. 

The definition \eqref{eq:nucnorm_def} is also equivalent to 
\begin{equation}
\|A\|_*=\inf\left\{\sum\limits_{i=1}^r|\lambda_i| \hspace{0.5em}\Big| \hspace{0.5em}A=\sum\limits_{i=1}^r \lambda_i u_i\otimes v_i,\|u_i\|=\|v_i\|=1,r\in\mathbb{N}\right\}.
\label{eq:nuc_norm_matrix_alt}
\end{equation}

A possible extension of \eqref{eq:nuc_norm_matrix_alt} to higher order is  \cite{friedland2018nuclear}
\begin{equation}
\|\mathcal{A}\|_{*,\mathcal T}=\inf\left\{\sum\limits_{i=1}^r|\lambda_i| \hspace{0.5em}\Big| \hspace{0.5em}\mathcal{A}=\sum\limits_{i=1}^r \lambda_i u^1_i \otimes u^2_i\otimes\cdots \otimes u_i^d,\|u_i^j\|=1,r\in\mathbb{N}\right\}.
\label{eq:nuc_norm_matrix_alt_ten}
\end{equation}
Notice that $u_{i}^{j_1},u_{i}^{j_2}$ need not be orthogonal for $j_1\neq j_2$. 

Definition \eqref{eq:nuc_norm_matrix_alt_ten} retains many of the favorable properties of the matrix nuclear norm. Specifically, it is the dual norm of the tensor spectral norm.  

Define for two tensors $\mathcal{A,B}\in\mathbb{R}^{n_1\times n_2\times\cdots\times n_d}$ their inner product as
\begin{equation}
\langle \mathcal{A,B}\rangle=\sum\limits_{i_1=1}^{n_1}\sum\limits_{i_2=1}^{n_2}\cdots\sum\limits_{i_d=1}^{n_d}a^*_{i_1,i_2,\dots,i_d}b_{i_1,i_2,\dots,i_d}
\label{eq:ten_inner}
\end{equation}
The spectral norm is naturally defined via
\begin{equation}
\|\mathcal{A}\|_\sigma=\max\limits_{\|x_k\|_2\le 1,k=1,\cdots,d}\langle \mathcal{A},x_1\otimes x_2\otimes\cdots\otimes x_d\rangle
\end{equation}
and the nuclear norm is then defined as its dual
\begin{equation}
\|\mathcal{A}\|_{*,\mathcal T}=\max\limits_{\|\mathcal{X}\|_\sigma\le 1}\langle \mathcal{A,X}\rangle.
\label{eq:ten_nuc_norm_def}
\end{equation}

These definitions reduce to the usual matrix definition in two dimensions. However, the computation of both spectral and nuclear norm has been shown to be NP-hard problems in general \cite{friedland2018nuclear}.

Commonly used tensor decompositions do not necessarily provide natural extensions for the nuclear norm, or approximate bounds for it. For example,  we cannot take the core tensor of \eqref{eq:hosvd} as a nuclear norm estimate, since $\sum\limits_\mathbf{J}|\sigma_{\mathbf{J}}|$  does not obey the triangle inequality, and hence is not a norm (or a convex objective). We thus look for other extensions of singular values to higher dimensions, presented in the main paper.
\section{Proof of Theorem~\ref{theorem_block} and Corollary~\ref{col:cor_1}}
We will first prove Theorem~\ref{theorem_block} that we recall next. 
\begin{theorem*}
For matrices $A_1,A_2,\dots,A_k,\quad A_i\in\mathbb{C}^{m\times n_i}$,  for the matrix  
\begin{equation}
\pmb{A}=[A_1,A_2,\dots,A_k]\in \mathbb{C}^{m\times N},\quad N=\sum\limits_{i=1}^k n_i,
\end{equation}
the following inequalities hold
\begin{equation}
\left(\sum\limits_{i=1}^k\|A_i\|_*^2\right)^{1/2}\le	\|\pmb{A}\|_*\le \sum\limits_{i=1}^k\|A_i\|_*
\end{equation}
\end{theorem*}
\begin{proof}
	Write 
\begin{equation}
\pmb{A}=[A_1,\underbrace{0,\dots,0}_{k-1 \text{ times} }]+[0,A_2,\underbrace{0,\dots,0}_{k-2 \text{ times} }]+\dots+[\underbrace{0,\dots,0}_{k-1 \text{ times} },A_k],
\end{equation}
and apply the triangle inequality to get the upper bound. 

The lower bound is attained following \cite{li2016bounds}. First, we can bound the spectral norm in the following way:
write
\begin{equation}
\pmb{x}^T=[y_1^T,y_2^T,\dots,y_k^T],\quad \pmb{x}\in\mathbb{C}^N,y_i\in\mathbb{C}^{n_i}
\end{equation}
\begin{equation}
\begin{split}
\|\pmb{A}\|_\sigma=\max\limits_{\pmb{x}\in\mathbb{C}^{N},\|\pmb{x}\|_2^2\le 1}\|\pmb{A} \pmb{x}\|_2&=\max\limits_{y_i\in\mathbb{C}^{n_i},\sum\limits_{i=1}^k\|y_i\|_2^2	\le 1}\|\sum\limits_{i=1}^kA_iy_i\|_2\le \max\limits_{y_i\in\mathbb{C}^{n_i},\sum\limits_{i=1}^k\|y_i\|_2^2	\le 1}\sum\limits_{i=1}^k\|A_iy_i\|_2.
\end{split}
\end{equation}

Write $\xi_i=\|y_i\|_2,\quad y_i=\xi_i\tilde{y_i}$

Then
\begin{equation}
\begin{split}
\|\pmb{A}\|_\sigma&\le \max\limits_{\sum\limits_{i=1}^k\xi_i^2\le 1,\hspace{0.1em}\|\tilde{y_i}\|\le 1}\sum\limits_{i=1}^k\xi_i\|A_i\tilde{y_i}\|_2=\max\limits_{{\sum\limits_{i=1}^k}\xi_i^2\le 1}\sum\limits_{i=1}^k\xi_i\max\limits_{\|\tilde{y_i}\|\le 1}\|A_i\tilde{y_i}\|_2=\max\limits_{\sum\limits_{i=1}^k \xi_i^2\le 1} \sum\limits_{i=1}^k \xi_i\|A_i\|_\sigma=\left(\sum\limits_{i=1}^k\|A_i\|_\sigma^2\right)^{1/2}.
\end{split}
\end{equation}
The last result is achieved by noting that, identifying $u_i=\xi_i,v_i=\|A_i\|_\sigma$, we have for $v\in\mathbb{R}^k_+$
\begin{equation}
\max\limits_{u\in\mathbb{R}_+^k,\|u\|_2\le 1}\left\langle u,v\right\rangle=\left\langle \frac{v}{\|v\|_2},v\right\rangle= \|v\|_2=\left(\sum\limits_{i=1}^d v_i^2\right)^{1/2}
\end{equation}

The nuclear norm is defined as
\begin{equation}
\|\pmb{A}\|_*=\max\limits_{\|\pmb{X}\|_\sigma\le1}\langle \pmb{A},\pmb{X}\rangle=\max\limits_{\|\pmb{X}\|_\sigma\le1}\sum\limits_{i=1}^k\langle A_i,Y_i\rangle,\quad \pmb{X}=[Y_1,Y_2,\dots,Y_k]\in\mathbb{C}^{m\times N}, Y_i\in\mathbb{C}^{m\times n_i}.
\end{equation}
From the bound on the spectral norm we have 
\begin{equation}
\left\{\pmb{X}\in\mathbb{C}^{m\times N}\Big|\sum\limits_{i=1}^k\|Y_i\|_\sigma^2\le 1\right\}\subseteq \left\{\pmb{X}\in\mathbb{C}^{m\times N}\Big| \|X\|_\sigma\le 1\right\}.
\end{equation}

Hence,
\begin{equation}
\begin{split}
\|\pmb{A}\|_*=\max\limits_{\|\pmb{X}\|_\sigma\le1}\sum\limits_{i=1}^k\langle A_i,Y_i\rangle\ge \max\limits_{\sum\limits_{i=1}^k\|Y_i\|_\sigma^2\le1}\sum\limits_{i=1}^k\langle A_i,Y_i\rangle
\end{split}
\end{equation}
Define $y_i=\|Y_i\|_\sigma,\quad Y_i=y_i\tilde{Y_i}$.

Then,
\begin{equation}
\begin{split}
\|\pmb{A}\|_*\ge &\max\limits_{\sum\limits_{i=1}^k y_i^2\le1,\hspace{0.1em}\|\tilde{Y_i}\|_\sigma\le1 }\sum\limits_{i=1}^ky_i\langle A_i,\tilde{Y_i}\rangle=\max\limits_{\sum\limits_{i=1}^k y_i^2\le1}\sum\limits_{i=1}^ky_i \max\limits_{\|\tilde{Y_i}\|_\sigma\le1}\langle A_i,\tilde{Y_i}\rangle
=\max\limits_{\sum\limits_{i=1}^k y_i^2\le1}\sum\limits_{i=1}^ky_i\|A_i\|_*=\left(\sum\limits_{i=1}^k\|A_i\|_*^2\right)^{1/2}.
\end{split}
\end{equation}
\end{proof}

We next prove Corollary~\ref{col:cor_1} that we recall first. 
\begin{corollary*}
	For $A_1,\dots,A_k$ in \eqref{eq:nuc_norm_mat}
	\begin{enumerate}[(a)]
	\item If all the matrices are mutually orthogonal  
	\begin{equation}
	\sum\limits_{i=1}^k \text{rank}(A_i)\le m, \quad A_i^HA_j=0 , \quad \forall i\neq j,
	\end{equation}
	the upper bound of \eqref{eq:nuc_norm_bounds} is attained.
	\item If $A_i=\beta_iA, \hspace{0.5em} \beta_i\in\mathbb{C}$, then the lower bound of \eqref{eq:nuc_norm_bounds} is attained and
	\begin{equation}
	\|\pmb{A}\|_*=\|\pmb{\beta}\|_2\|A\|_*,\quad \|\pmb{\beta}\|_2=\left(\sum\limits_{i=1}^k|\beta_i|^2\right)^{1/2}.
	\end{equation}	
	\end{enumerate}
	\label{}
	\nonumber
\end{corollary*}
\label{app:cor_proof}
\begin{proof}
\begin{enumerate}[(a)]
	\item In this case $\pmb{A}^H\pmb{A}$ is block diagonal
	\begin{equation}
	\pmb{A}^H\pmb{A}=\begin{pmatrix}A_1^HA_1&0&\cdots &0\\0&A_2^HA_2&\cdots&0\\
	&&\ddots&\\0&&&A_k^HA_k\end{pmatrix},
	\end{equation}
	and the result is trivial, since
	\begin{equation}
	\|\pmb{A}\|_*=\sum\limits_{i=1}^N \left(\lambda_i(\pmb{A}^H\pmb{A})\right)^{1/2}=\sum\limits_{i=1}^k \sum\limits_{j=1}^{n_i}\left(\lambda_j(A_i^HA_i)\right)^{1/2}=\sum\limits_{i=1}^k\|A_i\|_*.
	\end{equation}
	\item
	Let us look at $\pmb{A}^H\pmb{A}$
	\begin{equation}
	\pmb{A}^H\pmb{A}=\begin{pmatrix}|\beta_1|^2A^{H}A&\beta_1^*\beta_2 A^HA&\cdots &\beta_1^*\beta_{k}A^HA\\\beta_2^* \beta_1A^HA&|\beta_2|^2A^HA&\cdots&\cdots\\
	&&\ddots&\\ \beta_{k}^* \beta_1A^HA&&&|\beta_{k}|^2A^HA\end{pmatrix}=\mathcal{B}\otimes A^HA,
	\end{equation}
\end{enumerate}
where $\mathcal{B}_{ij}=\beta_i^*\beta_j$, and $\otimes$ is the tensor product
\begin{equation}
A\in\mathbb{R}^{m\times n},B\in\mathbb{R}^{p\times 	qy},\quad C=B\otimes A =\begin{pmatrix} b_{11}A&b_{12}A&\cdots&b_{1q}A\\
b_{21}A&b_{22}A&\cdots&a_{2q}A\\
\vdots&\vdots & &\vdots \\b_{p1}A&b_{p2}A&\cdots&b_{pq}A\end{pmatrix}\in\mathbb{R}^{mp\times nq}
\end{equation}

$\mathcal{B}$ has an eigenvalue $\|\pmb{\beta}\|_2^2=\left(\sum\limits_{i=1}^k|\beta_i|^2\right)$  with corresponding eigenvector $v_j=\beta_j^*$, and it is easy to see it is an Hermitian, rank one matrix. i.e.,
\begin{equation}
\mathcal{B}=T\Lambda T^H, \Lambda=\text{diag}(\|\pmb{\beta}\|_2^2,0,\dots,0).
\end{equation}
Thus,  $\pmb{A}^H\pmb{A}$ can be block-diagonalized by $T\otimes\mathbb{I}$, to get
\begin{equation}
(T\otimes\mathbb{I})^HA^HA(T\otimes\mathbb{I})=\begin{pmatrix}\|\pmb{\beta}\|_2^2A^{H}A&0&\cdots &0\\0&0&\cdots&0\\\vdots
&\vdots&&\vdots\\0&0&\cdots&0\end{pmatrix}.
\end{equation}
Hence $\sigma_j(\pmb{A})=\left(\lambda_j(\pmb{A}^H\pmb{A}))\right)^{1/2}=\|\pmb{\beta}\|_2\sigma_j(A)$, and $\|\pmb{A}\|_*=\|\pmb{\beta}\|_2\|A\|_*$.
\end{proof}

\section{Estimate of the the cross terms for two panels}
\label{app:cross_terms}
We assume an overall quadratic dependence of $\Delta \tau(s)$ on $s$,
\begin{equation}
\Delta\tau(s)=a+bs+cs^2.
\end{equation}
Since the total aperture is decoupled into many, small, sub-apertures, it is reasonable to approximate $\Delta\tau(s)$ as being linear in each panel, with
\begin{equation}
\Delta\tau_\ell(s)=\Delta\tau_\ell(s+\ell\theta)\approx a_\ell+b_\ell s ,\quad a_\ell=a+b\ell\theta +c\ell^2\theta^2, b_\ell=b+2c\ell \theta.
\end{equation}
The inner product of two columns is
\begin{equation}
\begin{split} \label{Dstat}
&\sum\limits_s e^{-\frac{B^2}{2}(t_j-\Delta\tau_\ell(s))^2}e^{i\omega \Delta \tau_\ell(s)}e^{-\frac{B^2}{2}(t_k-\Delta\tau_{\ell'}(s))^2}e^{-i\omega \Delta \tau_{\ell'}(s)}\\
=&e^{-\frac{B^2}{2}((t_j-a_\ell)^2+(t_k-a_{\ell'})^2)}e^{i\omega(a_\ell-a_{\ell'})}
\sum\limits_s e^{-\frac{B^2}{2}(-2(t_j-a_\ell )b_\ell s+b_\ell^2 s^2)}e^{i\omega b_\ell s}e^{-\frac{B^2}{2}(-2(t_k-a_{\ell'})b_{\ell'} s+b_{\ell'}^2 s^2)}e^{-i\omega b_{\ell'} s}\\
=&e^{-\frac{B^2}{2}((t_j-a_\ell)^2+(t_k-a_{\ell'})^2)}e^{i\omega(a_\ell-a_{\ell'})}
\sum\limits_s e^{-\frac{B^2}{2}((-2(t_j-a_{\ell})b_\ell -2(t_k-a_{\ell'})b_{\ell'}+i\frac{2\omega}{B^2}(b_\ell-b_{\ell'}))s+(b_\ell^2+b_{\ell'}^2)s^2)}\\
\approx& e^{-\frac{B^2}{2}((t_j-a_\ell)^2+(t_k-a_{\ell'})^2)}e^{i\omega(a_\ell-a_{\ell'})}\frac{1}{\Delta s}\int\limits_{s^-}^{s^+}e^{-\frac{B^2}{2}\phi(s)}ds\\
&\phi(s)\equiv (-2(t_j-a_{\ell})b_\ell -2(t_k-a_{\ell'})b_{\ell'}+i\frac{2\omega}{B^2}(b_\ell-b_{\ell'}))s+(b_\ell^2+b_{\ell'}^2)s^2
\end{split}
\end{equation}
For a quadratic form $\phi(s)=\alpha s^2+\beta s+\gamma$ stationary phase value is at $s^*=-\frac{\beta}{2\alpha}$, and $\phi(s^*)=\gamma-\frac{\beta^2}{4\alpha}$, which implies 
\begin{equation}
\begin{split}
s^*=\frac{(t_j-a_\ell)b_\ell+(t_k-a_{\ell'})b_{\ell'}}{b_\ell^2+b_{\ell'}^2}+i\frac{\omega(b_\ell-b_{\ell'})}{B^2(b_\ell^2+b_{\ell'}^2)},\\
\phi(s^*)=-\frac{((t_j-a_\ell)b_\ell+(t_k-a_{\ell'})b_{\ell'}-i\frac{\omega}{B^2}(b_\ell-b_{\ell'}))^2}{b_\ell^2+b_{\ell'}^2}.
\end{split}
\end{equation}
Plugging into \eqref{Dstat},  
we get for the real part of the argument of the exponent
\begin{equation}
\begin{split}
&-\frac{B^2}{2}\left[(t_j-a_\ell)^2+(t_k-a_{\ell'})^2-\frac{(t_j-a_\ell)^2 b_\ell^2-2(t_j-a_\ell) (t_k-a_{\ell'}) b_\ell b_{\ell'}+(t_k-a_{\ell'})^2 b_{\ell'}^2}{b_\ell^2+b_{\ell'}^2}+\frac{\omega^2(b_\ell-b_{\ell'})^2}{B^4(b_\ell^2+b_{\ell'}^2)}+i\cdots\right]\\
=&-\frac{B^2}{2}\frac{(b_{\ell'}(t_j-a_\ell)-b_\ell (t_k-a_{\ell'}))^2}{b_\ell^2+b_{\ell'}^2}-\frac{\omega^2}{2B^2}\frac{(b_\ell-b_{\ell'})^2}{b_\ell^2+b_{\ell'}^2}
\end{split}
\end{equation}
We see there is a suppression term, independent of the specific column indices, proportional to the difference in the slopes. Since $\omega\gg B$, a small difference is enough to create a large suppression- effectively making traces with different slopes orthogonal. We also see that a smaller overlap (larger $\theta$, and more distance between the panels improves the results).

\bibliographystyle{abbrv} \bibliography{SBIR}

\begin{thebibliography}{10}

\bibitem{appellof1981strategies}
C.~J. Appellof and E.~R. Davidson.
\newblock Strategies for analyzing data from video fluorometric monitoring of
  liquid chromatographic effluents.
\newblock {\em Analytical Chemistry}, 53(13):2053--2056, 1981.

\bibitem{barbarossa1998ambiguityautofocus}
S.~Barbarossa and A.~Scaglione.
\newblock {Autofocusing of SAR images based on the product of high-order
  ambiguity function}.
\newblock {\em IEE Proc.-Radar, Sonar Navig.}, 145(5):269--273, 1998.

\bibitem{borcea2013motion}
L.~Borcea, T.~Callaghan, and G.~Papanicolaou.
\newblock Motion estimation and imaging of complex scenes with synthetic
  aperture radar.
\newblock {\em Inverse Problems}, 29(5):054011, 2013.

\bibitem{borcea2013synthetic}
L.~Borcea, T.~Callaghan, and G.~Papanicolaou.
\newblock Synthetic aperture radar imaging and motion estimation via robust
  principal component analysis.
\newblock {\em SIAM Journal on Imaging Sciences}, 6(3):1445--1476, 2013.

\bibitem{braman2010third}
K.~Braman.
\newblock Third-order tensors as linear operators on a space of matrices.
\newblock {\em Linear Algebra and its Applications}, 433(7):1241--1253, 2010.

\bibitem{candesRPCA}
E.~J. Cand\`{e}s, X.~Li, Y.~Ma, and J.~Wright.
\newblock {Robust Principal Component Analysis?}
\newblock {\em Journal of ACM}, 58(1):1--37, 2009.

\bibitem{Willsky2014}
M.~Cetin, I.~Stojanovic, O.~Onhon, K.~Varshney, S.~Samadi, W.~C. Karl, and
  A.~S. Willsky.
\newblock Sparsity-driven synthetic aperture radar imaging: Reconstruction,
  autofocusing, moving targets, and compressed sensing.
\newblock {\em IEEE Signal Processing Magazine}, 31(4):27--40, July 2014.

\bibitem{cichocki2015tensor}
A.~Cichocki, D.~Mandic, L.~De~Lathauwer, G.~Zhou, Q.~Zhao, C.~Caiafa, and H.~A.
  Phan.
\newblock Tensor decompositions for signal processing applications: From
  two-way to multiway component analysis.
\newblock {\em IEEE Signal Processing Magazine}, 32(2):145--163, 2015.

\bibitem{de2000multilinear}
L.~De~Lathauwer, B.~De~Moor, and J.~Vandewalle.
\newblock A multilinear singular value decomposition.
\newblock {\em SIAM journal on Matrix Analysis and Applications},
  21(4):1253--1278, 2000.

\bibitem{ender1996}
J.~Ender.
\newblock Detection and estimation of moving target signals by multi-channel
  {SAR}.
\newblock {\em AEU International Journal of Electronic Communication},
  50(2):150--156, 1996.

\bibitem{fienup}
J.~R. Fienup.
\newblock {Detecting Moving Targets in SAR Imagery by Focusing}.
\newblock {\em IEEE Transactions on Aerospace and Electronic Systems},
  37(3):794--809, 2001.

\bibitem{friedland2018nuclear}
S.~Friedland and L.-H. Lim.
\newblock Nuclear norm of higher-order tensors.
\newblock {\em Mathematics of Computation}, 87(311):1255--1281, 2018.

\bibitem{golub2012matrix}
G.~H. Golub and C.~F. Van~Loan.
\newblock {\em Matrix computations}, volume~3.
\newblock JHU Press, 2012.

\bibitem{guo2016video}
H.~Guo and N.~Vaswani.
\newblock Video denoising via online sparse and low-rank matrix decomposition.
\newblock In {\em 2016 IEEE Statistical Signal Processing Workshop (SSP)},
  pages 1--5. IEEE, 2016.

\bibitem{hillar2013most}
C.~J. Hillar and L.-H. Lim.
\newblock Most tensor problems are {NP}-hard.
\newblock {\em Journal of the ACM (JACM)}, 60(6):45, 2013.

\bibitem{hore2016tensor}
V.~Hore, A.~Vi{\~n}uela, A.~Buil, J.~Knight, M.~I. McCarthy, K.~Small, and
  J.~Marchini.
\newblock Tensor decomposition for multiple-tissue gene expression experiments.
\newblock {\em Nature genetics}, 48(9):1094, 2016.

\bibitem{jondeau2018moment}
E.~Jondeau, E.~Jurczenko, and M.~Rockinger.
\newblock Moment component analysis: An illustration with international stock
  markets.
\newblock {\em Journal of Business \& Economic Statistics}, 36(4):576--598,
  2018.

\bibitem{kilmer2013third}
M.~E. Kilmer, K.~Braman, N.~Hao, and R.~C. Hoover.
\newblock Third-order tensors as operators on matrices: A theoretical and
  computational framework with applications in imaging.
\newblock {\em SIAM Journal on Matrix Analysis and Applications},
  34(1):148--172, 2013.

\bibitem{kilmer2011factorization}
M.~E. Kilmer and C.~D. Martin.
\newblock Factorization strategies for third-order tensors.
\newblock {\em Linear Algebra and its Applications}, 435(3):641--658, 2011.

\bibitem{kolda2009tensor}
T.~G. Kolda and B.~W. Bader.
\newblock Tensor decompositions and applications.
\newblock {\em SIAM review}, 51(3):455--500, 2009.

\bibitem{leib2018RPCA}
M.~{Leibovich}, G.~{Papanicolaou}, and C.~{Tsogka}.
\newblock Low rank plus sparse decomposition of synthetic aperture radar data
  for target imaging.
\newblock {\em IEEE Transactions on Computational Imaging}, 2019.

\bibitem{li2007}
G.~Li, Y.-N. Peng, and X.-G. Xia.
\newblock {Moving target location and imaging using dual-speed velocity SAR}.
\newblock {\em IET Radar Sonar Navig.}, 1(2):158--163, 2007.

\bibitem{li2016bounds}
Z.~Li.
\newblock Bounds on the spectral norm and the nuclear norm of a tensor based on
  tensor partitions.
\newblock {\em SIAM Journal on Matrix Analysis and Applications},
  37(4):1440--1452, 2016.

\bibitem{lin2010augmented}
Z.~{Lin}, M.~{Chen}, and Y.~{Ma}.
\newblock {The Augmented Lagrange Multiplier Method for Exact Recovery of
  Corrupted Low-Rank Matrices}.
\newblock {\em ArXiv e-prints}, Sept. 2010.

\bibitem{lu2016tensor}
C.~Lu, J.~Feng, Y.~Chen, W.~Liu, Z.~Lin, and S.~Yan.
\newblock Tensor robust principal component analysis: Exact recovery of
  corrupted low-rank tensors via convex optimization.
\newblock In {\em Proceedings of the IEEE Conference on Computer Vision and
  Pattern Recognition}, pages 5249--5257, 2016.

\bibitem{cetin2019joint}
M.~Moradikia, S.~Samadi, and M.~Cetin.
\newblock Joint sar imaging and multi-feature decomposition from 2-d
  under-sampled data via low-rankness plus sparsity priors.
\newblock {\em IEEE Transactions on Computational Imaging}, 5(1):1--16, 2019.

\bibitem{moreira2013tutorial}
A.~Moreira, P.~Prats-Iraola, M.~Younis, G.~Krieger, I.~Hajnsek, and K.~P.
  Papathanassiou.
\newblock A tutorial on synthetic aperture radar.
\newblock {\em IEEE Geoscience and remote sensing magazine}, 1(1):6--43, 2013.

\bibitem{muehe2000displaced}
C.~E. Muehe and M.~Labitt.
\newblock Displaced-phase-center antenna technique.
\newblock {\em Lincoln Laboratory Journal}, 12(2):281--296, 2000.

\bibitem{shashua2005non}
A.~Shashua and T.~Hazan.
\newblock Non-negative tensor factorization with applications to statistics and
  computer vision.
\newblock In {\em Proceedings of the 22nd international conference on Machine
  learning}, pages 792--799. ACM, 2005.

\bibitem{stimson}
G.~W. Stimson.
\newblock {\em Introduction to airborne radar 2nd edition}.
\newblock Scitech Publishing, Inc, 1998.

\bibitem{tucker1966some}
L.~R. Tucker.
\newblock Some mathematical notes on three-mode factor analysis.
\newblock {\em Psychometrika}, 31(3):279--311, 1966.

\bibitem{wang2006}
G.~Wang, X.~Xia, and V.~Chen.
\newblock {Dual-Speed SAR Imaging of Moving Targets}.
\newblock {\em IEEE Transactions on Aerospace and Electronic Systems},
  42(1):368--379, 2006.

\bibitem{waters2011sparcs}
A.~E. Waters, A.~C. Sankaranarayanan, and R.~Baraniuk.
\newblock Sparcs: Recovering low-rank and sparse matrices from compressive
  measurements.
\newblock In {\em Advances in neural information processing systems}, pages
  1089--1097, 2011.

\bibitem{zhang2014novel}
Z.~Zhang, G.~Ely, S.~Aeron, N.~Hao, and M.~Kilmer.
\newblock Novel methods for multilinear data completion and de-noising based on
  tensor-svd.
\newblock In {\em Proceedings of the IEEE conference on computer vision and
  pattern recognition}, pages 3842--3849, 2014.

\bibitem{Zhou2010StablePC}
Z.~Zhou, X.~Li, J.~N. Wright, E.~J. Cand{\`e}s, and Y.~Ma.
\newblock Stable principal component pursuit.
\newblock {\em 2010 IEEE International Symposium on Information Theory}, pages
  1518--1522, 2010.

\end{thebibliography}
 	
\end{document}